\documentclass{article}

\usepackage{amsmath, amsthm, amssymb, dsfont, mathtools}
\usepackage{graphicx}
\usepackage{verbatim}
\usepackage{natbib}
\usepackage{caption}
\usepackage{subcaption}
\usepackage{fancyvrb}
\usepackage[inline]{enumitem}
\usepackage{relsize}
\usepackage{hyperref}
\usepackage[margin=1.4in]{geometry}
\hypersetup{colorlinks,citecolor=blue,urlcolor=blue,linkcolor=blue}
\usepackage{diagbox}
\usepackage{nikos_tex}
\usepackage{float}
\usepackage{adjustbox}
\usepackage{multirow}
\usepackage{cleveref}

\usepackage[utf8]{inputenc}
\usepackage[english]{babel}

\usepackage{array, booktabs, longtable}
\usepackage[table,x11names]{xcolor}

\usepackage{tikz}
\usetikzlibrary{arrows.meta, positioning, matrix}

\usepackage[skins,theorems]{tcolorbox}
\tcbset{highlight math style={enhanced,
  colframe=red,colback=white,arc=0pt,boxrule=1pt}}

\newcommand{\SRatio}{\hat{\lambda}}
\newcommand{\VRatio}{\lambda}
\newcommand{\ZMean}{\hat{\mu}_{iA}}
\newcommand{\YMean}{\hat{\mu}_{iB}}

\newcommand{\Tbf}{T^{\text{BF}}}
\newcommand{\TbfNull}{\mathring{T}^{\text{BF}}}

\newcommand{\Twbf}{T^{\text{WBF}}}

\newcommand{\Tpool}{T^{\text{pool}}}

\newcommand{\Ppool}{P^{\text{pool}}}

\newcommand{\Tdm}{T^{\text{dm}}}

\newcommand{\PVRFunc}{\mathrm{P}^{\text{VR}}}
\newcommand{\PVR}{P^{\text{VR}}}
\newcommand{\PDVFunc}{\mathrm{P}^{\text{DV}}}
\newcommand{\PDV}{P^{\text{DV}}}

\newcommand{\Plimma}{P^{\mathrm{limma}}}

\newcommand{\TwbfNull}{\mathring{T}^{\text{WBF}}}

\newcommand{\PWVRFunc}{\mathrm{P}^{\text{WVR}}}
\newcommand{\PWVR}{P^{\text{WVR}}}
\newcommand{\PWDVFunc}{\mathrm{P}^{\text{WDV}}}
\newcommand{\PWDV}{P^{\text{WDV}}}

\newcommand{\Pwbf}{P^{\text{WBF}}}

\graphicspath{{./figures/}}

\theoremstyle{definition}
\newtheorem{prop}{Proposition}

\newtheorem{lemm}[prop]{Lemma}
\newtheorem{theo}[prop]{Theorem}
\newtheorem{rema}[prop]{Remark}

\usepackage{algorithm2e}
\RestyleAlgo{boxruled}
\DontPrintSemicolon
\SetAlFnt{\normalsize}
\SetKwInOut{Input}{Input}

\usepackage[section]{placeins}

\date{Draft manuscript: October 2025}

\title{Empirical partially Bayes two sample testing}

\author{
\begin{tabular}{lll}
 Wanyi Ling\thanks{These authors contributed equally to this work.} & Wufang Hong\footnotemark[1] & Nikolaos Ignatiadis \\
 \texttt{wanyiling@uchicago.edu} & \texttt{hwufang@uchicago.edu} & \texttt{ignat@uchicago.edu}
\end{tabular}
\vspace{1em}
}

\begin{document}

\maketitle

\begin{abstract}
A common task in high-throughput biology is to test for differences in means between two samples across thousands of features (e.g., genes or proteins), often with only a handful of replicates per sample. Moderated t-tests handle this problem by assuming normality and equal variances, and by applying the empirical partially Bayes principle: a prior is posited and estimated for the nuisance parameters (variances) but not for the primary parameters (means). This approach has been highly successful in genomics, yet the equal variance assumption is often violated in practice. Meanwhile, Welch’s unequal variance t-test with few replicates suffers from inflated type-I error and low power. Taking inspiration from moderated t-tests, we extend the empirical partially Bayes paradigm to two-sample testing with unequal variances. We develop two procedures: one that models the ratio of the two sample-specific variances and another that models the two variances jointly, with prior distributions estimated by nonparametric maximum likelihood. Our empirical partially Bayes methods yield p-values that are asymptotically uniform as the number of features grows while the number of replicates remains fixed, ensuring asymptotic type-I error control. Simulations and applications to genomic data demonstrate substantial gains in power.\\

\end{abstract}

\section{Introduction} 
\label{sec:introduction}

We seek to conduct thousands of parallel tests for equality in means between two normal samples with small and unbalanced degrees of freedom and unknown, possibly unequal variances---a large-scale, parallel form of the classical Behrens-Fisher problem \citep{Behrens1929test, fisher1935fiducial}.  Our motivation to revisit this classical problem stems from empirical practice in high-throughput biology wherein it is usually assumed that gene-specific residual variances are identical across conditions. Indeed, recent work~\citep{chen2018umicount, you2023modeling, chatterjee2024group} has partially attributed inflation of the false discovery rate in single cell and pseudo-bulk RNA-Seq studies to the aforementioned practice and has called for methods that account for variance heterogeneity.

Salient aspects of the problem are that we test for thousands of molecular features  (e.g., gene expression throughout this paper) simultaneously, while the number of biological replicates may be in the single digits, e.g., due to ethical considerations in animal studies, and by the exploratory nature of initial screening studies. We assume normality throughout, which is a common assumption in RNA-seq differential expression analysis after normalization and transformation \citep{law2014voom}, yet this problem remains challenging.
We illustrate this challenge through a simple simulation in which natural baseline attempts fail. In the simulation, we consider $n = 5000$ genes. For each gene, we have $K_A = 3$ normal observations from the first sample and $K_B \in \{3, 4, \dots, 9\}$ normal observations from the second sample with unequal sample-wise variances (we defer more details on the simulation to Section~\ref{sec: simulation}). We consider the following approaches: compute p-values using the equal variance t-test, the Welch unequal variance t-test with an asymptotic approximation of degrees of freedom \citep{Welch1938BF, Welch1947Approximation}, and the Behrens-Fisher test \citep{Behrens1929test, fisher1935fiducial}, then pass these p-values to the Benjamini-Hochberg (BH) procedure \citeyearpar{benjamini1995controlling} to control the false discovery rate (FDR) at level $\alpha=0.1$. Results are shown in Figure \ref{fig:unequal simulation 1} and demonstrate that the FDR is inflated for the equal variance t-test (as soon as $K_B \geq K_A+1$) and for Welch (for $K_B \geq  K_A + 2$), while Behrens-Fisher is powerless.

\begin{figure}
    \centering
    \includegraphics[width=0.9\textwidth]{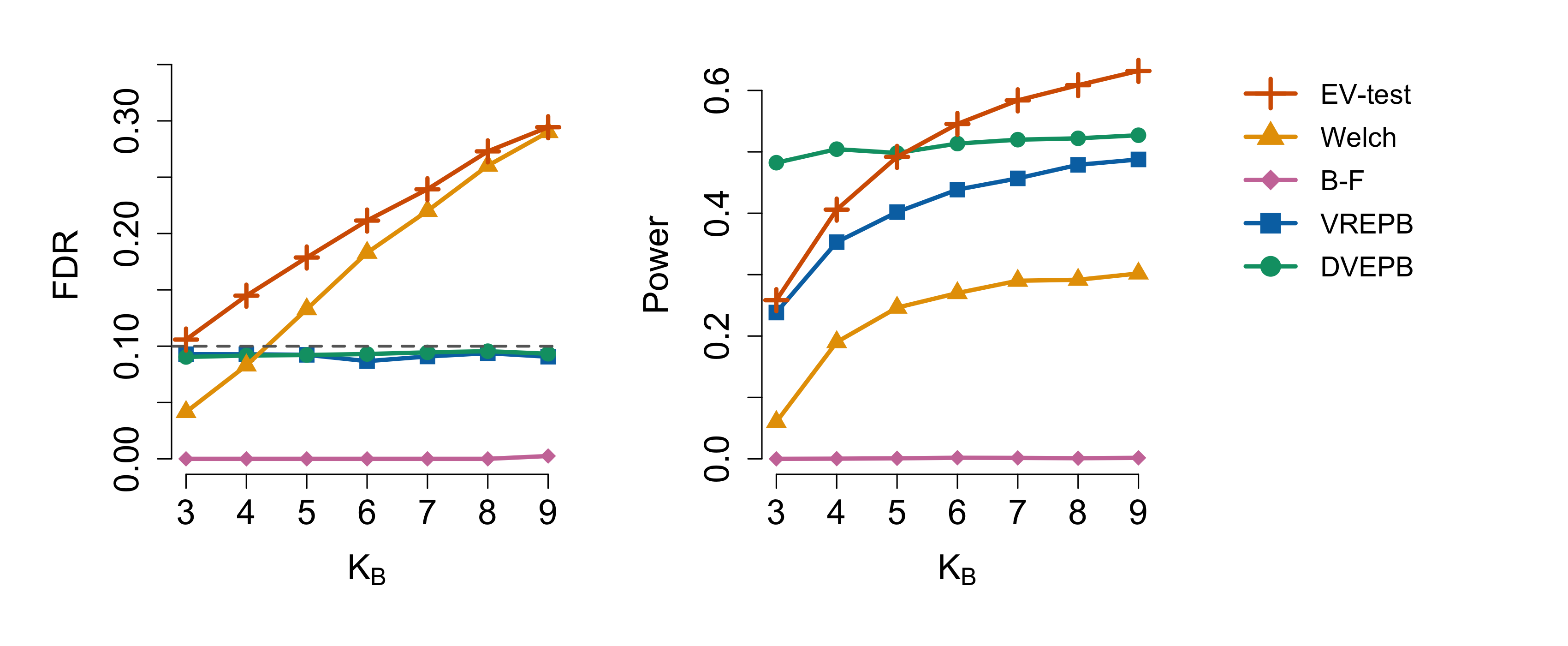}
   \caption{Comparison of five methods for testing equality of means with unequal variances in a multiple testing setting with 5000 simultaneous tests: equal variance t-test (EV-test), Welch approximation (Welch), Behrens-Fisher test (B-F), and our proposed VREPB and DVEPB methods. Each test compares 3 observations in the first sample against varying sample sizes $K_B \in \{3, \ldots, 9\}$ in the second sample. The left panel shows false discovery rate (FDR) after Benjamini-Hochberg correction at $\alpha=0.1$; the right panel shows power (expected proportion of discoveries among true alternatives). The EV-test and Welch methods exhibit severe FDR inflation, while Behrens-Fisher maintains FDR control but with negligible power. VREPB and DVEPB achieve both FDR control and substantial power gains.}
    \label{fig:unequal simulation 1}
\end{figure}

One explanation for the poor performance of these methods is that, by treating each test in isolation, they fail to correctly account for uncertainty in estimated nuisance parameters when degrees of freedom are small. The main thrust of our paper is that calibration to finite-sample nuisance parameter uncertainty becomes possible by sharing information across tests following a two-step empirical partially Bayesian blueprint: first, use empirical Bayes \citep{robbins1956empirical, efron2010largescale, stephens2016false} to learn the nuisance parameter distribution from the ensemble of tests, and second, compute partially Bayes~\citep{cox1975note, mccarthy2009testing} p-values based on the estimated distribution.

We instantiate this blueprint through two distinct methods: Variance-Ratio Empirical Partially Bayes (VREPB, Section~\ref{sec: Variance ratios as nuisance parameters}) and Dual-Variance Empirical Partially Bayes (DVEPB, Section \ref{sec: Groupwise variances as nuisance parameters}). VREPB follows~\citet{brown1965secondarily} in treating the ratio of sample-wise variances as the nuisance parameter, while DVEPB treats both variances as separate nuisance parameters. As previewed in Figure~\ref{fig:unequal simulation 1}, both methods successfully control FDR while achieving substantial power, with DVEPB typically outperforming VREPB.

The FDR control in the simulation is not accidental. Under the assumption that nuisance parameters (variance ratios for VREPB, variance pairs for DVEPB) are drawn from a common frequency distribution, we establish that VREPB p-values are asymptotically uniform under the null hypothesis (Theorem~\ref{thm:1D asymptotic uniformity}), as are DVEPB p-values (Theorem~\ref{thm:2D asymptotic uniformity}). Our asymptotic regime is non-standard compared to existing asymptotics for the same problem (see Section~\ref{sec:related_work} for a brief review of related work): we let the number of tests grow to infinity while keeping the number of observations per test fixed, a regime~\citet{hero2016foundational} call ``sample-starved'' large-scale inference. Moreover, we establish robustness to the distributional assumption on the nuisance parameters: Theorem~\ref{theo:vrepb_compound} shows that even when variance ratios are treated as fixed parameters (as in a standard frequentist analysis), VREPB p-values remain valid as asymptotic compound p-values~\citep{ignatiadis2025asymptotic}. We complement these theoretical guarantees with extensive simulations (Section~\ref{sec: simulation}) and empirical studies (Section~\ref{sec:applications}).

Methodologically,
our proposed approaches build directly on the empirical partially Bayes framework that underpins \texttt{limma}~\citep[$>50,000$ combined citations]{smyth2004linear, ritchie2015limma}. The \texttt{limma} package---the de facto standard for microarray analysis and widely used for RNA-seq studies---addresses the problem of unreliable variance estimates with small sample sizes by positing a distribution on nuisance parameters (variances), estimating it via empirical Bayes, then computing empirical partially Bayes p-values. This framework applies when variances are equal across samples, or more generally, when variance heterogeneity can be accurately modeled up to a single unknown scalar (see Section~\ref{subsec:epb_genomics} for more details about limma). While the approach has proven highly effective in practice, the formal study of its theoretical properties in the large-scale testing regime (where the number of tests grows while sample sizes remain fixed) was initiated only recently by~\citet{ignatiadis2025empirical}. We extend the empirical partially Bayes framework to handle arbitrary unequal variances by modeling the frequency distribution of either variance ratios (VREPB) or joint sample-wise variance pairs (DVEPB) with nonparametric priors, and establish formal asymptotic validity guarantees for these extensions.

Overall, in this paper, we revisit the classical Behrens–Fisher problem from an empirical partially Bayes perspective, leveraging ideas that proved successful for equal variance t-tests in genomics, and extending them to heteroscedastic two-sample testing.

\section{Statistical setting}
\label{sec:statiatical setting}
We have data for $n \in \mathbb N$ independent units (e.g., genes). For the $i$-th unit, we observe $K_A \in \mathbb N$ iid observations $Z_{i1}, \dotsc, Z_{iK_A} \sim \mathrm{N}(\mu_{iA}, \sigma_{iA}^2)$ from the first sample, and $K_B \in \mathbb N$ iid observations $Y_{i1}, \dotsc, Y_{iK_B} \sim \mathrm{N}(\mu_{iB},\sigma_{iB}^2)$ from the second sample, with the two samples independent of each other. All observations for the $i$-th unit may be collapsed to their complete and sufficient statistics: the sample average and sample variance of $Z_{i1},\ldots,Z_{iK_A}$ denoted as $\ZMean$ and $\hat{\sigma}_{iA}^2$, and their counterparts based on $Y_{i1},\ldots, Y_{iK_A}$ denoted as $\YMean$ and $\hat{\sigma}_{iB}^2$. The sufficient statistics are distributed as follows,
\begin{equation}
    \label{eq:ss_distribution}
        (\ZMean, \hat{\sigma}_{iA}^2, \YMean, \hat{\sigma}_{iB}^2)\overset{\text{ind}}{\sim}  \mathrm{N}\p{\mu_{iA}, \frac{\sigma_{iA}^2}{K_A}} \otimes \frac{\sigma_{iA}^2}{\nu_A} \chi_{\nu_A}^2 \otimes  \mathrm{N}\p{\mu_{iB}, \frac{\sigma_{iB}^2}{K_B}} \otimes \frac{\sigma_{iB}^2}{\nu_B} \chi_{\nu_B}^2,
\end{equation}
where $\nu_A = K_A - 1, \nu_B = K_B - 1$, $\chi^2_{\nu_A}$ (resp. $\chi^2_{\nu_B}$) denotes the chi-squared distribution with $\nu_A$ (resp. $\nu_B$) degrees of freedom and we use $\otimes$ to denote product measures. The unknown parameters are $\mu_{iA}, \mu_{iB} \in \RR$ and $\sigma_{iA}^2, \sigma_{iB}^2 >0$. Our goal is to test for equality of means for each $i$, that is, to test the null hypotheses $H_i: \mu_{iA} = \mu_{iB}$. We write $\mathcal{H}_0 := \{i \in \{1, \dots, n\}: \mu_{iA} = \mu_{iB}\}$ for the indices of  null units. We think of $\mu_{iA}, \mu_{iB}$ as the primary parameters, while $\sigma_{iA}^2,\sigma_{iB}^2$ are nuisance parameters.

The scope of our methods is broader than may be immediately apparent: they apply whenever data can be summarized in the form of~\eqref{eq:ss_distribution}, which encompasses settings beyond the sample-wise iid normal observations described at the start of this section. For instance, through the use of precision weights~\citep{law2014voom}, our methods become applicable to settings such as the single-cell RNA-Seq study analyzed in Section~\ref{sec: Single-cell RNA-seq Pseudo-bulk Data}; see Supplement~\ref{sec: weighted model setup} for the extension to precision-weighted data.

We note that $\hat{\sigma}_{iA}^2$ and $\hat{\sigma}_{iB}^2$ are ancillary for the primary parameters, that is, their distribution does not depend on $\mu_{iA},\mu_{iB}$, but only on the nuisance parameters $\sigma_{iA}^2$ and $\sigma_{iB}^2$. 
It will be convenient to also define the variance ratio and the sample variance ratio \smash{$\SRatio_i$}:
\begin{equation}
\label{eq:lambda}
\VRatio_i := \frac{\sigma_{iA}^2}{\sigma_{iB}^2},\qquad \SRatio_i := \frac{\hat{\sigma}_{iA}^2}{\hat{\sigma}_{iB}^2}.
\end{equation}
Note that $\lambda_i$ is also an unknown (nuisance) parameter and $\hat{\lambda}_i$ is a further statistic that can be computed based on the sufficient statistics. The distribution of $\hat{\lambda}_i$ only depends on $\lambda_i$,
\begin{equation}
\label{eq:tau_given_lambda}
    \SRatio_i \; \cond \; \VRatio_i \overset{\text{ind}}{\sim} \VRatio_i F_{\nu_A, \nu_B},
\end{equation}
where $F_{\nu_A, \nu_B}$ is the F-distribution with $\nu_A, \nu_B$ degrees of freedom. In particular,  $\hat{\lambda}_i$ is also ancillary for the primary parameters $\mu_{iA}, \mu_{iB}$.

Later, when we consider asymptotics, we will allow the number of tests (units) to grow by taking $n \to \infty$, while keeping $K_A$ and $K_B$ fixed throughout. When we use the term ``asymptotically,'' unless further clarified, we refer to the former regime.
This distinction is important: from the alternate asymptotic perspective with $K_A, K_B \to \infty$, it is uncontroversial (see e.g.,~\citealp[Example 13.5.4]{lehmann2005testing}, for a large sample optimality result) to proceed with the unequal variance Behrens-Fisher t-statistic,
\begin{equation}
\label{eq:BF}
\Tbf_i := \frac{\ZMean -\YMean }{\cb{\p{\hat{\sigma}_{iA}^2/K_A} + \p{\hat{\sigma}_{iB}^2/K_B}}^{1/2}},\;\;\;\; \TbfNull_i := \frac{(\ZMean -\YMean) - (\mu_{iA}-\mu_{iB})}{\cb{\p{\hat{\sigma}_{iA}^2/K_A} + \p{\hat{\sigma}_{iB}^2/K_B}}^{1/2}},
\end{equation}
where $\TbfNull_i$ is the null centered counterpart of $\Tbf_i$. As $K_A,K_B \to \infty$, $\TbfNull_i$ is approximately standard normal, and so is $\Tbf_i$ under the null hypothesis.

On occasion, we drop the subscript $i$ and e.g., may write $\TbfNull$ or $\SRatio$ to refer to $\TbfNull_i$ or $\SRatio_i$.

\section{Known variance ratio and empirical partially Bayes}
\label{sec:kmown_variance_ratio}

\subsection{Background}
It is common in practice to assume that the sample-wise variances are equal, that is, $\sigma_{iA}^2 = \sigma_{iB}^2$. This is equivalent to assuming that the variance ratio satisfies $\lambda_i = 1$. In this case, the equal variance t-test yields a p-value familiar from any introductory textbook. The driving force behind this construction, however, is not that the variances are equal, but rather that the variance ratio $\lambda_i$ is known~\citep{sprott1993difference, schechtman2007twosample}.  In what follows, we rederive the solution for known $\lambda_i$ (that does not need to be equal to $1$). Our motivation is two-fold: first, the construction with known $\lambda_i$ is an important building block for the VREPB p-values that we develop in Section~\ref{sec: Variance ratios as nuisance parameters}. Second, it will enable us to introduce the empirical partially Bayes argument, as currently ubiquitously used in high-throughput biology.

\subsection{The t-test when the variance ratio is known}
\label{subsec:ttest_known_variance_ratio}

Suppose $\lambda_i$ is known. Then only $(\mu_{iA}, \mu_{iB}, \sigma_{iA}^2)$ are unknown parameters in~\eqref{eq:ss_distribution}. The complete and sufficient statistics for the $i$-th test are the sample means $\hat{\mu}_{iA}$, $\hat{\mu}_{iB}$, and the pooled variance estimator
\begin{equation}
\hat{\sigma}_{iA, \text{pool}}^2 \equiv \hat{\sigma}_{iA, \text{pool}}^2(\lambda_i) := \frac{\nu_A \hat{\sigma}_{iA}^2 + \nu_B \hat{\sigma}_{iB}^2 \VRatio_i}{\nu_A + \nu_B}.
\label{eq:known_variance_sufficient}
\end{equation}
Then, since $\sigma_{iB}^2 = \sigma_{iA}^2/\lambda_i$, the distribution of the complete sufficient statistics is the following:
\begin{equation}
\label{eq:known_vr_suffstats}
(\ZMean, \YMean, \hat{\sigma}_{iA, \text{pool}}^2) \overset{\text{ind}}{\sim}  \mathrm{N}\p{\mu_{iA}, \frac{\sigma_{iA}^2}{K_A}} \otimes \mathrm{N}\p{\mu_{iB}, \frac{\sigma_{iA}^2}{K_B \lambda_i}} \otimes \frac{\sigma_{iA}^2}{\nu_A + \nu_B} \chi_{\nu_A+\nu_B}^2.
\end{equation}
Since for $i \in \mathcal{H}_0$, $\ZMean - \YMean \sim \mathrm{N}(0,\;  ((1/K_A) + 1/(K_B \lambda_i))\sigma_{iA}^2 )$ and $\ZMean - \YMean$  is independent of $\hat{\sigma}_{iA, \text{pool}}^2$, it follows that we can test $H_i: \mu_{iA} = \mu_{iB}$ via the pooled t-statistic
\begin{equation}
\label{eq: T_lambda}
\Tpool_i \equiv \Tpool_i(\lambda_i):=\frac{\ZMean - \YMean}{\sqrt{\frac{1}{K_A} + \frac{1}{K_B \VRatio_i}}\hat{\sigma}_{iA, \text{pool}}(\lambda_i)} \sim t_{\nu_A+ \nu_B} \;\text{ for } \; i \in \mathcal{H}_0,
\end{equation}
where $t_{\nu_A+ \nu_B}$ is the t-distribution with $\nu_A + \nu_B$ degrees of freedom. From here we can compute a p-value via,
$$
\Ppool_i \equiv \Ppool_i(\lambda_i) := 2F_{t, \nu_A + \nu_B}( - |\Tpool_i(\lambda_i)|), 
$$
where $F_{t, \nu_A + \nu_B}(\cdot)$ is the cumulative distribution function of a t-variate with $\nu_A + \nu_B$ degrees of freedom.
We can also check that $\Tpool_i$ is independent of $\SRatio_i$ (by ancillarity of $\SRatio_i$ when $\lambda_i$ is known and invoking Basu's theorem~\citep{basu1955statistics}). The implication is that $\Ppool_i$ is in fact a conditionally valid p-value,
$$
\PP[\lambda_i]{ \Ppool_i(\lambda_i) \leq \alpha \; \cond \;  \SRatio_i} = \alpha \; \text{ for all }\; \alpha \in (0,1),
$$
and this property is desirable by the conditionality principle~\citep{cox1974theoretical} (since $\SRatio_i$ is ancillary). The independence of $\SRatio_i$ and $\Tpool_i$ also has implications for the Behrens-Fisher statistic in~\eqref{eq:BF}. It holds that
$$
\Tbf_i = \Tpool_i(\lambda_i) \cdot  \phi(\VRatio_i, \SRatio_i, K_A, K_B),
$$
for a function $\phi$ with an explicit form (cf.~\citet{sprott1993difference}, also see Supplement \ref{appendix: supplementary formulas}). Thus $\Ppool_i$ may also be interpreted as a p-value based on $\Tbf_i$ conditional on $\SRatio_i$,
\begin{equation}
\Ppool_i = \PVRFunc(\Tbf_i, \SRatio_i; \lambda_i),\;\;\;\PVRFunc(t, l; \lambda) := \PP[\lambda]{ \abs{\TbfNull} \geq \abs{t} \; \cond \; \hat{\lambda}=l}.
\label{eq:PVR_known_nuisance}
\end{equation}
We note that the tail area function $\PVRFunc$ will be important in our development of the VREPB method in Section~\ref{sec:vrepb}. We record important tail area functions for our development in Table~\ref{tab:stats_results}.

\begin{table}
\centering
\caption{Most important tail area functions and p-values defined in this paper.}
\begin{tabular}{lccc}
 & Oracle tail area & Oracle tail area & Empirical partially \\ 
              & (known nuisance) & (partially Bayes) & Bayes p-value \vspace{5pt} \\
Variance Ratio (VR) & $\PVRFunc(t, l; \lambda)$~\eqref{eq:PVR_known_nuisance} &   $\PVRFunc(t, l; G)$~\eqref{eq:oracle_1D_p_value} & $\PVR_i$~\eqref{eq:epb_pvalues}  \vspace{2.5pt}\\ 
Dual Variance (DV) &  $\PDVFunc(t, s_A^2, s_B^2; \sigma_A^2, \sigma_B^2)$ \eqref{eq:PDV_known_nuisance} &   $\PDVFunc(t, s_A^2, s_B^2; H)$ \eqref{eq:oracle_2D_p_value} & $\PDV_i$ \eqref{eq:DV_epb_pvalues}
\end{tabular}
\label{tab:stats_results}
\end{table}

So far we have considered the case with known $\lambda_i$.
Analysis becomes more complicated when the variance ratio $\lambda_i$ is unknown.
As one illustration, suppose we seek to conduct a level $\alpha$ test for $H_i: \mu_{iA}=\mu_{iB}$ using a rejection region of the form
\begin{equation}
\label{eq:natural_tbf}
\abs{\Tbf_i} \geq g(\SRatio),
\end{equation}
for some function $g(\cdot)$. Such a form may be justified using invariance considerations~\citep[Chapter 6.6]{lehmann2005testing}. Then, a classical line of work, pioneered and summarized by~\citet[Chapter VIII]{linnik1968statistical}, shows that unless $g(\cdot)$ is pathological (e.g., discontinuous in the sufficient statistics), then the probability of the event in~\eqref{eq:natural_tbf} cannot be equal to $\alpha$ under the null for all values of $\lambda_i >0$.

\citet{barnard1984comparing} writes that ``perhaps the best solution is to accept that in this problem the unknown $\rho^2$ [$\lambda$ in our notation] is a `confounded nuisance parameter,' which cannot be fully eliminated.'' As such,  \citet{barnard1984comparing} proposes to compute p-values at several plausible values of $\VRatio_i$ as a form of sensitivity analysis. Building upon \citet{barnard1984comparing}'s work, \citet{sprott1993difference} propose that the plausible values of $\VRatio_i$ be restricted to a confidence interval for $\lambda_i$, which can be constructed solely based on $\hat{\lambda}_i$ alongside the distributional result in~\eqref{eq:tau_given_lambda}. \citet{berger1994values}\footnote{In the context of their method,~\citet{berger1994values} discuss whether p-values based on $\Tpool_i(\lambda_i)$ or $\Tbf_i$  are preferable (for a fixed value of $\VRatio_i$). Their discussion is not in contradiction with~\eqref{eq:PVR_known_nuisance}. One can interpret their discussion as comparing p-values based on $\Tbf_i$ computed conditionally on \smash{$\SRatio_i$} or unconditionally.
} propose the same profiling approach over a pilot confidence interval for $\VRatio_i$, followed by reporting the largest p-value among all candidates. Several other approaches to the Behrens-Fisher problem may be interpreted as carefully handling $\VRatio_i$; we will review the most relevant approaches to our development below.

\subsection{Empirical partially Bayes and its ubiquitous use in genomics}
\label{subsec:epb_genomics}

Even when $\lambda_i$ is known (say $\lambda_i=1$), the t-test solution just described is not fully satisfactory in applications to high throughput biology.
Small sample sizes (small $K_A, K_B$) yield unreliable variance estimates \smash{$\hat{\sigma}_{iA, \text{pool}}^2$} in~\eqref{eq:known_variance_sufficient}. Coupled with the t-distribution's heavy tails (which account for fluctuations in \smash{$\hat{\sigma}_{iA, \text{pool}}^2$}) and the need for multiple testing correction, this ``textbook'' approach often lacks power. The key innovation underlying limma~\citep{smyth2004linear}---a standard software package in genomics as explained in Section~\ref{sec:introduction}---is to pursue an empirical partially Bayes approach in which the nuisance parameters (here, $\sigma_{iA}^2$) are distributed according to a frequency distribution $H_A$ supported on $ \mathbb{R}_+$,
\begin{equation}
\label{eq:equal_var_partially_Bayes}
\sigma_{iA}^2 \simiid H_A.
\end{equation}
The first step is to estimate $H_A$ using all of \smash{$\hat{\sigma}_{1A,\text{pool}}^2,\ldots,\hat{\sigma}_{nA,\text{pool}}^2$} with \smash{$\hat{H}_A$}. Next, the goal is to shrink estimates of $\sigma_{iA}^2$ toward \smash{$\hat{H}_A$}. The shrinkage of $\sigma_{iA}^2$ is targeted toward the testing task by directly incorporating the empirical Bayes step into the p-value computation,
\begin{equation*}
\label{eq:limma_pvalue}
\Plimma_i := \mathrm{P}^{\text{limma}}(\Tbf_i, \hat{\sigma}_{iA,\text{pool}}^2; \hat{H}_A),\;\; \mathrm{P}^{\text{limma}}(t, s^2; H_A) := \PP[H_A]{\abs{\TbfNull} \geq \abs{t}  \; \cond \; \hat{\sigma}_{A,\text{pool}}^2= s^2}.
\end{equation*}
The right-hand side is a null tail area function computed conditionally on \smash{$\hat{\sigma}_{iA,\text{pool}}^2$} and integrating over randomness in both the sufficient statistics in~\eqref{eq:known_vr_suffstats} and $\sigma_{iA}^2$ in~\eqref{eq:equal_var_partially_Bayes}. Our notation makes this explicit through the subscript $H_A$ in the expression for the probability $\PP[H_A]{\cdot}$.
Another way to write the tail area function is as,
\begin{equation}
\mathrm{P}^{\text{limma}}(t, s^2; H_A)  = 2\EE[H_A]{\Phi\p{ - \abs{t}\frac{\cb{\p{1/K_A} + \p{1/(\lambda K_B)}}^{1/2}s} {\cb{\p{1/K_A} + \p{1/(\lambda K_B)}}^{1/2}\sigma_A}} \; \cond \;  \hat{\sigma}_{A,\text{pool}}^2=s^2},
\label{eq:limma_pvalue_function}
\end{equation}
where $\Phi$ is the standard normal distribution function. The expression inside the expectation is the tail area function if we also knew $\sigma_{iA}^2$ in addition to $\lambda_i$. Since $\sigma_{iA}^2$ is unknown, we integrate over this quantity with respect to the posterior distribution of $\sigma_{iA}^2$ given \smash{$\hat{\sigma}_{iA,\text{pool}}^2$}. It is in this sense that limma is shrinking \smash{$\hat{\sigma}_{iA,\text{pool}}^2$} toward the estimated prior \smash{$\hat{H}_A$}.

The approach is termed partially Bayes because a prior is posited only on the nuisance parameter $\sigma_{iA}^2$ in~\eqref{eq:equal_var_partially_Bayes}, while the primary parameters $\mu_{iA}, \mu_{iB}$ are treated in a frequentist manner. It is empirical partially Bayes because the prior $H_A$ is estimated from the data using all variance estimates \smash{$\hat{\sigma}_{1A,\text{pool}}^2,\ldots,\hat{\sigma}_{nA,\text{pool}}^2$}, then applied via the plug-in principle with $\hat{H}_A$ in place of the true prior. While limma implements this using a parametric specification for $H_A$, recent work has extended the framework to nonparametric priors \citep{lu2016variance, ignatiadis2025empirical}.
This empirical partially Bayes framework is a standard practice in differential expression analysis, showing that the core principle we will apply to our problem---leveraging information across tests to handle nuisance parameter uncertainty in small-sample settings---is already in widespread use.

\section{Variance ratio empirical partially Bayes (VREPB)}
\label{sec:vrepb}
\label{sec: Variance ratios as nuisance parameters}

\subsection{Oracle partially Bayes}
\label{sec: Oracle partially Bayes p-values}

Our goal in this section is to propose an alternative to the conditional p-value in~\eqref{eq:PVR_known_nuisance} when $\lambda_i$ is not known. We will do so by 
pursuing the partially Bayes principle introduced in Section~\ref{subsec:epb_genomics}. Since the distribution of the p-value in~\eqref{eq:PVR_known_nuisance} depends only on $\VRatio_i$, we will impose a prior only on $\VRatio_i$; later, in Section~\ref{sec: Groupwise variances as nuisance parameters} we will assign a prior on the pair $(\sigma_{iA}^2, \sigma_{iB}^2)$:
\begin{equation}
    \label{eq:varation_dbn}
\VRatio_i \simiid G.
\end{equation}
According to~\eqref{eq:varation_dbn}, we treat $\VRatio_1,\ldots,\VRatio_n$ as exchangeable. The distribution $G$ can be interpreted as a device through which we share information about $\VRatio_i$ across units. If we knew $G$, then, in analogy to~\eqref{eq:PVR_known_nuisance}, we could compute a p-value via $\PVRFunc(\Tbf_i, \SRatio_i; G)$, where the tail-area is defined as:
\begin{equation}
    \label{eq:oracle_1D_p_value}
 \PVRFunc(t, l; G) := \PP[G]{ \abs{\TbfNull} \geq \abs{t} \; \cond \; \hat{\lambda}=l} = \EE[G]{  \PVRFunc(t, l; \lambda) \; \cond \; \SRatio = l}.
\end{equation}
Above, the subscript $G$ in the probability ($\PP[G]{\cdot})$ and expectation ($\EE[G]{\cdot}$) statements indicates that we are also integrating over randomness in~\eqref{eq:varation_dbn}. We overload notation for  $\PVRFunc$: $\PVRFunc(t,l;\lambda)$ with $\lambda >0$ as the last argument refers to the tail area at fixed value of $\lambda$, while $\PVRFunc(t,l;G)$ with a distribution $G$ supported on $\RR_+$ as the last argument refers to also integrating over~\eqref{eq:varation_dbn}. This overloading of notation is justified by the observation that if $G=\delta_{\lambda}$ is a Dirac mass at $\lambda$, then $\PVRFunc(t,l;\lambda) = \PVRFunc(t,l; \delta_{\lambda})$. 
The right-hand side expression in~\eqref{eq:oracle_1D_p_value} establishes the following interpretation: $\PVRFunc(t, l; G)$ is the expectation of $\PVRFunc(t, l; \lambda)$ with respect to the posterior distribution of $\lambda$ given $\SRatio = l$. Shrinkage of $\VRatio$ is targeted toward the testing task.

The construction in~\eqref{eq:oracle_1D_p_value} is presented in the prescient PhD thesis of~\citet{brown1965secondarily}, under the guidance and encouragement of John Tukey; also see~\citet{brown1967twomeans}. Brown introduces the secondarily Bayes approach---specifically the p-values in~\eqref{eq:oracle_1D_p_value}---to furnish an improved solution to the Behrens-Fisher problem. Along the way, \citet[Chapter 10]{brown1965secondarily} also develops the partially Bayes approach for the case where $\lambda_i=1$ is known (the setting of Section~\ref{subsec:epb_genomics}), yielding the oracle limma partially Bayes p-values in~\eqref{eq:limma_pvalue_function} that are widely used in genomics. We also note that p-value functions as in~\eqref{eq:limma_pvalue_function} and~\eqref{eq:oracle_1D_p_value} are independently developed and called conditional predictive p-values by~\citet{bayarri2000values}. Their motivation is different, namely to provide an alternative to posterior predictive p-values~\citep{meng1994posterior} with better calibration properties.

\begin{rema}[Partially Bayes and Behrens-Fisher]
\label{rema:pb_and_bf}
\cite{brown1965secondarily} establishes a connection of~\eqref{eq:oracle_1D_p_value} with the (fiducial) Behrens-Fisher solution: if $G$ is chosen as the improper prior with Lebesgue density $g(\lambda) = 1/\lambda$, then $\PVRFunc(\Tbf_i, \SRatio_i; G)$ is identical to the Behrens-Fisher p-value~\citep{Behrens1929test, fisher1935fiducial}. We note that the Behrens-Fisher p-values are conjectured to be conservative~\citep{robinson1976properties}, i.e., to provide type-I error control for all values of $\lambda_i$.
\end{rema}

\citet{brown1965secondarily} and~\citet{bayarri2000values} consider choices of $G$ in~\eqref{eq:varation_dbn} that are uninformative, objective, and sometimes even improper (as in Remark~\ref{rema:pb_and_bf}). By contrast, for our development, it is crucial that $G$ be given the interpretation of the frequency distribution of the nuisance parameters $\VRatio_i$. The following result is similar in spirit to the notional replications in~\citet[Theorem 1]{meng1994posterior}, however, the conclusion is stronger (exact conditional uniformity).
\begin{prop}
    \label{prop:1D_oracle_uniform}
    Let $i\in \mathcal{H}_0$ and suppose that $\lambda_i \sim G$ as in~\eqref{eq:oracle_1D_p_value}. Then, 
    \smash{$\PVRFunc(\Tbf_i, \SRatio_i; G)$} follows the uniform distribution conditionally on \smash{$\SRatio_i$},
       $$\PP[G]{\PVRFunc(\Tbf_i, \SRatio_i; G) \leq \alpha \; \; \cond \;  \; \SRatio_i} = \alpha \; \text{ for all }\; \alpha \in (0, 1)\; \text{ almost surely},$$
    and thus it also follows the uniform distribution unconditionally,
        $$\PP[G]{\PVRFunc(\Tbf_i, \SRatio_i; G) \leq \alpha} = \alpha \; \text{ for all }\; \alpha \in (0,1).$$
\end{prop}
We emphasize that the above uniformity properties only hold if we also integrate over $\lambda_i \sim G$. In other words, for the result of the proposition to be meaningful, we must compute the partially Bayes p-values with respect to the frequency distribution of the nuisance parameters. If we are conducting a single test, then it is unrealistic that we would know this frequency distribution. However, since we are conducting thousands of tests, we can approximately learn $G$ through empirical Bayes, as we demonstrate next. 

\subsection{Empirical partially Bayes implementation}
\label{subsec:vrepb_eb_implementation}
We seek a data-driven implementation of the oracle partially Bayes p-values $\PVRFunc(\Tbf_i, \SRatio_i; G)$. We build on ideas presented in~\citet{ignatiadis2025empirical} for the empirical partially Bayes setting of Section~\ref{subsec:epb_genomics}.
First, we estimate $G$ using \smash{$\SRatio_1,\ldots,\SRatio_n$} only via the nonparametric maximum likelihood estimator (NPMLE) of \cite{robbins1950generalization} and \cite{kiefer1956consistency}. The NPMLE is defined as the maximizer of the marginal likelihood over all distributions,
\begin{equation}
\label{eq:1D_optimization}
    \hat{G} \in \argmax \left\{\sum_{i=1}^n \log\left(f_G(\SRatio_i)\right) \; : \; \text{G distribution supported on } (0, +\infty)\right\},
\end{equation}
where \smash{$f_G(\SRatio_i)$} is the marginal density of \smash{$\SRatio_i$ } when $\VRatio_i \simiid G$ as in~\eqref{eq:varation_dbn} and \smash{$\SRatio_i \; \cond \; \VRatio_i$} follows the scaled $F$ distribution in \eqref{eq:tau_given_lambda}. Formally, for $\SRatio_i >0$:
\begin{equation}
\label{eq:tau_density}
    f_G(\SRatio_i) \equiv f_G(\SRatio_i; \; \nu_A, \nu_B) = \int_{0}^{\infty} p(\SRatio_i \; \cond \; \VRatio_i, \nu_A, \nu_B) \dd G(\VRatio_i),
\end{equation}
where the likelihood is equal to
\begin{equation}
\label{eq:tau_given_lambda_density}
p(\SRatio \; \cond \; \lambda, \nu_A, \nu_B) := \frac{1}{\VRatio} \frac{1}{B\left(\frac{\nu_A}{2}, \frac{\nu_B}{2}\right)}
\left(\frac{\nu_A}{\nu_B}\right)^{\nu_A/2}
\left(\frac{\SRatio}{\VRatio}\right)^{\nu_A/2-1}
\left(1 + \frac{\nu_A \SRatio}{\nu_B \VRatio}\right)^{-(\nu_A + \nu_B)/2},
\end{equation}
and $B(\cdot,\cdot)$ is the Beta function.  

\begin{rema}[Computation for VREPB]
\label{rema:vrepb_npmle}
In practice, we compute the NPMLE using the recipe laid out by~\citet{koenker2014convex}. Instead of optimizing~\eqref{eq:1D_optimization} over all distributions, we only consider distributions supported on a finite grid $u_1,\ldots,u_B$. Given this discretization,~\eqref{eq:1D_optimization} turns into a finite conic programming problem, which we can solve using the Mosek interior point solver~\citep{aps2020mosek}.
In our concrete implementation, we choose the grid as follows: we take $B=1000$ grid points, logarithmically spaced between the smallest and largest value of \smash{$\SRatio_1, \dots, \SRatio_n$}. In our theoretical development below, we ignore this discretization and study the optimizer over all distributions (as in~\eqref{eq:1D_optimization}), noting that previous authors have studied the effects of discretization on the statistical accuracy of the NPMLE~\citep{dicker2016highdimensional, soloff2024multivariate}. 
\end{rema}

We record a few properties of the NPMLE in~\eqref{eq:1D_optimization} in Supplement \ref{sec:Properties of VR NPMLE}. Our first result is that the NPMLE is consistent for the true frequency distribution of $\VRatio_i$. Our proof (in Supplement \ref{proof:prop_1D_weak_convergence}) builds on the argument of \cite{jewell1982mixtures} alongside an identifiability result of~\citet{teicher1961identifiability}.

\begin{prop}
\label{prop:1D_weak_convergence}
Suppose that $\nu_A, \nu_B \geq 2$. Also suppose that there exist $L,U \in (0,\infty)$ such that $\lambda_i \in [L, U]$ almost surely. Then:
$$ \hat{G} \cd G \; \text{ as }\; n \to \infty\;\text{ almost surely}.
$$
\end{prop}
With $\hat{G}$ in hand, we compute the empirical partially Bayes p-values via the plug-in principle:
\begin{equation}
\PVR_i := \PVRFunc(\Tbf_i, \SRatio_i; \hat{G}).
\label{eq:epb_pvalues}
\end{equation}
Building on Proposition~\ref{prop:1D_weak_convergence} we can show that the empirical partially Bayes p-values are consistent for the corresponding oracle partially Bayes p-values (that have knowledge of the true frequency distribution $G$).  We provide the proof in Supplement \ref{proof:thm_1D_uniform_convergence_p_value}.
\begin{theo}
    \label{thm:1D_uniform_convergence_p_value}
Under the conditions of Proposition~\ref{prop:1D_weak_convergence}, it holds that:
    $$\lim_{n \to \infty} \max_{1 \leq i \leq n} \EE[G]{\abs{\PVR_i - \PVRFunc(\Tbf_i, \SRatio_i; G)}} = 0.$$
\end{theo}
Finally, as a consequence of Theorem \ref{thm:1D_uniform_convergence_p_value} and Proposition \ref{prop:1D_oracle_uniform}, we conclude that the empirical partially Bayes p-values $\PVR_i$ are asymptotically (in $n$, but with $K_A,K_B$ fixed) uniform. We provide the proof in Supplement \ref{proof:thm_1D asymptotic uniformity}.
\begin{theo}
    \label{thm:1D asymptotic uniformity}
    Let $i \in \mathcal{H}_0$.
Under the conditions of Proposition~\ref{prop:1D_weak_convergence}, $\PVR_i$ is asymptotically uniform conditional on all sample variance ratios \smash{$\SRatio_1,\ldots, \SRatio_n$},
$$\lim_{n \to \infty} \max_{i \in \mathcal{H}_0} \cb{ \EE[G]{\sup_{\alpha \in [0, 1]} \abs{ \PP[G]{\PVR_i \leq \alpha \; \Big | \; \SRatio_1, \dots, \SRatio_n} - \alpha }}} = 0.$$
As a consequence, $\PVR_i$ is asymptotically uniform,
$$\lim_{n \to \infty} \max_{i \in \mathcal{H}_0}\cb{ \sup_{\alpha \in [0,1]} \abs{ \PP[G]{\PVR_i \leq \alpha} - \alpha }} = 0.$$
\end{theo}
We note that in our empirical studies and simulations below, we use these asymptotic p-values $\PVR_1,\ldots,\PVR_n$ alongside the \cite{benjamini1995controlling} (BH) procedure for false discovery rate control to produce a subset $\mathcal{D}$ of $\{1, \dots, n\}$ representing the indices of rejected hypotheses. Our full VREPB algorithm is summarized in Algorithm~\ref{algo:vrepb} of Supplement \ref{appendix: algorithm}.

\subsection{Frequency interpretation: asymptotic compound p-values}
\label{subsec:vrepb_compound}

The result of Theorem~\ref{thm:1D asymptotic uniformity} relies on the validity of~\eqref{eq:varation_dbn} that $\lambda_i \sim G$. In this section we provide an asymptotic type-I error guarantee for  the VREPB p-values $\PVR_i$ in~\eqref{eq:epb_pvalues} when~\eqref{eq:varation_dbn} does not hold and the vector of variance ratios $\boldVRatio = (\lambda_1,\ldots,\lambda_n)$ is fixed (as in a standard frequentist analysis). We provide the proof of the next Theorem in Supplement \ref{proof:vrepb_compound}.

\begin{theo}[Asymptotic compound p-values]
Suppose that $\nu_A, \nu_B \geq 2$. Also suppose that there exist $L,U \in (0,\infty)$ such that $\lambda_i \in [L, U]$ for all $i$. Then, the VREPB p-values are asymptotic compound p-values, i.e., they satisfy the following property:
$$\limsup_{n \to \infty}\cb{ \frac{1}{n} \sum_{i \in \mathcal{H}_0} \PP[\boldVRatio]{\PVR_i \leq \alpha} } \leq \alpha\; \text{ for all }\; \alpha \in (0,1).$$
\label{theo:vrepb_compound}
\end{theo}
The term (asymptotic) compound p-values comes from~\citet{ignatiadis2025asymptotic}, though~\citet{armstrong2022false} had earlier introduced the same concept under the name average significance controlling p-values. The subscript $\boldVRatio$ in the probability expression ($\PP[\boldVRatio]{\cdot}$) indicates that we are treating all $\VRatio_i$ as fixed. In this case, for any single $i \in \mathcal{H}_0$, it is not the case in general that $\PVR_i \approx \mathrm{Unif}[0,1]$. This may not be surprising in view of the difficulties in determining critical values for $\Tbf_i$ when $\lambda_i$ is unknown that we expounded in Section~\ref{subsec:ttest_known_variance_ratio}. However, according to Theorem~\ref{theo:vrepb_compound}, $\PVR_i$ is still approximately uniform on average: for some $i \in \mathcal{H}_0$, we may have $\PP[\boldVRatio]{\PVR_i \leq \alpha} \gg \alpha$, but such effects are offset by $\PP[\boldVRatio]{\PVR_i \leq \alpha} < \alpha$ for other $i \in \mathcal{H}_0$. 

When researchers use a fixed threshold (e.g., $P_i \leq 0.001$) to reject hypotheses~\citep{klaus2018end}, compound p-values provide the same guarantee about the expected number of false discoveries as standard p-values. When applied alongside the Benjamini-Hochberg procedure, asymptotic compound p-values control FDR under certain asymptotic regimes~\citep{armstrong2022false, ignatiadis2025empirical} (see~\citet{barber2025false} for finite sample results with bona-fide compound p-values).

\section{Dual variance empirical partially Bayes (DVEPB)}
\label{sec: Groupwise variances as nuisance parameters}

\subsection{Context}
In Section~\ref{subsec:epb_genomics}, we noted that even with known variance ratios $\VRatio_i$, it is common practice—as exemplified by limma—to impose a partially Bayes prior on the (single) unknown variance, $\sigma_{iA}^2 \sim H_A$, cf.~\eqref{eq:equal_var_partially_Bayes}. By contrast, in the VREPB approach of Section~\ref{sec: Variance ratios as nuisance parameters} we  dispensed with the assumption that the variance ratios are known and instead
placed a prior explicitly on the variance ratios, $\lambda_i \sim G$~\eqref{eq:varation_dbn}. A natural generalization that incorporates both strategies is to treat the pair $(\sigma_{iA}^2, \lambda_i)$ jointly as nuisance parameters, imposing a partially Bayes prior directly on this bivariate space. After a reparameterization to sample-wise variances $(\sigma_{iA}^2, \sigma_{iB}^2)$, this joint prior assumption can be expressed as:
\begin{equation}
\label{eq:sigma_prior}
(\sigma_{iA}^2,\sigma_{iB}^2) \simiid \ H,
\end{equation}
where $H$ is an unknown bivariate distribution. The modeling assumption underlying~\eqref{eq:sigma_prior} will form the basis of the dual-variance empirical partially Bayes approach (DVEPB) that we develop below.

To be more explicit, the connection to alternative partially Bayes modeling assumptions described above is as follows. First, the limma model described in Section~\ref{subsec:epb_genomics} with $\sigma_{iA}^2 \sim H_A$ and equal variances ($\VRatio_i=1$) may be recovered as the special case of \eqref{eq:sigma_prior} where $H$ places all mass on the diagonal $(\sigma_{iA}^2 = \sigma_{iB}^2)$, i.e., $H(\{\sigma_{iA}^2 = \sigma_{iB}^2\}) = 1$. Second, the bivariate prior assumption \eqref{eq:sigma_prior} implies the variance-ratio prior assumption $\lambda_i\sim G$ in~\eqref{eq:varation_dbn}: $G$ can be recovered as the pushforward distribution of the mapping $(\sigma_{iA}^2,\sigma_{iB}^2) \mapsto \sigma_{iA}^2/\sigma_{iB}^2$ under $H$.

From these relationships, we view the DVEPB methodology as combining the advantages of both alternative empirical partially Bayes approaches: DVEPB retains limma's benifit of pooling information across tests regarding the magnitude of the variances. For instance, if $\sigma_{iA}^2=1$ for all $i$, then DVEPB will be able to learn this relationship and incorporate it in the construction of the empirical partially Bayes p-values, while VREPB is not able to learn such variance patterns. At the same time, like VREPB, DVEPB avoids the restrictive equal variance assumption: in case the assumption indeed holds, then DVEPB (and VREPB) can learn this relationship automatically.

\begin{figure}
\centering
\adjustbox{width=0.8\textwidth}{\begin{tikzpicture}[
    node distance=3.2cm, 
    every node/.style={font=\sffamily},
    fixed/.style={text=black},
    exchangeable/.style={text=blue, font=\bfseries},
    eb/.style={draw=orange, thin, fill=orange!10, text=orange!80!black, 
               font=\sffamily\small, rounded corners, inner sep=4pt}
]

\draw[-{Stealth[length=3mm]}, line width=1.5pt] (-1.0,0) -- (13.9,0);

\foreach \x/\pos in {0/0, 3.9/4.7, 8.1/9.3, 12/14} {
    \draw[thick] (\x,-0.1) -- (\x,0.1);
}

\node[align=center, font=\bfseries] at (0,0.5) {Frequentist};
\node[align=center, font=\bfseries] at (3.9,0.5) {Partially Bayes (VR)};
\node[align=center, font=\bfseries] at (8.1,0.5) {Partially Bayes (DV)};
\node[align=center, font=\bfseries] at (12,0.5) {Bayes};

\matrix[matrix of math nodes, nodes={anchor=west}, row sep=0.2cm] at (0,-1.7) {
    \node[fixed] {\sigma_{i,A}^2}; & \node[fixed] {\text{fixed}}; \\
    \node[fixed] {\sigma_{i,B}^2}; & \node[fixed] {\text{fixed}}; \\
    \node[fixed] {\mu_{i,A}}; & \node[fixed] {\text{fixed}}; \\
    \node[fixed] {\mu_{i,B}}; & \node[fixed] {\text{fixed}}; \\
};

\matrix[matrix of math nodes, nodes={anchor=west}, row sep=0.2cm] at (3.9,-1.7) {
    \node[exchangeable] (lambda) {\lambda_i := \frac{\sigma_{i,A}^2}{\sigma_{i,B}^2}}; & \node[exchangeable] {\text{exchangeable}}; \\
    \node[fixed] {\sigma_{i,A}^2}; & \node[fixed] {\text{fixed}}; \\
    \node[fixed] {\mu_{i,A}}; & \node[fixed] {\text{fixed}}; \\
    \node[fixed] {\mu_{i,B}}; & \node[fixed] {\text{fixed}}; \\
};

\matrix[matrix of math nodes, nodes={anchor=west}, row sep=0.2cm] at (8.1,-1.7) {
    \node[exchangeable] (sigmaA) {\sigma_{i,A}^2}; & \node[exchangeable] {\text{exchangeable}}; \\
    \node[exchangeable] (sigmaB) {\sigma_{i,B}^2}; & \node[exchangeable] {\text{exchangeable}}; \\
    \node[fixed] {\mu_{i,A}}; & \node[fixed] {\text{fixed}}; \\
    \node[fixed] {\mu_{i,B}}; & \node[fixed] {\text{fixed}}; \\
};

\matrix[matrix of math nodes, nodes={anchor=west}, row sep=0.2cm] at (12,-1.7) {
    \node[exchangeable] {\sigma_{i,A}^2}; & \node[exchangeable] {\text{exchangeable}}; \\
    \node[exchangeable] {\sigma_{i,B}^2}; & \node[exchangeable] {\text{exchangeable}}; \\
    \node[exchangeable] {\mu_{i,A}}; & \node[exchangeable] {\text{exchangeable}}; \\
    \node[exchangeable] {\mu_{i,B}}; & \node[exchangeable] {\text{exchangeable}}; \\
};

\draw[rounded corners, draw=orange, thick, dashed] 
    (1.6,-3.3) rectangle (10.0,-0.1);
\node[above, text=orange, font=\sffamily\bfseries] 
    at (6.0,-0.1) {Proposed Approaches};

\draw[->, orange, thick] (3.5,-3.6) -- (3.5,-3.15);
\node[eb, align=center, text width=5.2cm, minimum height=1.5cm] at (2.8,-4.3) 
    {Empirical Partially Bayes:\\ $\lambda_i \sim G$, \\
    $\hat{G}=\mathrm{NPMLE}(\{\hat{\lambda}_i: i=1,\dotsc,n)\})$};

\draw[->, orange, thick] (8.1,-3.6) -- (8.1,-3.15);
\node[eb, align=center, text width=6cm, minimum height=1.5cm] at (8.9,-4.3) 
    {Empirical Partially Bayes:\\ $(\sigma_{i,A}^2, \sigma_{i,B}^2) \sim H$,  \\
    $\hat{H}=\mathrm{NPMLE}(\{(\hat{\sigma}_{i,A}^2, \hat{\sigma}_{i,B}^2): i=1,\dotsc,n\})$};
\end{tikzpicture}}
\caption{Variance ratio empirical partially Bayes (VREPB) and dual variance empirical partially Bayes (DVEPB) along the frequentist–Bayesian spectrum: VREPB models $\lambda_i$ as exchangeable under a shared prior $G$; DVEPB models $(\sigma_{iA}^2,\sigma_{iB}^2)$ as exchangeable under a bivariate prior $H$. Both use the nonparametric maximum likelihood estimator (NPMLE) for prior estimation. DVEPB captures richer variance structure.}
\label{fig:Bayes-axis}
\end{figure}
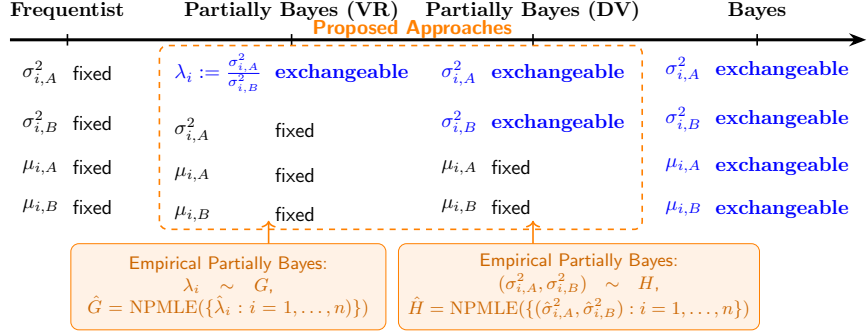

Before proceeding with the description of DVEPB, we situate it in a spectrum of frequentist and Bayesian  methods; see Figure~\ref{fig:Bayes-axis} for a graphical illustration. In a fully frequentist analysis, we would treat all parameters $\mu_{iA},\mu_{iB},\sigma_{iA}^2,\sigma_{iB}^2$ as fixed. Meanwhile, in the partially Bayes approach we either treat the variance ratio $\lambda_{i}=\sigma_{iA}^2/\sigma_{iB}^2$ as random and drawn from a prior, or alternatively both variances $\sigma_{iA}^2,\sigma_{iB}^2$. Thus the dual variance approach brings us closer to fully Bayesian modeling in which all parameters are drawn from a prior. We note that we do not pursue a fully empirical Bayes approach here since our focus is on extending the partially Bayes approach that is common in genomics.

\subsection{Oracle partially Bayes}

\label{sec: 2D Oracle partially Bayes p-values}

In Section~\ref{subsec:ttest_known_variance_ratio}, we started with a thought experiment of computing p-values based on $\Tbf_i$ when $\VRatio_i$ is known. Here we carry out the analogous thought experiment when the pair $(\sigma^2_{iA}, \sigma^2_{iB})$ is fixed and known. Under the null, $\hat{\mu}_{iB} - \hat{\mu}_{iA} \sim \mathrm{N}(0, \sigma_{iA}^2/K_A + \sigma_{iB}^2/K_B)$ and this result also holds conditional on $\hat{\sigma}^2_{iA}, \hat{\sigma}^2_{iB}$. In analogy to~\eqref{eq:PVR_known_nuisance}, we can define the tail area of $\Tbf_i$ with $\sigma_{iA}^2,\sigma_{iB}^2$ fixed and known, and conditioning on $\hat{\sigma}^2_{iA}, \hat{\sigma}^2_{iB}$,
\begin{equation}
    \label{eq:PDV_known_nuisance}
    \PDVFunc(t, s_A^2, s_B^2; \sigma_A^2,\sigma_B^2) := \PP[\sigma_A^2,\sigma_B^2]{ \abs{\TbfNull} \geq \abs{t} \; \cond \; \hat{\sigma}_A^2 = s_A^2,\; \hat{\sigma}_B^2=s^2_B},
\end{equation}
which by the above considerations is equal to
$$
\PDVFunc(t, s_A^2, s_B^2; \sigma_A^2,\sigma_B^2) = 2\Phi\p{ - \abs{t}\cdot \cb{\p{s_A^2/K_A} + \p{s_B^2/K_B}}^{1/2}  \big / \cb{\p{\sigma_A^2/K_A} + \p{\sigma_B^2/K_B}}^{1/2}}.
$$
Such explicit knowledge about $(\sigma^2_{iA}, \sigma^2_{iB})$ is typically unavailable. However, as before, we can turn~\eqref{eq:PDV_known_nuisance} into a practical method by pursuing an empirical partially Bayes approach by assuming that~\eqref{eq:sigma_prior} holds, that is, $(\sigma_{iA}^2,\sigma_{iB}^2)\sim H$. If $H$ was known, we could compute the $i$-th p-value as $\PDVFunc(\Tbf_i, \hat{\sigma}^2_{iA}, \hat{\sigma}^2_{iB}; H)$, where the tail area (in analogy to~\eqref{eq:oracle_1D_p_value}) is defined as:
\begin{align}
    \label{eq:oracle_2D_p_value}
    \PDVFunc(t, s^2_A, s^2_B; H) &:= \PP[H]{ \abs{\TbfNull} \geq \abs{t} \; \cond \; \hat{\sigma}_A^2 = s_A^2,\; \hat{\sigma}_B^2=s^2_B} \nonumber \\ 
    &= \EE[H]{\PDVFunc(t, s^2_A, s^2_B; \sigma^2_{A}, \sigma^2_{B}) \; \cond \; \hat{\sigma}^2_{A} = s^2_A,\; \hat{\sigma}^2_{B} = s^2_B}.
\end{align}

Above, the subscript $H$ in the probability $\PP[H]{\cdot}$ and expectation $\EE[H]{\cdot}$ statements indicates that we are also integrating over randomness in~\eqref{eq:sigma_prior}. The second equality establishes the interpretation that $\PDVFunc(t, s^2_{iA}, s^2_{iB}; H)$ is the expectation of $\PDVFunc(t, s^2_A, s^2_B; \sigma^2_{A}, \sigma^2_{B})$ with respect to the posterior distribution of $(\sigma^2_{A}, \sigma^2_{B})$ given $(\hat{\sigma}^2_{A}, \hat{\sigma}^2_{B}) = (s^2_A, s^2_B)$. In equations~\eqref{eq:PDV_known_nuisance} and ~\eqref{eq:oracle_2D_p_value} we abuse notation with the same rationale as for $\PVRFunc$. For \smash{$H=\delta_{(\sigma_A^2,\sigma_{B}^2)}$}, it holds that \smash{$\PDVFunc(t, s^2_A, s^2_B; \delta_{(\sigma_A^2,\sigma_{B}^2)}) = \PDVFunc(t, s^2_A, s^2_B; \sigma_A^2,\sigma_{B}^2).$}

Similar to the result of Proposition~\ref{prop:1D_oracle_uniform}, when $H$ is the frequency distribution of the nuisance parameters $(\sigma^2_{iA}, \sigma^2_{iB})$, the p-value function in \eqref{eq:oracle_2D_p_value} produces p-values that are exactly conditional uniform under the null; see Supplement \ref{proof:prop_2D_oracle_uniform} for a proof.

\begin{prop}
    \label{prop:2D_oracle_uniform}
    Let $i\in \mathcal{H}_0$ and suppose that $(\sigma^2_{iA}, \sigma^2_{iB}) \sim H$ as in~\eqref{eq:sigma_prior}. Then, 
    \smash{$\PDVFunc(\Tbf_i, \hat{\sigma}^2_{iA}, \hat{\sigma}^2_{iB}; H)$} follows the uniform distribution conditionally on \smash{$(\hat{\sigma}^2_{iA}, \hat{\sigma}^2_{iB})$},
       $$\PP[H]{\PDVFunc(\Tbf_i, \hat{\sigma}^2_{iA}, \hat{\sigma}^2_{iB}; H) \leq \alpha \; \; \cond \;  \; \hat{\sigma}^2_{iA}, \hat{\sigma}^2_{iB}} = \alpha \; \text{ for all }\; \alpha \in (0, 1)\; \text{ almost surely},$$
    and thus it is also follows the uniform distribution unconditionally,
        $$\PP[H]{\PDVFunc(\Tbf_i, \hat{\sigma}^2_{iA}, \hat{\sigma}^2_{iB}; H) \leq \alpha} = \alpha \; \text{ for all }\; \alpha \in (0, 1).$$
\end{prop}
We emphasize that for the above result to be meaningful, we must compute the partially Bayes p-values with respect to the frequency distribution of the nuisance parameters $(\sigma^2_{iA}, \sigma^2_{iB})$. Since we are conducting thousands of tests in parallel, we can approximately learn $H$ through empirical Bayes. 

\begin{rema}[Independence of $\sigma_{iA}^2, \sigma_{iB}^2$]
In developing the partially Bayes approach for the two-sample problem, \citet{brown1965secondarily} also considers the possibility of conducting inference conditional on \smash{$\hat{\sigma}_{iA}^2,\hat{\sigma}_{iB}^2$} rather than conditional on \smash{$\SRatio_i$}. However, in his implementation, he considers a specification of $H$ under which $\sigma_{iA}^2$ and $\sigma_{iB}^2$ are independent. This assumption appears to be inappropriate for real applications.
\end{rema}

\subsection{Empirical partially Bayes implementation}
\label{subsec:eb_dvepb}
Here we propose a data-driven implementation of the oracle partially Bayes p-values $\PDVFunc(\Tbf_i, \\ \hat{\sigma}^2_{iA}, \hat{\sigma}^2_{iB}; H)$ following the empirical partially Bayes blueprint of Section~\ref{subsec:vrepb_eb_implementation}
First, we estimate $H$ using $(\hat{\sigma}_{iA}^2, \hat{\sigma}_{iB}^2), i = 1, \dots, n$ via the NPMLE defined as the maximizer of the marginal likelihood over all distributions,
\begin{equation}
\label{eq:2D_optimization}
    \hat{H} \in \argmax \left\{\sum_{i=1}^n \log\left(f_H(\hat{\sigma}^2_{iA}, \hat{\sigma}^2_{iB})\right) \; : \; \text{H distribution supported on } \mathbb R_+^2 \right\},
\end{equation}
where $f_H(\hat{\sigma}_{iA}^2, \hat{\sigma}_{iB}^2)$ is the joint marginal density of $(\hat{\sigma}_{iA}^2, \hat{\sigma}_{iB}^2)$ when $(\sigma^2_{iA}, \sigma^2_{iB}) \simiid H$ as in \eqref{eq:sigma_prior}, and both $\hat{\sigma}_{iA}^2 \; \mid \; \sigma^2_{iA}, \hat{\sigma}_{iB}^2 \; \mid \; \sigma^2_{iB}$ follow independent scaled $\chi^2$ distributions as in~\eqref{eq:ss_distribution},
\begin{equation}
\label{eq:s2_density}
    f_H(\hat{\sigma}_{iA}^2, \hat{\sigma}_{iB}^2) \equiv f_H(\hat{\sigma}_{iA}^2, \hat{\sigma}_{iB}^2; \; \nu_A, \nu_B) := \int_{\mathbb{R}^2_+} p(\hat{\sigma}_{iA}^2, \hat{\sigma}_{iB}^2 \; \cond \; \sigma^2_{iA}, \sigma^2_{iB}, \nu_A, \nu_B) \dd H(\sigma^2_{iA}, \sigma^2_{iB}),
\end{equation}
see Supplement \ref{appendix: supplementary formulas} for a precise expression of  $p(\hat{\sigma}_{iA}^2, \hat{\sigma}_{iB}^2 \; \mid \; \sigma^2_{iA}, \sigma^2_{iB}, \nu_A, \nu_B)$. Our optimization follows a similar discretization strategy as laid out for VREPB in Remark~\ref{rema:vrepb_npmle}:

\begin{rema}[Computation for DVEPB]
\label{rema:dvepb_npmle}
In practice, instead of optimizing~\eqref{eq:2D_optimization} over all distributions, we only consider distributions supported on a finite two-dimensional grid $\{u_1,\ldots,u_{B_1}\} \times \{v_1,\ldots,v_{B_2}\}$. Given this discretization,~\eqref{eq:2D_optimization} again turns into a finite conic programming problem, which we can solve using the Mosek interior point solver~\citep{aps2020mosek}.
In our concrete implementation, we choose the grid as follows: we take $B_1 = 80$ grid points, logarithmically spaced between the 1\% quantile and largest value of \smash{$(\hat{\sigma}^2_{1A}, \dots, \hat{\sigma}^2_{nA})$} and $B_2 = 80$ grid points, logarithmically spaced between the 1\% quantile and largest value of \smash{$(\hat{\sigma}^2_{1B}, \dots \hat{\sigma}^2_{nB})$}.
We use a coarser grid than for VREPB (Remark~\ref{rema:vrepb_npmle}) because the two-dimensional optimization in~\eqref{eq:2D_optimization}  is computationally more expensive than the one-dimensional case; the discretization allows the problem to remain tractable while still capturing the key features of the variance distribution.
We ignore discretization issues in our theoretical development.
\end{rema}

We record a few properties of the NPMLE in~\eqref{eq:2D_optimization} in Supplement \ref{sec:properties DV NPMLE}.
The NPMLE \eqref{eq:2D_optimization} is consistent for the true frequency distribution of $(\sigma^2_{iA}, \sigma^2_{iB})$; see Supplement \ref{proof:prop_2D_weak_convergence} for the proof that builds on identifiability results of~\citet{teicher1961identifiability, teicher1967identifiability} and \cite{nielsen1965identifiability}.
\begin{prop}
\label{prop:2D_weak_convergence}
Suppose that $\nu_A, \nu_B \geq 2$. Also suppose that there exist $L_A,U_A,L_B,U_B \in (0,\infty)$ such that $(\sigma^2_{iA}, \sigma^2_{iB}) \in [L_A, U_A] \times [L_B, U_B]$ almost surely. Then:
$$ \hat{H} \cd H \; \text{ as }\; n \to \infty\;\text{ almost surely}.
$$
\end{prop}
With $\hat{H}$ in hand, we compute the empirical partially Bayes p-values via the plug-in principle:
\begin{equation}
\PDV_i := \PDVFunc(\Tbf_i, \SRatio_i; \hat{H}).
\label{eq:DV_epb_pvalues}
\end{equation}
Building on Proposition~\ref{prop:2D_weak_convergence} we can show that the above empirical partially Bayes p-values are consistent for the corresponding oracle partially Bayes p-values with the knowledge of the underlying $H$.  We provide the proof in Supplement \ref{proof:thm_2D_uniform_convergence_p_value}.
\begin{theo}
    \label{thm:2D_uniform_convergence_p_value}
Under the conditions of Proposition~\ref{prop:2D_weak_convergence}, it holds that:
    $$\lim_{n \to \infty} \max_{1 \leq i \leq n} \EE[H]{\abs{\PDV_i - \PDVFunc(\Tbf_i, \hat{\sigma}^2_{iA}, \hat{\sigma}^2_{iB}; H)}} = 0.$$
\end{theo}
Finally, as a consequence of Theorem \ref{thm:2D_uniform_convergence_p_value} and Proposition \ref{prop:2D_oracle_uniform}, we conclude that the empirical partially Bayes p-values $\PDV_i$ are asymptotically (in $n$, but with $K_A,K_B$ fixed) uniform. We provide the proof in Supplement \ref{proof:thm_2D asymptotic uniformity}.
\begin{theo}
    \label{thm:2D asymptotic uniformity}
    Let $i \in \mathcal{H}_0$.
Under the conditions of Proposition~\ref{prop:2D_weak_convergence}, $\PDV_i$ is asymptotically uniform conditional on all groupwise sample variances \smash{$(\hat{\sigma}^2_{1A}, \hat{\sigma}^2_{1B}), \dots, (\hat{\sigma}^2_{nA}, \hat{\sigma}^2_{nB})$},
$$\lim_{n \to \infty} \max_{i \in \mathcal{H}_0} \cb{ \EE[H]{\sup_{\alpha \in [0, 1]} \abs{ \PP[H]{\PDV_i \leq \alpha \; \Big | \; (\hat{\sigma}^2_{1A}, \hat{\sigma}^2_{1B}), \dots, (\hat{\sigma}^2_{nA}, \hat{\sigma}^2_{nB})} - \alpha }}} = 0.$$
As a consequence, $\PDV_i$ is asymptotically uniform,
$$\lim_{n \to \infty} \max_{i \in \mathcal{H}_0}\cb{ \sup_{\alpha \in [0,1]} \abs{ \PP[H]{\PDV_i \leq \alpha} - \alpha }} = 0.$$
\end{theo}
In our empirical studies and simulations below, we use these asymptotic p-values $\PDV_1,\ldots,\\\PDV_n$ alongside the BH procedure for false discovery rate control to produce a subset $\mathcal{D}$ of $\{1, \dots, n\}$ representing the indices of rejected hypotheses. Our full DVEPB algorithm is summarized in Algorithm~\ref{algo:dvepb} of Supplement~\ref{appendix: algorithm}.

\section{Simulation study}

\label{sec: simulation}

\begin{figure}
    \centering
    \setlength{\tabcolsep}{0pt}
    \renewcommand{\arraystretch}{0.95}

    \begin{tabular}{@{}l@{}}
    (a) Unequal variances, $K_A=5$, $K_B\in\{5,\dots,11\}$.\\
        \includegraphics[width=.85\linewidth,trim=4mm 5mm 4mm 4mm,clip]{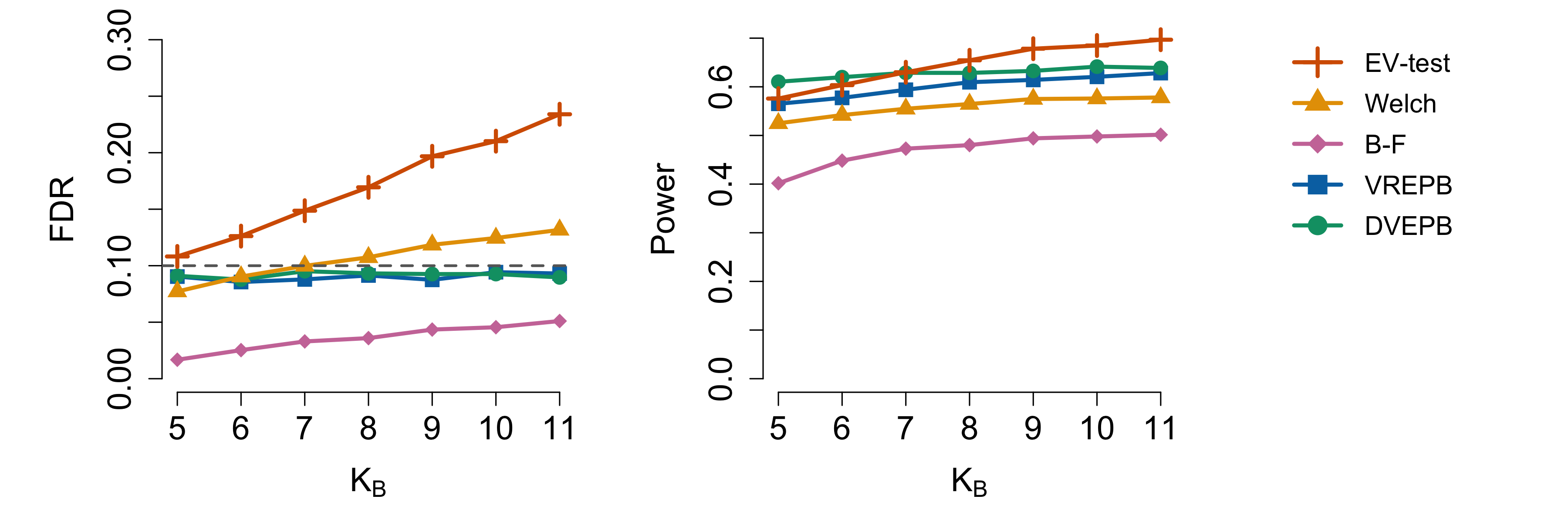} \\
    (b) Equal variances, $K_A=3$, $K_B\in\{3,\dots,9\}$.\\
        \includegraphics[width=.85\linewidth,trim=4mm 5mm 4mm 4mm,clip]{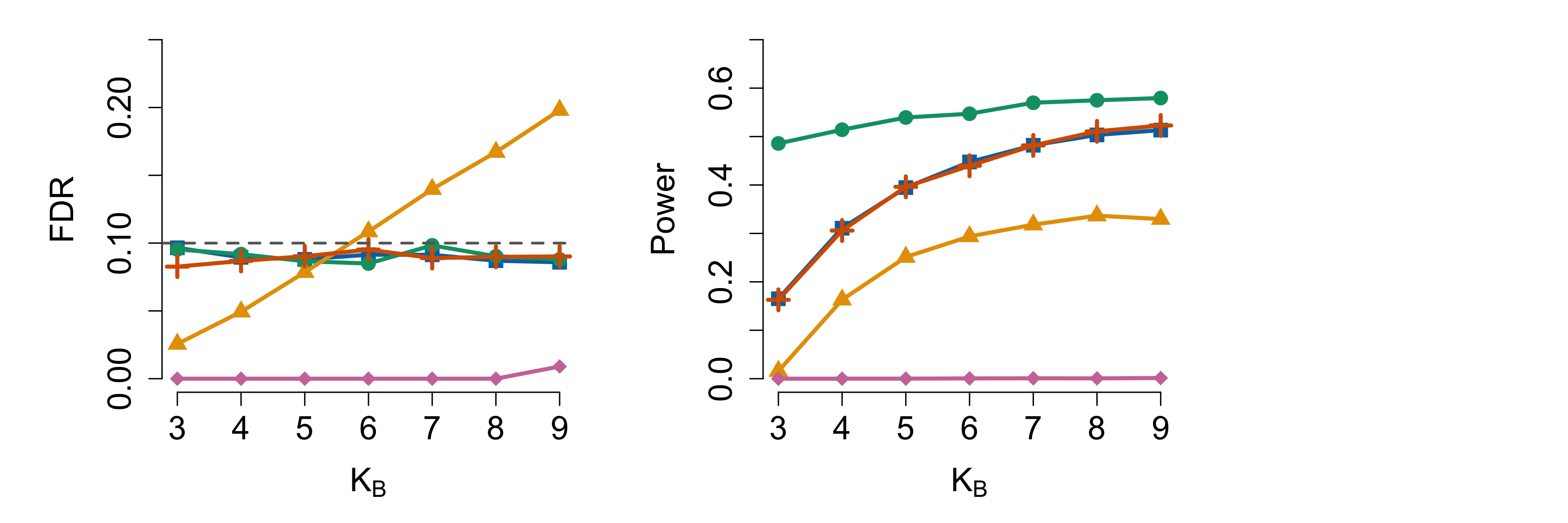} \\
    (c) Equal variances, $K_A=5$, $K_B\in\{5,\dots,11\}$. \\
        \includegraphics[width=.85\linewidth,trim=4mm 5mm 4mm 4mm,clip]{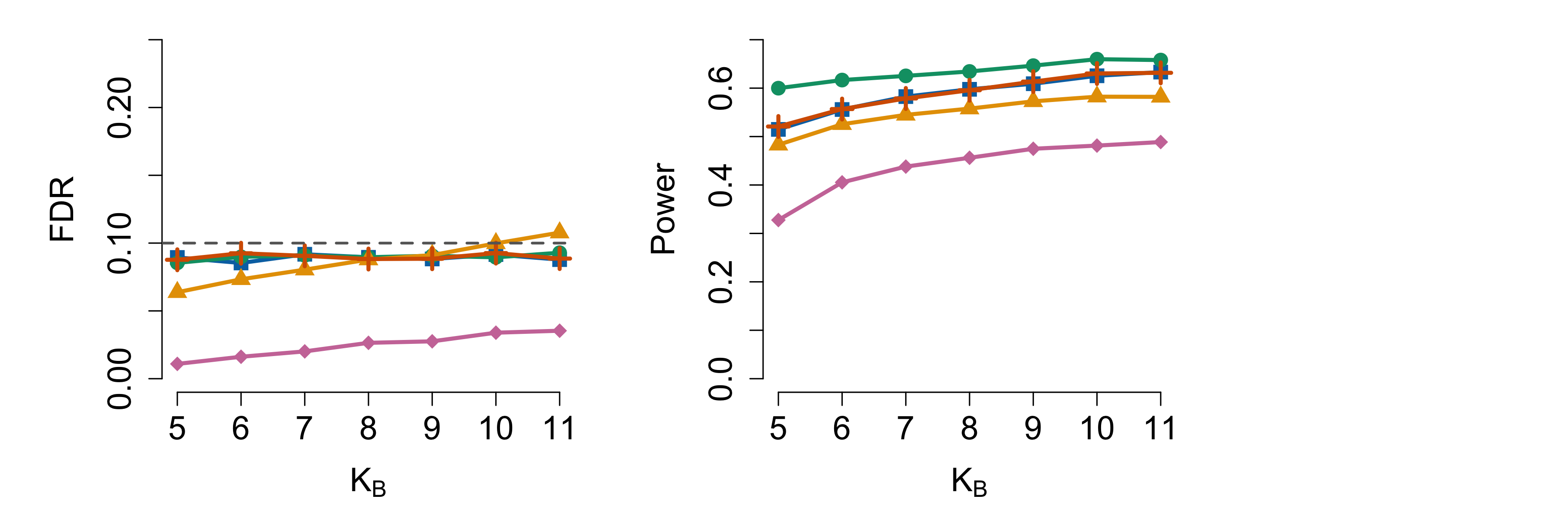} \\
    (d) Diffuse variance ratios, $K_A=5$, $K_B\in\{5,\dots,11\}$.\\
        \includegraphics[width=.85\linewidth,trim=4mm 5mm 4mm 4mm,clip]{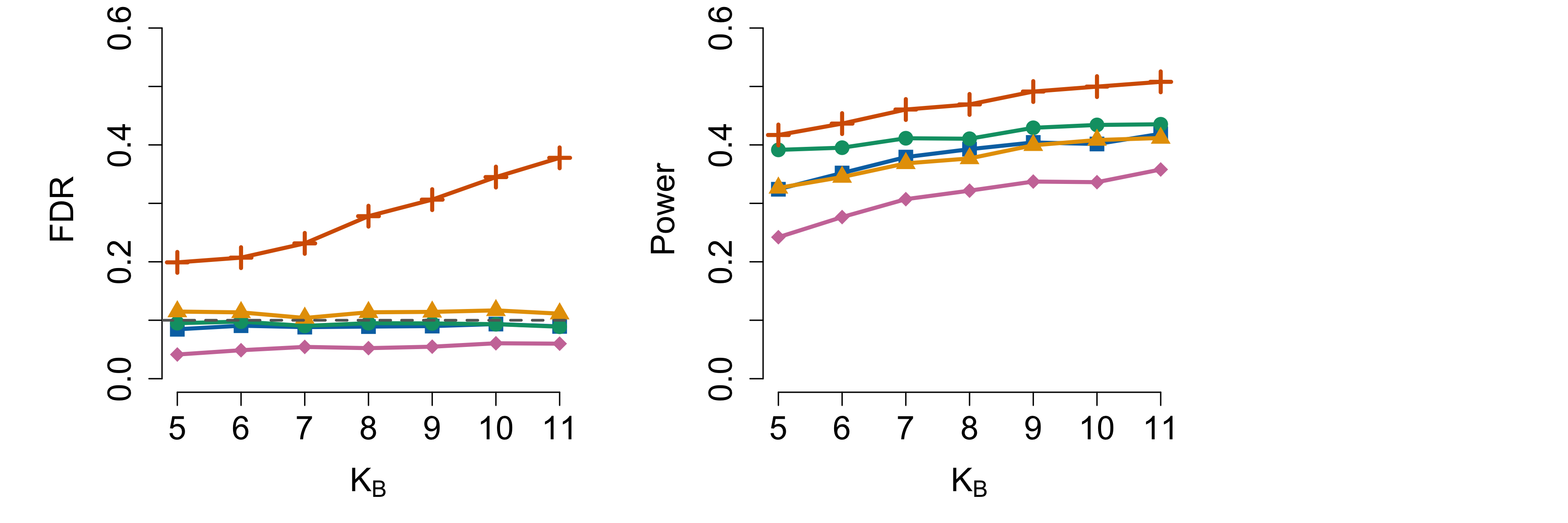}
    \end{tabular}

    \caption{Simulation results in four scenarios with varying $K_A$ and heteroscedasticity. Each panel plots FDR (left) and power (right) versus $K_B \in \{K_A, \ldots, K_A + 6 \}$ for EV-test, Welch, B-F, VREPB (ours), DVEPB (ours); the dashed line marks the BH level $\alpha=0.1$.}
    \label{fig:sim-overview}
    \vspace{-6pt}
\end{figure}

We use simulations to assess finite-sample error control and power in large-scale two-sample testing with small, possibly unbalanced replication and heterogeneous variances. 
Each setting has $n=5000$ parallel tests with null proportion $\pi_0=0.9$. Null units have $\mu_{iA}=\mu_{iB}=0$; alternatives have $\mu_{iB}=0$ and $\mu_{iA}\sim\mathrm{N}(0,12\sigma_{iA}^2)$. We draw $\sigma_{iB}^2\sim\lvert\mathrm{N}(6,4)\rvert$ and set $\sigma_{iA}^2=\lambda_i\sigma_{iB}^2$, where $\lambda_i$ is generated according to one of the following three settings:
\begin{enumerate}[itemsep=0pt, parsep=0pt]
    \item unequal variances: $\lambda_i\sim 2F_{8,12}$, where $F_{8,12}$ denotes the $F$-distribution with 8 and 12 degrees of freedom;
    \item equal variances: $\lambda_i\equiv 1$; 
    \item diffuse variance ratios: $\lambda_i\sim\exp(\mathrm{Unif}[-5,5])$. (This last choice is motivated by Remark~\ref{rema:pb_and_bf}. The induced density of $\lambda_i$ is the truncation of the improper prior with density $g(\lambda)=1/\lambda$ to the interval $[\exp(-5),\exp(5)]$.)
\end{enumerate} 
We consider $K_A\in\{3,5\}$ and vary $K_B\in\{K_A,\dots,K_A+6\}$. For each simulation setting, we generate 50 Monte Carlo replicates, compute method-specific p-values, apply Benjamini–Hochberg at $\alpha=0.1$, and report FDR and power, defined as the expected proportion of true alternatives correctly rejected, averaged over replicates.

We compare five methods for computing p-values:
\begin{enumerate}[itemsep=0pt, parsep=0pt]
    \item the equal-variance t-test (EV);
    \item Welch’s unequal-variance t-test (Welch);
    \item the Behrens–Fisher test (B–F);
    \item our VREPB;
    \item and DVEPB procedures.
\end{enumerate}
As previewed in Fig.~\ref{fig:unequal simulation 1} (which corresponds to the first simulation setting with $K_A=3$ and unequal variances), EV and Welch can inflate FDR under heteroscedasticity with small, unbalanced samples. Figure~\ref{fig:sim-overview} collates results for all remaining settings; each panel reports FDR (left) and power (right).
Under unequal variances (Fig.~\ref{fig:sim-overview}, panel (a)), EV and Welch inflate FDR; B–F controls FDR but is nearly powerless; VREPB and DVEPB both control FDR and substantially increase power, with DVEPB consistently most powerful. Under equal variances (Fig.~\ref{fig:sim-overview}, panels (b)–(c)), all methods (except Welch under severe imbalance) control FDR; DVEPB achieves the highest power, while VREPB $\approx$ EV as expected when $\lambda_i\equiv 1$. For the diffuse scenario (Fig.~\ref{fig:sim-overview}, panel (d)), B–F is less conservative and closer to VREPB than in heteroscedastic settings (consistent with Remark~\ref{rema:pb_and_bf}). Welch and EV fail to control FDR (although the violations for Welch while controlling FDR.

\section{Applications}
\label{sec:applications}

\subsection{Gene expression in macrophages after lipoprotein stimulation}
\label{sec: Microarrays Data: ALL}

\citet{Rangel_GSE9101} profile human THP-1 macrophage cells with $K_A = 3$ unstimulated controls versus $K_B = 9$ lipoprotein-stimulated samples. Gene expression of each sample is measured on microarrays (Affymetrix Human Genome U133 Plus 2.0 Array). After standard preprocessing (see Supplement~\ref{sec:preprocessing} for details), we obtain $n = 20{,}989$ genes for analysis. The assumption of normality following such preprocessing, as required by our statistical setting (described in Section~\ref{sec:statiatical setting}), is widely considered a well-founded approximation in microarray analysis \citep{lonnstedt2002replicated, smyth2004linear}.

\begin{table}
\caption{Differential gene expression with microarray data measured in unstimulated and lipoprotein-stimulated macrophages. The table records the number of discoveries for each method.}
\centering
\begin{tabular}{ccccc}
VREPB & DVEPB & Welch & EV-test & B-F \\ 
2322  & 2940  & 1614  & 2180    & 0   \\ 
\end{tabular}
\label{table:macrophage}
\end{table}

\begin{figure}
  \centering
  \setlength{\tabcolsep}{4pt}
  \renewcommand{\arraystretch}{0}

  \begin{tabular}{@{}c c @{}}
    \subcaptionbox{\label{fig:GSE9101 G}}{\includegraphics[width=0.45\linewidth]{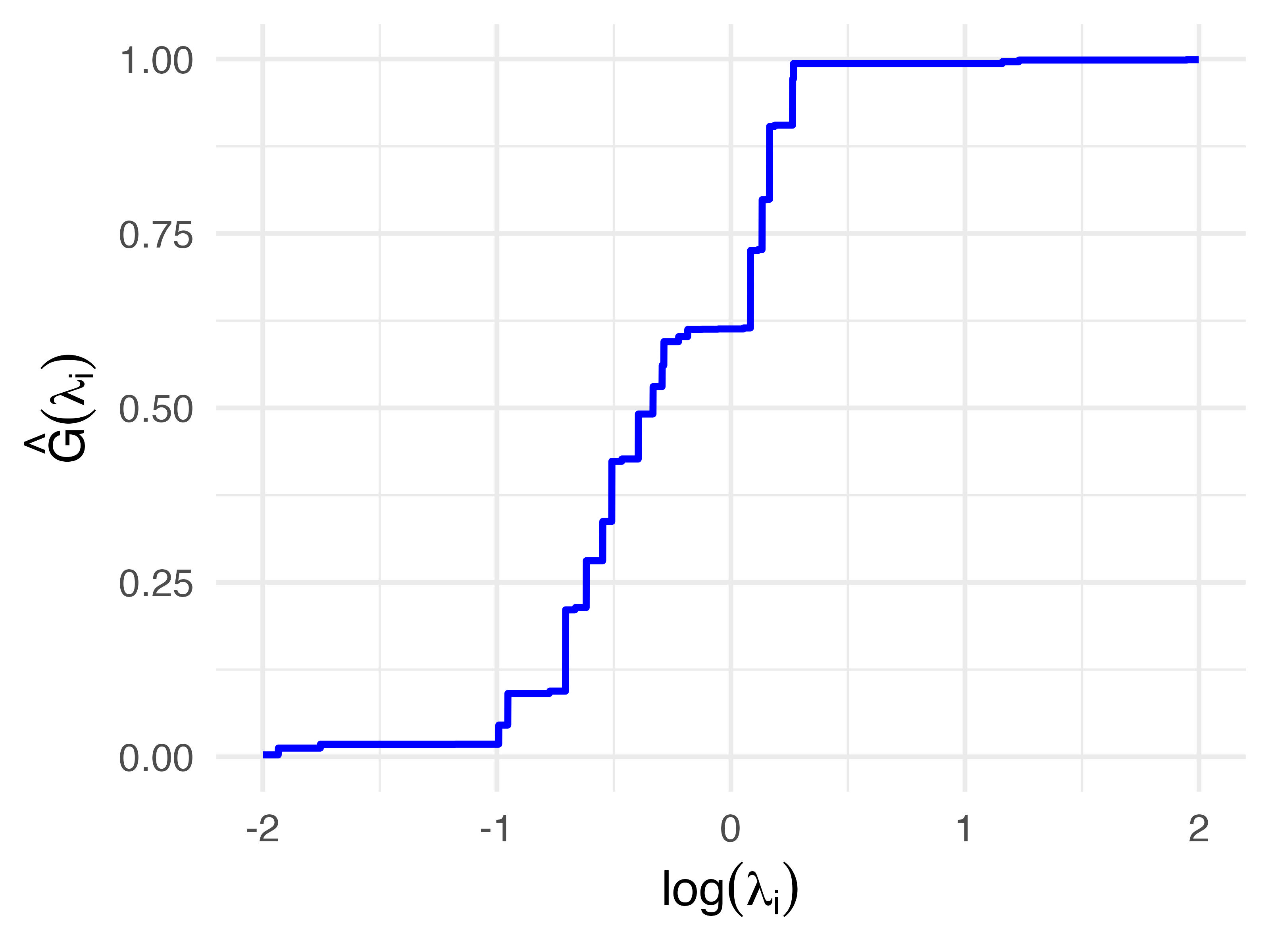}} &
    \subcaptionbox{\label{fig:GSE9101 H}}{\includegraphics[width=0.45\linewidth]{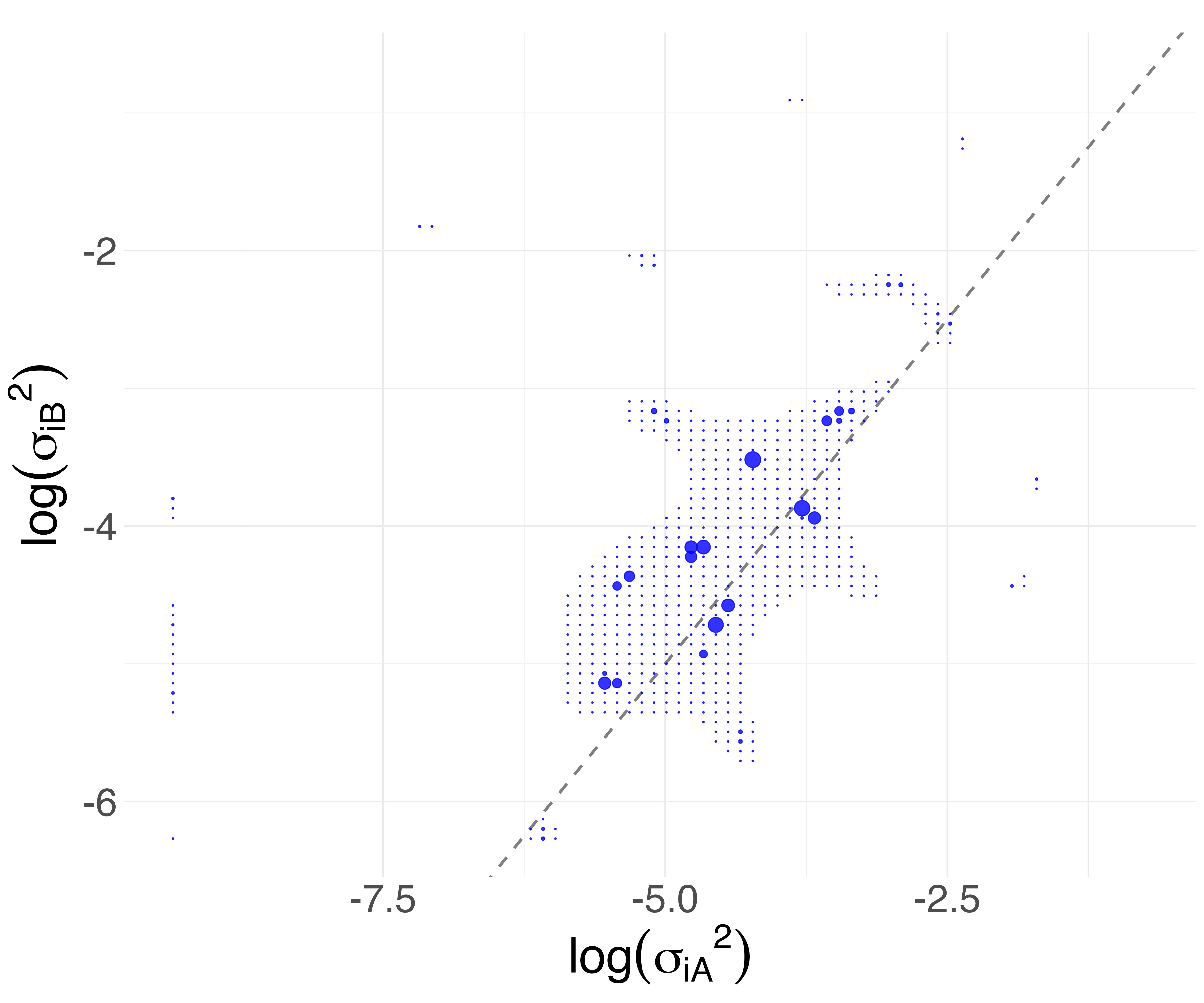}}
  \end{tabular}

  \caption{Estimated nuisance parameter frequency distributions for the macrophage dataset. (a) Cumulative distribution function of $\hat{G}$ (variance ratio distribution) used by VREPB. (b) Illustration of $\hat{H}$  (bivariate variance distribution) used by DVEPB, with point sizes proportional to probability mass and the gray dashed diagonal indicating equal variances. (Recall by Remarks~\ref{rema:vrepb_npmle} and~\ref{rema:dvepb_npmle} that both $\hat{G}$ and $\hat{H}$ are discrete distributions with a finite number of support points.)}
  \label{fig:ALL}
\end{figure}

After computing p-values with the same five methods as in Section~\ref{sec: simulation} (histograms in Figure~\ref{fig:p-values GSE9101}, Supplement~\ref{sec:p-value histograms}), we apply Benjamini-Hochberg to control the FDR at $\alpha = 0.01$. Table~\ref{table:macrophage} shows that DVEPB achieves the highest power, with VREPB also substantially outperforming Welch, while Behrens-Fisher makes no discoveries. Although the equal-variance t-test appears competitive, Figure~\ref{fig:ALL}(a) reveals that $\hat{G}$—the variance ratio distribution estimated and used by VREPB—places non-negligible mass away from $\lambda_i = 1$. This heteroscedasticity may compromise FDR control for the equal-variance method.
 
Figure~\ref{fig:ALL}(b) shows $\hat{H}$, the estimated frequency distribution of sample-wise variance pairs used by DVEPB. The probability mass spans from points on the diagonal (equal variances) to points near it (unequal variances), consistent with the mixture of variance ratios seen in $\hat{G}$. However, unlike $\hat{G}$ which captures only variance ratios, $\hat{H}$ also reveals structure in the variance magnitudes themselves, explaining DVEPB's superior power compared to VREPB.

\subsection{Single cell pseudobulk gene expression in peripheral blood mononuclear cells of COVID-19 patients}
\label{sec: Single-cell RNA-seq Pseudo-bulk Data}

\citet{Zhao2021PBMC} use single-cell pseudo-bulk RNA-seq to profile peripheral blood mononuclear cells (PBMC) from COVID-19 patients with varying symptom severity and healthy controls. In the analysis below, we compare gene expression between patients with moderate severity versus healthy controls. After quality control, the dataset contains $n = 9234$ genes, with $K_A = 5$ moderate samples and $K_B = 3$ healthy control samples.

In this application, the raw data consists of counts, and so the modeling assumptions of this paper (as described in Section~\ref{sec:statiatical setting}) are not directly applicable. Nevertheless, one of the most popular approaches for RNA-Seq data, called Voom~\citep{law2014voom}, involves transforming the counts into log-counts per million (log-cpm) such that a weighted generalization of the initial model is approximately applicable,
\begin{equation}
Z_{ij} \overset{\text{ind}}{\sim} \mathrm{N}\left(\mu_{iA}, v_{ij}^A \sigma_{iA}^2\right), \;\; Y_{ij} \overset{\text{ind}}{\sim} \mathrm{N}\left(\mu_{iB}, v_{ij}^B \sigma_{iB}^2\right),
\label{eq:model_with_weights}
\end{equation}
for weights $v_{ij}^A, v_{ij}^B >0$. In this weighted framework (Supplement~\ref{sec: weighted model setup}), the weights absorb known sources of variance heterogeneity across observations, and ``homoscedasticity'' refers to $\sigma_{iA}^2=\sigma_{iB}^2$ rather than equality of the raw observation variances, \smash{$\Var{Z_{ij}}=\Var{Y_{ij}}$}. Substantial research in bioinformatics has addressed the specification of these weights. Below we illustrate results using both the default Voom weights~\citep{law2014voom} and the VoomByGroup weights developed by~\citet{you2023modeling} to explicitly address sample-specific heteroscedasticity in k-sample comparisons.
We construct p-values using VREPB, DVEPB, Welch, the equal variance t-test (EV-test), and the Behrens-Fisher test (B-F). We explain how these methods can be modified to incorporate precision weights in Supplement~\ref{sec: weighted methods}.

\begin{table}
\caption{Differential gene expression with single cell pseudobulk data measured in peripheral blood mononuclear cells of moderate severity COVID-19 patients versus healthy controls. The table records the number of discoveries for each method.}
\centering
\begin{tabular}{lccccc}
            & VREPB & DVEPB & Welch & EV-test & B-F \\ 
Voom        & 5     & 260   & 0     & 4       & 0 \\ 
VoomByGroup & 7     & 318   & 0     & 7       & 0 \\ 
\end{tabular}
\label{table:PBMC1}
\end{table}

For each combination of weights (Voom, VoomByGroup) and method, we first compute p-values (histograms available at Figure \ref{fig:p-values PBMC1 voom} and Figure \ref{fig:p-values PBMC1 vbg} in Supplement \ref{sec:p-value histograms}) and then apply the BH procedure at $\alpha = 0.1$. Table \ref{table:PBMC1} records the number of discoveries of each method. Among the methods compared, DVEPB demonstrates the highest statistical power under both the Voom and VoomByGroup settings. In contrast, VREPB and E-V identify a comparable number of significant findings, while Behrens–Fisher and Welch yield no discoveries for both cases. For all tests except Welch and B-F, incorporating weight information estimated by VoomByGroup leads to more discoveries than Voom.

\begin{figure}
  \centering
  \setlength{\tabcolsep}{4pt}
  \renewcommand{\arraystretch}{0}

  \begin{tabular}{@{}c c @{}}
    \multicolumn{2}{c}{\subcaptionbox{\label{fig:PBMC1 G}}%
  {\includegraphics[width=0.45\linewidth]{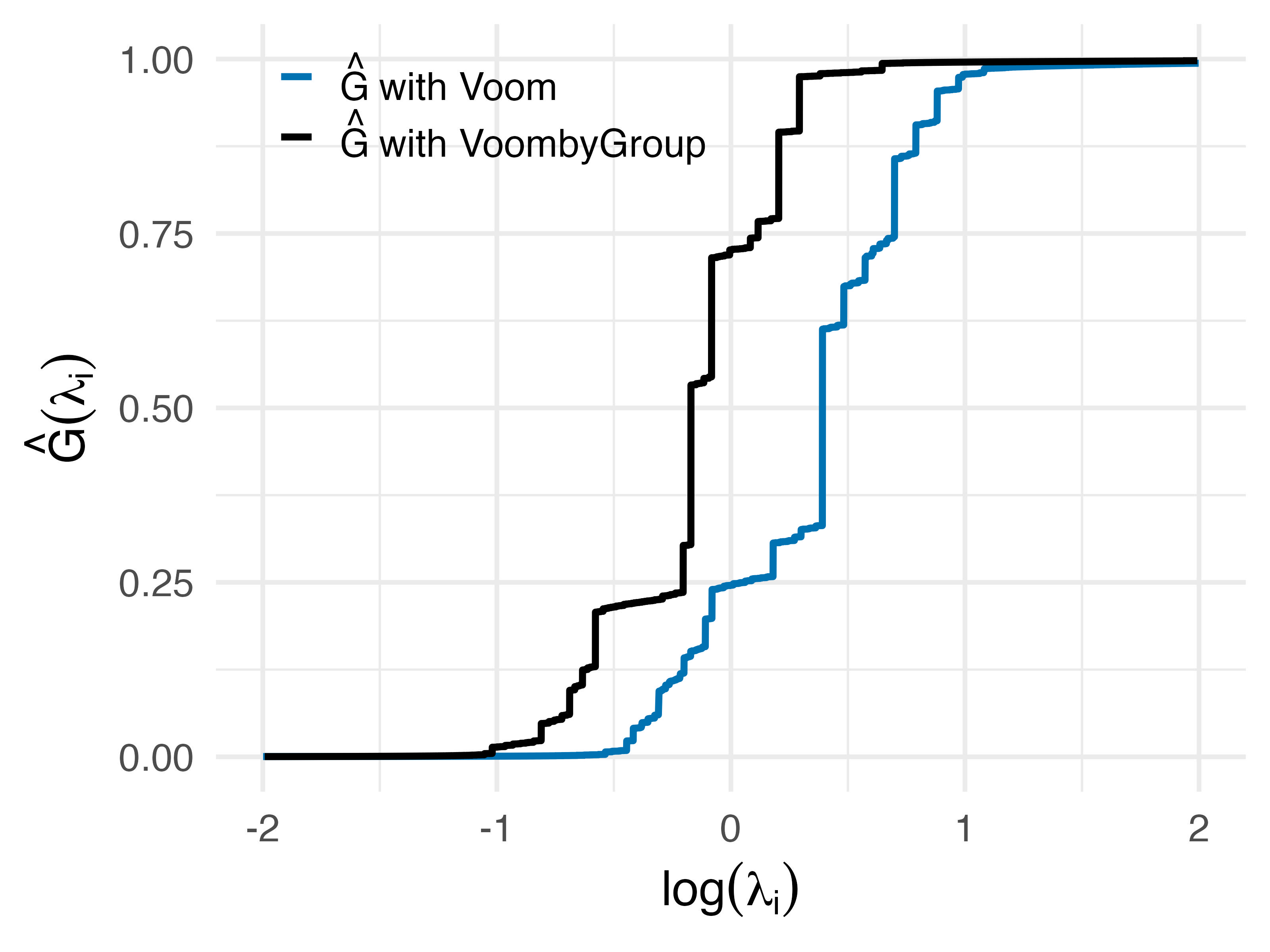}}} \\
      \subcaptionbox{\label{fig:PBMC1 H voom}}{\includegraphics[width=0.45\linewidth]{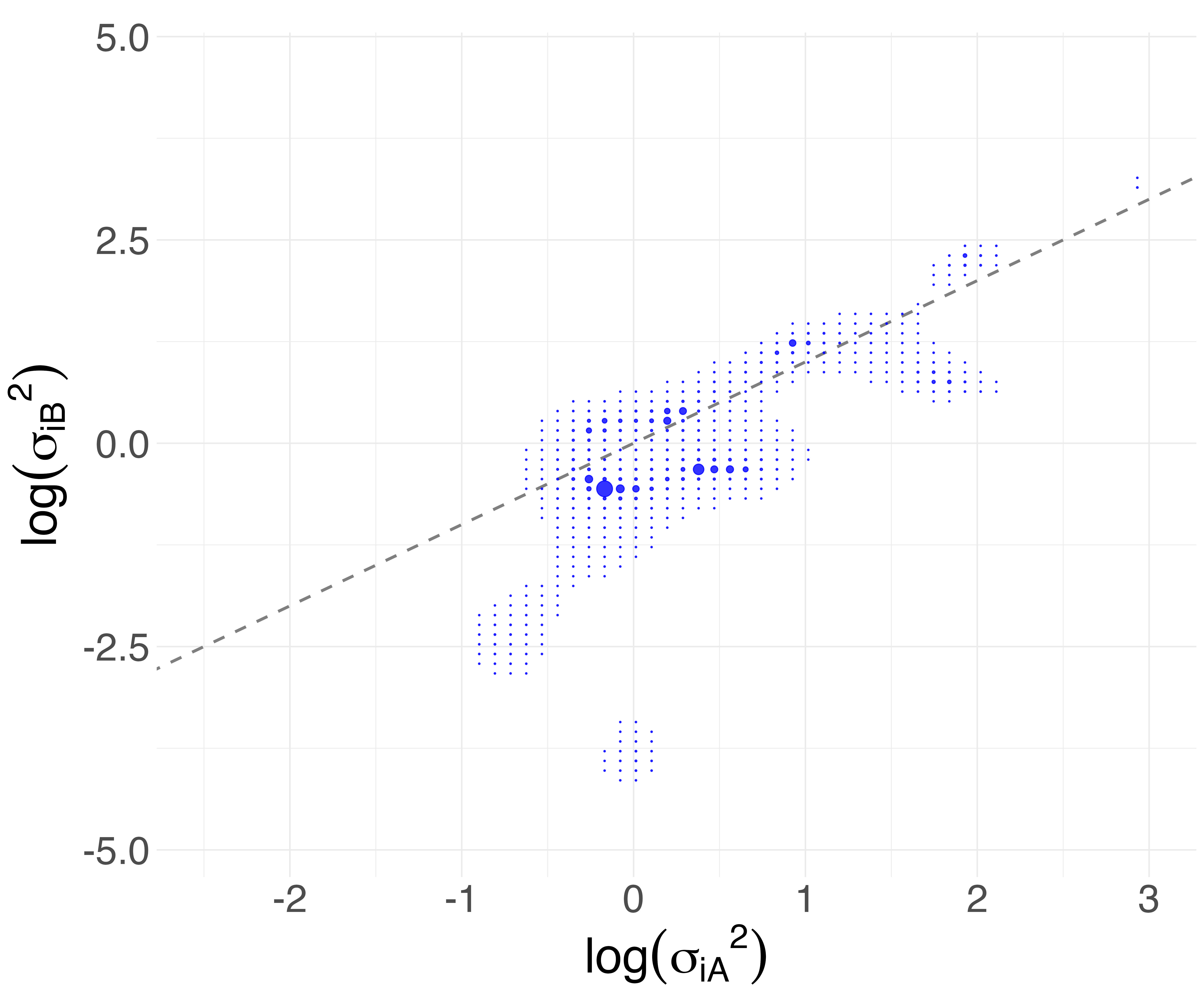}} &
    \subcaptionbox{\label{fig:PBMC1 H vbg}}{\includegraphics[width=0.45\linewidth]{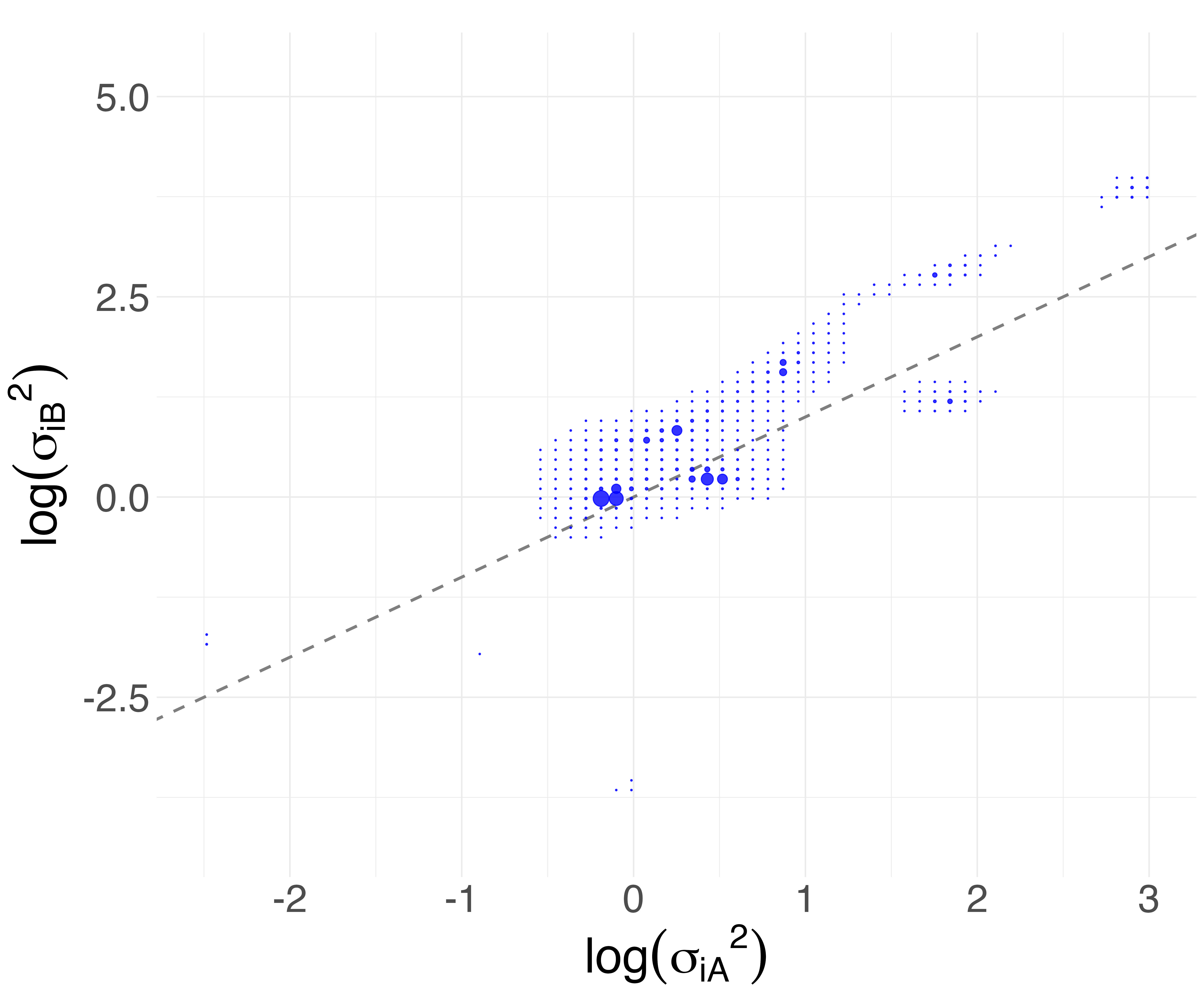}} 
  \end{tabular}
    \caption{Estimated nuisance parameter frequency distributions for the COVID-19 peripheral blood mononuclear cell dataset. (a) Cumulative distribution functions of $\hat{G}$ (variance ratio distribution) used by VREPB with Voom and VoomByGroup weights. (b)-(c) Illustration of $\hat{H}$ (bivariate variance distribution) used by DVEPB with (b) Voom weights and (c) VoomByGroup weights. Point sizes are proportional to probability mass, and the gray dashed diagonal indicates equal variances. (Recall by Remarks~\ref{rema:vrepb_npmle} and~\ref{rema:dvepb_npmle} that both $\hat{G}$ and $\hat{H}$ are discrete distributions with a finite number of support points.)}  
    \label{fig:PBMC1}
\end{figure}

In Figure~\ref{fig:PBMC1}, we show in panel (a) the estimated distribution $\hat{G}$ of $\lambda_i$ used by VREPB (using Voom and VoomByGroup weights), and the bivariate distribution $\hat{H}$ of DVEPB with Voom weights (panel (b)) and VoomByGroup weights (panel (c)). In interpreting this figure, it is helpful to think of the weights in~\eqref{eq:model_with_weights} as playing the role of guesses for $\Var{Z_{ij}}$ and $\Var{Y_{ij}}$. Indeed, if $v_{ij}^A = \Var{Z_{ij}}$ and $v_{ij}^B = \Var{Y_{ij}}$, then we would have $\sigma_{iA}^2=\sigma_{iB}^2=1$. Moreover, if the weights are correct up to a common proportionality factor, then we would have $\lambda_i = \sigma_{iA}^2/\sigma_{iB}^2=1$. Figure~\ref{fig:PBMC1} shows that VoomByGroup, designed specifically to handle sample-specific heteroscedasticity, substantially reduces variance heterogeneity compared to default Voom (see~\citet{you2023modeling}). For instance, in panel (a), $\hat{G}$ with VoomByGroup is centered around $\lambda_i=1$ and is relatively concentrated around it, whereas for Voom, $\hat{G}$ is shifted toward values of $\lambda_i > 1$. We can observe an analogous trend in the two-dimensional priors $\hat{H}$ (panels b-c), wherein VoomByGroup leads to probability mass more closely concentrated around the diagonal. Nevertheless, some residual heterogeneity remains, as we discuss below.

While VREPB and the equal variance t-test yield similar numbers of discoveries in this application, this does not indicate equivalent performance. When variance ratios are truly equal ($\VRatio_i = 1$ for all $i$), VREPB p-values converge to equal variance t-test p-values (Theorem~\ref{thm:1D_uniform_convergence_p_value}), and the two methods perform comparably. However, when variance ratios deviate from unity, the equal variance t-test can suffer inflated FDR, whereas VREPB adapts by producing appropriately calibrated p-values. For this dataset, the residual dispersion in $\hat{G}$ away from $\lambda_i=1$---particularly pronounced with Voom weights---suggests that applying the equal variance t-test risks FDR inflation. VREPB addresses this concern by accounting for variance ratio heterogeneity that remains even after VoomByGroup's improved weight specification.
Finally, the two-dimensional prior estimated by DVEPB captures additional information about variance magnitudes, enabling it to make by far the most discoveries with both Voom and VoomByGroup weights.

\section{Further related work}
\label{sec:related_work}
Throughout the paper, we have provided references to existing work where appropriate; in this brief section, we provide some additional context and references.
The Behrens-Fisher problem is one of the most fundamental statistical problems, and in the previous century, it sparked feuds about the foundations of statistical modeling. We refer to~\citet{scheffe1970practical,lee1975size,wallace1980behrensfisher,robinson2014behrensfisher} for reviews of the issues involved. The partially Bayes approach for the Behrens-Fisher problem developed by~\citet{brown1965secondarily,brown1967twomeans}, which is central to this paper, seems to be less well known (with only 6 citations to date). 

The approach of~\citet{duong1992empirical, duong1996behrensfisher} is perhaps the most similar to ours. Although not stated as such, their approach indeed can be interpreted as an empirical partially Bayes approach. They also construct p-values similar to the VREPB p-values in~\eqref{eq:epb_pvalues}. The principal differences are as follows: first, they consider a parametric specification for $G$ in~\eqref{eq:varation_dbn} and second, they consider only a single test instead of an asymptotically growing number of tests. In this way, their proposal does not have a type-I error guarantee akin to Theorem~\ref{thm:1D asymptotic uniformity} in which $K_A, K_B$ remain fixed.

Additional approaches include: \citet*{fraser1999simple, fraser2009three, rukhin2015higher} use higher-order asymptotic theory,~\citet*{hannig2006fiducial} provide an asymptotical justification of the fiducial argument via their theory for fiducial generalized pivotal quantities, \citet*{ullah2019significance} refine Welch's critical values through simulation, and \citet{barber2022testing} use approximate co-sufficient sampling. All of these approaches achieve provable nominal type-I error as $K_A, K_B \to \infty$.

Several recent works address Behrens-Fisher testing in the multiple testing context. \citet{demissie2008unequal} and \citet{zhang2021simultaneous} study the Behrens-Fisher problem alongside the Benjamini-Hochberg procedure: ~\citet*{demissie2008unequal} also seek to extend the empirical partially Bayes argument of limma (Section~\ref{subsec:epb_genomics}) to the unequal variance case (as we do in Section~\ref{sec: 2D Oracle partially Bayes p-values}). However, it is unclear if their approach has any type-I error guarantee. \citet*{zhang2021simultaneous} study the Behrens-Fisher t-statistic in~\eqref{eq:BF} and (a calibrated modification of) the pooled t-statistic in~\eqref{eq: T_lambda} in the multiple testing setting without assuming normality but considering an asymptotic regime with $n \to \infty$, $K_A\equiv K_A(n)\to \infty$, $K_A\equiv K_B(n) \to \infty$, in such a way that $K_A/K_B \to \gamma \in (0,\infty)$ and $\log n / (K_A+K_B)^{1/3} \to 0$.
Other recent works \citep{zou2020new, ge2021clipper, tian2025conformalized} sidestep p-value construction and achieve finite-sample FDR control using symmetry arguments under the null hypothesis alongside the SeqStep+ procedure of~\citet{barber2015controlling}. We instead develop empirical partially Bayes methods to construct asymptotic (compound) p-values.

\section*{Reproducibility}

All numerical results in this paper are fully reproducible. All code and datasets are provided on Github under the following link:
\noindent \href{https://github.com/Frankbest18/TwoSampleEPB-Paper}{https://github.com/Frankbest18/TwoSampleEPB-Paper}

\section*{Acknowledgments}
We thank Sida Li and Zixuan Wu for helpful feedback on an earlier version of this manuscript. We thank Dennis L. Sun for pointing us to the work of Linnik, and Jan Hannig for noting that it is still unknown whether the Behrens-Fisher test has type-I error control. Part of the computing for this project was conducted on UChicago's Data Science Institute cluster. N.I. gratefully acknowledges support from NSF (DMS 2443410).

\bibliographystyle{abbrvnat}
\bibliography{twosampleEB}

\appendix

\setcounter{equation}{0}
\setcounter{figure}{0}
\setcounter{table}{0}
\setcounter{prop}{0}

\renewcommand{\theequation}{S\arabic{equation}}
\renewcommand{\thefigure}{S\arabic{figure}}
\renewcommand{\thetable}{S\arabic{table}}
\renewcommand{\thealgocf}{S\arabic{algocf}}
\renewcommand{\theprop}{S\arabic{prop}}
\renewcommand{\thetheo}{S\arabic{theo}}
\renewcommand{\thelemm}{S\arabic{lemm}}

\makeatletter
\renewcommand{\theHprop}{S.\arabic{prop}}
\makeatother

\newpage
\section{P-value histograms of real datasets}
\label{sec:p-value histograms}

\begin{figure}[H]
  \centering
  \setlength{\tabcolsep}{4pt}
  \renewcommand{\arraystretch}{0}

  \begin{tabular}{@{}c c@{}}
    \subcaptionbox{VREPB\label{fig:GSE9101 VREPB}}{\includegraphics[width=0.45\linewidth]{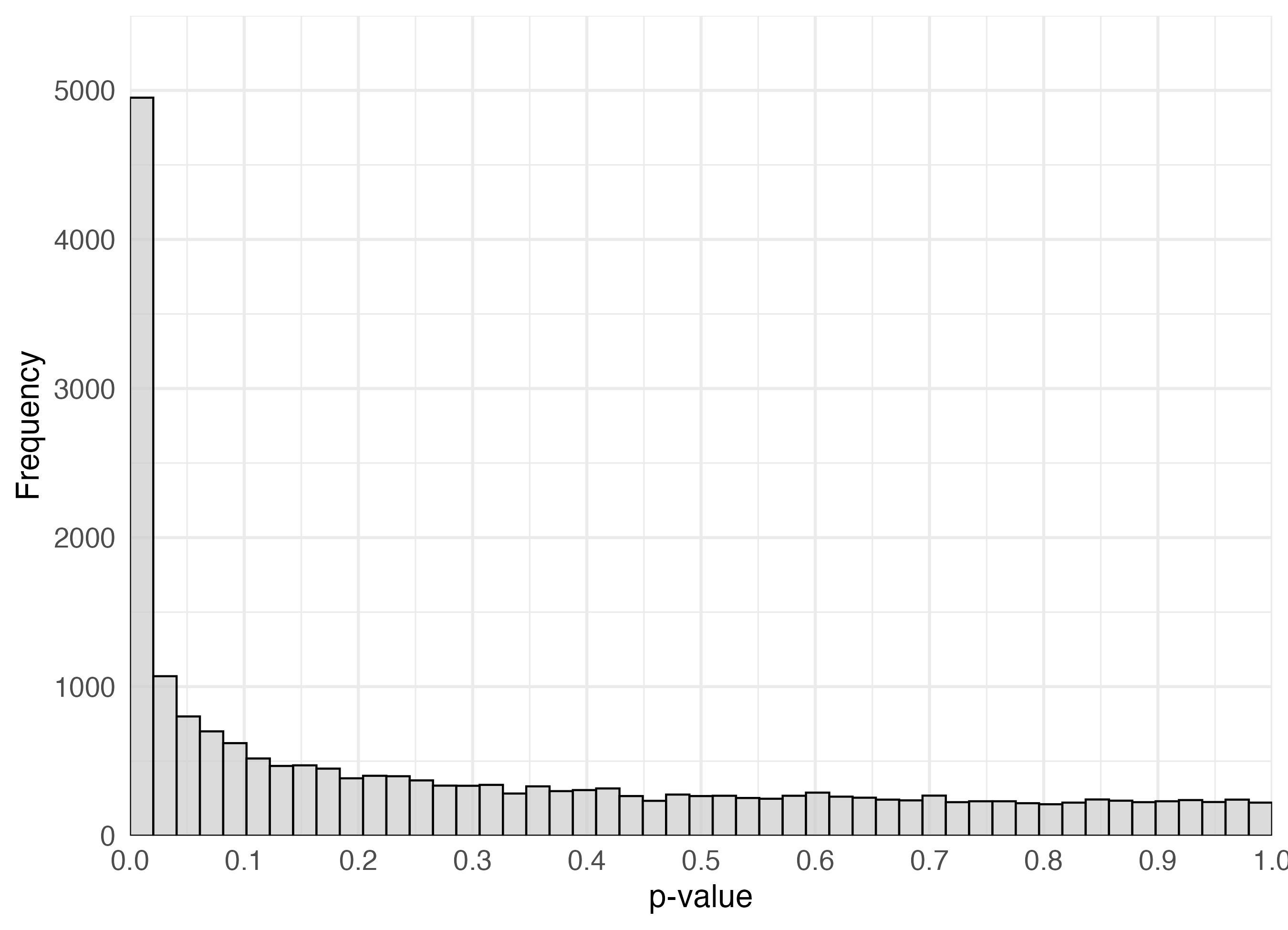}} &
    \subcaptionbox{DVEPB\label{fig:GSE9101 DVEPB}}{\includegraphics[width=0.45\linewidth]{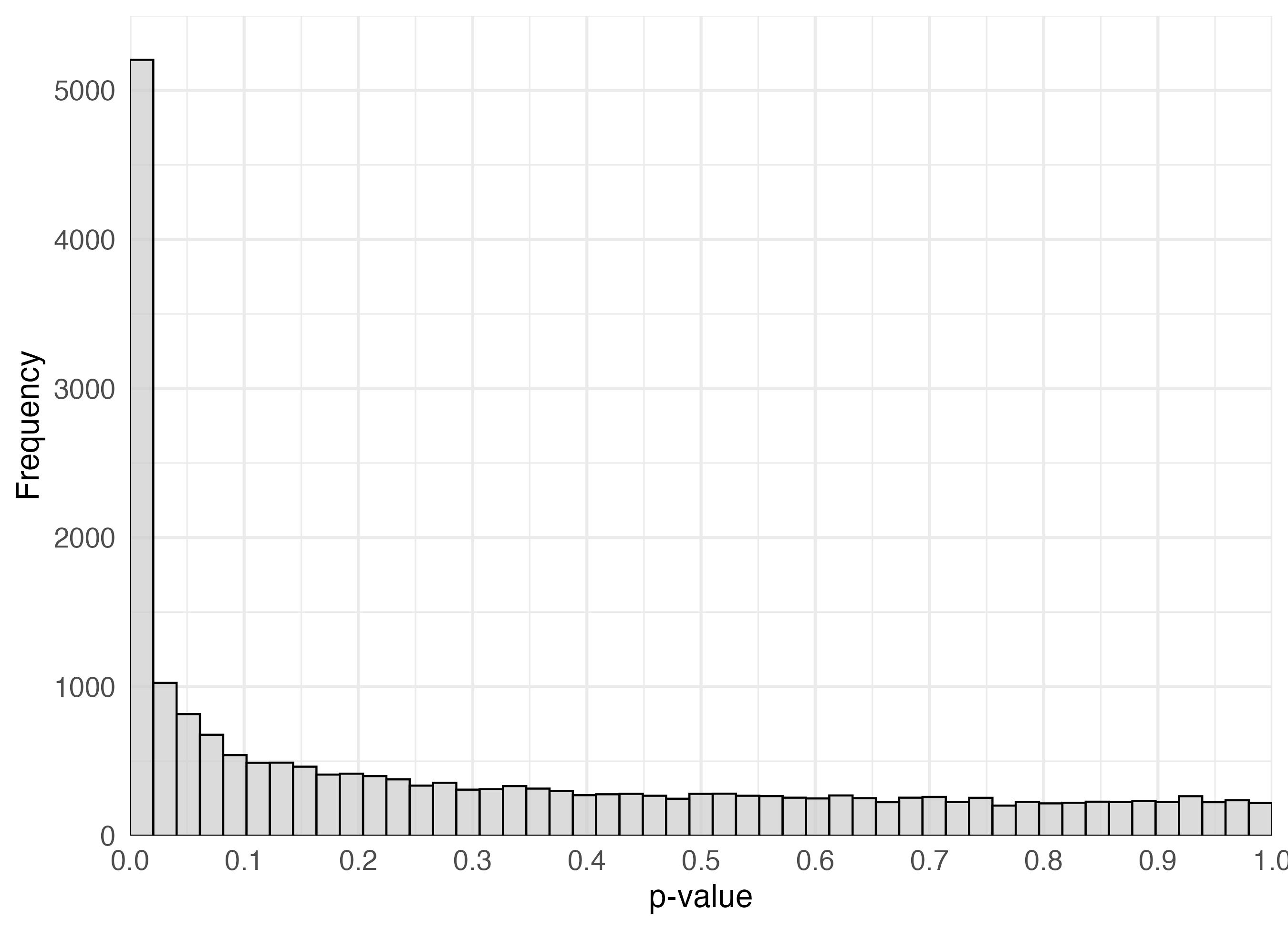}} \\
    \subcaptionbox{Welch\label{fig:GSE9101 Welch}}{\includegraphics[width=0.45\linewidth]{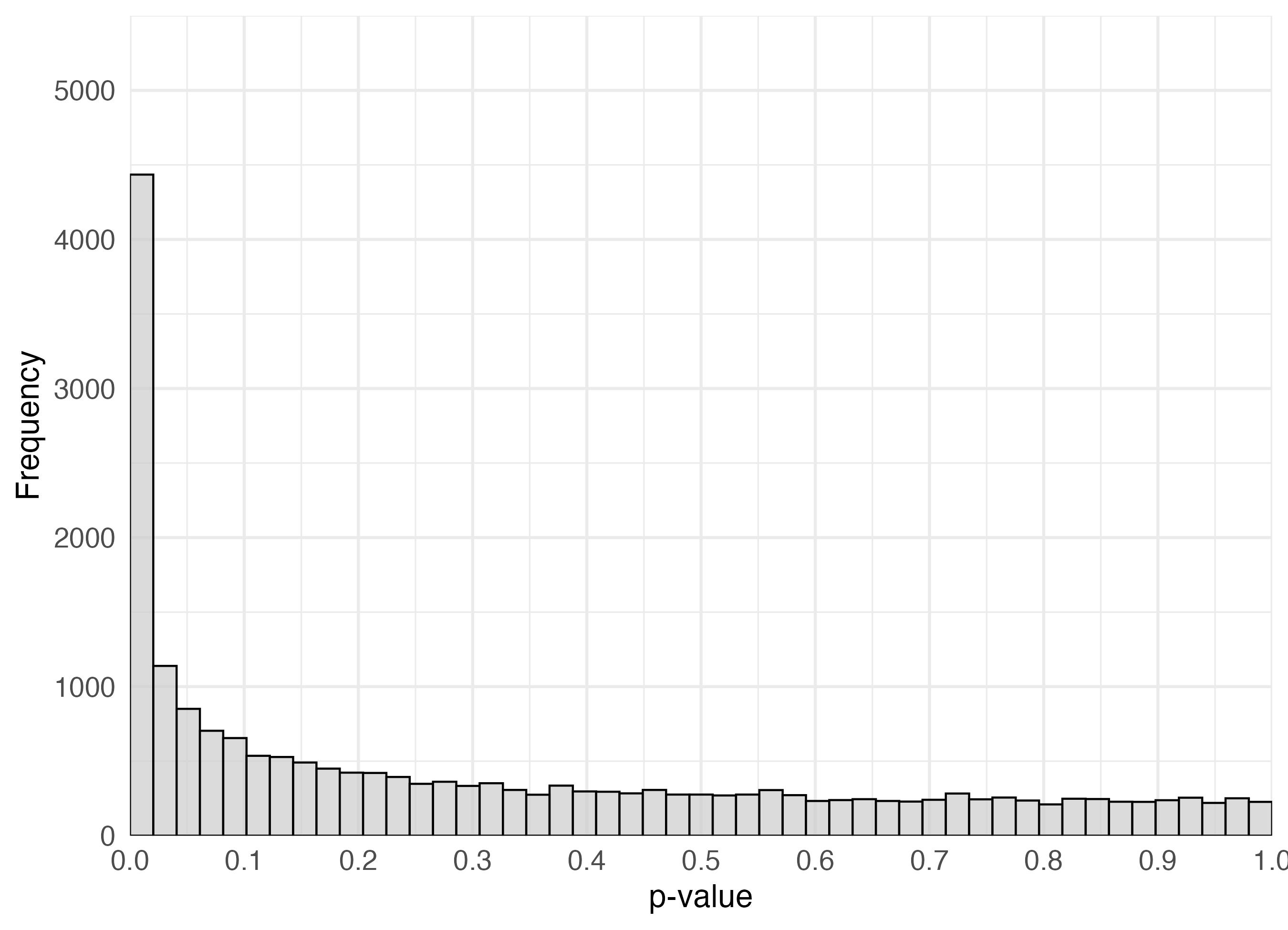}} &
    \subcaptionbox{EV-test\label{fig:GSE9101 EV}}{\includegraphics[width=0.45\linewidth]{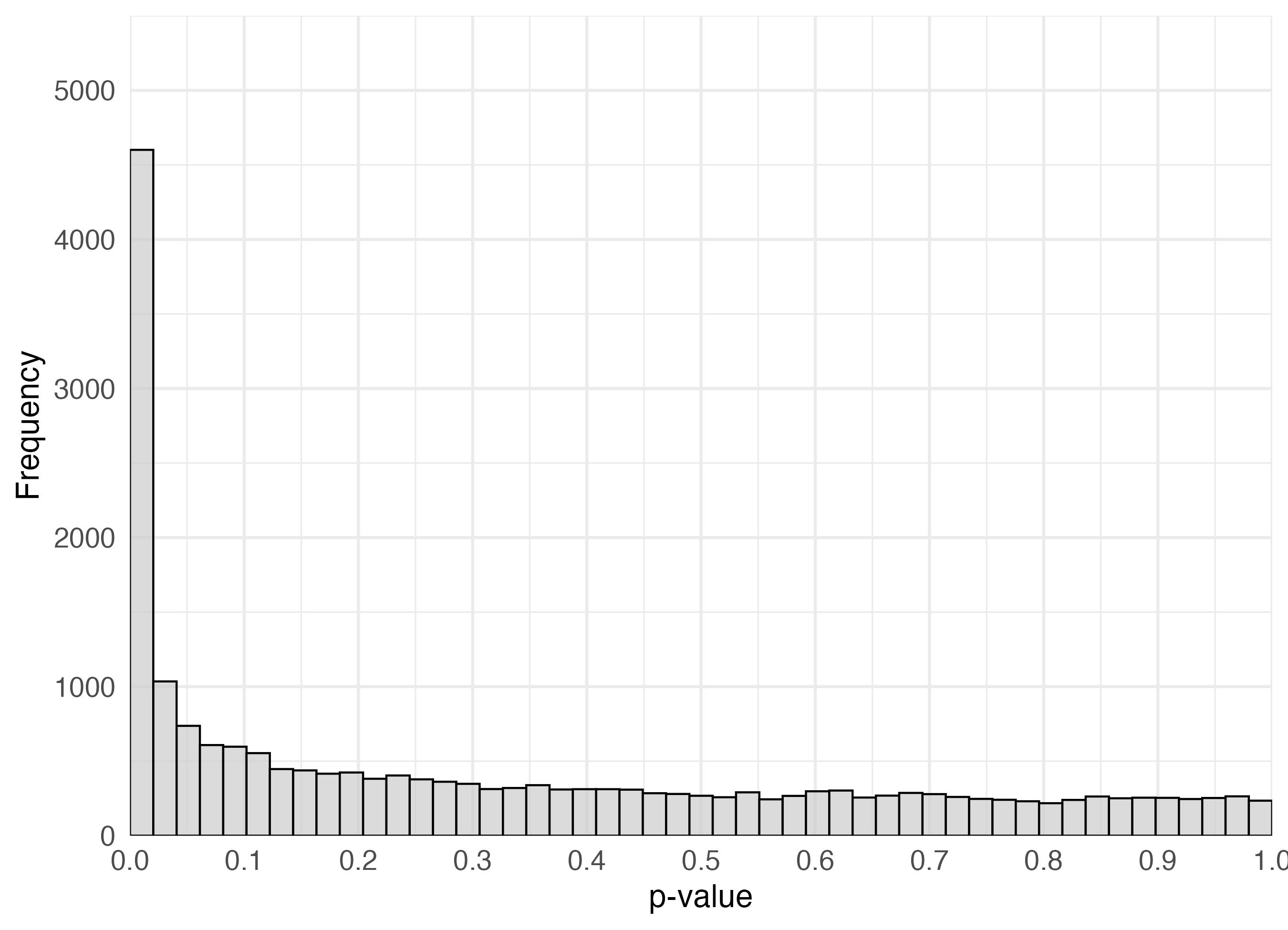}} \\
    \subcaptionbox{B-F\label{fig:GSE9101 BF}}{\includegraphics[width=0.45\linewidth]{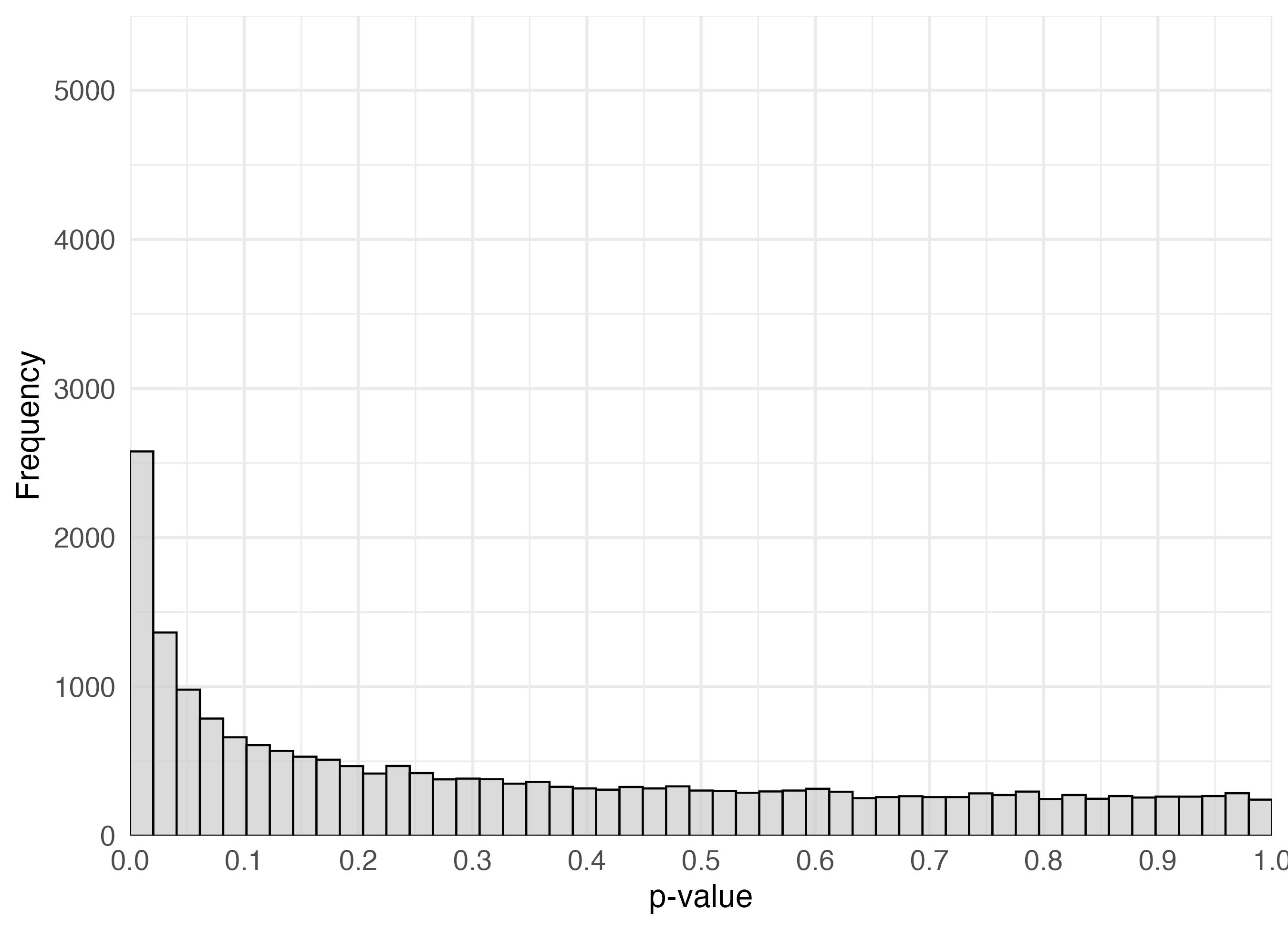}} &
    \multicolumn{1}{c}{}
  \end{tabular}

  \caption{P-value histograms for VREPB, DVEPB, Welch, EV-test, and B-F based on the macrophage dataset.}
  \label{fig:p-values GSE9101}
\end{figure}

\begin{figure}
  \centering
  \setlength{\tabcolsep}{4pt}
  \renewcommand{\arraystretch}{0}

  \begin{tabular}{@{}c c @{}}
    \subcaptionbox{VREPB with Voom\label{fig:PBMC1 VREPB voom}}{\includegraphics[width=0.45\linewidth]{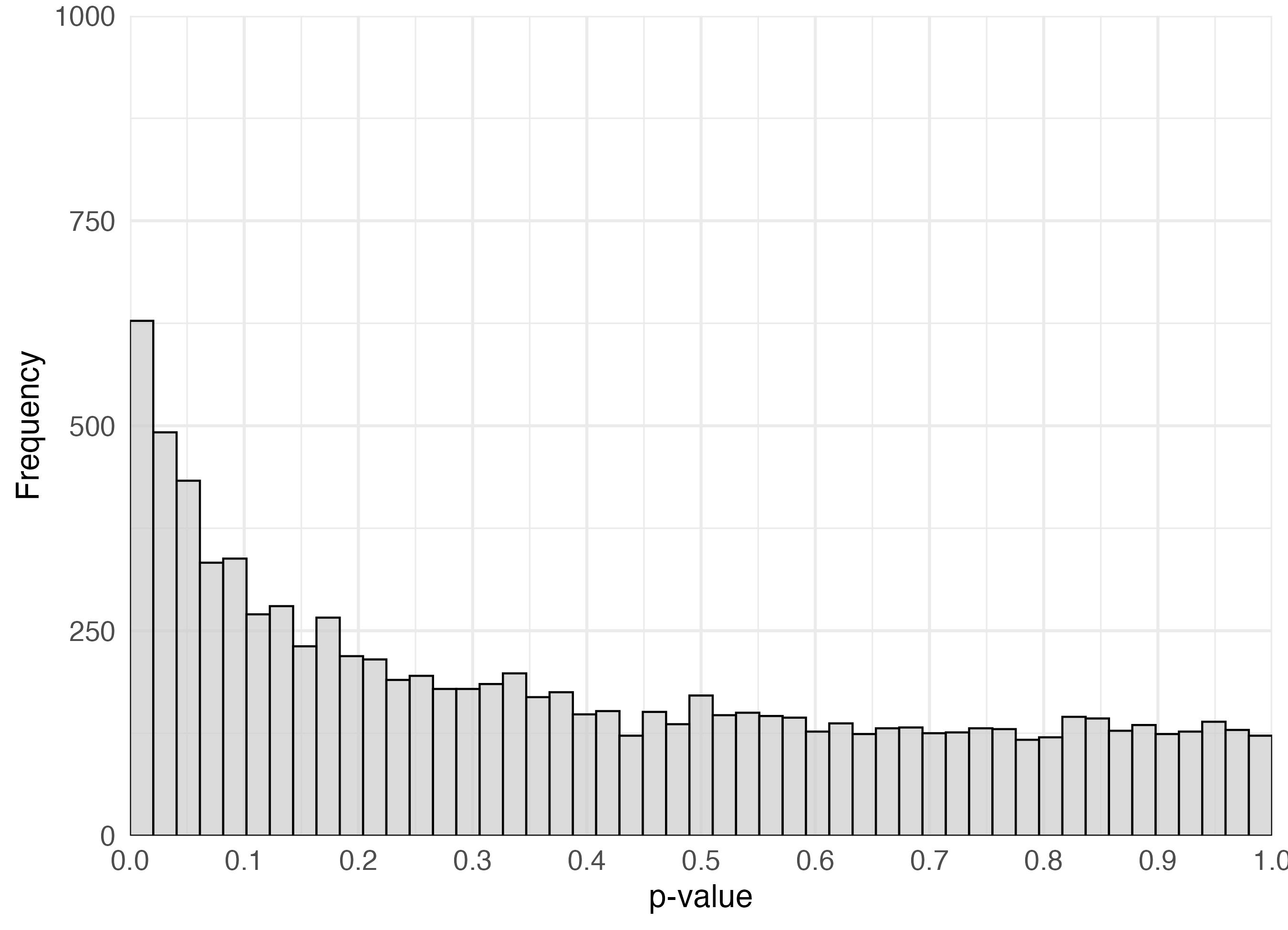}} &
    \subcaptionbox{DVEPB with Voom\label{fig:PBMC1 DVEPB voom}}{\includegraphics[width=0.45\linewidth]{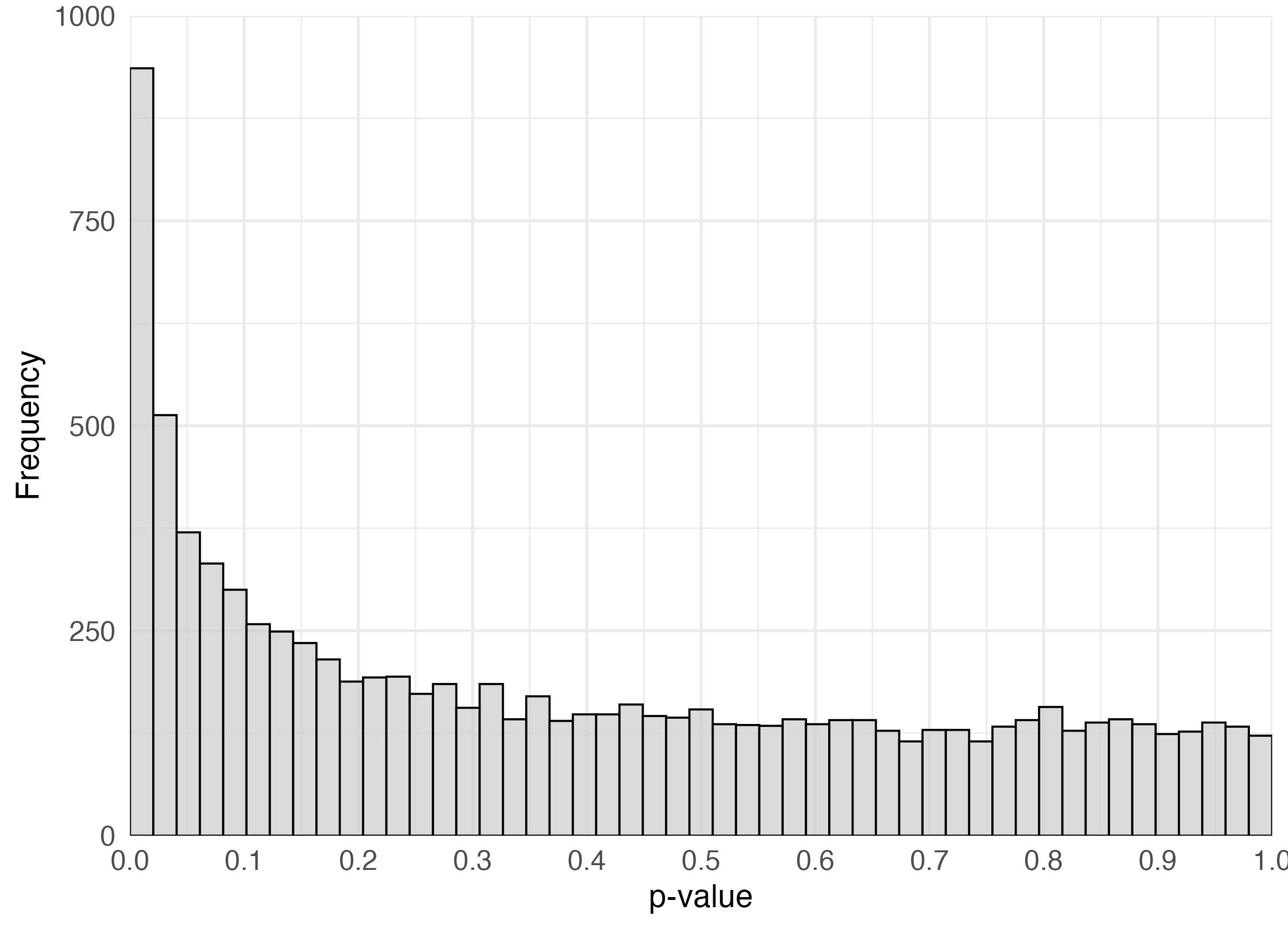}} \\
    \subcaptionbox{Welch with Voom\label{fig:PBMC1 Welch voom}}{\includegraphics[width=0.45\linewidth]{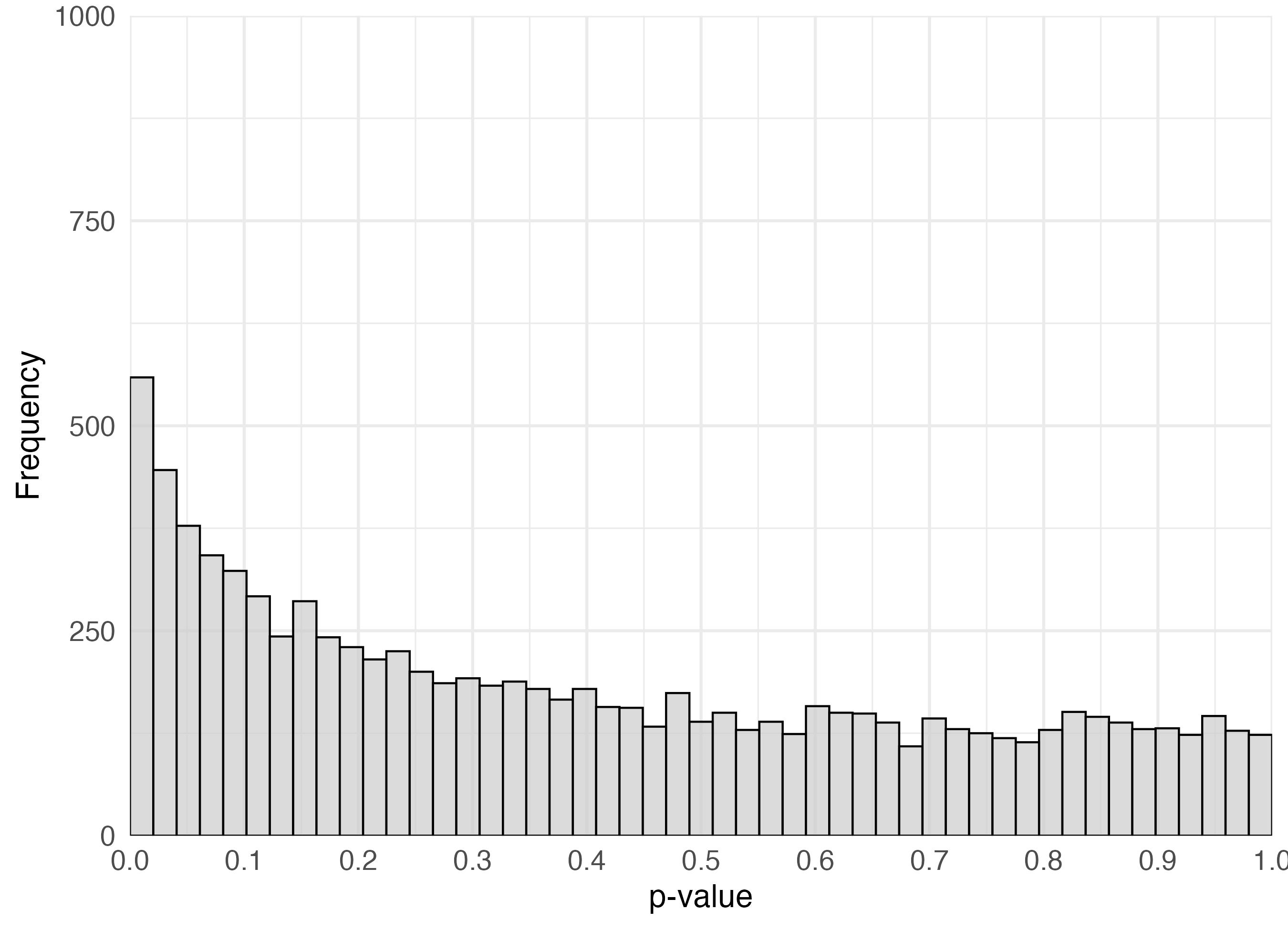}}&
    \subcaptionbox{EV-test with Voom \label{fig:PBMC1 EV voom}}{\includegraphics[width=0.45\linewidth]{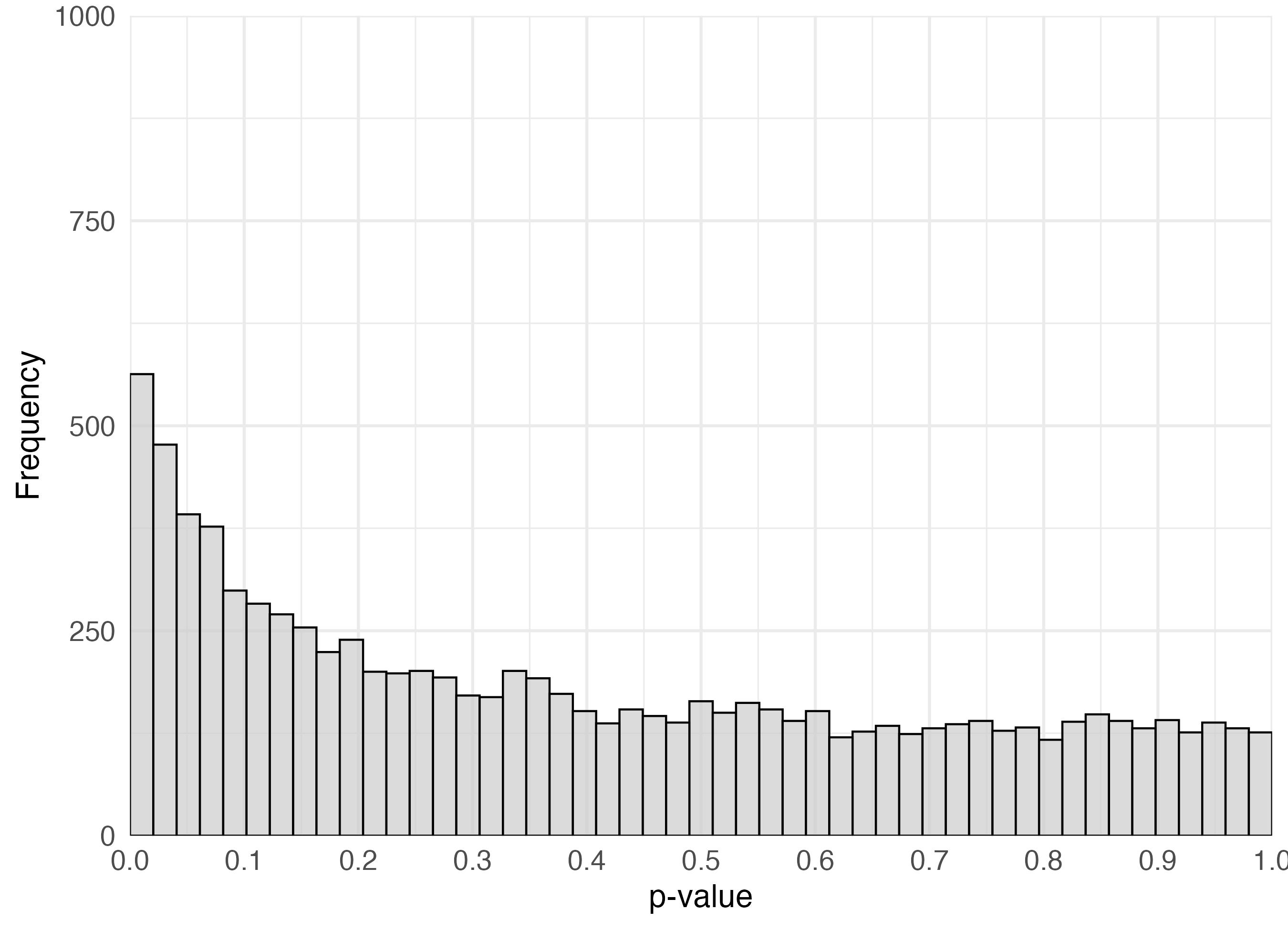}}\\
    \subcaptionbox{B-F with Voom \label{fig:PBMC1 BF voom}}{\includegraphics[width=0.45\linewidth]{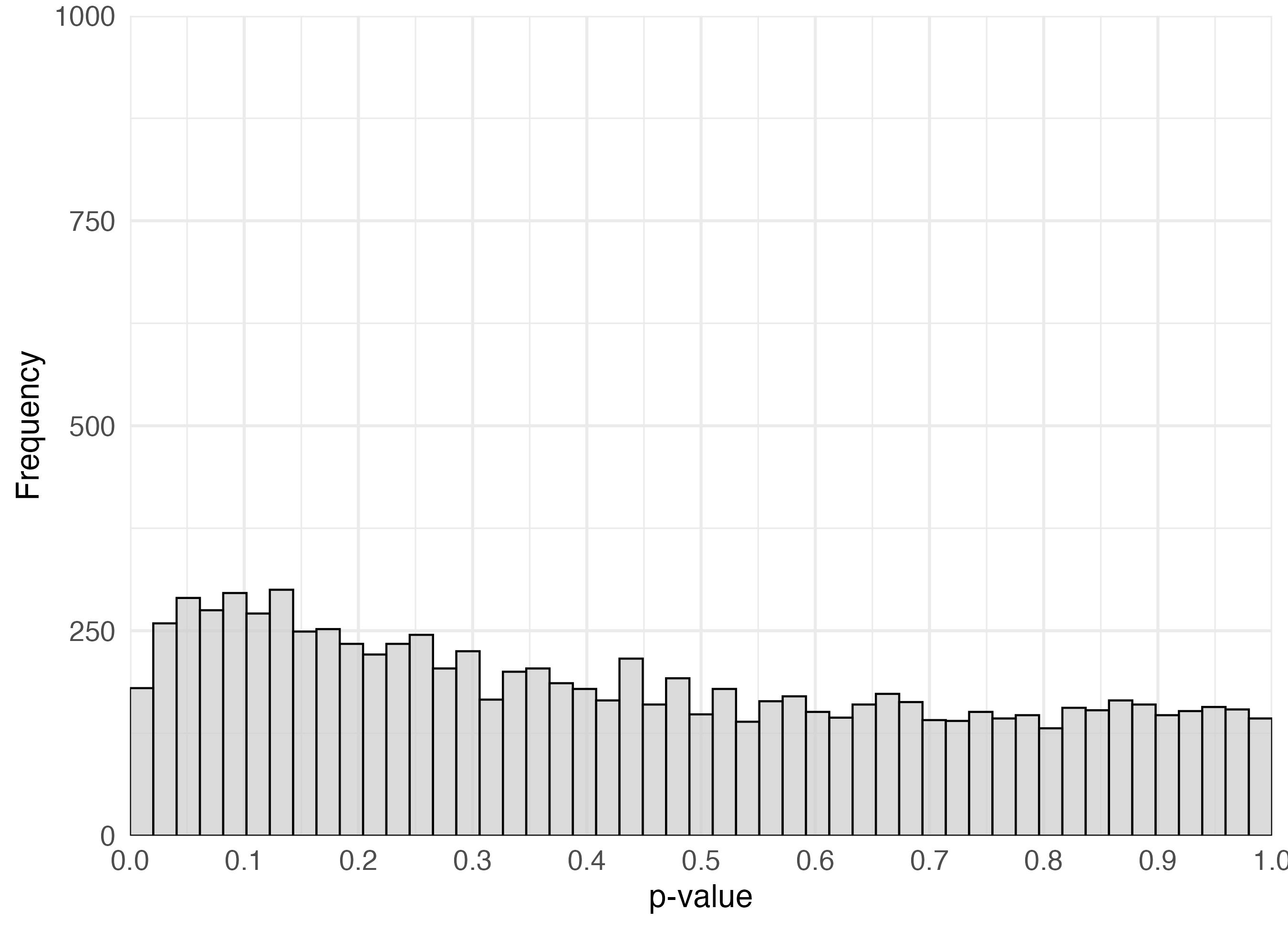}}&
    \multicolumn{1}{c}{}
  \end{tabular}

  \caption{P-value histograms for VREPB, DVEPB, Welch, EV-test, B-F with Voom weights based on the COVID-19 peripheral blood mononuclear cell dataset.}
  \label{fig:p-values PBMC1 voom}
\end{figure}

\begin{figure}
  \centering
  \setlength{\tabcolsep}{4pt}
  \renewcommand{\arraystretch}{0}

  \begin{tabular}{@{}c c @{}}
    \subcaptionbox{VREPB with VoomByGroup \label{fig:PBMC1 VREPB vbg}}{\includegraphics[width=0.45\linewidth]{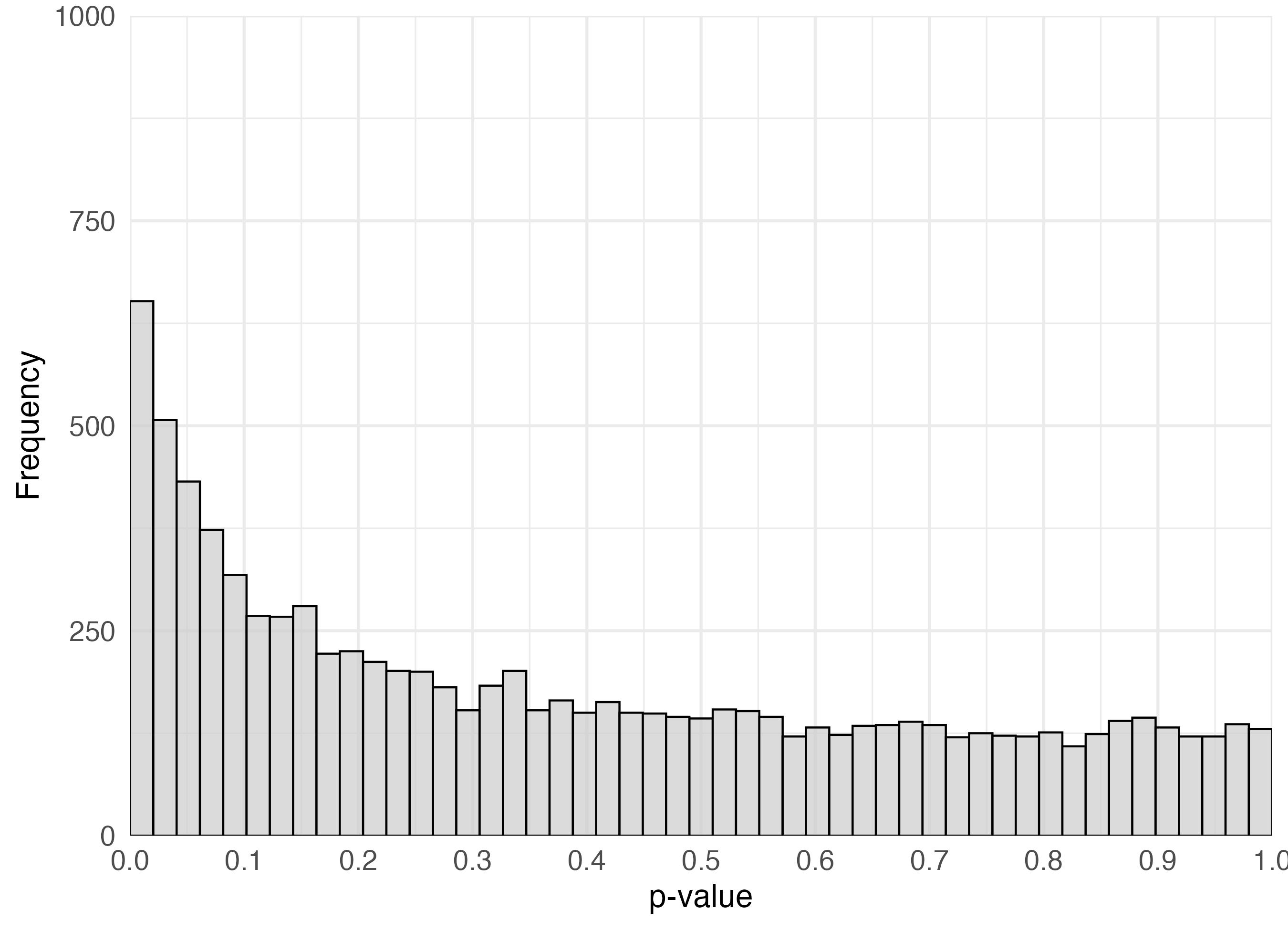}} &
    \subcaptionbox{DVEPB with VoomByGroup \label{fig:PBMC1 DVEPB vbg}}{\includegraphics[width=0.45\linewidth]{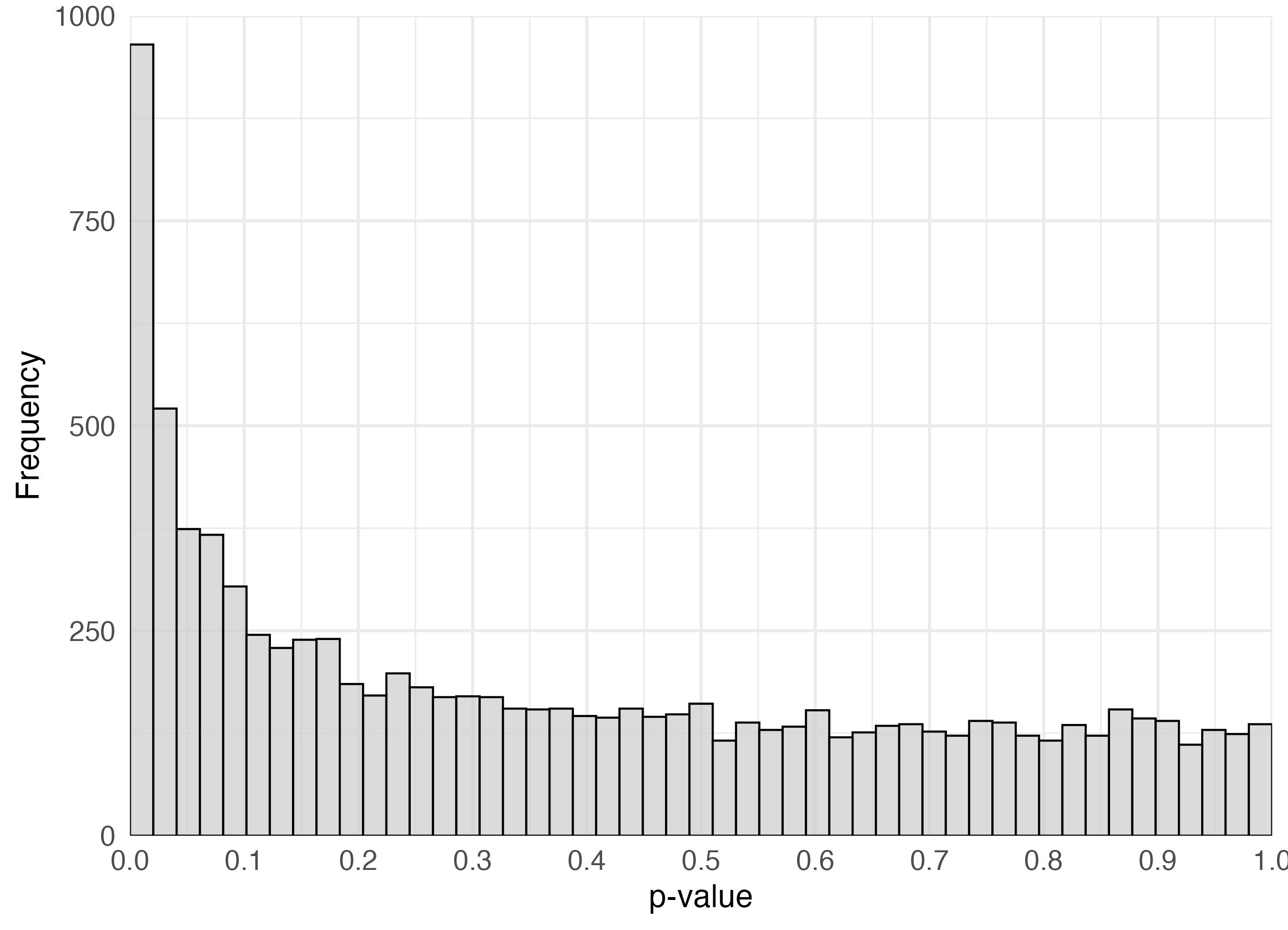}} \\
    \subcaptionbox{Welch with VoomByGroup\label{fig:PBMC1 Welch vbg}}{\includegraphics[width=0.45\linewidth]{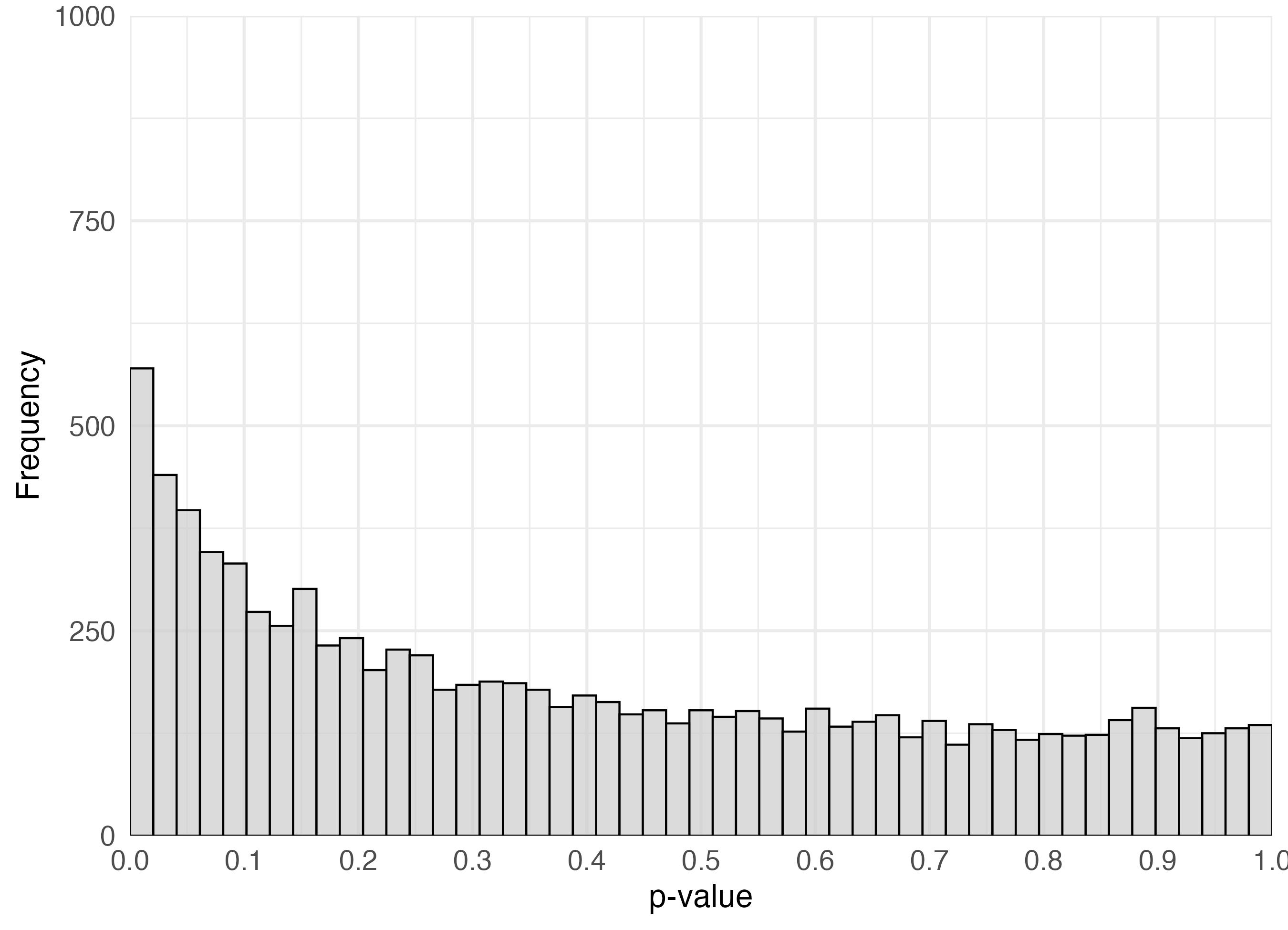}} &
    \subcaptionbox{EV-test with VoomByGroup\label{fig:PBMC1 EV vbg}}{\includegraphics[width=0.45\linewidth]{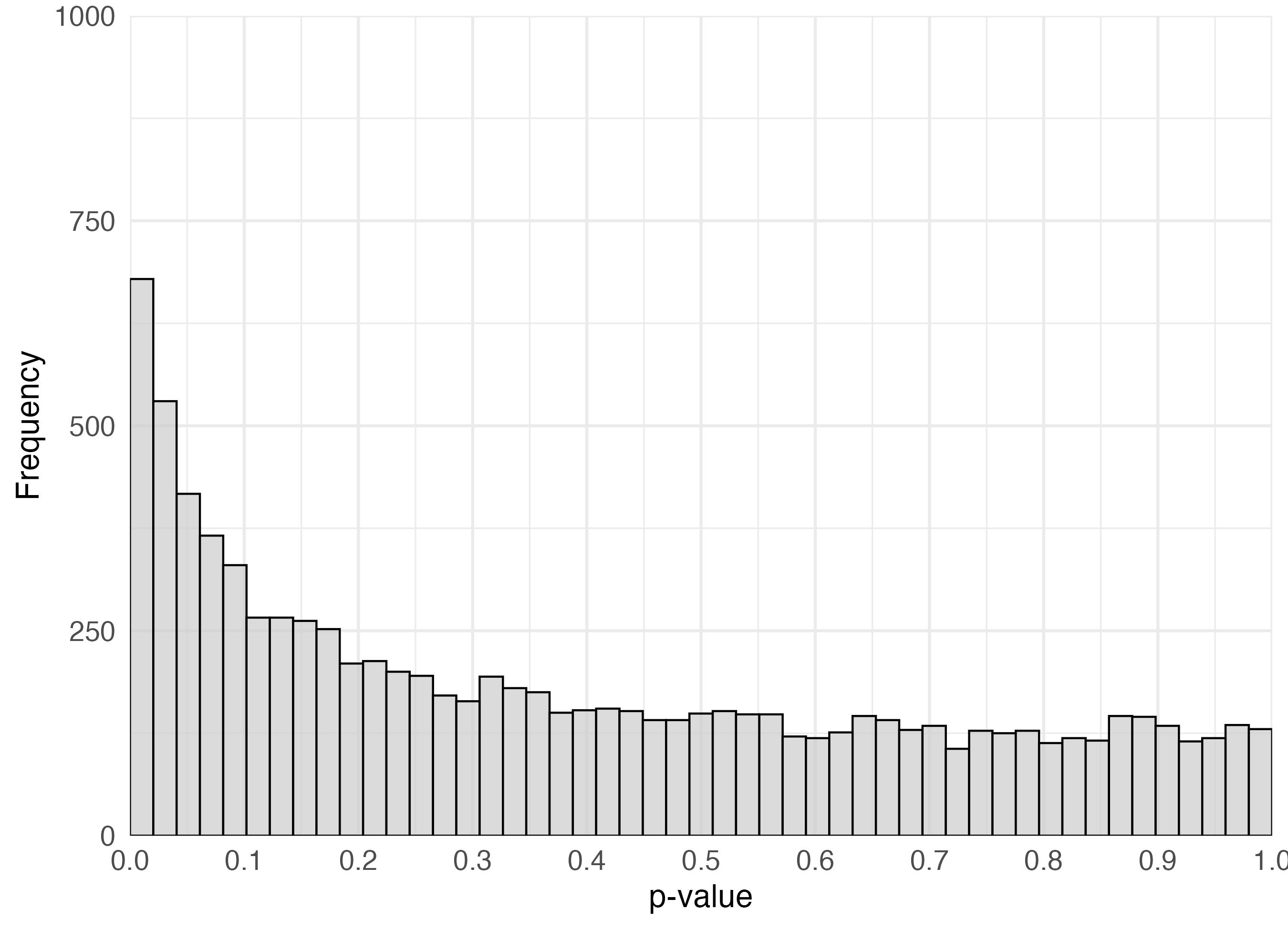}} \\
    \subcaptionbox{B-F with VoomByGroup\label{fig:PBMC1 BF vbg}}{\includegraphics[width=0.45\linewidth]{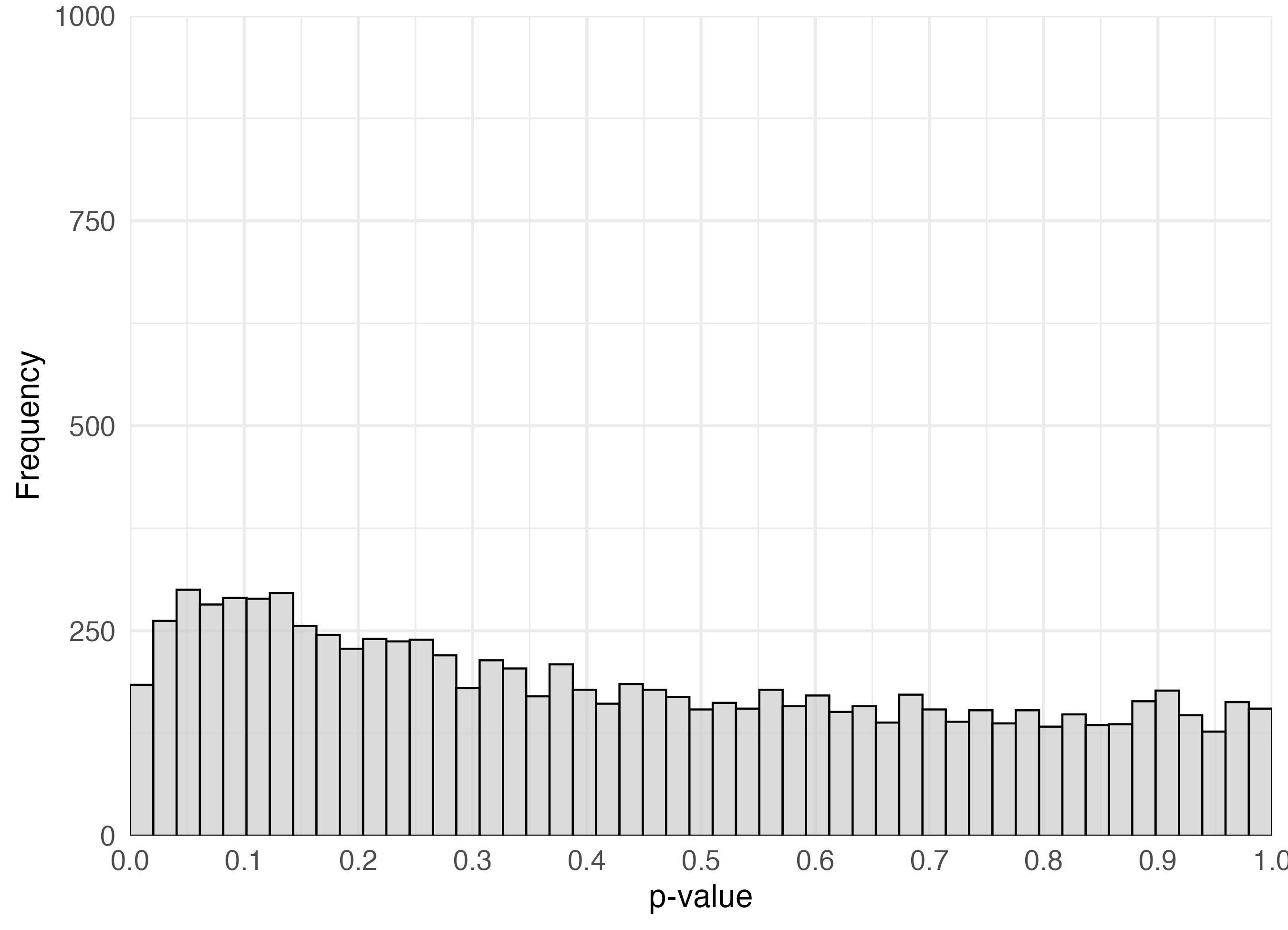}} &
    \multicolumn{1}{c}{}
  \end{tabular}

  \caption{P-value histograms for VREPB, DVEPB, Welch, EV-test, B-F with VoombyGroup weights based on the COVID-19 peripheral blood mononuclear cell dataset}
  \label{fig:p-values PBMC1 vbg}
\end{figure}

\section{Algorithms}
\label{appendix: algorithm}

\begin{algorithm} 
    \SetAlgoNoLine
    \caption{Variance Ratio Empirical Partially Bayes (VREPB in Section \ref{sec: Variance ratios as nuisance parameters})}
    \textbf{Input:} Pairs $(T_i^{\text{BF}}, \hat{\VRatio}_i),\;\; i = 1, \dots, n$. \\
    Let $\hat{G}$ be the NPMLE \eqref{eq:1D_optimization} of G \eqref{eq:varation_dbn} based on $(\SRatio_1, \dots, \SRatio_n)$.\\
    Let $\PVR_i = \PVRFunc(\Tbf_i, \SRatio_i; \hat{G})$, where $\PVRFunc(t, l; G)$ is defined in \eqref{eq:oracle_1D_p_value}.\\
    (Optional.) Apply the BH procedure with p-values $\PVR_1, \dots, \PVR_n$ at level $\alpha$.
    \label{algo:vrepb}
\end{algorithm}

\begin{algorithm} 
    \SetAlgoNoLine
    \caption{Dual Variance Empirical Partially Bayes (DVEPB in Section \ref{sec: Groupwise variances as nuisance parameters})}
    \textbf{Input:} Triples $(T_i^{\text{BF}}, \hat{\sigma}^2_{iA}, \hat{\sigma}^2_{iB}),\;\; i = 1, \dots, n$. \\
    Let $\hat{H}$ be the NPMLE \eqref{eq:2D_optimization} of $H$ \eqref{eq:sigma_prior} based on $((\hat{\sigma}^2_{iA}, \hat{\sigma}^2_{iB}), i = 1, \dots, n)$.\\
    Let $\PDV_i = \PDVFunc(\Tbf_i, \hat{\sigma}^2_{iA}, \hat{\sigma}^2_{iB}; \hat{H})$, where $\PDVFunc(t, s^2_{iA}, s^2_{iB}; H)$ is defined in \eqref{eq:oracle_2D_p_value}.\\
    (Optional.)  Apply the BH procedure with p-values $\PDV_1, \dots, \PDV_n$ at level $\alpha$.
    \label{algo:dvepb}
\end{algorithm}

\section{Microarray preprocessing}
\label{sec:preprocessing}

We downloaded the raw data from the Gene Expression Omnibus (GEO, \url{http://www.ncbi.nlm.nih.gov/geo/}) with accession number \texttt{GSE9101}.
Starting from the raw Affymetrix CEL files, we computed RMA expression measures \citep{irizarry2003exploration} to obtain a probe-level $\log_2$ matrix. From this matrix, we followed the guidelines of \citep{klaus2018end} and (i) removed Affymetrix QC control probes, (ii) filtered low-expression probes by retaining those expressed above a global 5th-percentile intensity threshold in at least half the samples $(\geq 6 / 12)$, (iii) aggregated probe sets to gene symbols using the GPL570 annotation and collapsed multiple probes per gene by their mean to produce a gene-level expression matrix, and (iv) excluded genes with zero variance in either group.

\section{Supplementary formulas for the main text}

\label{appendix: supplementary formulas}

In Section \ref{subsec:ttest_known_variance_ratio}, we noted that the independence of $\SRatio_i$ and $\Tpool_i(\VRatio_i)$ allows us to bridge between $\Tpool_i(\VRatio_i)$ and $\Tbf_i$ through a function $\phi$. Citing from~\citet{sprott1993difference}, we provide the explicit form of $\phi$ below,
\begin{equation}
    \label{eq:phi}
    \phi(\VRatio_i, \SRatio_i, K_A, K_B) = \sqrt{\frac{\nu_A \frac{c_i}{\gamma_i} + \nu_B \frac{1-c_i}{1-\gamma_i}}{\nu_A+\nu_B}} \text{ with } c_i = \frac{\SRatio_i}{\SRatio_i + \frac{K_A}{K_B}}, \;\; \gamma_i = \frac{\VRatio_i}{\VRatio_i + \frac{K_A}{K_B}}.
\end{equation}
Here we also provide the precise expression of  $p(\hat{\sigma}_{iA}^2, \hat{\sigma}_{iB}^2 \; \cond \; \sigma^2_{iA}, \sigma^2_{iB}, \nu_A, \nu_B)$. Due to independence between the two groups, we can decompose the above expression as $p(\hat{\sigma}_{iA}^2 \; \cond \; \sigma^2_{iA}, \nu_A) \times p(\hat{\sigma}_{iB}^2 \; \cond \; \sigma^2_{iB}, \nu_B)$, where under model \eqref{eq:ss_distribution}, each part of the likelihood is equal to
\begin{equation}
\label{eq:explicit_s2_given_sigma}
p(s^2 \; \cond \; \sigma^2, \nu) = \frac{\nu}{\sigma^2} \frac{1}{2^{\nu/2} \Gamma(\nu/2)} \left(\frac{\nu s^2}{\sigma^2}\right)^{\nu/2-1} \exp\left(-\frac{1}{2} \frac{\nu s^2}{\sigma^2}\right),
\end{equation}
and $\Gamma(\cdot)$ is the Gamma function.

\section{Statistical setting with precision weights}

As discussed in Section \ref{sec: Single-cell RNA-seq Pseudo-bulk Data}, the analysis workflow for RNA-seq data accompanies each transformed observation with a precision weight~\citep{law2014voom}. The incorporation of additional information carried by weights calls for a further study beyond the statistical setting in Section \ref{sec:statiatical setting}. We will start by studying the properties of the weighted normal distribution. We note that the following are textbook results on generalized least squares, but we provide them below for self-containedness.

\subsection{Properties of the weighted normal distribution}

We observe $n \in \mathbb N$ iid observations, \begin{equation}
    \label{eq:weighted_normal}
    X_i \overset{\text{ind}}{\sim} \mathrm{N}\left(\mu, \frac{\sigma^2}{w_i}\right), \text{ for }i = 1, \dots, n, 
\end{equation}
from a unit where $\mu, \sigma^2$ are unknown parameters, $w_i$'s are known precision weights. In the following proposition, we show that all observations for the unit may be collapsed to their complete and sufficient statistics, which serve as the guideline for further analysis.
\begin{prop}
\label{prop:weighted_sufficiency}
    $(\sum_{i=1}^n w_i X_i, \sum_{i=1}^n w_i X_i^2)$ is the sufficient statistic for $(\mu,\sigma^2)$.
\end{prop}

\begin{proof}
Due to independence between each observation, we have that
\begin{align*}
    &\PP[]{X_1, \dots, X_n \; \cond \; \mu, \sigma^2, w_1, \dots, w_n} = \prod_{i = 1}^{n} \PP[]{X_i \; \cond \; \mu, \sigma^2, w_i} \\
    &\quad\quad = \prod_{i=1}^n \frac{1}{\sqrt{2 \pi \sigma^2 / w_i}} \exp\left(-\frac{1}{2} \frac{(X_i - \mu)^2}{\sigma^2 / w_i}\right) \\
    &\quad\quad = (2\pi)^{-\frac{n}{2}} \left(\prod_{i=1}^n w_i\right)^{\frac{1}{2}} \text{exp}\left(\frac{\mu}{\sigma^2} \sum_{i=1}^n w_i X_i -\frac{1}{2\sigma^2} \sum_{i=1}^n w_i X_i^2  - \frac{n}{2} \text{log}(\sigma^2) - \frac{\mu}{\sigma^2} \sum_{i=1}^n w_i\right).
\end{align*}
Based on the Factorization Theorem, $(\sum_{i=1}^n w_i X_i, \sum_{i=1}^n w_i X_i^2)$ is the sufficient statistic for $(\mu,\sigma^2)$.
\end{proof}

Following Proposition \ref{prop:weighted_sufficiency}, we can focus on studying the complete and sufficient statistic $(\bar{X}, S^2)$ defined as
$$\bar{X} = \frac{\sum_{i=1}^n w_i X_i}{\sum_{i=1}^n w_i}, \quad S^2 = \frac{1}{\nu}\sum_{i=1}^n w_i(X_i - \bar{X})^2,$$
where $\nu = n-1$. Then the following proposition reveals the distribution of $(\bar{X}, S^2)$, and establishes the independence between $\bar{X}$ and $S^2$.
\begin{prop}
\label{prop:weighted_normal_ss}
    $\bar{X} \sim \mathrm{N}\left(\mu, \frac{\sigma^2}{\sum_{i=1}^n w_i}\right),\; S^2 \sim \frac{\sigma^2}{\nu}\chi^2_{\nu}$, and $\bar{X}$ is independent of $S^2$.
\end{prop}

\begin{proof}

First, we focus on $\bar{X}$. Since $X_i$'s follow the normal distribution, $\bar{X}$ also follows a normal distribution. Thus, we need the expectation and the variance of $\bar{X}$ to determine its distribution.
$$\EE{\bar{X}} = \EE{\frac{\sum_{i=1}^n w_i X_i}{\sum_{i=1}^n w_i}} = \frac{\sum_{i=1}^n w_i \EE{X_i}}{\sum_{i=1}^n w_i} = \mu,$$
\begin{align*}
    \text{Var}(\bar{X}) = \text{Var}\left(\frac{\sum_{i=1}^n w_i X_i}{\sum_{i=1}^n w_i}\right) 
    = \frac{\sum_{i=1}^n w_i^2 \text{Var}(X_i) + 2\sum_{i=1}^n \sum_{j \neq i} w_iw_j\text{Cov}(X_i, X_j)}{(\sum_{i=1}^n w_i)^2} = \frac{\sigma^2}{\sum_{i=1}^n w_i}.
\end{align*}
Therefore, we conclude that $\bar{X} \sim \mathrm{N}(\mu, \sigma^2/\{\sum_{i=1}^n w_i\})$.

Now, we will prove the independence between $\bar{X}$ and $S^2$. First, fix $\sigma^2$ and consider a family of distributions $X_i \sim \mathrm{N}(\mu, \sigma^2 / w_i), \text{for } i = 1, \dots, n$ with $\mu \in R$. It is easy to derive that $\bar{X}$ is a complete and sufficient statistic for $\mu$.

Next, we study $S^2$ by considering $Z_i = \sqrt{w_i} (X_i - \bar{X})$. Since $X_i$'s follow a normal distribution, and $\bar{X}$ is a linear combination of $X_i$, $Z_i$ also follows a normal distribution, whose mean and variance are as follows:
$$\EE{Z_i} = \EE{\sqrt{w_i} (X_i - \bar{X})} = \sqrt{w_i} (\EE{X_i} - \EE{\bar{X}}) = 0,$$
\begin{align*}
    \text{Var}(Z_i) &= \text{Var}(\sqrt{w_i} (X_i - \bar{X})) = w_i (\text{Var}(X_i) + \text{Var}(\bar{X}) - 2 \text{Cov}(X_i, \bar{X})) \\
    &= w_i \left(\frac{\sigma^2}{w_i} + \frac{\sigma^2}{\sum_{i=1}^n w_i} - 2 \text{Cov}\left(X_i, \frac{\sum_{i=1}^n w_iX_i}{\sum_{i=1}^n w_i}\right)\right) \\
    &= w_i \left(\frac{\sigma^2}{w_i} - \frac{\sigma^2}{\sum_{i=1}^n w_i}\right).
\end{align*}
Therefore, we derive that 
$$Z_i \sim \mathrm{N}\left(0, w_i \left(\frac{\sigma^2}{w_i} - \frac{\sigma^2}{\sum_{i=1}^n w_i}\right)\right),$$
and its distribution is independent of $\mu$. Thus, $S^2$ is ancillary to $\mu$. Since the choice of $\sigma^2$ is arbitrary, based on Basu \citep{basu1955statistics}, then $\bar{X}$ is independent of $S^2$.

The last step is to obtain the distribution of $S^2$. Notice that $\sqrt{w_i}(X_i - \mu) \big /\sigma \sim \mathrm{N}(0, 1), \text{for } i = 1, \dots, n$, we consider the following random variable:
\begin{align*}
    U &= \frac{\sum_{i=1}^n w_i (X_i - \mu)^2}{\sigma^2} \\
    &= \frac{\sum_{i=1}^n w_i (X_i - \bar{X} + \bar{X} -\mu)^2}{\sigma^2} \\
    &= \frac{\sum_{i=1}^n w_i (X_i - \bar{X})^2}{\sigma^2} + \frac{(\bar{X} - \mu)^2 \sum_{i=1}^n w_i}{\sigma^2} + \frac{2(\bar{X} - \mu) \sum_{i=1}^n w_i(X_i - \bar{X})}{\sigma^2} \\
    &= \frac{\nu S^2}{\sigma^2} + V,
\end{align*}
where $V = \sum_{i=1}^n w_i (\bar{X} - \mu)^2 \big / \sigma^2$. Notice that $U \sim \chi_{n}^2$ due to the independence between $X_i$ and $X_j$ for $i \neq j$, and $V \sim \chi_{1}^2$ due to the distribution of $\bar{X}$. We can study the moment generating function of $U$ for $t < \frac{1}{2}$:
\begin{align*}
    (1 - 2t)^{-\frac{n}{2}}=M_U(t) = \EE{\exp(tU)} = \EE{\exp\left(t\frac{\nu S^2}{\sigma^2}\right)} \EE{\exp(tV)} = \EE{\exp\left(t\frac{\nu S^2}{\sigma^2}\right)} (1 - 2t)^{-\frac{1}{2}} ,
\end{align*}
where the third equality follows from the independence between $\bar{X}$ and $S^2$. Therefore, the moment generating function of $\nu S^2 \big / \sigma^2$ is $(1 - 2t)^{-(n-1)/2}$, which is the moment generating function of $\chi_{\nu}^2$. Thus, we derive that $S^2 \sim \frac{\sigma^2}{\nu}\chi^2_{\nu}$.

\end{proof}

Proposition \ref{prop:weighted_normal_ss} implies that properties of the weighted normal distribution \eqref{eq:weighted_normal} shares a similar structure with that of the normal distribution, which indicates potential extensions of our VR and DV approaches by incorporating precision weights. Now we are ready to study the weighted version of the statistical model in Section \ref{sec:statiatical setting}.

\subsection{Statistical setting with precision weights}
\label{sec: weighted model setup}

We have data for $n \in \mathbb N$ independent units (e.g., transformed observations of genes). For the $i$-th unit, we observe $K_A \in \mathbb N$ independent observations $Z_{i1}, \dotsc, Z_{iK_A}$ from the first sample, and $K_B \in \mathbb N$ independent observations $Y_{i1}, \dotsc, Y_{iK_B}$ from the second sample, with the two samples independent of each other. Formally, the law of data generation of the $i$-th unit is:
\begin{equation}
\label{eq:weighted_model}
    Z_{ik} \overset{\text{ind}}{\sim} \mathrm{N}\left(\mu_{iA}, \frac{\sigma_{iA}^2}{w_{ik}^A}\right),\text{for } k = 1, \dots, K_A; \quad Y_{ik} \overset{\text{ind}}{\sim} \mathrm{N}\left(\mu_{iB}, \frac{\sigma_{iB}^2}{w_{ik}^B}\right), \text{for }k = 1, \dots, K_B,
\end{equation}
where $\mu_{iA}, \mu_{iB} \in \RR$, $\sigma_{iA}^2, \sigma_{iB}^2 >0$ are unknown parameters, and $(w_{ik}^A, \text{for } k = 1, \dots, K_A; \; w_{ik}^B, \\ \text{for } k = 1, \dots K_B)$ are known precision weights. Here we emphasize that data within each sample are not identically distributed due to the additional weight information. But based on Proposition \ref{prop:weighted_sufficiency}, all observations for the $i$-th unit can still be collapsed to their complete and sufficient statistics: the weighted sample average and weighted sample variance of $Z_{i1},\ldots,Z_{iK_A}$ denoted as $\ZMean$ and $\hat{\sigma}_{iA}^2$, and their counterparts based on $Y_{i1},\ldots, Y_{iK_A}$ denoted as $\YMean$ and $\hat{\sigma}_{iB}^2$.
\begin{align}
\label{eq:weighted_ss_calculation}
   \ZMean &= \frac{\sum_{k=1}^{K_A} w_{ik}^A Z_{ik}}{\sum_{k=1}^{K_A} w_{ik}^A}, \quad\hat{\sigma}_{iA}^2 = \frac{1}{\nu_A } \sum_{k=1}^{K_A} w_{ik}^A (Z_{ik} - \ZMean)^2, \nonumber \\
   \YMean &= \frac{\sum_{k=1}^{K_B} w_{ik}^B Y_{ik}}{\sum_{k=1}^{K_B} w_{ik}^B}, \quad \hat{\sigma}_{iB}^2 = \frac{1}{\nu_B} \sum_{k=1}^{K_B} w_{ik}^B(Y_{ik} - \YMean)^2,
\end{align}
where $\nu_A = K_A - 1, \nu_B = K_B - 1$. Based on Proposition \ref{prop:weighted_normal_ss}, the sufficient statistics are distributed as follows,
\begin{equation}
\label{eq:weighted_ss_distribution}
    (\ZMean, \hat{\sigma}_{iA}^2, \YMean, \hat{\sigma}_{iB}^2) \overset{\text{ind}}{\sim} \mathrm{N}\left(\mu_{iA}, \frac{\sigma_{iA}^2}{n_{iA}}\right) \otimes \frac{\sigma_{iA}^2}{\nu_A} \chi_{\nu_A}^2 \otimes \mathrm{N}\left(\mu_{iB}, \frac{\sigma_{iB}^2}{n_{iB}}\right) \otimes \frac{\sigma_{iB}^2}{\nu_B} \chi_{\nu_B}^2,
\end{equation}
where $\nu_A = K_A - 1, \nu_B = K_B - 1$, $n_{iA} = \sum_{k=1}^{K_A} w_{ik}^A, n_{iB} = \sum_{k=1}^{K_B} w_{ik}^B$, $\chi^2_{\nu_A}$ (resp. $\chi^2_{\nu_B}$) denotes the chi-squared distribution with $\nu_A$ (resp. $\nu_B$) degrees of freedom and we use $\otimes$ to denote product measures. Our goal is to test for equality of means for each $i$, that is, to test the null hypothesis $H_i: \mu_{iA} = \mu_{iB}$ based on $(\ZMean, \hat{\sigma}_{iA}^2, \YMean, \hat{\sigma}_{iB}^2)$ drawn from the model \eqref{eq:weighted_ss_distribution} for $i = 1, \dots, n$. We think of $\mu_{iA}, \mu_{iB}$ as the primary parameters, while $\sigma_{iA}^2,\sigma_{iB}^2$ are nuisance parameters. For convenience,  we also define the variance ratio $\VRatio_i$ and the sample variance ratio $\SRatio_i$ as \eqref{eq:lambda}.

\subsection{Proposed methods with weights}

\label{sec: weighted methods}

For the above weighted statistical setting, the sufficient statistics \eqref{eq:weighted_ss_calculation} and their laws \eqref{eq:weighted_ss_distribution} are of the similar structure to those of the statistical setting studied in Section \ref{sec:statiatical setting}. This can be expected, as the statistical setting in Section \ref{sec:statiatical setting} can be regarded as a special case of the weighted statistical setting in Supplement \ref{sec: weighted model setup}, where all precision weights equal one. Such an understanding led us to focus on the weighted Behrens-Fisher statistics,
\begin{equation}
\label{eq:WBF}
\Twbf_i := \frac{\ZMean -\YMean }{\cb{\p{\hat{\sigma}_{iA}^2/n_{iA}} + \p{\hat{\sigma}_{iB}^2/n_{iB}}}^{1/2}},\;\;\;\; \TwbfNull_i := \frac{(\ZMean -\YMean) - (\mu_{iA}-\mu_{iB})}{\cb{\p{\hat{\sigma}_{iA}^2/n_{iA}} + \p{\hat{\sigma}_{iB}^2/n_{iB}}}^{1/2}}.
\end{equation}
$\TwbfNull_i$ is the null-centered counterpart of $\Twbf_i$. In the following sections, we will present some properties of $\Twbf_i$ under model \eqref{eq:weighted_model} that are analogous to those of $\Tbf_i$ under \eqref{eq:ss_distribution}. Such analogues facilitate the extension of VREPB, DVEPB, Welch and EV-test to incorporate precision weights.

\subsubsection{Variance ratio empirical partially Bayes with weights}

\label{sec:VREPB weighted}

For this part, we will not start from the study of the pooled t-statistic as we did in Section \ref{subsec:ttest_known_variance_ratio} to avoid redundancy. Instead, we will study the distribution of the $\Twbf_i$ by providing a result that is analogous to the connection between $\Tpool_i$ and $\Tbf_i$ under \eqref{eq:ss_distribution} from \citep{sprott1993difference}.

For the $i$-th unit, under model \eqref{eq:weighted_ss_distribution}, based on Proposition \ref{prop:weighted_normal_ss}, we have 
$$Z_i = \frac{(\ZMean - \YMean) - (\mu_{iA} - \mu_{iB})}{\sqrt{\frac{\sigma^2_{iA}}{n_{iA}} + \frac{\sigma^2_{iB}}{n_{iB}}}} \sim \mathrm{N}\p{0, 1} , \quad V_i = \frac{\nu_{A} \hat{\sigma}_{iA}^2}{\sigma^2_{iA}} + \frac{\nu_{B} \hat{\sigma}_{iB}^2}{\sigma^2_{iB}} \sim \chi_{\nu_A + \nu_B}^2,$$
and that $Z_i$ is independent of $V_i$, which leads to $Z_i/\sqrt{\frac{V_i}{\nu_A + \nu_B}} \sim t_{\nu_A + \nu_B}.$
Combining the above with the fact that under the null $H_i: \mu_{iA}=\mu_{iB}$,
$$\Twbf_i = Z_i \cdot \sqrt{\left(\frac{\sigma^2_{iA}}{n_{iA}} + \frac{\sigma^2_{iB}}{n_{iB}}\right) \Big / \left(\frac{\hat{\sigma}_{iA}^2}{n_{iA}} + \frac{\hat{\sigma}_{iB}^2}{n_{iB}}\right)},$$
we can derive the following proposition:
\begin{prop}
\label{prop:weighted_BF_distribution}
    Under the model \eqref{eq:weighted_ss_distribution} and the null hypothesis $H_i: \mu_{iA} = \mu_{iB}$,
    $$\Twbf_i \equiv \Twbf_i(\VRatio_i) \sim \phi^{\text{W}}(\VRatio_i, \SRatio_i, K_A, K_B, n_{iA}, n_{iB}) \cdot t_{\nu_A+\nu_B},$$
    where
    $t_{\nu_A+\nu_B}$ is the t-distribution with $\nu_A + \nu_B$ degrees of freedom.
\end{prop}
Notice that $\phi^{\text{W}}(\VRatio_i, \SRatio_i, K_A, K_B, n_{iA}, n_{iB})$ is the analogue of $\phi(\VRatio_i, \SRatio_i, K_A, K_B)$ under the weighted statistical setting \eqref{eq:weighted_model}, and the explicit form is:

\begin{equation}
    \label{eq:weighted phi}
    {\phi^{\text{W}}}(\VRatio_i, \SRatio_i, K_A, K_B, n_{iA}, n_{iB}) = \sqrt{\frac{\nu_A \frac{c_i}{\gamma_i} + \nu_B \frac{1-c_i}{1-\gamma_i}}{\nu_A+\nu_B}} \text{ with } c_i = \frac{\SRatio_i}{\SRatio_i + \frac{n_{iA}}{n_{iB}}}, \; \gamma_i = \frac{\VRatio_i}{\VRatio_i + \frac{n_{iA}}{n_{iB}}}.
\end{equation}
Since precision weights are known, given the $i$-th unit, $n_{iA}, n_{iB}, \SRatio_i, \nu_A, \nu_B$ are known. Thus, conditional on $\VRatio_i$, we can compute a p-value via,
$$\Pwbf_i \equiv \Pwbf_i(\VRatio_i) := 2F_{t, \nu_A + \nu_B}( - |\Twbf_i(\VRatio_i) / \phi^W(\VRatio_i, \SRatio_i, K_A, K_B, n_{iA}, n_{iB})|),$$
where $F_{t, \nu_A + \nu_B}(\cdot)$ is the cumulative distribution function of a t-variate with $\nu_A + \nu_B$ degrees of freedom. The above result can be interpreted as a p-value based on $\Twbf_i$ conditional on $\VRatio_i$,
\begin{equation}
    \label{eq:PWBF_VR}
    \Pwbf_i \equiv \Pwbf_i(\VRatio_i) := \PWVRFunc(\Twbf_i, \SRatio_i; \VRatio_i),
\end{equation}
where
\begin{equation}
    \PWVRFunc(t, l; \VRatio) := \PP[\lambda]{ \abs{\TwbfNull} \geq \abs{t} \; \cond \; \hat{\lambda}=l}.
\label{eq:PWVR_known_nuisance}
\end{equation}
When $\VRatio_i$ is not known, we pursue the same partially Bayes principle that leads to VREPB in Section \ref{sec: Variance ratios as nuisance parameters}. We will impose a prior only on $\VRatio_i$,
\begin{equation}
\label{eq:weighted_lambda}
    \VRatio_i \simiid G.
\end{equation}
If we knew $G$, then, in analogy to \eqref{eq:PWVR_known_nuisance}, we can compute a p-value via $\PWVRFunc(\Twbf_i, \SRatio_i; G)$, where the tail-area is defined as:
\begin{equation}
    \label{eq:weighted_oracle_1D_p_value}
    \PWVRFunc(t, l; G) := \PP[G]{ \abs{\TwbfNull} \geq \abs{t} \; \cond \; \hat{\lambda}=l} = \EE[G]{  \PWVRFunc(t, l; \lambda) \; \cond \; \SRatio = l}.
\end{equation}
Above, the subscript $G$ in the probability ($\PP[G]{\cdot})$ and expectation ($\EE[G]{\cdot}$) statements indicates that we are also integrating over randomness in~\eqref{eq:weighted_lambda}.

Next, we seek a data-driven implementation of the weighted oracle partially Bayes p-values $\PWVRFunc(\\ \Twbf_i, \SRatio_i; G)$. Based on the conditional distributions of $\hat{\sigma}_{iA}^2  \; \cond \;  \sigma^2_{iA}$ and $\hat{\sigma}_{iB}^2  \; \cond \;  \sigma^2_{iB}$ under \eqref{eq:weighted_ss_distribution}, it is easy to verify that the distribution of $\SRatio_i$ only depends on $\VRatio_i$:
\begin{equation}
\label{eq:weighted_tau_given_lambda}
    \SRatio_i \; \cond \; \VRatio_i \overset{\text{ind}}{\sim} \VRatio_i F_{\nu_A, \nu_B},
\end{equation}
which is the same as the conditional distribution in \eqref{eq:tau_given_lambda}. Therefore, we need to solve the same optimization problem in \eqref{eq:1D_optimization}, and the solution is the NPMLE $\hat{G}$ under the weighted statistical setting. With $\hat{G}$ in hand, we follow the plug-in principle and obtain the empirical partially Bayes p-value:
    \begin{equation}
\PWVR_i := \PWVRFunc(\Twbf_i, \SRatio_i; \hat{G}).
\label{eq:epb_pvalues_wvr}
\end{equation}

\subsubsection{Dual variance empirical partially Bayes with weights}

\label{sec:DVEPB weighted}

If $(\sigma^2_{iA},\sigma^2_{iB})$ were known, then under $H_i$, the standardized mean difference 
$$\frac{\hat{\mu}_{iA} - \hat{\mu}_{iB}}{\sqrt{\frac{\sigma^2_{iA}}{n_{iA}} + \frac{\sigma^2_{iB}}{n_{iB}}}} \sim \mathrm{N}(0, 1),$$ and, by Basu’s theorem, is independent of $(\hat{\sigma}^2_{iA},\hat{\sigma}^2_{iB})$. Hence, given the $i$-th unit, we can take the two-sided standard normal tail at $|\Twbf_i|$ after rescaling by the known-versus-estimated variance scale and compute the conditional tail probability by $\PWDVFunc(\Tbf_i, \hat{\sigma}^2_{iA}, \hat{\sigma}^2_{iB}; \sigma^2_{iA}, \sigma^2_{iB})$ where the definition of the function $\PDVFunc$ above can be unpacked as follows,
\begin{equation}
    \label{eq:PWDV_known_nuisance}
    \PWDVFunc(t, s_A^2, s_B^2; \sigma_A^2,\sigma_B^2) := \PP[\sigma_A^2,\sigma_B^2]{ \abs{\TwbfNull} \geq \abs{t} \; \cond \; \hat{\sigma}_A^2 = s_A^2,\; \hat{\sigma}_B^2=s^2_B}.
\end{equation}
Notice that we can write
$$
\PWDVFunc(t, s_A^2, s_B^2; \sigma_A^2,\sigma_B^2) = 2\Phi\p{ - \abs{t} \cdot \sqrt{\p{s_A^2/n_{iA} + s_B^2/n_{iB}}  \; \big / \; \p{\sigma_A^2/n_{iA} + \sigma_B^2/n_{iB}}}}.
$$
When $(\sigma^2_{iA}, \sigma^2_{iB})$ is not known, we pursue the same partially Bayes principle. We will impose a prior on the pair $(\sigma^2_{iA}, \sigma^2_{iB})$,
\begin{equation}
\label{eq:weighted_sigma_prior}
    (\sigma^2_{iA}, \sigma^2_{iB}) \simiid H.
\end{equation}
If we knew $H$, then, in analogy to \eqref{eq:PWDV_known_nuisance}, we can compute a p-value via $\PWDVFunc(\Twbf_i, \hat{\sigma}^2_{iA}, \hat{\sigma}^2_{iA}; H)$, where the tail-area is defined as:
\begin{align}
    \label{eq:weighted_oracle_2D_p_value}
    \PWDVFunc(t, s^2_A, s^2_B; H) &:= \PP[H]{ \abs{\TwbfNull} \geq \abs{t} \; \cond \; \hat{\sigma}^2_{iA}=s^2_A, \hat{\sigma}^2_{iB} = s^2_B} \\ \nonumber
    &= \EE[H]{  \PWDVFunc(t, s^2_A, s^2_B; \sigma^2_A, \sigma^2_B) \; \cond \; \hat{\sigma}^2_{iA} = s^2_A, \hat{\sigma}^2_{iB} = s^2_B}.
\end{align}
Above, the subscript $H$ in the probability ($\PP[H]{\cdot})$ and expectation ($\EE[H]{\cdot}$) statements indicates that we are also integrating over randomness in~\eqref{eq:weighted_sigma_prior}. Next, we seek a data-driven implementation of the oracle partially Bayes p-values $\PWDVFunc(\Twbf_i, \hat{\sigma}^2_{iA}, \hat{\sigma}^2_{iB}; H)$. Based on the conditional distributions of $\hat{\sigma}_{iA}^2  \; \cond \;  \sigma^2_{iA}$ and $\hat{\sigma}_{iB}^2  \; \cond \;  \sigma^2_{iB}$ under \eqref{eq:weighted_ss_distribution}, we verify that it is of the same distributions under \eqref{eq:ss_distribution}. Therefore, we need to solve the same optimization problem in \eqref{eq:2D_optimization}, and the solution is the NPMLE $\hat{H}$ under the weighted statistical setting. With $\hat{H}$ in hand, we follow the plug-in principle and obtain the empirical partially Bayes p-value:
\begin{equation}
\PWDV_i := \PWDVFunc(\Twbf_i, \hat{\sigma}^2_{iA}, \hat{\sigma}^2_{iB}; \hat{H}).
\label{eq:epb_pvalues_wdv}
\end{equation}

\subsubsection{Welch approximation with weights}

We continue in the weighted statistical setting in \eqref{eq:weighted_normal} for the weighted Welch approximation. For the $i$-th unit, we first collapse all observations in to complete and sufficient statistics defined as (\ref{prop:weighted_normal_ss}). Based on Proposition \ref{prop:weighted_normal_ss}, it is easy to verify that:
$$Z_i = \frac{(\ZMean - \YMean) - (\mu_{iA} - \mu_{iB})}{\sqrt{\frac{\sigma^2_{iA}}{n_{iA}} + \frac{\sigma^2_{iB}}{n_{iB}}}} \sim \mathrm{N}\p{0, 1}.$$
Motivated by the vanilla Welch approximation, we attempt to estimate the distribution of $\Twbf_i$, which leads us to consider the random variable
$$U_i = \Big(\frac{\sigma^2_{iA}}{n_{iA}} + \frac{\sigma^2_{iB}}{n_{iB}}\Big) \Big / \Big(\frac{\hat{\sigma}^2_{iA}}{n_{iA}} + \frac{\hat{\sigma}^2_{iB}}{n_{iB}}\Big).$$
As it has been discussed in Supplement \ref{sec:DVEPB weighted}, it holds that
$\Twbf_i = Z_i \big /\sqrt{U_i}$. Based on Proposition \ref{prop:weighted_normal_ss}, we already have that $U_i$ is independent of $Z_i$. If we further approximate the distribution of $U_i$ by $\chi_{\nu_i}^2 \big / \nu_i$, then we can approximate the distribution of $\Twbf_i$ by a $t$-distribution $t_{\nu_i}$ with $\nu_i$ degrees of freedom, which will allow us to compute an approximated p-value. Under such an approximation, our next goal is to estimate $\nu_i$. Similar to the vanilla Welch approximation, we find the estimator $\hat{\nu}_i$ of $\nu_i$ with $\text{Var}(\chi_{\nu_i}^2) = \text{Var}(\nu_i U).$
After some simplification, we arrive at the following result:
$$\nu_i = \frac{\left(\frac{\sigma^2_{iA}}{n_{iA}} + \frac{\sigma^2_{iB}}{n_{iB}}\right)^2}{\frac{1}{\nu_A} \left(\frac{\sigma^2_{iA}}{n_{iA}}\right)^2 + \frac{1}{\nu_B} \left(\frac{\sigma^2_{iB}}{n_{iB}}\right)^2},\quad \hat{\nu}_i = \frac{\left(\frac{\hat{\sigma}_{iA}^2}{n_{iA}} + \frac{\hat{\sigma}_{iB}^2}{n_{iB}}\right)^2}{\frac{1}{\nu_A} \left(\frac{\hat{\sigma}_{iA}^2}{n_{iA}}\right)^2 + \frac{1}{\nu_B} \left(\frac{\hat{\sigma}_{iB}^2}{n_{iB}}\right)^2},$$
where $\hat{\nu}_i$ is derived only by plugging in $\hat{\sigma}_{iA}^2$ for $\sigma^2_{iA}$, $\hat{\sigma}_{iB}^2$ for $\sigma^2_{iB}$.

Summarizing above results, under the model \eqref{eq:weighted_model} and the null hypothesis $H_i: \mu_{iA} = \mu_{iB}$, we can approximate the p-value by first assuming that $\Twbf_i$ approximately follows a t-distribution with $\hat{\nu}_i$ degrees of freedom, and then computing a two-sided tail probability.

\subsubsection{Equal variance t-test with weights}

For the equal variance t-test with weights, we only consider the $i$-unit in the weighted statistical setting \eqref{eq:weighted_model}. For notation simplicity, we will neglect the subscript $i$. Then, we make the additional equal variance assumption that $\sigma^2_{A} = \sigma^2_{B} = \sigma^2$, which leads us tothe  following law of data generation:
\begin{equation}
    \label{eq:model_t_test}
    Z_{k} \overset{\text{ind}}{\sim} \mathrm{N}\left(\mu_{A}, \frac{\sigma^2}{w_k^A}\right), \text{for } k = 1, \dots, K_A;  \quad Y_k \overset{\text{ind}}{\sim} \mathrm{N}\left(\mu_{B}, \frac{\sigma^2}{w_k^B}\right), \text{for } k = 1, \dots, K_B.
\end{equation}
Following Proposition \ref{prop:weighted_sufficiency}, we have the following complete and sufficient statistics:
\begin{align*}
    &\hat{\mu}_{A} = \frac{\sum_{k=1}^{K_A} w_k^A Z_k}{n_A}, \quad \hat{\sigma}^2_{A} = \frac{1}{\nu_A}\sum_{k=1}^{K_A} w_k^A(Z_k - \hat{\mu}_{A})^2, \\
    &\hat{\mu}_{B} = \frac{\sum_{k=1}^{K_B} w_i^B Y_k}{n_B}, \quad \hat{\sigma}^2_{B} = \frac{1}{\nu_B}\sum_{k=1}^{K_B} w_k^B(Y_k - \hat{\mu}_{B})^2,
\end{align*}
with $n_A = \sum_{k=1}^{K_A} w_k^A, n_B = \sum_{k=1}^{K_B} w_k^B$. Based on Proposition \ref{prop:weighted_normal_ss}, it is easy to verify that:
$$Z = \frac{(\hat{\mu}_A - \hat{\mu}_B) - (\mu_{A} - \mu_{B})}{\sqrt{\frac{\sigma^2}{n_{A}} + \frac{\sigma^2}{n_{B}}}} \sim \mathrm{N}(0, 1) \text{, and } U = \frac{\nu_A \hat{\sigma}_{A}^2}{\sigma^2} + \frac{\nu_B \hat{\sigma}_{B}^2}{\sigma^2} \sim \chi_{\nu_A + \nu_B}^2.$$
Since $U$ is independent of $Z$, we have
$$\frac{Z}{\sqrt{\frac{U}{\nu_A + \nu_B}}} \sim t_{\nu_A + \nu_B},$$
which leads us to the following proposition below.
\begin{prop}
    \label{prop:weighted_t_test}
    Under the model \eqref{eq:model_t_test} and the null hypothesis $H_0: \mu_{A} = \mu_{B}$,

    $$T^{\text{WEV}} = \frac{\hat{\mu}_A - \hat{\mu}_B}{\sqrt{\frac{\nu_A \hat{\sigma}_{A}^2 + \nu_B \hat{\sigma}_{B}^2}{\nu_A + \nu_B}} \sqrt{\frac{1}{n_A} + \frac{1}{n_B}}} \sim t_{\nu_A + \nu_B}.$$
\end{prop}

\subsubsection{Behrens-Fisher test with weights}

The last method we consider is the Behrens-Fisher test. To mimic the same fiducial argument under the statistical setting without weights \eqref{eq:ss_distribution}, we will start by studying the first sample from the $i$-th unit.

For the first sample, under the weighted setting \eqref{eq:weighted_ss_distribution} and based on Proposition \ref{prop:weighted_normal_ss}, it holds that
$$
    \frac{\ZMean - \mu_{iA}}{\sqrt{\hat{\sigma}^2_{iA} \big / n_{iA}}} \sim t_{\nu_A},
$$
and it follows that $\mu_{iA} = \ZMean - \sqrt{\hat{\sigma}^2_{iA} \big / n_{iA}} \times t_{\nu_A}.$ Thus, for any given value $t$, the probability that $t_{\nu_A}$ is greater than $t$ is equivalent to the probability that $\mu_{iA}$ exceeds $\ZMean - \sqrt{\hat{\sigma}^2_{iA} \big / n_{iA}} \times t$.

Analogously, for the second sample of the $i$-th unit, we have that:
$$\mu_{iB} = \YMean - \sqrt{\hat{\sigma}^2_{iB} \big / n_{iB}} \times t_{\nu_B},$$
which allows us to take the difference and obtain:
$$(\ZMean - \YMean) - (\mu_{iA} - \mu_{iB}) = \sqrt{\hat{\sigma}^2_{iA} \big / n_{iA}} \times t_{\nu_A} - \sqrt{\hat{\sigma}^2_{iB} \big / n_{iB}} \times t_{\nu_B},$$
where the quantity on the right is the known fiducial distribution of $(\ZMean - \YMean) - (\mu_{iA} - \mu_{iB})$.

After dividing both sides of the equality by $\sqrt{\hat{\sigma}^2_{iA} \big / n_{iA} + \hat{\sigma}^2_{iB} \big / n_{iB}}$, we recover the $\Twbf_i$ on the left side with
$$\Twbf_i = \sin(R_i) \times t_{\nu_A} - \cos(R_i) \times t_{\nu_B},$$
where $R_i$ is a known angle given all samples and weights with $\tan(R_i) = \sqrt{(\hat{\sigma}^2_{iA} / n_{iA})\big /(\hat{\sigma}^2_{iB} / n_{iB})}$. Thus, for a fixed value $t$,
$$\PP{\Twbf_i \geq t} = \iint \mathbf{1}\{\sin(R_i)u-\cos(R_i)v \geq t \}f_{t_{\nu_A}}(u)f_{t_{\nu_B}}(v)\dd u\dd v,$$
where the right side evolves integrating out the randomness in both $t_{\nu_A}$ and $t_{\nu_B}$, and $f_{t_\nu}$ is the density function of the t-distribution with $\nu$ degrees of freedom. The p-value is therefore calculated by inserting $t$ with $(\ZMean - \YMean)/\sqrt{\hat{\sigma}^2_{iA} \big / n_{iA} + \hat{\sigma}^2_{iB} \big / n_{iB}}$.

\section{Properties of the variance ratio NPMLE (\ref{eq:1D_optimization})}
\label{sec:Properties of VR NPMLE}

Before showing properties of the NPMLE \eqref{eq:1D_optimization}, we first set up the context and provide definitions that are important for clarity. 

First, for $\SRatio > 0$ and $\VRatio > 0$, we let
$p(\SRatio \mid \VRatio; \nu_A, \nu_B)$ denote the density of $\SRatio$ given $\VRatio$
(as in \eqref{eq:tau_density}). For any finite positive measure $G$ on $[0,\infty)$, we define
the mixture density
$$f_{G}(\SRatio) \equiv f_{G}(\SRatio;\nu_A, \nu_B) := \int_0^{\infty}p(\SRatio \; \cond \; \VRatio, \nu_A, \nu_B)\dd G(\VRatio).$$
We further define $\mathcal{G}$ to be the class of finite positive measures on $[0,\infty)$ with total mass
$\le 1$, and let $\mathcal{G}_{\backslash0} \subset \mathcal{G}$ to be a subset of $\mathcal{G}$ within which each measure has no mass at 0.

Now, recall our definition of the NPMLE \eqref{eq:1D_optimization}: given independent observations $\{\SRatio_i\}_{i=1}^n$ from the mixture distribution $f_G(\SRatio)$, and any $G' \in \mathcal{G}$, the marginal log-likelihood is $\sum_{i=1}^n \log f_{G'}(\SRatio_i),$ and we call any maximizer $\hat{G} \in \mathcal{G}$ such that
$$\hat{G}\in\argmax\left\{\sum_{i=1}^n \log f_{G'}(\SRatio_i): G' \in \mathcal{G}\right\},$$
a variance ratio NPMLE.

Next, under the above context, we record three basic properties of the NPMLE $\hat{G}$ \eqref{eq:1D_optimization} that will be used throughout.

\begin{lemm}
    \label{lemma:properties of NPMLE VR} 
    For the variance ratio NPMLE $\hat{G}$, the following properties hold:
    \begin{enumerate}
        
         \item \label{lemma:existence NPMLE VR} $\hat{G}$ exists with total mass $1$; furthermore, $\hat{G} \in \mathcal{G}_{\backslash0}$. 

        \item \label{lemma:inequality NPMLE VR} For every $G \in \mathcal{G}$,
        \begin{equation}
            \label{eq:beta_inequality}
            \frac{1}{n} \sum_{i=1}^n  \frac{f_G(\SRatio_i)}{f_{\hat{G}}(\SRatio_i)} \leq 1 .
        \end{equation}

        \item \label{lemma:support NPMLE VR} $\hat{G}$ is supported on the interval $[\min_i\{\SRatio_i\}, \max_i\{\SRatio_i\}]$.
    \end{enumerate}
\end{lemm}

We comment that the Lemma~\ref{lemma:properties of NPMLE VR}(\ref{lemma:existence NPMLE VR}) ensures a valid maximizing distribution
$\hat{G} \in \mathcal{G}_{\backslash0}$ exists, and the Lemma~\ref{lemma:properties of NPMLE VR}(\ref{lemma:inequality NPMLE VR}) is a
optimality inequality at $\hat{G}$ that leads to the Lemma~\ref{lemma:properties of NPMLE VR}(\ref{lemma:support NPMLE VR}). Lemma~\ref{lemma:properties of NPMLE VR}(\ref{lemma:support NPMLE VR}) lays the foundation for our discretization strategy for the computation of NPMLE in Remark~\ref{rema:vrepb_npmle}. Lemma~\ref{lemma:properties of NPMLE VR}(\ref{lemma:existence NPMLE VR}), Lemma~\ref{lemma:properties of NPMLE VR}(\ref{lemma:inequality NPMLE VR}), together with the identifiability result on the mixture distribution $f_{G}(\SRatio_i)$ from \citet{teicher1961identifiability}, yields the Proposition \ref{prop:1D_weak_convergence} in the main text that establishes the consistency of the NPMLE \eqref{eq:1D_optimization} for the true frequency distribution of $\VRatio_i$.

Following the arguments of \citet{jewell1982mixtures}, we first provide the proof for Lemma~\ref{lemma:properties of NPMLE VR} in Supplement \ref{sec:Properties of VR NPMLE}. Then we will provide the formal proof for Proposition \ref{prop:1D_weak_convergence} in Supplement \ref{proof:prop_1D_weak_convergence}.


\label{proof:lemma_1D_exist_mle}

\begin{proof}
We start by defining the map
$$\psi: \mathcal{G} \rightarrow \mathbb{R}^n, \text{ where } \psi(G') = \left(\psi_1(G'), \dots, \psi_n(G')\right) \text{ with } \psi_i(G') \equiv f_{G'}(\SRatio_i;\nu_A, \nu_B).$$
If $A \subseteq \mathbb{R}^n$ is the image of $\mathcal{G}$ under $\psi$, it is easy to verify that $A$ is convex. We then focus on studying the $p(\SRatio_i \; \cond \; \VRatio, \nu_A, \nu_B)$ defined in \eqref{eq:tau_given_lambda_density}. Consider the image $q: \VRatio \rightarrow p(\SRatio \; \cond \; \VRatio, \nu_A, \nu_B)$, we can write $q(\VRatio)$ and the derivative $q'(\VRatio)$ explicitly as
$$q(\VRatio) = \frac{C \VRatio^{\frac{\nu_B}{2}} \SRatio^{\frac{\nu_A}{2} - 1}} {\left(\VRatio + \frac{\nu_A}{\nu_B} \SRatio\right)^{\frac{\nu_A + \nu_B}{2}}},\quad q'(\VRatio) = \frac{C \VRatio^{\frac{\nu_B}{2} - 1} \SRatio^{\frac{\nu_A}{2} - 1}  \frac{\nu_A}{2} \left(\SRatio - \VRatio\right)}{\left(\VRatio + \frac{\nu_A}{\nu_B} \SRatio\right)^{\frac{\nu_A + \nu_B}{2}+1}}.$$
Therefore, $q(\VRatio) \uparrow$ for $\VRatio < \SRatio$ and $q(\VRatio) \downarrow$ for $\VRatio \geq \SRatio$, meaning $q(\VRatio)$ is maximized when $\VRatio = \SRatio$, i.e. 
\begin{equation}
    \label{eq:f(lambda) bounded}
    q(\VRatio) \leq q(\SRatio) = \frac{C}{\left(1 + \frac{\nu_A}{\nu_B}\right)^\frac{\nu_A + \nu_B}{2} \SRatio} = \frac{C'}{\SRatio},
\end{equation}
where $C'$ is a constant given $\nu_A, \nu_B$. Thus, $q(\VRatio)$ is bounded and continuous. Based on the Helly Bray selection Theorem \citep{breiman1992probability}, we can show that $A$ is compact. Furthermore, define the function $h: A \subseteq \mathbb{R}^n \to \mathbb{R}$ by:
$$
 h(\beta):= \sum_{i=1}^n \log(\beta_i). 
$$
Then $h$ is strictly concave on the compact set $A$ due to the concavity of the log function. Therefore, there must exist a unique $\hat{\beta} \in \psi(\mathcal{G})$ such that:
\begin{equation}
    \label{eq:beta_1D_optimization}
    \hat{\beta} \in \text{argmax}\left\{h(\beta): \beta \in \psi(\mathcal{G})\right\}.
\end{equation}
Let $\hat{G} \in \mathcal{G}$ be the measure with $\psi(\hat{G}) = \hat{\beta}$, we then want to show that $\hat{G}$ is a probability measure. Suppose that the mass of $\hat{G}$ is $1 - \epsilon$ where $\epsilon > 0$. Let $t$ be an arbitrary point on $(0, \infty)$ and consider the measure $\Tilde{G} = \hat{G} + \epsilon\delta_t$. Then $\Tilde{G} \in \mathcal{G}$ and $h(\psi(\Tilde{G})) > h(\psi(\hat{G}))$, which is a contradiction to the definition of $\hat{G}$. Thus, all measures $\hat{G} \in \mathcal{G}$ with $\psi(\hat{G}) = \hat{\beta}$ have total mass 1. Furthermore, we can also show that no measure $\hat{G} \in \mathcal{G}$ with $\psi(\hat{G}) = \hat{\beta}$ has positive mass at 0. Suppose there exists such a $\Tilde{G}^{'}$ with mass $\eta$ at 0 where $\eta > 0$, then $\psi(\Tilde{G}^{'} - \eta\delta_0) = \psi(\Tilde{G}^{'}) = \hat{\beta}$. Based on the argument above, $\Tilde{G}^{'} - \eta\delta_0$ has total mass 1, which contradicts the fact that $\Tilde{G}^{'}$ has total mass 1. Therefore, we have proved Lemma~\ref{lemma:properties of NPMLE VR}(\ref{lemma:existence NPMLE VR}) and established the existence of the MLE $\hat{G}$, a probability measure which is in $\mathcal{G}_{\backslash0}$.

Let $G' \in \mathcal{G}$ and $\psi(G') = \beta' \in A$. We define the map
$$g: [0,1] \rightarrow R, g(\epsilon) = h\left(\psi\left((1-\epsilon)\hat{G} + \epsilon G'\right)\right) = \sum_{i=1}^n \log\left((1-\epsilon)\hat{\beta}_i + \epsilon \beta'_i\right).$$
Notice that $g$ is a concave function of $\epsilon$ with maximum taken at $\epsilon = 0$. Thus, differentiation of $g(\epsilon)$ with respect to $\epsilon$ at $\epsilon=0$ leads to
\begin{equation}
    \label{eq:NPMLE inequality}
    \sum_{i=1}^n \frac{\beta'_i}{\hat{\beta}_i} \leq n,
\end{equation}
which is equivalent to the inequality presented in Lemma~\ref{lemma:properties of NPMLE VR}(\ref{lemma:inequality NPMLE VR}).

At last, consider the special case with $G' = \delta_{\VRatio}$, Lemma~\ref{lemma:properties of NPMLE VR}(\ref{lemma:inequality NPMLE VR}) leads to
\begin{equation}
    \label{eq:xi_lambda_inequality}
    \xi(\VRatio) := \sum_{i=1}^n \frac{p(\SRatio_i;\VRatio, \nu_A, \nu_B)}{f_{\hat{G}}(\SRatio_i)} \leq n \text{ for all } \VRatio \geq 0.
\end{equation}
Integration of LHS of the inequality with respect to $\hat{G}$ is equal to n, and so we obtain the equality of the integrated LHS and RHS of \eqref{eq:xi_lambda_inequality} with respect to $\hat{G}$, which implies that 
\begin{equation}
    \label{eq:xi_lambda_equality}
    \xi(\VRatio) = \sum_{i=1}^n \frac{p(\SRatio_i;\VRatio, \nu_A, \nu_B)}{f_{\hat{G}}(\SRatio_i)} = n, \quad \hat{G}\text{-almost surely}.
\end{equation}
Absorbing constants with respect to $\VRatio$ for the $i$-th summand into a constant $c_i$, we rewrite $\xi(\VRatio)$ and also the derivative of $\xi$ with respect to $\VRatio$:
$$\xi(\VRatio) = \sum_{i=1}^n c_i\frac{\VRatio^{\frac{\nu_B}{2}}}{\left(\VRatio + \frac{\nu_A}{\nu_B}\SRatio_i\right)^{\frac{\nu_A+\nu_B}{2}}},\quad \xi'(\VRatio) = \sum_{i=1}^n c_i \frac{\nu_A}{2} \frac{\VRatio^{\frac{\nu_B}{2}-1}(\SRatio_i - \VRatio)}{\left(\VRatio + \frac{\nu_A}{\nu_B}\SRatio_i\right)^{\frac{\nu_A+\nu_B}{2}+1}}.$$
Now, we are ready to derive further conclusions about the support of $\hat{G}$. Based on the same argument that has been made to derive \eqref{eq:f(lambda) bounded}, it holds that $\xi'(\VRatio) > 0$ in the interval $(0, \min_i\{\SRatio_i\})$, which means that $\xi(\cdot)$ strictly increases in that interval. Now, suppose there exists such $\VRatio \in (0, \min_i \{\SRatio_i\})$ in the support of $\hat{G}$. By \eqref{eq:xi_lambda_equality}, we would have $\xi(\VRatio) = n$, and by the above analysis, for sufficiently small $\epsilon >0$ such that $\lambda+\epsilon \in (0, \min_i \{\SRatio_i\}) $, we would then have $\xi(\VRatio + \epsilon) > \xi(\VRatio) = n$, which is in contradiction to \eqref{eq:xi_lambda_inequality}. Thus, it can be concluded that $\hat{G}$ should assign zero mass on $(0, \min_i \{\SRatio_i\})$. Analogously, it can be argued that $\hat{G}$ assigns zero mass in the interval $(\max_i \{\SRatio_i\}, \infty)$ as for the latter interval, it always holds that $\xi'(\VRatio) < 0$. Thus, we can conclude with Lemma~\ref{lemma:properties of NPMLE VR}(\ref{lemma:support NPMLE VR}) that $\hat{G}$ is supported on the interval $[\min_i\{\SRatio_i\}, \max_i\{\SRatio_i\}]$.
\end{proof}

\section{Proofs for properties of VREPB in Section \ref{sec: Variance ratios as nuisance parameters}}

\subsection{Proof of Proposition \ref{prop:1D_weak_convergence} }

\label{proof:prop_1D_weak_convergence}

For this section, we will provide the proof for weak convergence of NPMLE to the $G$ under the hierarchical setting with $\VRatio_i \simiid G$ \eqref{eq:varation_dbn}. We will provide the proof for the argument in the compound setting (Proposition \ref{prop:compound_1D_weak_convergence}) in Supplement \ref{proof:vrepb_compound}, which $\boldsymbol{\lambda} = (\lambda_1, \dots, \lambda_n)$ is fixed, in which case $G(\boldsymbol{\lambda})$ is defined as the empirical distribution of $\VRatio_i$ \eqref{eq: empirical G}. 

\begin{proof}

Lemma~\ref{lemma:properties of NPMLE VR}(\ref{lemma:existence NPMLE VR}) ensures that we can derive the MLE $\hat{G}$ given n independent observations $(\SRatio_1, \dots, \SRatio_n)$, which we will denote as $\hat{G}_n$. The identifiability results from \cite{teicher1961identifiability} ensures that if $f_{G_1} = f_{G_2}$ almost everywhere, then $G_1 = G_2$. We will show that there exists a set $B$ with probability 1 such that for each $\omega\in B$, given any subsequence $\{\hat{G}_{n_k}(\cdot;\omega)\}$ of $\{\hat{G}_{n}(\cdot;\omega)\}$ there exists a further subsequence which converges weakly to $G$, then we will complete the proof \citep[Theorem 2.6]{billingsley1999probability}.

By Helly's theorem \citep[Theorem 25.9]{billingsley1995probability}, the sequence $\{\hat{G}_{n_k}(\cdot;\omega)\}$ has a subsequence $\{\hat{G}_{n_{k_l}}(\cdot;\omega)\}$ converging vaguely to a sub-distribution function $\tilde{G}$. We then have to show that $\tilde{G} = G$.

Let $\hat{F}_n$ be the empirical distribution function associated with ($\SRatio_1, \dots, \SRatio_n$), $F$ be the true underlying distribution function of $\SRatio$ under $G$, i.e., $F$ has density $f_G \equiv f_{F}$. By the Glivenko-Cantelli theorem, we have
$$\PP[]{B} = 1, \text{ where } B := \left\{\omega: \left\lVert \hat{F}_n(\cdot;\omega)-F(\cdot)\right\rVert_{\infty}:=\sup_{\substack{u\in \mathbb{R}}} \abs{\hat{F}_n(u;\omega)-F(u)} \to 0, \text{ as }n \to \infty\right\}.$$
For $\kappa > 1$, define a compact set 
$$
A_{\kappa} := \left\{\hat{\lambda} \in (0,\infty): f_G(\hat{\lambda}) \geq \frac{1}{\kappa}\right\}\cap\left[\frac{1}{\kappa},\kappa\right] \text{ \quad s.t. }F(A_{\kappa})=c_{\kappa} >0.
$$
Fixed $\kappa$, and let $\omega \in B$, then
\begin{equation}
\label{limit:A_kappa}
    \hat{F}_n(A_{\kappa};\omega) \to F(A_{\kappa}) > 0 \text{ as } n \to \infty.
\end{equation}
Fix $\omega \in B$, and we write $\hat{G}_n(\cdot)=\hat{G}_n(\cdot;\omega)$, $ \hat{F}_n(\cdot)=\hat{F}_n(\cdot;\omega)$. Recall Lemma~\ref{lemma:properties of NPMLE VR}(\ref{lemma:inequality NPMLE VR}) and by the definition of $\hat{F}_n$,
$$\frac{1}{n}\sum_{i=1}^n \frac{f_{G}(\SRatio_i)}{f_{\hat{G}_n}(\SRatio_i)} =\int_0^{\infty} \frac{f_{G}(\SRatio)}{f_{\hat{G}_n}(\SRatio)}  \dd \hat{F}_n(\SRatio) = \int_0^\infty r_n(\SRatio)  \dd \hat{F}_n(\SRatio) \leq 1,$$
where $r_n(\SRatio) = \frac{f_{G}(\SRatio)}{f_{\hat{G}_n}(\SRatio)}$. Also define $r(\hat{\lambda}) = \frac{f_{G}(\SRatio)}{f_{\tilde{G}}(\SRatio)}$.

\begin{lemm}
\label{limit:KW_l}
For any fixed $\kappa$, we have:
    $$\lim_{\substack{l \to \infty}} \int_{A_{\kappa}}r_{n_{k_l}}(\hat{\lambda})\dd \hat{F}_{n_{k_l}}(\hat{\lambda}) = \int_{A_{\kappa}} r(\hat{\lambda})\dd F(\hat{\lambda}) \leq 1.$$
\end{lemm}
\begin{proof}
For notation simplicity, we rename $\{\hat{G}_{n_{k_l}}\}$ to $\{\hat{G}_n\}$, and assume that $\hat{G}_n$ converges vaguely to $\tilde{G}$. Also, rename $\{\hat{F}_{n_{k_l}}\}$ to $\{\hat{F}_n\}$, and \eqref{limit:A_kappa} holds. Denote $t(\hat{\lambda},\lambda) := p(\hat{\lambda}\;\cond\;\lambda,\nu_A,\nu_B)$. For every fixed $\hat{\lambda} \in (0,\infty)$, the function $t(\hat{\lambda},\lambda)$ is continuous on $(0,\infty)$, with $\lim_ {\lambda \to 0}t(\hat{\lambda},\lambda) = 0$ and $\lim_{\lambda\to \infty}t(\hat{\lambda},\lambda)=0$. By the definition of vague convergence,
$$
f_{\hat{G}_{n}}(\hat{\lambda})=\int t(\hat{\lambda},\lambda) d \hat{G}_{n}(\lambda) \longrightarrow \int t(\hat{\lambda},\lambda) d \widetilde{G}(\lambda)=f_{\widetilde{G}}(\hat{\lambda}), \quad \text{as } n \to \infty.
$$
Also notice that $t(\hat{\lambda},\lambda)$ is strictly positive and bounded on $\lambda \in (0,\infty)$, therefore both $f_{\hat{G}_{n}}$ and $f_{G}$ is positive finite, we have
$$
r_n(\hat{\lambda}) =  \frac{f_{G}(\hat{\lambda})}{f_{\hat{G}_n}(\hat{\lambda})} \to \frac{f_{G}(\hat{\lambda})}{f_{\tilde{G}}(\hat{\lambda})} = r(\hat{\lambda}),\quad \text{as } n \to \infty.
$$
For a fixed $\kappa > 1$,
\begin{equation}
    \label{eq:kappa_inequality}
    \int_{A_{\kappa}} r_n(\SRatio)  \dd \hat{F}_n(\SRatio) \leq 1.
\end{equation}
There exist constants $a, b > 0$ such that the measure $\hat{G}_n$ has mass of at least $\delta > 0$ on $[a,b]$ for a large enough $n$, otherwise \eqref{eq:kappa_inequality} will be violated. To see this,  for $\hat{\lambda }\in A_{\kappa}$ and a chosen $\eta>0$, we can find $a,b >0$ such that $t(\hat{\lambda},\lambda) \leq \eta$ for all $(\hat{\lambda},\lambda) \in A_{\kappa} \times(0,a) \cup(b,\infty)$, since $\lim_ {\lambda \to 0}t(\hat{\lambda},\lambda) = 0$ and $\lim_{\lambda\to \infty}t(\hat{\lambda},\lambda)=0$ uniformly in $\hat{\lambda} \in A_{\kappa}$. Suppose there is a sequence $\{n_r\}$ such that 
$$
\gamma_r := \hat{G}_{n_r}([a,b]) \to 0 \quad \text{as } r \to \infty.
$$
Define $M:= \max_{\hat{\lambda}\in A_{\kappa},\lambda \in [a,b]}t(\hat{\lambda},\lambda) >0$ (maximum for a continuous function exists on a compact set), and for all $\hat{\lambda} \in A_{\kappa}$, consider 
\begin{align*}
f_{\hat{G}_{n_r}}(\hat{\lambda}) &= \int_{[a,b]}t(\hat{\lambda},\lambda) \dd \hat{G}_{n_r}(\lambda) + \int_{(0,\infty)\backslash[a,b]}t(\hat{\lambda},\lambda)\dd \hat{G}_{n_r}(\lambda) \leq \gamma_rM + (1-\gamma_r)\eta.
\end{align*}
Since $\gamma_r \to 0$, we can choose $r$ large enough such that $\gamma_r \leq \eta/2M$, then $f_{\hat{G}_{n_r}}(\hat{\lambda}) \leq 3\eta/2$.  Therefore, 
$$
r_{n_r}(\hat{\lambda}) = \frac{f_{G}(\hat{\lambda})}{f_{\hat{G}_{n_r}}(\hat{\lambda})} \geq \frac{1/\kappa}{3\eta/2} := \alpha >0, \quad\text{for all } \hat{\lambda} \in A_{\kappa}.$$
For $r$ sufficiently large, we have $\hat{F}_{n_r}(A_{\kappa}) \geq \hat{F}(A_{\kappa})/2$, therefore, we have
$$
\int_{A_{\kappa}} r_{n_r}(\hat{\lambda})\dd \hat{F}_{n_r}(\hat{\lambda}) \geq \alpha \hat{F}_{n_r}(A_{\kappa}) \geq \frac{c_\kappa \alpha}{2}.
$$
Choose $\eta $ such that $c_{\kappa}\alpha >2$, i.e. $\eta <c_{\kappa}/(3\kappa)$, hence $\int_{A_{\kappa}}r_{n_r}(\hat{\lambda})\dd \hat{F}_{n_r}(\hat{\lambda}) > 1$, which contradicts \eqref{eq:kappa_inequality}. Therefore, for the chosen $a,b$, there exists $\delta > 0$ such that $\hat{G}_n([a,b]) \geq \delta$ for all large $n$.

Thus, for $n$ large enough, $f_{\hat{G}_n}(\SRatio) \geq \delta \text{min}\{t(\hat{\lambda},a), t(\hat{\lambda},b)\}:=\delta m_{a,b}(\hat{\lambda})$. Based on \eqref{eq:f(lambda) bounded}, we have $f_{G}(\SRatio) \leq C'/\SRatio \leq \kappa C^{\prime}$ on $A_{\kappa}$. Hence, we can show that 
\begin{equation}
\label{ineq:bound_rn}
    r_n(\SRatio) \leq \frac{\kappa C'}{\delta m_{a,b}(\hat{\lambda})} \leq  \frac{\kappa C'}{\delta\min_{\substack{t\in A_{\kappa}}}m_{a,b}(t)} :=M_{\kappa},\quad \text{for } \hat{\lambda} \in A_{\kappa},
\end{equation}
where $M_{\kappa}$ is a constant determined by $\delta,a,b,\nu_A,\nu_B,\kappa$, and therefore bounded uniformly in $n$ by a constant on $A_{\kappa}$. 

Now, consider 
$$
\int_{A_\kappa}r_{n}(\hat{\lambda})\dd \hat{F}_{n}(\hat{\lambda}) = \int_{A_\kappa}r_n(\hat{\lambda}) \dd F(\hat{\lambda}) + \int_{A_\kappa} r_n(\hat{\lambda})\dd (\hat{F}_n(\hat{\lambda})-F(\hat{\lambda})):= I_{1,n}+I_{2,n}.
$$
For $I_{1,n}$, by \eqref{ineq:bound_rn} we define $g_{\kappa}(\hat{\lambda}) \equiv M_{\kappa} $, then $r_n(\hat{\lambda}) \leq g_{\kappa}(\hat{\lambda})$ for all large $n$, and $\int_{A_\kappa} g_{\kappa}(\hat{\lambda}) \dd F(\hat{\lambda}) = M_{\kappa}F(A_{\kappa}) <\infty$. Combined with the pointwise convergence of $r_n(\hat{\lambda})$ to $r(\hat{\lambda})$ on $A_{\kappa}$, by applying the dominated convergence theorem, we obtain 
$$
\lim_{\substack{n\to \infty}} I_{1,n} = \int_{A_\kappa}r(\hat{\lambda}) \dd F(\hat{\lambda}).
$$
For $I_{2,n}$, by \eqref{limit:A_kappa}
$$
\abs{I_{2,n}} \leq \sup_{\substack{\hat{\lambda} \in A_{\kappa}}} r_n(\hat{\lambda})\abs{\hat{F}_n(A_{\kappa})-F(A_{\kappa})} \leq M_{\kappa}\abs{\hat{F}_n(A_{\kappa})-F(A_{\kappa})} \to 0, \quad \text{as } n \to \infty.
$$
Combining the above results, we have
$$
\lim_{\substack{n \to \infty}} \int_{A_{\kappa}}r_{n}(\hat{\lambda})\dd \hat{F}_{n}(\hat{\lambda}) = \int_{A_{\kappa}} r(\hat{\lambda})\dd F(\hat{\lambda}) \leq 1
$$
\end{proof}

Since the choice of $\kappa$ is arbitrary, based on the monotone convergence theorem, by taking $\kappa \to \infty$, we have
$$
\lim_{\substack{\kappa\to \infty}}\int_{A_{\kappa}} r(\hat{\lambda}) \dd \hat{F}(\hat{\lambda})=\int \frac{f_G(\hat{\lambda})}{f_{\tilde{G}}(\hat{\lambda})} \dd \hat{F}(\hat{\lambda}) = \int\frac{f_G^2(\hat{\lambda})}{f_{\tilde{G}}(\hat{\lambda})}\dd \hat{\lambda} \leq 1.
$$
By rearranging the inequality,
$$
0 \geq\int\frac{f_G^2(\hat{\lambda})}{f_{\tilde{G}}(\hat{\lambda})}\dd \hat{\lambda} -1 = \int \left(\frac{f_G(\hat{\lambda})}{f_{\tilde{G}}(\hat{\lambda})}-1\right)^2 f_{\tilde{G}}(\hat{\lambda}) \dd \hat{\lambda} \geq 0,
$$
with the equality holds iff $f_G = f_{\tilde{G}}$ almost surely. Since both $f_G$ and $f_{\tilde{G}}$ are continuous, we must have $f_G = f_{\tilde{G}}$. Combined with the identifiability results, we have $\tilde{G} = G$. 

Therefore, we can conclude there is a set $B$ with probability $1$ such that for each $\omega \in B$, every subsequence of the sequence $\{\hat{G}_n(\cdot;\omega)\}$ has a convergent subsequence, and all these subsequences have the same weak limit $G$. Then with probability 1, 
$$
 \hat{G} \cd G \; \text{ as }\; n \to \infty.
$$

\end{proof}

\subsection{Proof of Proposition \ref{prop:1D_oracle_uniform}}

\label{proof:prop_1D_oracle_uniform}

\begin{proof}

The first property follows from the probability integral transform applied conditional on $\SRatio$. Fix $l > 0$ and condition throughout on $\SRatio_i = l$. Under $H_i$ we have
$\Tbf_i = \TbfNull_i$, the null-centered statistic.
Let $X := |\TbfNull_i|$ and for $x > 0$, define
$$F_{G,l}(x)\ :=\ \PP[G]{X < x \; \cond \; \SRatio_i = l},$$
be the conditional c.d.f. of $X$ under the hierarchical model
$\lambda_i\sim G$. Because the conditional law of $\TbfNull_i$ given $(\VRatio_i, \SRatio_i)$ has a continuous density (hence so does $X$),
$F_{G,l}$ is continuous and strictly increasing.

By definition of the VREPB tail area,
$$\PVRFunc(t, l; \lambda) := \PP[\lambda]{ \abs{\TbfNull} \geq \abs{t} \; \cond \; \hat{\lambda}=l} = 1 - F_{G,l}(|t|).$$ 
Evaluating this at a random $t = \Tbf_i = \TbfNull_i$ under $H_i$ gives
$$
U := \PVRFunc(\Tbf_i, l; G) = 1 - F_{G,l}(X).
$$
By the probability integral transform, if $X$ has a continuous c.d.f. $F_{G,l}$ , then $F_{G,l}(X) \sim \text{Unif}(0,1)$. Hence, $U = 1 - F_{G,l}(X)$ is
also $\text{Unif}(0,1)$. Therefore, for all $\alpha\in(0,1)$,
$$\PP[G]{U < \alpha \; \cond \; \SRatio_i = l} = \alpha.$$
Since $l$ is arbitrary, the statement holds almost surely with respect to
$\SRatio_i$.

The second property follows from the conditional uniformity and iterated expectation, i.e.
$$\PP[G]{U \leq \alpha} = \EE[G]{\PP[G]{U \leq \alpha \; \cond \; \SRatio}} = \EE[G]{\alpha} = \alpha.$$
\end{proof}

\subsection{Proof of Theorem \ref{thm:1D_uniform_convergence_p_value} and Theorem \ref{thm:compound_1D_uniform_convergence_p_value}}

\label{proof:thm_1D_uniform_convergence_p_value}

We provide two versions for the consistency of the empirical partially Bayes p-values $\PVR_i$. The first version, Theorem \ref{thm:1D_uniform_convergence_p_value} in the main text applies to the hierarchical setting with $\VRatio_i \simiid G$ \eqref{eq:varation_dbn}. And the second starred version, Theorem \ref{thm:compound_1D_uniform_convergence_p_value} in Supplement \ref{proof:vrepb_compound} applies to the compound decision setting in which $\boldsymbol{\lambda} = (\lambda_1, \dots, \lambda_n)$ is fixed, in which case $G \equiv G(\boldsymbol{\VRatio})$ is defined as the empirical distribution of $\VRatio_i$ \eqref{eq: empirical G}. In the main text and Supplement \ref{proof:vrepb_compound}, we clarify whether an expectation or probability is computed with respect to the hierarchical or compound setting by using the subscript $G (\EE[G]{\cdot}, \PP[G]{\cdot})$ or respectively $\boldsymbol{\lambda} (\EE[\boldsymbol{\lambda}]{\cdot}, \PP[\boldsymbol{\lambda}]{\cdot})$.

For this section, we provide the proof of these two theorems concurrently. For the unified treatment, we write $\PVRFunc(\Tbf_i, \SRatio_i; G)$ for the oracle p-values, where we use $G$ to indicate that an argument is to be understood as going through both the hierarchical and compound settings. For steps that require different arguments for the compound setting, we clarify within the proof and use $\PVRFunc(\Tbf_i, \SRatio_i; G(\boldsymbol{\VRatio}))$ defined in \eqref{eq:compound_oracle_1D_p_value} as the oracle compound p-values.

\begin{proof}

We start by providing the explicit form of the absolute value of $\Tpool$ in \eqref{eq: T_lambda}:
\begin{equation*}
    \abs{\Tpool} = \frac{\abs{\hat{\mu}_A - \hat{\mu}_B} \sqrt{\nu_A + \nu_B}}{\sqrt{\hat{\sigma}_{A}^2 (\frac{\nu_A}{\nu_A+1} + \frac{\nu_A}{\nu_B + 1} \frac{1}{\VRatio}) + \hat{\sigma}_{B}^2 (\frac{\nu_B}{\nu_B+1} + \frac{\nu_B}{\nu_A + 1} \VRatio)}}.
\end{equation*}
With $\SRatio = \hat{\sigma}_{A}^2/\hat{\sigma}_{B}^2$, we can divide the numerator and the denominator by $\sqrt{\hat{\sigma}_{B}^2}$ and arrive at:
\begin{equation*}
    \abs{\Tpool} = \frac{\frac{\abs{\hat{\mu}_A - \hat{\mu}_B}}{\sqrt{\hat{\sigma}_{B}^2 / (\nu_A + \nu_B)}}} {\sqrt{\SRatio (\frac{\nu_A}{\nu_A+1} + \frac{\nu_A}{\nu_B + 1} \frac{1}{\VRatio}) + (\frac{\nu_B}{\nu_B+1} + \frac{\nu_B}{\nu_A + 1} \VRatio)}} =: \frac{\hat{\zeta}}{\sqrt{\SRatio m(\VRatio) + n(\VRatio)}},
\end{equation*}
where 
$$\hat{\zeta} := \frac{\abs{\hat{\mu}_A - \hat{\mu}_B}}{\sqrt{\hat{\sigma}_{B}^2 / (\nu_A + \nu_B)}},\quad m(\VRatio) := \frac{\nu_A}{\nu_A+1} + \frac{\nu_A}{\nu_B + 1} \frac{1}{\VRatio},\quad n(\VRatio) := \frac{\nu_B}{\nu_B+1} + \frac{\nu_B}{\nu_A + 1} \VRatio. $$ 
Now, we are ready to express $\PVRFunc(t, l; \VRatio)$ explicitly as:
\begin{equation}
    \label{eq: p(lambda) explicit}
    \PVRFunc(t, l; \VRatio) \equiv I(z, l, \VRatio) := 2\int_{\frac{z}{\sqrt{l m(\VRatio) + n(\VRatio)}}}^{\infty} \frac{\Gamma(\frac{\nu+1}{2})}{\sqrt{\pi \nu} \Gamma(\frac{\nu}{2})} \left(1 + \frac{u^2}{\nu}\right)^{-\frac{\nu+1}{2}} \dd u,
\end{equation}
where $\nu = \nu_A + \nu_B$, $z = \abs{t} \sqrt{\nu (l/K_A + 1/K_B)}$. With the explicit form of $\PVRFunc(t, l; \VRatio)$, we can then define $\phi_{\tilde{G}}$ for all $\tilde{G} \in \mathcal{G}_{\backslash 0}$:
$$\phi_{\tilde{G} }:(0, \infty) \times (0, \infty) \to R, \;\phi_{\tilde{G} }(z, l) = \EE[\tilde{G} ]{I(z, l, \VRatio) \; \cond \; \SRatio = l} = \frac{\int I(z, l, \VRatio) p(l \; \cond \; \VRatio) \dd \tilde{G} (\VRatio)}{\int p(l \; \cond \; \VRatio) \dd \tilde{G} (\VRatio)},$$
where $p(l \; \cond \; \VRatio) = p(\SRatio = l \; \cond \; \VRatio)$ follows \eqref{eq:tau_given_lambda}. Thus, we have $\PVR_i = \PVRFunc(\Tbf_i, \SRatio_i; \hat{G}) = \phi_{\hat{G}}(\hat{\zeta}_i, \SRatio_i), \PVRFunc(\Tbf_i, \SRatio_i; G) = \phi_{G}(\hat{\zeta}_i, \SRatio_i)$. Notice that $I(z, l, \VRatio) \in (0, 1]$ by its definition as the two-sided tail area of a t-distribution. Therefore, it is easy to see that $\phi_{\tilde{G}} \in (0, 1]$, which leads to the property that $\PVR_i \in (0, 1],  \PVRFunc(\Tbf_i, \SRatio_i; G)\in (0, 1]$.

Next, we define the compact set $S = \{(z, l) \; \cond \; 0 \leq m_z \leq z \leq M_z < \infty,  0 < m_{l} \leq l \leq M_{l} < \infty \}$ where $m_z, m_{l},M_{z}, M_{l}$ are positive numbers.

For the hierarchical setting where $\lambda_i \simiid G$, we fix $i$ and define the event $A_i := \{(\hat{\zeta}_i, \SRatio_i) \in S\}$:

\begin{align*}
    &\EE[G]{\abs{\PVR_i - \PVRFunc(\Tbf_i, \SRatio_i; G)}} 
    = \EE[G]{\abs{\phi_{\hat{G}}(\hat{\zeta}_i, \SRatio_i) - \phi_{G}(\hat{\zeta}_i, \SRatio_i)}} \\
    &\quad\quad \leq \EE[G]{\abs{\phi_{\hat{G}}(\hat{\zeta}_i,\SRatio_i) - \phi_G(\hat{\zeta}_i, \SRatio_i)} \mathbf{1}(A_i)} + \EE[G]{\abs{\phi_{\hat{G}}(\hat{\zeta}_i,\SRatio_i) - \phi_G(\hat{\zeta}_i, \SRatio_i)} \mathbf{1}(A_i^c)}\\
    &\quad\quad \leq \EE[G]{\sup_S \abs{\phi_{\hat{G}}(z, l) - \phi_{G}(z, l)} \mathbf{1}(A_i)} + \EE[G]{\mathbf{1}(A_i^c)} \\
    &\quad\quad\leq \EE[G]{\sup_{S}\abs{\phi_{\hat{G}}(z, l) - \phi_{G}(z, l)}} + \PP[G]{A_i^c}.
\end{align*}
The second inequality comes from the fact that $\phi_{\hat{G}}(\hat{\zeta}_i, \SRatio_i) \in (0, 1],  \phi_{G}(\hat{\zeta}_i, \SRatio_i) \in (0, 1]$. Since the choice of $i$ is arbitrary, we further show that:
$$\max_{1\leq i \leq n} \EE[G]{\abs{\PVR_i - \PVRFunc(\Tbf_i, \SRatio_i; G)}} \leq \EE[G]{\sup_{S}\abs{\phi_{\hat{G}}(z, l) - \phi_{G}(z, l)}} + \PP[G]{A^c}.$$
Here, we omit the subscript $i$ in $\PP[G]{A^c}$ since all $\PP[G]{A_i^c}$ share the same value.

For the compound setting where $\boldsymbol{\lambda} = (\lambda_1, \dots, \lambda_n)$ is fixed, we have:
\begin{align*}
    \label{eq: compound_1D_limit}
    &\frac{1}{n}\sum_{i=1}^n\EE[\boldsymbol{\lambda}]{\abs{\PVR_i - \PVRFunc(\Tbf_i, \SRatio_i; G(\boldsymbol{\lambda}))}} = \frac{1}{n}\sum_{i=1}^n\EE[\lambda_i]{\abs{\PVR_i - \PVRFunc(\Tbf_i, \SRatio_i; G(\boldsymbol{\lambda}))}} \\
    &\quad \leq \frac{1}{n}\sum_{i=1}^n \Big \{\EE[\lambda_i]{\abs{\phi_{\hat{G}}(\hat{\zeta}_i, \SRatio_i) - \phi_{G(\boldsymbol{\lambda})}(\hat{\zeta}_i, \SRatio_i)} \mathbf{1}(A_i)} + \EE[\lambda_i]{\abs{\phi_{\hat{G}}(\hat{\zeta}_i, \SRatio_i) - \phi_{G(\boldsymbol{\lambda})}(\hat{\zeta}_i, \SRatio_i)} \mathbf{1}(A_i^c)} \Big \} \\
    &\quad \leq \frac{1}{n} \sum_{i=1}^n \Big \{\EE[\lambda_i]{\abs{\phi_{\hat{G}}(\hat{\zeta}_i, \SRatio_i) - \phi_{G(\boldsymbol{\lambda})}(\hat{\zeta}_i, \SRatio_i)} \mathbf{1}(A_i)} + \PP[\lambda_i]{A_i^c} \Big \}\\
    &\quad \leq \EE[G(\boldsymbol{\lambda})]{\sup_S\abs{\phi_{\hat{G}}(z, l) - \phi_{G(\boldsymbol{\lambda})}(z, l)}} + \PP[G(\boldsymbol{\lambda})]{A^c}.
\end{align*}
Here we again omit the subscript $i$ in $\PP[G(\boldsymbol{\lambda})]{A^c}$. Combining the above arguments, we can conclude that the key object of interest for both the hierarchical and the compound setting is:
$$\EE[G]{\sup_S\abs{\phi_{\hat{G}}(z, l) - \phi_{G}(z, l)}} + \PP[G]{A^c},$$
where $\VRatio_i \simiid G$ for the hierarchical setting and $G = G(\boldsymbol{\VRatio})$ for the compound setting.

Next, for any distribution $\Tilde{G}$ supported on $(0, \infty)$, and any $z \geq 0, l > 0$, let us write:
\begin{align*}
    N(z, l, \Tilde{G}) &:= \int I(z, l, \VRatio) p(l \; \cond \; \VRatio)  \dd \Tilde{G}(\VRatio), \\
    D(l, \Tilde{G}) &:= \int p(l \; \cond \; \VRatio)  \dd \Tilde{G}(\VRatio), \\
    \phi_{\Tilde{G}}(z, l) &:= N(z, l, \Tilde{G})/D(l, \Tilde{G}).
\end{align*}
\begin{lemm}
    \label{lemm:transformation}It holds that:
    $$|\phi_{\hat{G}}(z, l) - \phi_{G}(z, l)| \leq \frac{2\abs{N(z, l, \hat{G}) - N(z, l, G)}}{D(l, G)} + \frac{2\abs{D(l, \hat{G}) - D(l, G)}}{D(l, G)}.$$
\end{lemm}

\begin{proof}
We let $\hat{G}_0 = (\hat{G} + G)/2$, then we have:
\begin{align*}
    &|\phi_{\hat{G}}(z, l) - \phi_{G}(z, l)| = \abs{\frac{N(z, l, \hat{G})}{D(l, \hat{G})} - \frac{N(z, l, G)}{D(l, G)}} \\
    &\quad\quad = \abs{\frac{N(z, l, \hat{G})}{D(l, \hat{G})} -  \frac{N(z, l, \hat{G})}{D(l, \hat{G}_0)} + \frac{N(z, l, \hat{G})}{D(l, \hat{G}_0)} - \frac{N(z, l, G)}{D(l, \hat{G}_0)} + \frac{N(z, l, G)}{D(l, \hat{G}_0)} - \frac{N(z, l, G)}{D(l, G)}} \\
    &\quad\quad \leq \frac{N(z, l, \hat{G})}{D(l, \hat{G})} \frac{\abs{D(l, \hat{G}_0) - D(l, \hat{G})}}{D(l, \hat{G}_0)} + \frac{\abs{N(z, l, \hat{G}) - N(z, l, G)}}{D(l, \hat{G}_0)}  + \frac{N(z, l, G)}{D(l, G)} \frac{\abs{D(l, \hat{G}_0) - D(l, G)}}{D(l, \hat{G}_0)} \\
    &\quad\quad \leq \frac{\abs{N(z, l, \hat{G}) - N(z, l, G)}}{D(l, \hat{G}_0)} + \frac{\abs{D(l, \hat{G}) - D(l, G)}}{D(l, \hat{G}_0)} \\
    &\quad\quad \leq \frac{2\abs{N(z, l, \hat{G}) - N(z, l, G)}}{D(l, G)} + \frac{2\abs{D(l, \hat{G}) - D(l, G)}}{D(l, G)},
\end{align*}
where the second inequality comes from the fact that for any $\Tilde{G}$, $N(z, l, \Tilde{G})/D(l, \Tilde{G}) = \phi_{\Tilde{G}}(z, l) \\ \in (0, 1]$, and the fact that the map $\Tilde{G} \to D(l, \Tilde{G})$ is linear, which implies that:
$$D(l, \hat{G}_0) - D(l, \hat{G}) = (D(l, G) - D(l, \hat{G}))/2, D(l, \hat{G}_0) - D(l, G) = (D(l, \hat{G}) - D(l, G))/2.$$
And the last inequality also comes from the linearity of the map and the fact that for any $\Tilde{G}$, $D(l, \Tilde{G}) > 0$, then $D(l,\hat{G}_0) =(D(l,\hat{G})+D(l,G))/2\geq D(l,G)/2$.
\end{proof}
Combining all results above, we have:
\begin{align*}
    &\EE[G]{\sup_S\abs{\phi_{\hat{G}}(z, l) - \phi_{G}(z, l)}} + \PP[G]{A^c}
    \leq \PP[G]{A^c} + 2\EE[G]{\sup_S \frac{\abs{N(z, l, \hat{G}) - N(z, l, G)}}{D(l, G)}} \\
    &\hspace{8cm} + 2\EE[G]{\sup_{S_l} \frac{\abs{D(l, \hat{G}) - D(l, G)}}{D(l, G)}},
\end{align*}
where $S_l = \{l: \text{there exists }z\in \mathbb{R} \text{ s.t. } (z,l) \in S\} $. Since $G$ is supported on $[L, U]$, we can further study the lower bound of $D(l, G)$ on the compact set of $S_l$.

As it has been proved in the Supplement \ref{sec:Properties of VR NPMLE}, given any known $l > 0$, $\VRatio \to p(l \; \cond \; \VRatio)$ increases for $\VRatio < l$ and decreases for $\VRatio > l$. Thus, we have:
$$D(l, G) = \int p(l \; \cond \; \VRatio) \dd G(\VRatio) \geq \int \text{min}\left\{p(l \; \cond \; L), \; p(l \; \cond \; U)\right\} \dd G(\VRatio) = \text{min}\left\{p(l \; \cond \; L), p(l \; \cond \; U)\right\}.$$
Given any known $\VRatio > 0$, we focus on the property of the map $l \to p(l \; \cond \; \VRatio)$ by studying its derivative with respect with $l$:
\begin{equation}
    \label{eq:p(tau|lambda) derivative wrt tau}
    \frac{dp(l \; \cond \; \VRatio)}{dl} = C l^{\frac{\nu_A}{2} - 2} \left(1 + \frac{\nu_A}{\nu_B \VRatio} l\right)^{-\frac{\nu_A + \nu_B}{2} - 1} \left[\left(\frac{\nu_A}{2} - 1\right) - \frac{\nu_A}{\nu_B}\left(1 + \frac{\nu_B}{2}\right) \frac{l}{\VRatio}\right].
\end{equation}
For $\nu_A = 2$, $p(l \; \cond \; \VRatio) \downarrow$ for $l > 0$, therefore $p(l \; \cond \; \VRatio) \geq \text{min}(p(m_{l} \; \cond \; \VRatio),\; p(M_{l} \; \cond \; \VRatio))$ holds.

For $\nu_A \geq 3$, $p(l \; \cond \; \VRatio) \uparrow$ for $l < \left(\left(\frac{\nu_A}{2} - 1\right) / \left(\frac{\nu_A}{\nu_B} + \frac{\nu_A}{2}\right)\right) \VRatio$ and $p(l \; \cond \; \VRatio) \downarrow$ for $l \geq \left(\left(\frac{\nu_A}{2} - 1\right) / \left(\frac{\nu_A}{\nu_B} + \frac{\nu_A}{2}\right)\right) \VRatio$, therefore $p(l \; \cond \; \VRatio) \geq \text{min}(p(m_{l} \; \cond \; \VRatio),\; p(M_{l} \; \cond \; \VRatio))$ also holds. Thus we have that for all $\nu_A\geq 2$:
\begin{align*}
    D(l, G) &\geq  \text{min}\left\{p(l \; \cond \; L),\; p(l \; \cond \; U)\right\} \\
    &\geq \text{min}\left\{(p(m_{l} \; \cond \; L),\; p(M_{l} \; \cond \; L),\; p(m_{l} \; \cond \; U), p(M_{l} \; \cond \; U)\right\} = C_{S, L, U} > 0.
\end{align*}
Further combining the results above, we have:
\begin{align*}
    &\EE[G]{\sup_S\abs{\phi_{\hat{G}}(z, l) - f_{G}(z, l)}} + \PP[G]{A^c} \\
    &\quad\quad\leq \frac{2}{C_{S, L, U}} \left(\EE[G]{\sup_S \abs{N(z, l, \hat{G}) - N(z, l, G)}} +  \EE[G]{\sup_{S_l} \abs{D(l, \hat{G}) - D(l, G)}} \right)+ \PP[G]{A^c} \\
    &\quad\quad = \frac{2}{C_{S, L, U}} \left(\text{I} + \text{II} \right)+ \text{III}.
\end{align*}
In three upcoming Lemma, we study terms I, II, and III.

\begin{lemm}
    \label{lemma: bound I} It holds that
    $$\lim_{n \to \infty} \EE[G]{\sup_{(z,l)\in S} \abs{N(z, l, \hat{G}) - N(z, l, G)}} = 0.$$
\end{lemm}

\begin{proof}

First, we want to show that for any distribution $\tilde{G}$ supported on $(0,\infty)$, $N(z, l, \tilde{G}) = \int I(z, l, \VRatio) p(l \; \cond \; \VRatio)  \dd \tilde{G}(\VRatio)$ is bounded on the compact set $S = \{(z, l) \; \cond \; 0 \leq m_z \leq z \leq M_z < \infty,\;  0 < m_{l} \leq l \leq M_{l} < \infty \}$.

Since $I(z, l, \VRatio) \in (0, 1]$, it is easy to verify that $N(z, l, \VRatio) \in (0, \int p(l \; \cond \; \VRatio)  \dd \tilde{G}(\VRatio)] = (0, D(l, \tilde{G})]$, which means we need to upper bound $D(l, \tilde{G})$.

Similar to what we did in Supplement \ref{sec:Properties of VR NPMLE}, we can show that:
\begin{equation}
    \label{eq: bound D}
    D(l, \tilde{G}) = \int p(l \; \cond \; \VRatio)  \dd \tilde{G}(\VRatio) \leq \int \frac{C'}{l}  \dd \tilde{G}(\VRatio) = \frac{C'}{l} \leq \frac{C'}{m_{l}} < \infty.
\end{equation}
Thus, we have shown that $N(z, l, \tilde{G}) \in (0, C'/m_{l}]$.

Next, we show that for any $\tilde{G}$, both $\abs{\partial N(z, l, \tilde{G})/\partial z}$ and $\abs{\partial N(z, l, \tilde{G})/\partial l}$ are bounded.

We will start by studying $\abs{\partial N(z, l, \tilde{G})/\partial z}$:
\begin{align*}
    \abs{\frac{\partial N(z, l, \tilde{G})}{\partial z}} = \abs{\int -2 f_{t_{\nu}}\left(\frac{z}{\sqrt{l m(\VRatio) + n(\VRatio)}}\right)\frac{1}{\sqrt{l m(\VRatio) + n(\VRatio)}}p(l \; \cond \; \VRatio)\dd \tilde{G}(\VRatio)},
\end{align*}
where $\nu = \nu_A + \nu_B$, $f_{t_\nu}$ is the p.d.f. of t-distribution with $\nu$ degrees of freedom:

\begin{equation}
\label{pdf_t_distribution}
    f_{t_{\nu}}(x) = \frac{\Gamma(\frac{\nu+1}{2})}{\sqrt{\pi \nu} \Gamma(\frac{\nu}{2})} \left(1 + \frac{x^2}{\nu}\right)^{-\frac{\nu+1}{2}} = C_{\nu}(\nu + x^2)^{-\frac{\nu+1}{2}},
\end{equation}
where $C_\nu$ is a constant given $\nu$. Now, we write $b = \sqrt{l m(\VRatio) + n(\VRatio)}$ and define a map:
$$q: [0, \infty) \to R, \quad q(a) = \frac{1}{b} f_{t_{\nu}}\left(\frac{a}{b}\right) = \frac{C_{\nu} (\nu + \frac{a^2}{b^2})^{-\frac{\nu + 1}{2}}}{b}.$$
It is easy to verify that $q'(a) \leq 0$, then, $q(a) \leq q(0) = C_{\nu} \nu^{-\frac{\nu+1}{2}} \big / b$. Since $m(\VRatio) = \frac{\nu_A}{\nu_A + 1} + \frac{\nu_A}{\nu_B + 1} \frac{1}{\VRatio} > \frac{\nu_A}{\nu_A + 1}, n(\VRatio) = \frac{\nu_B}{\nu_B + 1} + \frac{\nu_B}{\nu_A + 1} \VRatio > \frac{\nu_B}{\nu_B + 1}$, we have $b \geq \sqrt{l \frac{\nu_A}{\nu_A + 1} + \frac{\nu_B}{\nu_B + 1}} \geq \sqrt{\frac{\nu_B}{\nu_B + 1}}$. Therefore, we have:
$$q(a) \leq \frac{C_{\nu} \nu^{-\frac{\nu+1}{2}}}{\sqrt{\frac{\nu_B}{\nu_B + 1}}} = C_{\nu_A, \nu_B},$$
where $C_{\nu_A, \nu_B}$ is a constant given $\nu_A, \nu_B$. Combining the above results, we have:
\begin{align*}
    \abs{\frac{\partial N(z, l, \tilde{G})}{\partial z}} \leq 2 \abs{\int C_{\nu_A, \nu_B} p(l \; \cond \; \VRatio)\dd \tilde{G}(\VRatio)} 
    = 2 C_{\nu_A, \nu_B} \abs{D(l, \tilde{G})} 
    \leq 2 C_{\nu_A, \nu_B} \frac{C'}{m_{l}} < \infty,
\end{align*}
where we applied \eqref{eq: bound D} in the last inequality.

Next, we focus on studying $\abs{\partial N(z, l, \tilde{G})/\partial l}$. First, it is easy to verify that:
\begin{align}
    \frac{\partial I(z, l, \VRatio)}{\partial l} &= f_{t_{\nu}}\left(\frac{z}{\sqrt{l m(\VRatio) + n(\VRatio)}}\right)\frac{z m(\VRatio)}{2(l m(\VRatio) + n(\VRatio))^\frac{3}{2}}, \label{eq:I wrt tau} \\
    \frac{dp(l \; \cond \; \VRatio)}{dl} &= \frac{\frac{\nu_A}{2}-1}{l} p(l \; \cond \; \VRatio) - \frac{\nu_A + \nu_B}{2} \frac{\nu_A}{\nu_B \VRatio+\nu_Al} p(l \; \cond \; \VRatio).\label{eq:p(tau|lambda) derivative wrt tau implicit}
\end{align}
Combining \eqref{eq:I wrt tau} and \eqref{eq:p(tau|lambda) derivative wrt tau implicit}, we can obtain that:
\begin{align*}
    \abs{\frac{\partial N(z, l, \tilde{G})}{\partial l}} &= 2\Big|\int f_{t_{\nu}}\left(\frac{z}{\sqrt{l m(\VRatio) + n(\VRatio)}}\right)\frac{z m(\VRatio)}{2(l m(\VRatio) + n(\VRatio))^\frac{3}{2}} p(l \; \cond \; \VRatio)  \dd \tilde{G}(\VRatio) \\
    &\quad\quad+ \int I(z, l, \VRatio) \frac{\frac{\nu_A}{2}-1}{l} p(l \; \cond \; \VRatio)  \dd \tilde{G}(\VRatio) - \int I(z, l, \VRatio) \frac{\nu_A + \nu_B}{2} \frac{\nu_A}{\nu_B \VRatio + \nu_A l} p(l \; \cond \; \VRatio)  \dd \tilde{G}(\VRatio)\Big| \\
    &\leq 2\abs{\int f_{t_{\nu}}\left(\frac{z}{\sqrt{l m(\VRatio) + n(\VRatio)}}\right)\frac{z m(\VRatio)}{2(l m(\VRatio) + n(\VRatio))^\frac{3}{2}} p(l \; \cond \; \VRatio)  \dd \tilde{G}(\VRatio)} \\
    &\quad\quad+ 2\abs{\int \frac{\frac{\nu_A}{2}-1}{l} p(l \; \cond \; \VRatio)  \dd \tilde{G}(\VRatio)} + 2\abs{\int \frac{\nu_A + \nu_B}{2l} p(l \; \cond \; \VRatio)  \dd \tilde{G}(\VRatio)} \\
    &:= 2T_1 + 2T_2 + 2T_3,
\end{align*}
where the last inequality uses the fact that $\nu_A/(\nu_B\lambda+\nu_Al) \leq 1/l$ and $I(z,l,\VRatio)\leq 1$. Our next target is to bound $T_1, T_2, T_3$. We will start studying $T_1$ by defining the map:
$$h: (0, \infty) \to R,\quad h(l) = f_{t_{\nu}}\left(\frac{z}{\sqrt{l m(\VRatio) + n(\VRatio)}}\right)\frac{z m(\VRatio)}{2(l m(\VRatio) + n(\VRatio))^\frac{3}{2}}.$$
It is easy to verify that with $l_0 =(\frac{\nu-2}{3}z^2 - n(\VRatio)\nu)/(m(\VRatio)\nu)$, $h'(l) > 0$ when $l < l_0$, and $h'(l) < 0$ when $l > l_0$, which leads to:
$$h(l) \leq h\left(l_0\right) = \frac{1}{2} C_{\nu} \left(\nu + \frac{3\nu}{\nu-2}\right)^{-\frac{\nu+1}{2}} \left(\frac{3\nu}{\nu-2}\right)^{\frac{3}{2}} \frac{m(\VRatio)}{z^2} = C'_{\nu}\frac{m(\VRatio)}{z^2} \leq \frac{C'_{\nu}}{m_z^2}m(\VRatio).$$
Thus, we can bound $T_1$ by:
\begin{align*}
    T_1 
    \leq \abs{\frac{C'_{\nu}}{m_z^2} \int \left(\frac{\nu_A}{\nu_A+1} + \frac{\nu_A}{\nu_B + 1} \frac{1}{\VRatio}\right) p(l \; \cond \; \VRatio)  \dd \tilde{G}(\VRatio)}
    \leq \frac{C'_{\nu}}{m_z^2} \frac{\nu_A}{\nu_A+1} \frac{C'}{m_{l}} + \frac{C'_{\nu}}{m_z^2} \frac{\nu_A}{\nu_B + 1} \abs{\int \frac{p(l \; \cond \; \VRatio)}{\VRatio}  \dd \tilde{G}(\VRatio)},
\end{align*}
where we use \eqref{eq: bound D} in the last inequality. Now, we study $p\left(l \; \cond \; \VRatio\right) \big / \VRatio$ by defining the map:
$r: (0, \infty) \to R, r(\VRatio) = p(l \; \cond \; \VRatio)/\VRatio.$ It is easy to verify that with $K = (\frac{\nu_A}{2} - \frac{\nu_A}{\nu_B})/(\frac{\nu_A}{2} + 1)$, $r'(\VRatio) > 0$ when $\VRatio <Kl$, and $r'(\VRatio) < 0$ when $\VRatio >Kl$, then

$$r(\VRatio) \leq r(Kl) = \frac{C K^{\frac{\nu_B}{2}}}{K \left(K + \frac{\nu_A}{\nu_B}\right)^{\frac{\nu_A + \nu_B}{2}}} \frac{1}{l^2} = \frac{C_K}{l^2} \leq \frac{C_K}{m_{l}^2}.$$
Thus, we have shown that:
$$T_1 \leq \frac{C'_{\nu}}{m_z^2} \frac{\nu_A}{\nu_A+1} \frac{C'}{m_{l}} + \frac{C'_{\nu}}{m_z^2} \frac{\nu_A}{\nu_B + 1} \frac{C_K}{m_{l}^2} < \infty.$$
Next, we study the upper bound of $T_2$:
$$
T_2 = \abs{\int \frac{\frac{\nu_A}{2}-1}{l} p(l \; \cond \; \VRatio)  \dd \tilde{G}(\VRatio)}= \frac{\frac{\nu_A}{2}-1}{l} \abs{D(l, \tilde{G})}\leq \frac{\frac{\nu_A}{2}-1}{m_{l}} \frac{C'}{m_{l}} < \infty.
$$
And then, we study the upper bound of $T_3$:
\begin{align*}
    T_3 = \abs{\int  \frac{\nu_A + \nu_B}{2l} p(l \; \cond \; \VRatio)  \dd \tilde{G}(\VRatio)} 
    = \frac{\nu_A + \nu_B}{2l}\abs{D(l,\tilde{G})} 
    \leq \frac{\nu_A + \nu_B}{2 m_{l}} \frac{C'}{m_{l}} < \infty.
\end{align*}
Combining the above results about the upper bounds of $T_1, T_2, T_3$, we have shown that:
$$\abs{\frac{\partial N(z, l, \tilde{G})}{\partial l}} < \infty.$$
Hence, we have shown the continuity of $N(z, l, \tilde{G})$ with respect to $z, l$. Applying the mean value theorem, it is easy to verify the equicontinuity of $\{N(z, l, \tilde{G}):\tilde{G} \text{ supported on }(0,\infty) \}$ on the compact set $S$. Thus, by Arzelà-Ascoli \citep[Theorem 7.25]{Rudin1976Analysis}, we can prove that on the compact set $S$, any sequence of $\{N(z, l, \tilde{G})\}$ has a uniformly convergent subsequence on $S$ whose limit is continuous. 

In the hierarchical setting, we have $\hat{G}_n \to G$, and the fact that $f_{z,l}(\lambda) := I(z,l,\lambda)p(l\;\cond\;\lambda)$ is bounded and continuous for $\lambda \in \mathbb{R}_{+}$ for any fixed $(z,l) \in S$. By the definition of weak convergence, we have
$$
N(z,l,\hat{G}_n) = \int f_{z,l}(u)\dd \hat{G}_n(u) \to \int f_{z,l}(u) \dd G(u) = N(z,l,G).
$$
In the compound setting, by Proposition \ref{prop:compound_1D_weak_convergence} that we will going to prove in Supplement \ref{proof:vrepb_compound}, we also have 
$$
N(z,l,\hat{G}_n) = \int f_{z,l}(u)\dd \hat{G}_n(u) \to \int f_{z,l}(u) \dd G(\boldsymbol{\lambda})(u) = N(z,l,G).
$$
Since the pointwise convergence ensures that the only possible limit for the uniformly convergent subsequence of $\{N(z,l,\hat{G}_n)\}$ is $N(z,l,G)$, it follows that the full sequence converges uniformly on $S$ for both the hierarchical and the compound setting:
$$\lim_{n \to \infty} \sup_S \abs{N(z, l, \hat{G}_n) - N(z, l, G)} = 0 \quad \text{almost surely}.$$
Thus, by Fatou's lemma \citep[Theorem 16.3]{billingsley1995probability}, 
$$\lim_{n \to \infty} \EE[G]{\sup_S\abs{N(z, l, \hat{G}) - N(z, l, G)}} \leq \EE[G]{\lim_{n \to \infty} \sup_S \abs{N(z, l, \hat{G}) - N(z, l, G)}} = 0.$$
\end{proof}

\begin{lemm}
    \label{lemma:bound II}It holds that
    $$\lim_{n \to \infty} \EE[G]{\sup_{l\in S_l} \abs{D(l, \hat{G}) - D(l, G)}} = 0.$$
\end{lemm}

\begin{proof}

First, we want to show that for any distribution $\tilde{G}$ supported on $(0,\infty)$, $D(l, \tilde{G}) = \int p(l\; \cond \;\VRatio)  \dd \tilde{G}(\VRatio)$ is bounded on the compact set $S_l$.

Since it has been proved in \eqref{eq: bound D}, we have:
$D(l, \tilde{G}) \in \left(0, C'/m_{l}\right].$ Now, we want to study the upper bound of $\abs{\partial D(l, \tilde{G})/\partial l}$ by applying the result in \eqref{eq:p(tau|lambda) derivative wrt tau implicit}:
\begin{align*}
    \abs{\frac{\partial D(l, \tilde{G})}{\partial l}} &= \abs{\int \frac{\frac{\nu_A}{2}-1}{l} p(l\; \cond \;\VRatio)  \dd \tilde{G}(\VRatio) - \int \frac{\nu_A + \nu_B}{2} \frac{\nu_A}{\nu_B \VRatio + \nu_A l} p(l\; \cond \;\VRatio)  \dd \tilde{G}(\VRatio)} \\
    &\leq \abs{\int \frac{\frac{\nu_A}{2}-1}{l} p(l\; \cond \;\VRatio)  \dd \tilde{G}(\VRatio)} + \abs{\int \frac{\nu_A + \nu_B}{2} \frac{\nu_A}{\nu_B \VRatio + \nu_A l} p(l\; \cond \;\VRatio)  \dd \tilde{G}(\VRatio)} \\
    &\leq \frac{\frac{\nu_A}{2}-1}{l} \abs{D(l, \tilde{G})} + \frac{\nu_A + \nu_B}{2l} \abs{D(l, \tilde{G})} \\
    &\leq \frac{\frac{\nu_A}{2}-1}{m_{l}} \frac{C'}{m_{l}} + \frac{\nu_A + \nu_B}{2 m_{l}} \frac{C'}{m_{l}} < \infty.
\end{align*}
Hence, we have shown the continuity of $D(l, \tilde{G})$ with respect to $l$ and the equicontinuity of $\{D(l, \tilde{G}): \tilde{G} \text{ supported on }(0,\infty)  \}$ on the compact set $S_l$ is therefore guaranteed by the mean value theorem. Thus, by Arzelà-Ascoli \citep[Theorem 7.25]{Rudin1976Analysis}, we can prove that on the compact set $S_l$, any sequence of $\{D(l, \tilde{G})\}$ has a uniformly convergent subsequence on $S_l$ whose limit is continuous.

In the hierarchical setting, we have $\hat{G}_n \to G$, and the fact that $p(l\;\cond\;\lambda)$ is bounded and continuous for $\lambda \in \mathbb{R}_{+}$ for any fixed $l \in S_l$. By the definition of weak convergence, we have
$$
D(l,\hat{G}_n) = \int p(l\;\cond\;u)\dd \hat{G}_n(u) \to \int p(l\;\cond\;u) \dd G(u) = D(l,G).
$$
In the compound setting, by Proposition \ref{prop:compound_1D_weak_convergence}, we also have
$$
D(l,\hat{G}_n) = \int p(l\;\cond\;u)\dd \hat{G}_n(u) \to \int p(l\;\cond\;u) \dd G(\boldsymbol{\lambda})(u) = D(l,G).
$$
Since the pointwise convergence ensures that the only possible limit for the uniformly convergent subsequence of $\{D(l, \hat{G}_n)\}$ is $D(l,G)$, it follows that the full sequence converges uniformly on both the hierarchical setting and the compound setting:

$$\lim_{n \to \infty} \sup_{S_l} \abs{D(l, \hat{G}_n) - D(l, G)} = 0 \quad\text{almost surely}.$$ 
Thus, by Fatou's lemma \citep[Theorem 16.3]{billingsley1995probability}, 
$$\lim_{n \to \infty} \EE[G]{\sup_{S_l}\abs{D(l, \hat{G}) - D(l, G)}} \leq \EE[G]{\lim_{n \to \infty} \sup_{S_l} \abs{D(l, \hat{G}) - D(l, G)}} = 0.$$
\end{proof}

\begin{lemm}
    \label{lemma: P(Ac) = 0}
    As $m_z, m_l\to 0$ and $M_z, M_l \to \infty$, $\PP[G]{A^c} \to 0$.
\end{lemm}

\begin{proof}

Recall that $A = \{(\hat{\zeta}, \SRatio) \in S\}$, where $S = \{(z, l) \; \cond \; 0 <m_z \leq z \leq M_z < \infty,  0 < m_{l} \leq l \leq M_{l} < \infty \}$. Notice that $\PP[G]{A^c} \leq \PP[G]{\hat{\zeta} < m_z}+\PP[G]{\hat{\zeta} > M_z} +\PP[G]{\hat{\lambda} < m_l} + \PP[G]{\hat{\lambda} > M_l}$, therefore we focus on studying these terms separately.

First, we study $\PP[G]{\hat{\zeta} < m_z}$ and $\PP[G]{\hat{\zeta} > M_z}$. Based on the law of the sufficient statistics \eqref{eq:ss_distribution} and under the null hypothesis, we have:
$$
\hat{\mu}_A - \hat{\mu}_B \sim \mathrm{N}\left(0, \frac{\sigma^2_A}{K_A} + \frac{\sigma^2_B}{K_B}\right),\quad \hat{\sigma}_{B}^2 \sim \frac{\sigma^2_B}{\nu_B} \chi^2_{\nu_B},
$$
and by Proposition \ref{prop:weighted_normal_ss}, $\hat{\mu}_B$ is independent of $\hat{\sigma}_{B}^2$, and therefore $\hat{\mu}_A - \hat{\mu}_B$ is independent of $\hat{\sigma}_{B}^2$, then we have
$$
\hat{\zeta}=\frac{\abs{\hat{\mu}_A - \hat{\mu}_B}}{\sqrt{\hat{\sigma}_{B}^2 / (\nu_A + \nu_B)}} \sim C(\VRatio)\abs{t_{\nu_B}}, \quad C(\VRatio) = \sqrt{(\nu_A + \nu_B)\left(\frac{\VRatio}{K_A} + \frac{1}{K_B}\right)},
$$
where $t_{\nu_B}$ represents a t-distribution with $\nu_B$ degrees of freedom.

Thus, we can express $\PP[G]{\hat{\zeta} < m_z}$ explicitly as:
\begin{align*}
    \PP[G]{\hat{\zeta}< m_z} &= \PP[G]{\abs{C(\VRatio) t_{\nu_B}} < m_z} \\
    &= 2 \int \int_0^{m_z} \frac{1}{C(\VRatio)} \frac{\Gamma(\frac{\nu_B +1}{2})}{\sqrt{\pi \nu_B} \Gamma(\frac{\nu_B}{2})} \left(1 + \frac{t^2}{C^2(\VRatio) \nu_B}\right)^{-\frac{\nu_B + 1}{2}} \dd t \dd G(\VRatio) \\
    &\leq \frac{ 2  C_{\nu_B} }{C(L)} \int \int_0^{m_z} \left(1 + \frac{t^2}{C^2(U) \nu_B}\right)^{-\frac{\nu_B + 1}{2}} \dd t \dd G(\VRatio)\\
    &= \frac{ 2  C_{\nu_B} }{C(L)} \int_0^{m_z} \left(1 + \frac{t^2}{C^2(U) \nu_B}\right)^{-\frac{\nu_B + 1}{2}} \dd t,
\end{align*}
where $C_{\nu_B} = \Gamma(\frac{\nu_B +1}{2}) \big/ \left(\Gamma(\frac{\nu_B}{2}) \sqrt{\pi \nu_B}\right) $, and the last inequality comes from the fact that G is supported on $[L, U]$, then $C(L) \leq C(\lambda) \leq C(U)$.

Since as $t$ increases, $\left(1 + t^2/(C^2(U) \nu_B)\right)^{-(\nu_B + 1)/2}$ decreases, we can show that for any $t \in (0, m_z)$, 
$$\left(1 + \frac{t^2}{C^2(U) \nu_B}\right)^{-\frac{\nu_B + 1}{2}} \leq \left(1 + \frac{0^2}{C^2(U) \nu_B}\right)^{-\frac{\nu_B + 1}{2}} = 1,$$ then we arrive at:
$$\PP[G]{\hat{\zeta} < m_z} \leq \frac{2 C_{\nu_B}}{C(L)} \int_0^{m_z} 1 \dd t = \frac{2 C_{\nu_B}}{C(L)} m_z.$$
Hence, we show that as $m_z \to 0$, $\PP[G]{\hat{\zeta} < m_z} \to 0$.

For $\PP[G]{\hat{\zeta} > M_z}$, we have
$$
\PP[G]{\hat{\zeta} > M_z} = \PP[G]{\abs{C(\lambda)t_{\nu_B}}>M_z} = 2\PP[G]{t_{\nu_B}>\frac{M_z}{C(\lambda)}}\leq 2\PP[]{t_{\nu_B}>\frac{M_z}{C(U)}},
$$
since $G$ is supported $[L,U]$ and $C(\lambda) \leq C(U)$. Recall the density function of the t-distribution in \eqref{pdf_t_distribution}, we have that for $x >0$, $f_{t_{\nu}}(x) =\mathcal{O}(x^{-(\nu+1)})$ as $x \to \infty$, therefore
$$
\PP[]{t_{\nu}>x} = \int_{t}^{\infty}f_{t_{\nu}}(u) \dd u = \mathcal{O}(t^{-\nu}) \quad \text{as } t \to \infty.
$$
Then we arrive at,
$$
\PP[G]{\hat{\zeta} > M_z} \leq 2\PP[]{t_{\nu_B}>\frac{M_z}{C(U)}} = \mathcal{O}(M_z^{-\nu_B}) \to 0 \quad \text{as } M_z \to \infty.
$$
Hence, as $M_z \to \infty$, $\PP[G]{\hat{\zeta} > M_z} \to 0$.

Now, we move on to study $\PP[G]{\SRatio < m_l}$ and $\PP[G]{\SRatio > M_l}$. Recall that $\SRatio \mid \VRatio \sim \VRatio F_{\nu_A,\nu_B}$, we have
$$
\PP[G]{\SRatio < m_l} =\PP[G]{F_{\nu_A,\nu_B}<\frac{m_l}{\VRatio}}\leq\PP[]{F_{\nu_A,\nu_B}<\frac{m_l}{L}}= \int_{0}^{\frac{m_l}{L}} f_{F_{\nu_A,\nu_B}}(t)\dd t,
$$
where $f_{F_{\nu_A,\nu_B}}$ is the density function for the F-distribution with $\nu_A,\nu_B$ degrees of freedom. The inequality comes from the fact that $\lambda \in [L,U]$ and $f_{F_{\nu_A,\nu_B}}(t) > 0$ for all $t >0$.
Note that when $\nu_A \geq2$, $f_{F_{\nu_A,\nu_B}}$ is bounded on the compact intervals $[0,m_l/L]$, we define
$$
M_1 := \sup_{\substack{t \in [0,m_\lambda/L]}}f_{F_{\nu_A,\nu_B}}.
$$
Then we arrive at:
$$
\PP[G]{\SRatio < m_l} \leq  \int_{0}^{\frac{m_l}{L}} M_1 \dd t = \frac{M_1}{L}m_l.
$$
Hence, as $m_l \to 0$, $\PP[G]{\SRatio < m_l} \to 0$.

For $\PP[G]{\SRatio > M_l}$, we have
$$
\PP[G]{\SRatio > M_l} =\PP[G]{F_{\nu_A,\nu_B} > \frac{M_l}{\VRatio}} \leq \PP[]{F_{\nu_A,\nu_B} > \frac{M_l}{U}}=\int_{\frac{M_l}{U}}^{\infty} f_{F_{\nu_A,\nu_B}}(t)\dd t,
$$
where the inequality comes from the fact that $\lambda \in [L,U]$ and $f_{F_{\nu_A,\nu_B}}(t) > 0$ for all $t >0$.
Consider the tail behavior of the F-distribution. Recall that the density of $F_{\nu_A,\nu_B}$ is
$$
f_{F_{\nu_A,\nu_B}}(t) = C_{\nu_A,\nu_B}t^{\frac{\nu_A}{2}-1}\left(1+\frac{\nu_A}{\nu_B}t\right)^{-\frac{\nu_A+\nu_B}{2}}= \mathcal{O}\left(t^{-\frac{\nu_B}{2}-1}\right) \; \text{ as }\; t \to \infty\;.
$$
Now consider the tail probability,
$$
\PP[]{F_{\nu_A,\nu_B} > t} = \int_{t}^\infty f_{F_{\nu_A,\nu_B}}(x) \dd x = \mathcal{O}\left(t^{-\frac{\nu_B}{2}}\right)  \; \text{ as }\; t \to \infty\;.
$$
Then, there exists a constant $M_2$ such that 
$$
\PP[G]{\SRatio > M_l}\leq\int_{\frac{M_l}{U}}^\infty f_{F_{\nu_A,\nu_B}}(x) \dd x \leq M_2\left(\frac{M_l}{U}\right)^{-\frac{\nu_B}{2}}.
$$
Hence, as $M_l \to \infty$, $\PP[G]{\SRatio > M_l} \to 0$.

Combining all the results above, we can now conclude that as $m_z,m_l \to 0$, $M_z ,M_l \to \infty$, 
$$
\PP[G]{A^c} \leq \PP[G]{\hat{\zeta} < m_z}+\PP[G]{\hat{\zeta} > M_z}+\PP[G]{\hat{\lambda} < m_l} + \PP[G]{\hat{\lambda} > M_l} \to 0.
$$
\end{proof}

\noindent With above three lemmas, we can continue our study of $\EE[G]{\abs{\PVR_i - \PVRFunc(\Tbf_i, \SRatio_i; G)}} + \PP[G]{A^c}$. Recall that in the hierarchical setting, since the choice of $i$ is arbitrary, we have
\begin{align*}
    &\max_{1 \leq i \leq n} E_{G}\Big [\abs{\PVR_i - \PVRFunc(\Tbf_i, \SRatio_i; G)}\Big ] \\
    &\quad \leq \frac{2}{C_{S, L, U}} \left(\EE[G]{\sup_S \abs{N(z, l, \hat{G}) - N(z, l, G)}}+ \EE[G]{\sup_{S_l} \abs{D(l, \hat{G}) - D(l, G)}} \right)+ \PP[G]{A^c}.
\end{align*}
Thus,
\begin{align*}
    &\lim_{n \to \infty}\max_{1 \leq i \leq n} \EE[G]{\abs{\PVR_i - \PVRFunc(\Tbf_i, \SRatio_i; G)}} \\
    &\quad\leq \frac{2}{C_{S, L, U}}\lim_{n \to \infty}\left(  \EE[G]{\sup_S \abs{N(z, l, \hat{G}) - N(z, l, G)}} +   \EE[G]{\sup_{S_l} \abs{D(l, \hat{G}) - D(l, G)}} \right)+ \PP[G]{A^c} \\
    &\quad= \PP[G]{A^c}.
\end{align*}
Since the choices of $m_z,m_l,M_z,M_l$ are arbitrary, as $m_z,m_l \to 0$ and $M_z,M_l \to\infty$, $\PP[G]{A^c} \to 0$, which leads to:
$$ \lim_{n \to \infty}\max_{1 \leq i \leq n} \EE[G]{\abs{\PVR_i - \PVRFunc(\Tbf_i, \SRatio_i; G)}} = 0.$$
Recall that in the compound setting, we have
\begin{align*}
    \frac{1}{n}\sum_{i=1}^n &\EE[\boldsymbol{\lambda}]{\abs{\PVR_i - \PVRFunc(\Tbf_i, \SRatio_i; G(\boldsymbol{\VRatio}))}} \\
    &\leq \frac{2}{C_{S, L, U}} \EE[G(\boldsymbol{\lambda})]{\sup_S \abs{N(z, l, \hat{G}) - N(z, l, G(\boldsymbol{\lambda}))}} \\
    &\hspace{0.35cm} + \frac{2}{C_{S,L,U}} \EE[G(\boldsymbol{\lambda})]{\sup_{S_l} \abs{D(l, \hat{G}) - D(l, G(\boldsymbol{\lambda}))}} + \PP[G(\boldsymbol{\lambda})]{A^c}.
\end{align*}
Thus,
\begin{align*}
    \lim_{n \to \infty}\frac{1}{n}\sum_{i=1}^n &\EE[\boldsymbol{\lambda}]{\abs{\PVR_i - \PVRFunc(\Tbf_i, \SRatio_i; G(\boldsymbol{\VRatio}))}} \\
    &\leq \frac{2}{C_{S, L, U}} \lim_{n \to \infty} \EE[G(\boldsymbol{\lambda})]{\sup_S \abs{N(z, l, \hat{G}) - N(z, l, G(\boldsymbol{\lambda}))}} \\
    &+ \frac{2}{C_{S,L,U}} \lim_{n \to \infty} \EE[G(\boldsymbol{\lambda})]{\sup_{S_l} \abs{D(l, \hat{G}) - D(l, G(\boldsymbol{\lambda}))}} + \PP[G(\boldsymbol{\lambda})]{A^c} \\
    &= \PP[G(\boldsymbol{\lambda})]{A^c}.
\end{align*}
Since the choices of $m_z,m_l,M_z,M_l$ are arbitrary, as $m_z,m_l \to 0$, $M_z,M_l \to \infty$, we have $\PP[G(\boldsymbol{\lambda})]{A^c} \to 0$, which leads to:
$$ \lim_{n \to \infty}\frac{1}{n}\sum_{i=1}^n\EE[\boldsymbol{\lambda}]{\abs{\PVR_i - \PVRFunc(\Tbf_i, \SRatio_i; G(\boldsymbol{\VRatio}))}} = 0.$$

\end{proof}

\subsection{Proof of Theorem \ref{thm:1D asymptotic uniformity}}

\label{proof:thm_1D asymptotic uniformity}

We first present a simple lemma.

\begin{lemm}
    \label{lem:indicator}
    Let $\delta \in (0, 1)$ and $\alpha \in [0, 1 - \delta]$. Then:
    \[ \mathbf{1}(P_i \leq \alpha) - \mathbf{1}(P_i^* \leq \alpha + \delta) \leq \frac{1}{\delta} \abs{P_i - P_i^*}.\]
\end{lemm}

\begin{proof} Note that:
$$\mathbf{1}(P_i \leq \alpha) - \mathbf{1}(P_i^* \leq \alpha + \delta) \leq \mathbf{1}(P_i \leq \alpha, P_i^* > \alpha + \delta) \leq \frac{1}{\delta} \abs{P_i - P_i^*}.$$
\end{proof}
Then, we are ready to prove Theorem \ref{thm:1D asymptotic uniformity}.
\begin{proof}
For simplicity, we would denote $P_i^{\text{VR}*} := \PVRFunc(\Tbf_i, \SRatio_i; G)$ in this proof.  Fix any $\delta \in (0, 1)$, by Lemma \ref{lem:indicator}, it can be shown that for any $\alpha \leq 1 - \delta$:
$$\mathbf{1}(\PVR_i \leq \alpha) \leq \frac{1}{\delta} \abs{\PVR_i - P_i^{\text{VR}*}} + \mathbf{1}(P_i^{\text{VR}*} \leq \alpha + \delta).$$
Take conditional expectation given $\SRatio_1,...,\SRatio_n$ on both sides of the inequality, and use the conditional uniformity of the oracle p-values $P_i^{\text{VR}*}$, we have the following inequality after rearranging:
$$
\PP[G]{\PVR_i \leq \alpha \; \cond \; \SRatio_1, \dots, \SRatio_n} - \alpha \leq  \delta + \frac{1}{\delta}\EE[G]{\abs{\PVR_i - P_i^{\text{VR}*}} \; \cond \; \SRatio_1, \dots, \SRatio_n}.
$$
For any $\alpha \in (1 - \delta, 1)$, we also have:
$\PP[G]{\PVR_i \leq \alpha \; \cond \; \SRatio_1, \dots, \SRatio_n} - \alpha \leq 1 - \alpha \leq \delta.$
Define $(a)_{+} = \max\{a, 0\}$. Above two inequalities lead to:
\begin{align*}
\sup_{\alpha \in [0, 1]} \left(\PP[G]{\PVR_i \leq \alpha \; \cond \; \SRatio_1, \dots, \SRatio_n} - \alpha\right)_{+}\leq \delta + \frac{1}{\delta} \EE[G]{\abs{\PVR_i - P_i^{\text{VR}*}} \; \cond \; \SRatio_1, \dots, \SRatio_n}.
\end{align*}
Analogously, by defining $(a)_{-} = \max\{-a, 0\}$, we can prove that
\begin{align*}
&\sup_{\alpha \in [0, 1]} \left(\PP[G]{\PVR_i \leq \alpha \; \cond \; \SRatio_1, \dots, \SRatio_n} - \alpha\right)_{-}\leq \delta + \frac{1}{\delta} \EE[G]{\abs{\PVR_i - P_i^{\text{VR}*}} \; \cond \; \SRatio_1, \dots, \SRatio_n}.
\end{align*}
Notice that $\abs{a} \leq (a)_{+} + (a)_{-}$, then:
\begin{align*}
    \sup_{\alpha \in [0, 1]} \abs{\PP[G]{\PVR_i \leq \alpha \; \cond \; \SRatio_1, \dots, \SRatio_n} - \alpha}\leq 2\left(\delta + \frac{1}{\delta} \EE[G]{\abs{\PVR_i - P_i^{\text{VR}*}} \; \cond \; \SRatio_1, \dots, \SRatio_n}\right).
\end{align*}
With iterated expectation, and the fact that $i$ was picked arbitrarily, the inequality above should hold for all $i \in \mathcal{H}_0$, which leads to:
\begin{align*}
    &\max_{i \in \mathcal{H}_0} \left\{\EE[G]{\sup_{\alpha \in [0, 1]} \abs{\PP[G]{\PVR_i \leq \alpha \; \cond \; \SRatio_1, \dots, \SRatio_n} - \alpha}}\right\} \leq 2\left(\delta + \frac{1}{\delta} \max_{i \in \mathcal{H}_0} \left\{\EE[G]{\abs{\PVR_i -P_i^{\text{VR}*}}}\right\}\right).
\end{align*}
Based on Theorem \ref{thm:1D_uniform_convergence_p_value}, by taking $n\to 0$, we have
$$
\lim_{n \to \infty}\max_{i \in \mathcal{H}_0} \left\{\EE[G]{\sup_{\alpha \in [0, 1]} \abs{\PP[G]{\PVR_i \leq \alpha \; \cond \; \SRatio_1, \dots, \SRatio_n} - \alpha}}\right\} \leq 2\delta .$$
Since the above inequality holds for any $\delta\in(0,1)$, taking $\delta \downarrow 0$ and we can conclude that:
$$
\lim_{n \to \infty}\max_{i \in \mathcal{H}_0} \left\{\EE[G]{\sup_{\alpha \in [0, 1]} \abs{\PP[G]{\PVR_i \leq \alpha \; \cond \; \SRatio_1, \dots, \SRatio_n} - \alpha}}\right\} =0 .$$
Asymptotic uniformity follows from the above result:
\begin{align*}
    &\max_{i \in \mathcal{H}_0} \left\{\sup_{\alpha \in [0,1]} \abs{\PP[G]{\PVR_i \leq \alpha} - \alpha}\right\} = \max_{i \in \mathcal{H}_0} \left\{\sup_{\alpha \in [0,1]} \abs{\EE[G]{\PP[G]{\PVR_i \leq \alpha \; \cond \; \SRatio_1, \dots, \SRatio_n} - \alpha}}\right\} \\
    &\quad\quad \leq \max_{i \in \mathcal{H}_0} \left\{\sup_{\alpha \in [0,1]} \EE[G]{\abs{\PP[G]{\PVR_i \leq \alpha \; \cond \; \SRatio_1, \dots, \SRatio_n} - \alpha}}\right\} \\
    & \quad\quad \leq \max_{i \in \mathcal{H}_0} \left\{\EE[G]{\sup_{\alpha \in [0,1]}\abs{\PP[G]{\PVR_i \leq \alpha \; \cond \; \SRatio_1, \dots, \SRatio_n} - \alpha}} \right\} \to 0 \text{ as } n \to \infty.
\end{align*}
\end{proof}

\subsection{Proof of Theorem \ref{theo:vrepb_compound}}
\label{proof:vrepb_compound}

The crux of our argument will be that the setting with $\boldsymbol{\VRatio}$ fixed is quantitatively similar to the hierarchical model \eqref{eq:tau_given_lambda} with a specific choice of prior $G$ for $\VRatio_i$, namely the empirical distribution of the $\VRatio_i$:
\begin{equation}
    \label{eq: empirical G}
    G \equiv G(\boldsymbol{\VRatio}) := \frac{1}{n} \sum_{i=1}^n \delta_{\VRatio_i}.
\end{equation}
Based on the above compound decision argument, we can consider the following oracle compound p-value in lieu of the conditional p-value \eqref{eq:oracle_1D_p_value}:
\begin{equation}
    \label{eq:compound_oracle_1D_p_value}
    \PVRFunc(t, l; G(\boldsymbol{\VRatio})) := \frac{\sum_{i=1}^n \PVRFunc(t, l; \VRatio_i) p(l \;\cond\;\VRatio_i)}{\sum_{i=1}^n p(l \;\cond\; \VRatio_i)},
\end{equation}
where $p(l \;\cond\; \VRatio_i)$ is follows \eqref{eq:tau_given_lambda_density}.

Suppose we know $\boldsymbol{\VRatio}$ and could use $\PVRFunc(\Tbf_i, \SRatio_i; G(\boldsymbol{\VRatio}))$ to test the null hypothesis $H_i: \mu_{iA} = \mu_{iB}$, we would still have the frequentist guarantees shown by the following theorem.

\begin{prop}
    \label{thm: compound_oracle_1D_average}
    Consider $n$ independent draws from the model \eqref{eq:ss_distribution} with fixed $\boldsymbol{\VRatio}$. Denote the null indices by $\mathcal{H}_0 := \{1 \leq i \leq n: \mu_{iA} = \mu_{iB}\}$. The oracle compound p-values $\PVRFunc(\Tbf_i, \hat{\lambda}_i; G(\boldsymbol{\VRatio}))$ \eqref{eq:compound_oracle_1D_p_value} satisfy the following guarantee:
    \begin{equation}
        \label{eq: compound_oracle_1D_average}
        \frac{1}{n} \sum_{i \in \mathcal{H}_0} \PP[\boldsymbol{\VRatio}]{\PVRFunc(\Tbf_i, \hat{\lambda}_i; G(\boldsymbol{\VRatio})) \leq \alpha} = \frac{1}{n} \sum_{i \in \mathcal{H}_0} \PP[\VRatio_i]{\PVRFunc(\Tbf_i, \hat{\lambda}_i; G(\boldsymbol{\VRatio})) \leq \alpha} \leq \alpha \text{ for all } \alpha \in [0, 1].
    \end{equation}
\end{prop}

\begin{proof}
Notice that for $i \in \mathcal{H}_0$, $\Tbf_i = \TbfNull_i$. Then,
\begin{align*}
    \frac{1}{n} \sum_{i \in \mathcal{H}_0} \PP[\VRatio_i]{\PVRFunc(\Tbf_i, \hat{\lambda}_i; G(\boldsymbol{\VRatio})) \leq \alpha} &\leq \frac{1}{n} \sum_{i=1}^n \PP[\VRatio_i]{\PVRFunc(\TbfNull_i, \hat{\lambda}_i; G(\boldsymbol{\VRatio})) \leq \alpha}\\ &= \PP[G(\boldsymbol{\VRatio})]{\PVRFunc(\TbfNull, \hat{\lambda}; G(\boldsymbol{\VRatio})) \leq \alpha}.
\end{align*}
The last equality is only formal. By Proposition \ref{prop:1D_oracle_uniform}, it holds 
that $\PP[G(\boldsymbol{\VRatio})]{\PVRFunc(\TbfNull, \hat{\lambda}; G(\boldsymbol{\VRatio})) \leq \alpha} = \alpha$ for all $\alpha \in [0, 1]$, which allows us to conclude.
\end{proof}

Next, we will drop the assumption that we know $\boldsymbol{\VRatio}$ and show that most of the asymptotic results of Section \ref{sec: Variance ratios as nuisance parameters} have an analogue in the compound setting wherein $\boldsymbol{\VRatio}$ is fixed. First, we state the matching result to the Proposition \ref{prop:1D_weak_convergence}:
\begin{prop}
    \label{prop:compound_1D_weak_convergence}
    Suppose that $\nu_A, \nu_B \geq 2$. Also suppose that there exist $L,U \in (0,\infty)$ such that $\lambda_i \in [L, U]$ for all $i = 1, \dots, n$. Then, for any bounded and continuous function $\psi:[0,\infty) \to \mathbb{R},$
    $$ \int \psi(u)\dd \hat{G}(u)-\int\psi(u)\dd G(\boldsymbol{\lambda)}(u) \to 0\; \text{ almost surely as }\; n \to \infty.
    $$
\end{prop}
\begin{proof}
We write $G(\boldsymbol{\lambda}_n)$ to clarify its dependency on $n$.
Lemma~\ref{lemma:properties of NPMLE VR}(\ref{lemma:existence NPMLE VR}) ensures that we can derive the NPMLE $\hat{G}$ given $n$ independent observations $(\SRatio_1, \dots, \SRatio_n)$, which we will denote as $\hat{G}_n$. The identifiability results from \cite{teicher1961identifiability} ensures that if $f_{G_1} = f_{G_2}$ almost everywhere, then $G_1 = G_2$. Let $\hat{F}_n$ be the empirical distribution function associated with ($\SRatio_1, \dots, \SRatio_n$). We will show that there exists a set $B$ with probability $1$, such that for each $\omega \in B$, given any subsequence $a_{n_k}(\omega):= \int \psi(u)\dd \hat{G}_{n_k}(u;\omega)-\int\psi(u)\dd G(\boldsymbol{\lambda}_{n_k})(u)$, there exists a further subsequence that converges to 0~\citep[Theorem 2.6]{billingsley1999probability}.

Since $\{G(\boldsymbol{\lambda}_{n})\}$ is supported on $[L,U]$, by Helly's theorem \citep[Theroem 25.9]{billingsley1995probability}, there exists a subsequence $\{n_k\}$ such that $\{G(\boldsymbol{\lambda}_{n_k})\}$ converges weakly to a positive measure $G_0$ with total mass $1$. Again by Helly's theorem, the sequence $\{\hat{G}_{n_k}(\cdot;\omega)\}$ has a subsequence $\{\hat{G}_{n_{k_l}}(\cdot;\omega)\}$ converging vaguely to a sub-distribution function $\tilde{G}$. We will show that $\tilde{G}=G_0$.

Notice that we can view $\hat{\lambda}_i = \lambda_iV_i$ with $\lambda_i \in [L,U]$ fixed, and $V_i$ i.i.d. having c.d.f. $F_{\nu_A,\nu_B}$. Let $P_i$ be the law of $\hat{\lambda}_i$, define $\hat{P}_n := (1/n)\sum_{i=1}^n\delta_{\hat{\lambda}_i}$, $\bar{P}_n := (1/n)\sum_{i=1}^nP_i$. Also define $$\hat{F}_n(t):=\int \mathbf{1}_{(-\infty,t]}\dd \hat{P}_n = \frac{1}{n}\sum_{i=1}^n\mathbf{1}_{\{\hat{\lambda}_i \leq t\}},\quad \bar{F}_n(t) :=\int \mathbf{1}_{(-\infty,t]}\dd \bar{P}_n=\frac{1}{n}\sum_{i=1}^np_{\lambda_i}(t),$$ where $p_{\lambda}(t) := F_{\nu_A,\nu_B}(t/\lambda)$. 

\begin{lemm}
    \label{lemm:GC-noniid}
    It holds that:
    $$
    \sup_{t\in\mathbb{R}}\abs{\hat{F}_n(t)-\bar{F}_n(t)} \to 0 \quad \text{almost surely}.
    $$
\end{lemm}
\begin{proof}
For any $M >0$, we have
$$
\bar{P}_n([0,M]^c) = \frac{1}{n}\sum_{i=1}^n\PP[]{V_i>\frac{M}{\lambda_i}} \leq \PP[]{V >\frac{M}{U}} \to 0 \quad \text{as } M \to \infty.
$$
Thus, $\{\bar{P}_n\}$ is tight. Therefore, by applying \citep[Theorem 1]{WELLNER1981309}, we have $\beta(\hat{P}_n,\bar{P}_n) \to 0$ almost surely, where 
$$
\beta(P,Q) = \sup\left\{\abs{\int g\dd(P-Q)}:\lVert g\rVert_{BL} \leq 1\right\}
$$
with $\lVert g\rVert_{\infty} = \sup_x\abs{g(x)}$, $\lVert g\rVert_{L} =\sup_{x\neq y}\abs{g(x)-g(y)}/d(x,y)$, $\lVert g\rVert_{BL} =\lVert g\rVert_{\infty}+ \lVert g\rVert_{L}$. Then for any bounded Lipschitz function $f$, 
\begin{equation}
\label{eq:bound_BL}
    \abs{\int f \dd(P-Q)} \leq \lVert f\rVert_{BL}\beta(P,Q).
\end{equation}
For $t \in \mathbb{R}$ and $\eta>0$, define the following two piecewise-linear functions
$$
\phi_{t, \eta}^{-}(x)=\left\{\begin{array}{ll}
1, & x \leq t-\eta, \\
1-\frac{x-(t-\eta)}{\eta}, & t-\eta<x<t, \\
0, & x \geq t,
\end{array} \quad \phi_{t, \eta}^{+}(x)= \begin{cases}1, & x \leq t, \\
1-\frac{x-t}{\eta}, & t<x<t+\eta, \\
0, & x \geq t+\eta .\end{cases}\right.
$$
They satisfy
$$
\phi_{t, \eta}^{-} \leq \mathbf{1}_{(-\infty, t]} \leq \phi_{t, \eta}^{+}, \quad\left\|\phi_{t, \eta}^{ \pm}\right\|_{\infty} \leq 1, \quad\left\|\phi_{t, \eta}^{ \pm}\right\|_L=\frac{1}{\eta}.
$$
Using the first inequality, we have
\begin{align*}
& \hat{F}_n(t)-\bar{F}_n(t) \leq \int \phi_{t, \eta}^{+} d\left(\hat{P}_n-\bar{P}_n\right)+\int\left(\phi_{t, \eta}^{+}-\mathbf{1}_{(-\infty, t]}\right) d \bar{P}_n, \\
& \bar{F}_n(t)-\hat{F}_n(t) \leq \int \phi_{t, \eta}^{-} d\left(\bar{P}_n-\hat{P}_n\right)+\int\left(\mathbf{1}_{(-\infty, t]}-\phi_{t, \eta}^{-}\right) d \bar{P}_n.
\end{align*}
The first term here can be bounded using \eqref{eq:bound_BL}:
$$
\left|\int \phi_{t, \eta}^{ \pm} d\left(\hat{P}_n-\bar{P}_n\right)\right| \leq\left\|\phi_{t, \eta}^{ \pm}\right\|_{BL} \beta\left(\hat{P}_n, \bar{P}_n\right)=\left(1+\frac{1}{\eta}\right) \beta\left(\hat{P}_n, \bar{P}_n\right).
$$
The second term can also be bounded by the definition of $\phi_{t, \eta}^{\pm}$:
$$
\int\left(\phi_{t, \eta}^{+}-\mathbf{1}_{(-\infty, t]}\right) d \bar{P}_n \leq \bar{F}_n(t+\eta)-\bar{F}_n(t), \quad \int\left(\mathbf{1}_{(-\infty, t]}-\phi_{t, \eta}^{-}\right) d \bar{P}_n \leq \bar{F}_n(t)-\bar{F}_n(t-\eta).
$$
Therefore, we have
\begin{align}
    \abs{\hat{F}_n(t)-\bar{F}_n(t)} &=\max\{\hat{F}_n(t)-\bar{F}_n(t),\bar{F}_n(t)-\hat{F}_n(t) \}\nonumber\\
    &\leq \left(1+\frac{1}{\eta}\right) \beta\left(\hat{P}_n, \bar{P}_n\right)+\max\left\{\bar{F}_n(t+\eta)-\bar{F}_n(t),\bar{F}_n(t)-\bar{F}_n(t-\eta)\right\} \nonumber\\
    &\leq \left(1+\frac{1}{\eta}\right) \beta\left(\hat{P}_n, \bar{P}_n\right)+ \sup_{u\in\mathbb{R}}\left\{\bar{F}_n(u+\eta)-\bar{F}_n(u-\eta)\right\}. \label{eq:BL-cdf-bound}
\end{align}
Since $F_{\nu_A,\nu_B}$ has bounded density when $\nu_A,\nu_B \geq 2$, and $\lambda \in [L,U]$, then $\abs{p_{\lambda}(t+\eta)-p_{\lambda}(t)}\leq\eta\lVert f_{\nu_A,\nu_B}\rVert_{\infty}/L$, hence 
$$
\sup_{u\in\mathbb{R}}\left(\bar{F}_n(u+\eta)-\bar{F}_n(u-\eta)\right) \leq\frac{2\eta}{L}\lVert f_{\nu_A,\nu_B}\rVert_{\infty} =: C_{F}\eta
$$
uniformly in $n$. Since \eqref{eq:BL-cdf-bound} holds for all $t\in \mathbb{R}$, by taking supremum over $t$ and $\beta(\hat{P}_n,\bar{P}_n) \to 0$, 
$$
\sup_{t\in\mathbb{R}}\abs{\hat{F}_n(t)-\bar{F}_n(t)}\leq\left(1+\frac{1}{\eta}\right) \beta\left(\hat{P}_n, \bar{P}_n\right)+C_{F}\eta \to C_{F}\eta \quad \text{as }n\to \infty \text{ almost surely.}
$$
Since the choice of $\eta$ is arbitrary, letting $\eta \downarrow 0$, then we prove that as $n \to \infty$,
$$
\sup_{t\in\mathbb{R}}\abs{\hat{F}_n(t)-\bar{F}_n(t)} \to 0 \quad \text{almost surely}.
$$
\end{proof}

By weak convergence of $G(\boldsymbol{\lambda}_{n_{k_l}})$ to $G_0$, then for fixed $t$,
$$
\bar{F}_{n_{k_l}}(t) = \int p(t \;\cond\; \lambda)\dd G(\boldsymbol{\lambda}_{n_{k_l}})(\lambda)\to \int p(t  \;\cond\; \lambda)\dd G_0 (\lambda)=: F_0(t),
$$
where we define $p(t\;\cond\;\lambda) = 0$ for all $t < 0$. Since both $G_0,G(\boldsymbol{\lambda)} \in \mathcal{G}$, similar to what we have proved in Lemma \ref{lemma:bound II}, we can show the continuity of $\bar{F}_{n_{k_l}}(t)$ with respect to $t$ and the equicontinuity of $\{\bar{F}_{n_{k_l}}\}$ on the compact set $[-T,T]$ for any $T>0$. Again, by Arzelà-Ascoli \citep[Theorem 7.25]{Rudin1976Analysis}, we can prove that any sequence of $\{\bar{F}_{n_{k_l}}\}$ has a uniformly convergent subsequence on the compact set $[-T,T]$ whose limit is continuous. Since the pointwise convergence result ensures that the only possible limit for the uniformly convergent subsequence of $\{\bar{F}_{n_{k_l}}\}$ is $F_0$, it is automatically uniform on $[-T,T]$. Let $T \to \infty$ and combined with the fact that both functions have a limit $1$ as $t\to \infty$, and a limit $0$ as $t\to -\infty$, we have
$$
\sup_{t \in \mathbb{R}}\abs{\bar{F}_{n_{k_l}}(t)-F_0(t)} \to 0.
$$
Therefore, by combining the results, we have
\begin{equation}
\label{eq:convergence_empirical_cdf}
    \left\lVert  \hat{F}_{n_{k_l}}(\cdot)-F_0(\cdot) \right\rVert_{\infty}\leq \left\lVert \hat{F}_{n_{k_l}}(\cdot)-\bar{F}_{n_{k_l}}(\cdot)\right\rVert_{\infty} +\left\lVert\bar{F}_{n_{k_l}}(\cdot)-F_0(\cdot)\right\rVert_{\infty} \to 0 \quad \text{almost surely}.
\end{equation}
By Lemma \ref{lemm:GC-noniid}, we have 
\begin{equation}
\label{noniid-Glivenko-Cantelli}
    \PP[]{B} = 1, B := \left\{\omega: \left\lVert \hat{F}_{n}(\cdot;\omega)-\bar{F}_n(\cdot)\right\rVert_{\infty}:=\sup_{\substack{t\in \mathbb{R}}} \abs{\hat{F}_{n}(t;\omega)-\bar{F}_n(t)} \to 0, \text{ as } n \to \infty\right\}.
\end{equation}
For $\kappa > 1$, define a compact set 
$$
A_{\kappa} := \left\{\hat{\lambda} \in (0,\infty): f_{G_0}(\hat{\lambda}) \geq \frac{1}{\kappa}\right\}\cap\left[\frac{1}{\kappa},\kappa\right] \text{ \quad s.t. }F_0(A_{\kappa})=c_{\kappa} >0.
$$
Fix $\kappa$, and let $\omega \in B$, then by \eqref{eq:convergence_empirical_cdf} and \eqref{noniid-Glivenko-Cantelli}, we have
\begin{equation}
\label{limit:A_kappa_F0}
    \hat{F}_{n_{k_l}}(A_{\kappa};\omega) \to F_0(A_{\kappa}) > 0 \text{ as } l \to \infty.
\end{equation}
Fix $\omega \in B$, and we write $\hat{G}_n(\cdot)=\hat{G}_n(\cdot;\omega)$, $ \hat{F}_n(\cdot)=\hat{F}_n(\cdot;\omega)$. Since $G_0\in\mathcal{G}$ , recall Lemma~\ref{lemma:properties of NPMLE VR}(\ref{lemma:inequality NPMLE VR}) and with the definition of $\hat{F}_n$,
$$\frac{1}{n}\sum_{i=1}^n \frac{f_{G_0}(\SRatio_i)}{f_{\hat{G}_n}(\SRatio_i)}=\int_0^{\infty} \frac{f_{G_0}(\SRatio)}{f_{\hat{G}_n}(\SRatio)}  \dd \hat{F}_n(\SRatio) = \int_0^\infty r_n(\SRatio)  \dd \hat{F}_n(\SRatio) \leq 1,$$
where $r_n(\SRatio) = \frac{f_{G_0}(\SRatio)}{f_{\hat{G}_n}(\SRatio)}$. Also define $r(\hat{\lambda}) = \frac{f_{G_0}(\SRatio)}{f_{\tilde{G}}(\SRatio)}$.

Therefore, the following statement also holds:

$$\lim_{\substack{l \to \infty}} \int_{A_{\kappa}}r_{n_{k_l}}(\hat{\lambda})\dd \hat{F}_{n_{k_l}}(\hat{\lambda}) = \int_{A_{\kappa}} r(\hat{\lambda})\dd F_0(\hat{\lambda}) \leq 1.$$
This can be proved using Lemma \ref{limit:KW_l}, only by replacing $F$ with $F_0$, $G$ with $G_0$, and use \eqref{limit:A_kappa_F0}.

Since the choice of $\kappa$ is arbitrary, based on the monotone convergence theorem, by taking $\kappa \to \infty$, we have
$$
\lim_{\substack{\kappa\to \infty}}\int_{A_{\kappa}} r(\hat{\lambda}) \dd F_0(\hat{\lambda})=\int \frac{f_{G_0}(\hat{\lambda})}{f_{\tilde{G}}(\hat{\lambda})} \dd F_0(\hat{\lambda}) = \int\frac{f_{G_0}^2(\hat{\lambda})}{f_{\tilde{G}}(\hat{\lambda})}\dd \hat{\lambda} \leq 1.
$$
By rearranging the inequality,
$$
0 \geq\int\frac{f_{G_0}^2(\hat{\lambda})}{f_{\tilde{G}}(\hat{\lambda})}\dd \hat{\lambda} -1 = \int \left(\frac{f_{G_0}(\hat{\lambda})}{f_{\tilde{G}}(\hat{\lambda})}-1\right)^2 f_{\tilde{G}}(\hat{\lambda}) \dd \hat{\lambda} \geq 0,
$$
with the equality holds iff $f_{G_0} = f_{\tilde{G}}$ almost surely. Since both $f_{G_0}$ and $f_{\tilde{G}}$ are continuous, we must have $f_{G_0} = f_{\tilde{G}}$. Combined with the identifiability results, we have $\tilde{G} = G_{0}$.

Now, we have that $\{\hat{G}_{n_{k_l}}\}$ converges vaguely to $G_0$, which has total mass $1$, then by definition we have $\{\hat{G}_{n_{k_l}}\}$ converges weakly to $G_0$. Thus, for any bounded and continuous function $\psi:[0,\infty) \to \mathbb{R}$, we have
$$
\int \psi(u)\dd \hat{G}_{n_{k_l}}(u) \to \int \psi(u)\dd G_0(u),\quad \int \psi(u)\dd G(\boldsymbol{\lambda}_{n_{k_l}})(u) \to \int \psi(u)\dd G_0(u) \quad \text{as } l \to \infty,
$$
and we can conclude
\begin{align*}
    &~~~~\abs{\int \psi(u)\dd \hat{G}_{n_{k_l}}(u)-\int \psi(u)\dd G(\boldsymbol{\lambda}_{n_{k_l}})(u)} \\
    &\leq \abs{\int \psi(u)\dd \hat{G}_{n_{k_l}}(u) - \int \psi(u)\dd G_0(u)} + \abs{\int \psi(u)\dd G(\boldsymbol{\lambda}_{n_{k_l}})(u) - \int \psi(u)\dd G_0(u)} \to 0 \quad \text{as } l \to \infty.
\end{align*}
Therefore, we show that there is a set $B$ with probability $1$ such that for each $\omega \in B$, every subsequence of the $\int \psi(u)\dd \hat{G}_{n_{k_l}}(u;\omega)-\int \psi(u)\dd G(\boldsymbol{\lambda}_{n_{k_l}})(u)$ has a convergent subsequence, and all these subsequences converge to 0. Then we complete the proof.

\end{proof}
After obtaining NPMLE $\hat{G}$ from \eqref{eq:1D_optimization}, we plug it in \eqref{eq:oracle_1D_p_value} and obtain the empirical partially Bayes p-values $\PVR_i = \PVRFunc(\Tbf_i, \SRatio_i; \hat{G})$. Building on Proposition~\ref{prop:compound_1D_weak_convergence} , we can further show that the empirical partially Bayes p-values are consistent for the corresponding oracle compound p-values (that have knowledge of the empirical distribution $G(\boldsymbol{\VRatio})$).  
The following theorem, which parallels Theorem \ref{thm:1D_uniform_convergence_p_value}, also follows immediately:
\begin{theo}
    \label{thm:compound_1D_uniform_convergence_p_value}
    Under the conditions of Proposition~\ref{prop:compound_1D_weak_convergence}, it holds that:
    $$\lim_{n \to \infty} \frac{1}{n} \sum_{i=1}^n \EE[\boldsymbol{\VRatio}]{\abs{\PVR_i - \PVRFunc(\Tbf_i, \SRatio_i; G(\boldsymbol{\VRatio}))}} = 0.$$
\end{theo}
The difference between Theorem \ref{thm:compound_1D_uniform_convergence_p_value} and Theorem \ref{thm:1D_uniform_convergence_p_value} is that the convergence concerns the expectation averaged over all $i = 1, \dots, n$. The formal proof is provided concurrently with that of Theorem \ref{thm:1D_uniform_convergence_p_value} in Supplement \ref{proof:thm_1D_uniform_convergence_p_value}.

Finally, as a consequence of Theorem \ref{thm:compound_1D_uniform_convergence_p_value} and Proposition \ref{thm: compound_oracle_1D_average}, we can conclude that empirical partially Bayes p-values $\PVR_i$ are asymptotic compound p-values. The proof is provided below.
 
\begin{proof}
    First, for any fixed $\delta \in (0, 1)$, by Lemma \ref{lem:indicator}, it can be shown that for any $\alpha \leq 1 - \delta$:
$$\mathbf{1}(\PVR_i \leq \alpha) \leq \frac{1}{\delta} \abs{\PVR_i - \PVRFunc(\Tbf_i, \SRatio_i; G(\boldsymbol{\VRatio}))} + \mathbf{1}(\PVRFunc(\Tbf_i, \SRatio_i; G(\boldsymbol{\VRatio})) \leq \alpha + \delta).$$
Thus, summing the above inequality over all null hypotheses $i \in \mathcal{H}_0$, and then taking the expectation on both sides of the inequality, we obtain:
\begin{align*}
\frac{1}{n} \sum_{i \in \mathcal{H}_0} \PP[\boldsymbol{\lambda}]{\PVR_i \leq \alpha} &\leq \frac{1}{n} \sum_{i \in \mathcal{H}_0} \PP[\boldsymbol{\lambda}] {\PVRFunc(\Tbf_i, \SRatio_i; G(\boldsymbol{\VRatio})) \leq \alpha + \delta} \\
& \quad + \frac{1}{\delta} \frac{1}{n}\sum_{i \in \mathcal{H}_0} \EE[\boldsymbol{\lambda}]{\abs{\PVR_i - \PVRFunc(\Tbf_i, \SRatio_i; G(\boldsymbol{\VRatio}))}}.
\end{align*}
Based on the average significance control property of the oracle compound p-values in Proposition \ref{thm: compound_oracle_1D_average}, i.e. $\frac{1}{n} \sum_{i \in \mathcal{H}_0} \PP[\boldsymbol{\lambda}]{\PVRFunc(\Tbf_i, \SRatio_i; G(\boldsymbol{\VRatio})) \leq \alpha'} \leq \alpha'$ for any $\alpha' \in [0,1]$, we further have:
$$\frac{1}{n} \sum_{i \in \mathcal{H}_0} \PP[\boldsymbol{\lambda}]{\PVR_i \leq \alpha} \leq \alpha + \delta + \frac{1}{\delta} \frac{1}{n}\sum_{i \in \mathcal{H}_0} \EE[\boldsymbol{\lambda}]{\abs{\PVR_i - \PVRFunc(\Tbf_i, \SRatio_i; G(\boldsymbol{\VRatio}))}}.$$
By Theorem \ref{thm:compound_1D_uniform_convergence_p_value}, we can show that as $n \to \infty$
\begin{align*}
\frac{1}{\delta} \frac{1}{n}\sum_{i \in \mathcal{H}_0} \EE[\boldsymbol{\lambda}]{\abs{\PVR_i - \PVRFunc(\Tbf_i, \SRatio_i; G(\boldsymbol{\VRatio}))}} \leq \frac{1}{\delta} \frac{1}{n}\sum_{i=1}^n \EE[\boldsymbol{\lambda}]{\abs{\PVR_i - \PVRFunc(\Tbf_i, \SRatio_i; G(\boldsymbol{\VRatio}))}} \to 0,
\end{align*}
which leads to:
$$\limsup_{n\to\infty}\frac{1}{n} \sum_{i \in \mathcal{H}_0} \PP[\boldsymbol{\lambda}]{\PVR_i \leq \alpha} \leq \alpha + \delta.$$
Since the choice of $\delta$ can be arbitrarily small, we can take $\delta \downarrow 0$ and conclude the proof.
\end{proof}

\section{Properties of the dual variance NPMLE (\ref{eq:2D_optimization})}
\label{sec:properties DV NPMLE}
For simplicity, we denote $\sigma^2_P = (\sigma^2_A, \sigma^2_B)$, $\hat{\sigma}^2_P = (\hat{\sigma}^2_A, \hat{\sigma}^2_B)$ throughout this section.

Before showing properties of the NPMLE \eqref{eq:2D_optimization}, we first setup the context and provide definitions that are important for clarity. 

First, for $\hat{\sigma}^2_A, \hat{\sigma}^2_B > 0$ and $\sigma^2_A, \sigma^2_B > 0$, we let
$p(\hat{\sigma}^2_A, \hat{\sigma}^2_B| \sigma^2_A, \sigma^2_B, \nu_A, \nu_B)$ denote the density of $(\hat{\sigma}^2_A, \hat{\sigma}^2_B)$ given $(\sigma^2_A, \sigma^2_B)$
(explicit form given in Supplement \ref{appendix: supplementary formulas}). For any finite positive measure $H$ on $[0,\infty) \times [0,\infty)$, we define
the mixture density
$$f_{H}\left(\hat{\sigma}^2_{A}, \hat{\sigma}^2_{B}\right) \equiv f_{H}\left(\hat{\sigma}^2_{A}, \hat{\sigma}^2_{B};\;\nu_A, \nu_B\right) = \int_{\mathbb{R}^2_{+}} p(\hat{\sigma}^2_{A}, \hat{\sigma}^2_{B} \; \cond \; \sigma^2_A, \sigma^2_B, \nu_A, \nu_B) \dd H(\sigma^2_A, \sigma^2_B).$$
We further define $\mathcal{H}$ to be the class of finite positive measures on $[0,\infty) \times [0,\infty)$ with total mass
$\le 1$, and let $\mathcal{H}_{\backslash0} \subset \mathcal{H}$ to be a subset of $\mathcal{H}$ within which each measure has no mass at set of points with $\sigma^2_A = 0$ or $\sigma^2_B = 0$.

Recall our definition of the NPMLE \eqref{eq:2D_optimization}: given independent observations
$\left\{(\hat{\sigma}^2_{iA}, \hat{\sigma}^2_{iB})\right\}_{i=1}^n$
from the mixture distribution $f_H(\hat{\sigma}^2_{A}, \hat{\sigma}^2_{B})$, and any $H' \in \mathcal{H}$, the marginal log-likelihood is
$$\sum_{i=1}^n \log f_{H'}(\hat{\sigma}^2_{iA}, \hat{\sigma}^2_{iB}),$$
and we call any maximizer $\hat{H} \in \mathcal{H}$ such that
$$\hat{H} \in \argmax\left\{\sum_{i=1}^n \log f_{H'}(\hat{\sigma}^2_{iA}, \hat{\sigma}^2_{iB}): H' \in \mathcal{H}\right\}$$
a dual variance NPMLE.

Next, under the above context, we record three basic properties of the NPMLE $\hat{H}$ \eqref{eq:2D_optimization} that will be used throughout.

\begin{lemm}
    \label{lemma:properties of NPMLE DV}
    For the dual variance NPMLE $\hat{H}$, the following properties hold:
    
    \begin{enumerate}
        \item \label{lemma:existence NPMLE DV} $\hat{H}$ exists with total mass $1$, furthermore, $\hat{H} \in \mathcal{H}_{\backslash0}$. 

        \item \label{lemma:inequality NPMLE DV} For every $H \in \mathcal{H}$,
        \begin{equation}
            \label{eq:beta_inequality_DV}
            \frac{1}{n} \sum_{i=1}^n  \frac{f_H(\hat{\sigma}^2_{iA}, \hat{\sigma}^2_{iB})}{f_{\hat{H}}(\hat{\sigma}^2_{iA}, \hat{\sigma}^2_{iB})} \leq 1 .
        \end{equation}

        \item \label{lemma:support NPMLE DV} $\hat{H}$ is supported on the interval $\left[\min_i \{\hat{\sigma}^2_{iA}\}, \max_i \{\hat{\sigma}^2_{iA}\}\right] \times \left[\min_i \{\hat{\sigma}^2_{iB}\}, \max_i \{\hat{\sigma}^2_{iB}\}\right]$.
    \end{enumerate}
\end{lemm}

We comment that the Lemma~\ref{lemma:properties of NPMLE DV}(\ref{lemma:existence NPMLE DV}) ensures a valid maximizing distribution
$\hat{H} \in \mathcal{H}_{\backslash0}$ exists, and the Lemma~\ref{lemma:properties of NPMLE DV}(\ref{lemma:inequality NPMLE DV}) is a
optimality inequality at $\hat{H}$ that leads to the Lemma~\ref{lemma:properties of NPMLE DV}(\ref{lemma:support NPMLE DV}). Lemma~\ref{lemma:properties of NPMLE DV}(\ref{lemma:support NPMLE DV}) lays the foundation for our discretization strategy for the computation of NPMLE in Remark~\ref{rema:dvepb_npmle}. Lemma~\ref{lemma:properties of NPMLE DV}(\ref{lemma:existence NPMLE DV}), Lemma~\ref{lemma:properties of NPMLE DV}(\ref{lemma:inequality NPMLE DV}), together with the identifiability result on the mixture distribution $f_{H}(\hat{\sigma}^2_{iA}, \hat{\sigma}^2_{iB})$ from \citet{teicher1961identifiability, teicher1967identifiability} and \cite{nielsen1965identifiability}, yields the Proposition \ref{prop:2D_weak_convergence} in the main text that establishes the consistency of the NPMLE \eqref{eq:2D_optimization} for the true frequency distribution of $(\sigma^2_{iA}, \sigma^2_{iB})$.

Following the arguments of \citet{jewell1982mixtures}, we first provide the proof for Lemma~\ref{lemma:properties of NPMLE DV} in Supplement \ref{sec:properties DV NPMLE}. Then we will provide the formal proof for Proposition \ref{prop:2D_weak_convergence} in Supplement \ref{prop:2D_weak_convergence}.



\begin{proof}
Then, under the same context above, we observe n independent observations, ($\hat{\sigma}^2_{1P}, \dots, \hat{\sigma}^2_{nP}$), from the mixing distribution $f_{H}(\hat{\sigma}^2_P) \equiv f_{H}(\hat{\sigma}^2_P;\nu_A, \nu_B) = \int_{\mathbb{R}_{+}^2} p(\hat{\sigma}^2_P \; \cond \; \sigma^2_P, \nu_A, \nu_B) \dd H(\sigma^2_P)$. We start by defining the map
$$\psi: \mathcal{H} \rightarrow \mathbb{R}^n, \text{ where } \psi(H') = \left(\psi_1(H'), \dots, \psi_n(H')\right) \text{ with } \psi_i(H') = f_{H'}(\hat{\sigma}^2_{iP};\nu_A, \nu_B).$$
If $A \subseteq \mathbb{R}^n$ is the image of $\mathcal{H}$ under $\psi$, it is easy to verify that $A$ is convex. We then focus on studying the $p(\hat{\sigma}^2_{iP} \; \cond \; \sigma^2_P, \nu_A, \nu_B)$ defined in \eqref{eq:s2_density}, which can be further decomposed as $p(\hat{\sigma}_{iA}^2 \; \cond \; \sigma^2_{A}, \nu_A) \times p(\hat{\sigma}_{iB}^2 \; \cond \; \sigma^2_{B}, \nu_B)$ due to independence, and $p(s^2 \; \cond \; \sigma^2, \nu)$ is the density of the scaled $\chi^2_{\nu}$ given explicitly in \eqref{eq:explicit_s2_given_sigma}.

Consider the image $q: \sigma^2 \rightarrow p(s^2 \; \cond \; \sigma^2, \nu_A, \nu_B)$, we can write $q(\sigma^2)$ explicitly as
$$q(\sigma^2) = C_{\nu} (s^2)^{\frac{\nu}{2} - 1} (\sigma^2)^{-\frac{\nu}{2}} \exp\left(-\frac{\nu s^2}{2 \sigma^2}\right).$$
Thus, it is easy to have 
$$q'(\sigma^2) = C_{\nu} (s^2)^{\frac{\nu}{2} - 1} \frac{\nu}{2}(\sigma^2)^{-\frac{\nu}{2}-1} \exp\left(-\frac{\nu s^2}{2 \sigma^2}\right) \left(\frac{s^2}{\sigma^2} - 1\right).$$
Therefore, $q(\sigma^2) \uparrow$ for $\sigma^2 < s^2$ and $q(\sigma^2) \downarrow$ for $\sigma^2 \geq s^2$, meaning $q(\sigma^2)$ is maximized when $\sigma^2 = s^2$, i.e. 
\begin{equation}
    \label{eq:f(sigma) bounded}
    q(\sigma^2) \leq q(s^2) = C_{\nu} (s^2)^{-\frac{1}{2}} \exp\left(-\frac{\nu}{2}\right) = C_{\nu}' (s^2)^{-1},
\end{equation}
where $C_{\nu}'$ is constant given $\nu$. Thus, $\sigma^2 \rightarrow p(s^2 \; \cond \; \sigma^2, \nu)$ is bounded and continuous, which leads to that $\sigma^2_P \rightarrow p(\hat{\sigma}^2_P \; \cond \; \sigma^2_P, \nu_A, \nu_B)$ is also bounded and continuous. Based on the Helly Bray selection Theorem \citep{breiman1992probability}, we can show that $A$ is compact. Furthermore, define the function $h: A \subseteq \mathbb{R}^n \to \mathbb{R}$ by:
$$
 h(\beta):= \sum_{i=1}^n \log(\beta_i). 
$$
Then $h$ is strictly concave on the compact set $A$ due to the concavity of the log function. Therefore, there must exists unique $\hat{\beta} \in \psi(\mathcal{H})$ such that:
\begin{equation}
    \label{eq:beta_2D_optimization}
    \hat{\beta} \in \text{argmax}\{h(\beta): \beta \in \psi(\mathcal{H})\}.
\end{equation}
Let $\hat{H} \in \mathcal{H}$ be the measure with $\psi(\hat{H}) = \hat{\beta}$, we then want to show that $\hat{H}$ is a probability measure. Suppose that the mass of $\hat{H}$ is $1 - \epsilon$ where $\epsilon > 0$. Let $(t,s)$ be an arbitrary point on $(0, \infty) \times (0, \infty)$ and consider the measure $\Tilde{H} = \hat{H} + \epsilon\delta_{(t,s)}$. Then $\Tilde{H} \in \mathcal{H}$ and $h(\psi(\Tilde{H})) > h(\psi(\hat{H}))$, which is a contradiction to the definition of $\hat{H}$. Thus, all measures $\hat{H} \in \mathcal{H}$ with $\psi(\hat{H}) = \hat{\beta}$ have total mass 1.

Furthermore, we can also show that no measure $\hat{H} \in \mathcal{H}$ with $\psi(\hat{H}) = \hat{\beta}$ have positive mass at any $\sigma^2_{P0}$ with $\sigma^2_A = 0$. First, it is easy to show that $\lim_{\sigma^2 \rightarrow 0} p(s^2 \; \cond \; \sigma^2, \nu) = 0$ for any $s^2 > 0$. Suppose there exists such a $\Tilde{H}^{'}$ with mass $\eta$ at $\sigma^2_{P0}$ with $\sigma^2_A = 0$ where $\eta > 0$, then $\psi(\Tilde{H}^{'} - \eta\delta_{\sigma^2_{P0}}) = \psi(\Tilde{H}^{'}) = \hat{\beta}$. Based on argument above, $\Tilde{H}^{'} - \eta\delta_{\sigma^2_{P0}}$ has total mass 1, which contradicts the fact that $\Tilde{H}^{'}$ has total mass 1. An analogous argument can be made with respect to any $\sigma^2_{P0}$ with $\sigma^2_B = 0$. Therefore, we have established the existence of the MLE $\hat{H}$ which is in $\mathcal{H}_{\backslash0}$, which allows us to conclude Lemma~\ref{lemma:properties of NPMLE DV}(\ref{lemma:existence NPMLE DV}).

Let $H' \in \mathcal{H}$ and $\psi(H') = \beta' \in A$. We define the map
$$g: [0,1] \rightarrow R, g(\epsilon) = h\left(\psi\left((1-\epsilon)\hat{H} + \epsilon H'\right)\right) = \sum_{i=1}^n \log\left((1-\epsilon)\hat{\beta}_i + \epsilon \beta'_i\right).$$
Notice that $g$ is a concave function of $\epsilon$ with the maximum at $\epsilon = 0$. Thus, differentiation with respect to $\epsilon$ leads to same inequality that we derived for the VREPB \eqref{eq:NPMLE inequality}, i.e.
$$\sum_{i=1}^n \frac{\beta'_i}{\hat{\beta}_i} \leq n,$$
which is equivalent to the inequality presented in Lemma~\ref{lemma:properties of NPMLE DV}(\ref{lemma:inequality NPMLE DV}):
\begin{equation}
            \frac{1}{n} \sum_{i=1}^n  \frac{f_{H^{\prime}}(\hat{\sigma}^2_{iP})}{f_{\hat{H}}(\hat{\sigma}^2_{iP})} \leq 1 .
        \end{equation}
At last, consider the special case with $H' = \delta_{\sigma^2_P}$, Lemma~\ref{lemma:properties of NPMLE DV}(\ref{lemma:inequality NPMLE DV}) leads to
\begin{equation}
    \label{eq:xi_lambda_inequality_2D}
    \xi(\sigma^2_P) := \sum_{i=1}^n \frac{p(\hat{\sigma}^2_{iP};\; \sigma^2_P, \nu_A, \nu_B)}{f_{\hat{H}}(\hat{\sigma}^2_{iP})} \leq n \text{ for all } \sigma^2_P \in [0, \infty) \times [0, \infty).
\end{equation}
Integration of LHS of the inequality with respect to $\hat{H}$ is equal to $n$, and so we obtain the equality of the integrated LHS and RHS of \eqref{eq:xi_lambda_inequality_2D} with respect to $\hat{H}$, which implies that
\begin{equation}
    \label{eq:xi_lambda_equality_2D}
    \xi(\sigma^2_P) = \sum_{i=1}^n \frac{p(\hat{\sigma}^2_{iP};\; \sigma^2_P, \nu_A, \nu_B)}{f_{\hat{H}}(\hat{\sigma}^2_{iP})} = n, \quad\hat{H}\text{-almost surely}.
\end{equation}
Absorbing constants with respect to $\sigma^2_P$ for the $i$-th summand into a constant $c_i$, we may write $\xi(\sigma^2_P)$ as:
$$\xi(\sigma^2_P) = \sum_{i=1}^n c_i (\sigma^2_A)^{-\frac{\nu_A}{2}} \exp\left(-\frac{\nu_A \hat{\sigma}^2_{iA}}{2\sigma^2_A}\right) (\sigma^2_B)^{-\frac{\nu_B}{2}} \exp\left(-\frac{\nu_B \hat{\sigma}^2_{iB}}{2\sigma^2_B}\right).$$
Next, the partial derivative of $\xi$ with respect to $\sigma^2_A$ is equal to:
$$\frac{\partial \xi}{\partial \sigma^2_A} = \sum_{i=1}^n c_i \frac{\nu_A}{2} (\sigma^2_A)^{-\frac{\nu_A}{2}-1} \exp\left(-\frac{\nu_A \hat{\sigma}^2_{iA}}{2\sigma^2_A}\right) (\sigma^2_B)^{-\frac{\nu_B}{2}} \exp\left(-\frac{\nu_B \hat{\sigma}^2_{iB}}{2\sigma^2_B}\right)\left(\frac{\hat{\sigma}^2_{iA}}{\sigma^2_A} - 1\right).$$
Now, we are ready to derive further conclusions about the support of $\hat{H}$. Based on the same argument that has been made to derive \eqref{eq:f(sigma) bounded}, it holds that $\partial \xi/\partial \sigma^2_A > 0$ in the range $(0, \min_i\{\hat{\sigma}^2_{iA}\}) \times (0, \infty)$, which means that for any $\sigma^2_B \in (0, \infty)$, $\xi(\sigma^2_P)$ strictly increases with respect to $\sigma^2_A$ in that interval. Now, suppose there exists such $\sigma^2_P = (\sigma^2_A, \sigma^2_B) \in (0, \min_i\{\hat{\sigma}^2_{iA}\}) \times (0, \infty)$ in the support of $\hat{H}$. By \eqref{eq:xi_lambda_equality_2D}, we would have $\xi(\sigma^2_P) = n$, and by the above analysis, for sufficiently small $\epsilon >0$ such that $\sigma_A^2 + \epsilon \in (0,\min_i\{\hat{\sigma}^2_{iA}\})$, we would then have $\xi(\sigma^{2'}_P) > \xi(\sigma^{2}_P) = n$ where $\sigma^{2'}_P = (\sigma^2_A + \epsilon, \sigma^2_B)$, which is in contradiction to \eqref{eq:xi_lambda_inequality_2D}. Thus, it can be concluded that $\hat{H}$ should assign zero mass on $(0, \min_i\{\hat{\sigma}^2_{iA}\}) \times (0, \infty)$. Analogously, it can be argued that $\hat{H}$ assigns zero mass in the range $(\max_i\{\hat{\sigma}^2_{iA}\}, \infty) \times (0, \infty)$ as for the latter interval, it always holds that $\partial \xi/\partial \sigma^2_A < 0$. Furthermore, an analogous analysis on $\partial \xi/\partial \sigma^2_B$ would conclude that $\hat{H}$ assigns zero mass in the range $(0, \infty) \times (0, \min_i\{\hat{\sigma}^2_{iB}\})$ and the range $(0, \infty) \times (\max_i\{\hat{\sigma}^2_{iB}\}, \infty)$, which allows us to conclude Lemma~\ref{lemma:properties of NPMLE DV}(\ref{lemma:support NPMLE DV}).
\end{proof}

\section{Proofs for properties of DVEPB in Section \ref{sec: Groupwise variances as nuisance parameters}}
For simplicity, we denote $\sigma^2_P = (\sigma^2_A, \sigma^2_B)$, $\hat{\sigma}^2_P = (\hat{\sigma}^2_A, \hat{\sigma}^2_B)$, $s^2_P = (s^2_A, s^2_B)$ throughout the proofs in this section.

\subsection{Proof of Proposition \ref{prop:2D_weak_convergence}}

\label{proof:prop_2D_weak_convergence}

\begin{proof}

Lemma~\ref{lemma:properties of NPMLE DV}(\ref{lemma:existence NPMLE DV}) ensures that we can derive the MLE $\hat{H}$ given $n$ independent observation $(\hat{\sigma}^2_{1P},  \dots, \hat{\sigma}^2_{nP})$, which we will denote as $\hat{H}_n$. The identifiability results from \citet{teicher1961identifiability, teicher1967identifiability} and \cite{nielsen1965identifiability} ensures that if $f_{H_1} = f_{H_2}$ almost everywhere, then $H_1 = H_2$. We will show that there exists a set $B$ with probability 1 such that for each $\omega\in B$, given any subsequence $\{\hat{H}_{n_k}(\cdot;\omega)\}$ of $\{\hat{H}_{n}(\cdot;\omega)\}$ there exists a further subsequence which converges weakly to $H$, then we will complete the proof~\citep[Theorem 2.6]{billingsley1999probability}.

By Helly's theorem \citep[Theorem 25.9]{billingsley1995probability}, the sequence $\{\hat{H}_{n_k}(\cdot;\omega)\}$ has a subsequence $\{\hat{H}_{n_{k_l}}(\cdot;\omega)\}$ converging vaguely to a sub-distribution function $\tilde{H}$. We then have to show that $\tilde{H} = H$.

Let $\hat{F}_n$ be the empirical distribution function associated with ($\hat{\sigma}^2_{1P}, \dots, \hat{\sigma}^2_{nP}$), $F$ be the true underlying distribution function of $\SRatio$ under $H$, i.e., F has density $f_H \equiv f_F$. By the Glivenko-Cantelli theorem, we have
$$\PP[]{B} = 1, \text{ where } B := \left\{\omega: \left\lVert \hat{F}_n(\cdot;\omega)-F(\cdot)\right\rVert_{\infty}:=\sup_{\substack{l\in \mathbb{R}^2}} \abs{\hat{F}_n(l;\omega)-F(l)} \to 0, \text{ as }n \to \infty\right\}.$$
For $\kappa > 1$, define a compact set 
$$
A_{\kappa} := \left\{\hat{\sigma}_P^2 \in \mathbb{R}_+^2: f_H(\hat{\sigma}_P^2) \geq \frac{1}{\kappa}\right\}\cap\left[\frac{1}{\kappa},\kappa\right]^2 \text{ \quad s.t. }F(A_{\kappa})=c_{\kappa} >0.
$$
Fixed $\kappa$, and let $\omega \in B$, then
\begin{equation}
\label{limit:A_kappa_H}
    \hat{F}_n(A_{\kappa};\omega) \to F(A_{\kappa}) > 0 \text{ as } n \to \infty.
\end{equation}
Fix $\omega \in B$, and we write $\hat{H}_n(\cdot)=\hat{H}_n(\cdot;\omega)$, $ \hat{F}_n(\cdot)=\hat{F}_n(\cdot;\omega)$. Recall Lemma~\ref{lemma:properties of NPMLE DV}(\ref{lemma:inequality NPMLE DV}) and by the definition of $\hat{F}_n$, it holds that
$$\frac{1}{n}\sum_{i=1}^n \frac{f_{H}(\hat{\sigma}_{iP}^2)}{f_{\hat{H}_n}(\hat{\sigma}_{iP}^2)} =\int_{\mathbb{R}_+^2} \frac{f_{H}(\hat{\sigma}_P^2)}{f_{\hat{H}_n}(\hat{\sigma}_P^2)}  \dd \hat{F}_n(\hat{\sigma}_P^2) = \int_{\mathbb{R}_+^2} r_n(\hat{\sigma}_P^2)  \dd \hat{F}_n(\hat{\sigma}_P^2) \leq 1,$$
where $r_n(\hat{\sigma}_P^2) = \frac{f_{H}(\hat{\sigma}_P^2)}{f_{\hat{H}_n}(\hat{\sigma}_P^2)}$. Also define $r(\hat{\sigma}_P^2) := \frac{f_{H}(\hat{\sigma}_P^2)}{f_{\tilde{H}}(\hat{\sigma}_P^2)}$.
\begin{lemm} For any fixed $\kappa$, we have:
    $$\lim_{\substack{l \to \infty}} \int_{A_{\kappa}}r_{n_{k_l}}(\hat{\sigma}_P^2)\dd \hat{F}_{n_{k_l}}(\hat{\sigma}_P^2) = \int_{A_{\kappa}} r(\hat{\sigma}_P^2)\dd F(\hat{\sigma}_P^2) \leq 1$$
\end{lemm}
\begin{proof}
For notation simplicity, we rename $\{\hat{H}_{n_{k_l}}\}$ to $\{\hat{H}_n\}$, and assume that $\hat{H}_n$ converges vaguely to $\tilde{H}$. Also rename $\{\hat{F}_{n_{k_l}}\}$ to $\{\hat{F}_n\}$ with \eqref{limit:A_kappa_H} holds. Denote $q_\nu(\hat{\sigma}^2,\sigma^2) := p(\hat{\sigma}^2|\sigma^2,\nu)$, therefore $t(\hat{\sigma}_P^2,\sigma_P^2) := p(\hat{\sigma}_P^2|\sigma_P^2,\nu_A,\nu_B) = q_{\nu_A}(\hat{\sigma}_A^2,\sigma_A^2)\times q_{\nu_B}(\hat{\sigma}_B^2,\sigma_B^2)$. For every fixed $\hat{\sigma}^2\in (0,\infty)$ when $\nu \geq 2$, $q_\nu(\hat{\sigma}^2,\sigma^2)$ is continuous on $(0,\infty)$, with $\lim_{\sigma^2\to0}q_\nu(\hat{\sigma}^2,\sigma^2)=0$ and $\lim_{\sigma^2\to\infty}q_\nu(\hat{\sigma}^2,\sigma^2)=0$. Properties also hold for $t(\hat{\sigma}_P^2,\sigma_P^2)$ for every fixed $\hat{\sigma}_P^2 \in \mathbb{R}_+^2$. By the definition of vague convergence,
$$
f_{\hat{H}_{n}}(\hat{\sigma}_P^2)=\int t(\hat{\sigma}_P^2,\sigma_P^2) d \hat{H}_{n}(\sigma_P^2) \longrightarrow \int t(\hat{\sigma}_P^2,\sigma_P^2) d \tilde{H}(\sigma_P^2)=f_{\tilde{H}}(\hat{\sigma}_P^2), \quad \text{as } n \to \infty.
$$
Also notice that $t(\hat{\sigma}_P^2,\sigma_P^2)$ is strictly positive and bounded on $\sigma_P^2 \in \mathbb{R}_+^2$, therefore both $f_{\hat{H}_{n}}$ and $f_{\tilde{H}}$ are positive finite, we have
$$
r_n(\hat{\sigma}_P^2) =  \frac{f_{H}(\hat{\sigma}_P^2)}{f_{\hat{H}_n}(\hat{\sigma}_P^2)} \to \frac{f_{H}(\hat{\sigma}_P^2)}{f_{\tilde{H}}(\hat{\sigma}_P^2)} = r(\hat{\sigma}_P^2),\quad \text{as } n \to \infty.
$$
For a fixed $\kappa > 1$,
\begin{equation}
    \label{eq:kappa_inequality_H}
    \int_{A_{\kappa}} r_n(\hat{\sigma}_P^2)  \dd \hat{F}_n(\hat{\sigma}_P^2) \leq 1.
\end{equation}
There exist constants $a, b > 0$ such that the measure $\hat{H}_n$ has mass of at least $\delta > 0$ on $[a,b]\times[a,b]$ for large enough $n$, otherwise \eqref{eq:kappa_inequality_H} will be violated. To see this, for $\hat{\sigma}_P^2 \in A_{\kappa}$ and a chosen $\eta>0$, we can find $a,b >0$ such that $t(\hat{\sigma}_P^2,\sigma_P^2) \leq \eta$ for all $(\hat{\sigma}_P^2,\sigma_P^2) \in A_{\kappa} \times \left((0,a) \cup(b,\infty)\right)^2$, since $\lim_ {\sigma_P^2 \to 0}t(\hat{\sigma}_P^2,\sigma_P^2) = 0$ and $\lim_ {\sigma_P^2 \to \infty}t(\hat{\sigma}_P^2,\sigma_P^2) = 0$ uniformly in $\hat{\sigma}_P^2 \in A_{\kappa}$. Suppose there is a sequence $\{n_r\}$ such that 
$$
\gamma_r := \hat{H}_{n_r}([a,b]\times[a,b]) \to 0 \quad \text{as } r \to \infty.
$$
Define $M:= \max_{\hat{\sigma}_P^2\in A_{\kappa},\sigma_P^2 \in [a,b]^2}t(\hat{\sigma}_P^2,\sigma_P^2) >0$ (maximum for a continuous function exists on a compact set), and for all $\hat{\sigma}_P^2 \in A_{\kappa}$, consider 
\begin{align*}
f_{\hat{H}_{n_r}}(\hat{\sigma}_P^2) &= \int_{[a,b]^2}t(\hat{\sigma}_P^2,\sigma_P^2) \dd \hat{H}_{n_r}(\sigma_P^2) + \int_{((0,\infty)\backslash[a,b])^2}t(\hat{\sigma}_P^2,\sigma_P^2)\dd \hat{H}_{n_r}(\sigma_P^2) \\
&\leq \gamma_rM + (1-\gamma_r)\eta.
\end{align*}
Since $\gamma_r \to 0$, we can choose $r$ large enough such that $\gamma_r \leq \eta/2M$, then $f_{\hat{H}_{n_r}}(\hat{\sigma}_P^2) \leq 3\eta/2$.  Therefore, 
$$
r_{n_r}(\hat{\sigma}_P^2) = \frac{f_{H}(\hat{\sigma}_P^2)}{f_{\hat{H}_{n_r}}(\hat{\sigma}_P^2)} \geq \frac{1/\kappa}{3\eta/2} := \alpha >0, \quad\text{for all } \hat{\sigma}_P^2 \in A_{\kappa}.$$
For $r$ sufficiently large, we have $\hat{F}_{n_r}(A_{\kappa}) \geq F(A_{\kappa})/2$, therefore, we have
$$
\int_{A_{\kappa}} r_{n_r}(\hat{\sigma}_P^2)\dd \hat{F}_{n_r}(\hat{\sigma}_P^2) \geq \alpha \hat{F}_{n_r}(A_{\kappa}) \geq \frac{c_\kappa \alpha}{2}.
$$
Choose $\eta$ such that $c_{\kappa}\alpha >2$, i.e. $\eta <c_{\kappa}/(3\kappa)$, hence $\int_{A_{\kappa}}r_{n_r}(\hat{\sigma}_P^2)\dd \hat{F}_{n_r}(\hat{\sigma}_P^2) > 1$, which contradicts \eqref{eq:kappa_inequality_H}. Therefore, for the chosen $a,b$, there exists $\delta > 0$ such that $\hat{H}_n([a,b]^2) \geq \delta$ for all large $n$.

Thus, for all large $n$, $f_{\hat{H}_n}(\hat{\sigma}_P^2) \geq \delta \text{min}\{q_{\nu_A}(\hat{\sigma}_A^2,a), q_{\nu_A}(\hat{\sigma}_A^2,b)\} \times \text{min}\{q_{\nu_B}(\hat{\sigma}_B^2,a), q_{\nu_B}(\hat{\sigma}_B^2,b)\}:=\delta m_{a,b}(\hat{\sigma}_P^2)$. Based on \eqref{eq:f(sigma) bounded}, we have $f_{H}(\hat{\sigma}_P^2) \leq C_{\nu_A}' C_{\nu_B}' (\hat{\sigma}^2_A \hat{\sigma}^2_B)^{-1} \leq C_{\nu_A}' C_{\nu_B}' \kappa^2$ on $A_{\kappa}$. Hence, we can show that 
\begin{equation}
\label{ineq:bound_rn_H}
    r_n(\hat{\sigma}_P^2) \leq \frac{C_{\nu_A}' C_{\nu_B}' \kappa^2}{\delta m_{a,b}(\hat{\lambda})} \leq  \frac{\kappa C'}{\delta\min_{\substack{t\in A_{\kappa}}}m_{a,b}(t)} :=M_{\kappa},\quad \text{for } \hat{\sigma}_P^2\in A_{\kappa},
\end{equation}
where $M_{\kappa}$ is a constant determined by $\delta,a,b,\nu_A,\nu_B,\kappa$, and therefore bounded uniformly in $n$ by a constant on $A_{\kappa}$. 

Now, consider 
$$
\int_{A_\kappa}r_{n}(\hat{\sigma}_P^2)\dd \hat{F}_{n}(\hat{\sigma}_P^2) = \int_{A_\kappa}r_n(\hat{\sigma}_P^2) \dd F(\hat{\sigma}_P^2) + \int_{A_\kappa} r_n(\hat{\sigma}_P^2)\dd (\hat{F}_n(\hat{\sigma}_P^2)-F(\hat{\sigma}_P^2)):= I_{1,n}+I_{2,n}.
$$
For $I_{1,n}$, by \eqref{ineq:bound_rn_H} we define $g_{\kappa}(\hat{\sigma}_P^2) \equiv M_{\kappa} $, then $r_n(\hat{\sigma}_P^2) \leq g_{\kappa}(\hat{\sigma}_P^2)$ for all large $n$, and $\int_{A_\kappa} g_{\kappa}(\hat{\sigma}_P^2) \dd F(\hat{\sigma}_P^2) = M_{\kappa}F(A_{\kappa}) <\infty$. Combined with the pointwise convergence of $r_n(\hat{\sigma}_P^2)$ to $r(\hat{\sigma}_P^2)$ on $A_{\kappa}$, by applying the dominated convergence theorem, we obtain 

$$
\lim_{\substack{n\to \infty}} I_{1,n} = \int_{A_\kappa}r(\hat{\sigma}_P^2) \dd F(\hat{\sigma}_P^2).
$$
For $I_{2,n}$, by \eqref{limit:A_kappa_H}
$$
\abs{I_{2,n}} \leq \sup_{\substack{\hat{\sigma}_P^2 \in A_{\kappa}}} r_n(\hat{\sigma}_P^2)\abs{\hat{F}_n(A_{\kappa})-F(A_{\kappa})} \leq M_{\kappa}\abs{\hat{F}_n(A_{\kappa})-F(A_{\kappa})} \to 0, \quad \text{as } n \to \infty.
$$
Combined above results, we have
$$
\lim_{\substack{n \to \infty}} \int_{A_{\kappa}}r_{n}(\hat{\sigma}_P^2)\dd \hat{F}_{n}(\hat{\sigma}_P^2) = \int_{A_{\kappa}} r(\hat{\sigma}_P^2)\dd F(\hat{\sigma}_P^2) \leq 1.
$$
\end{proof}

Since the choice of $\kappa$ is arbitrary, based on the monotone convergence theorem, by taking $\kappa \to \infty$, we have
$$
\lim_{\substack{\kappa\to \infty}}\int_{A_{\kappa}} r(\hat{\sigma}_P^2) \dd F(\hat{\sigma}_P^2)=\int_{\mathbb{R}_+^2} \frac{f_H(\hat{\sigma}_P^2)}{f_{\tilde{H}}(\hat{\sigma}_P^2)} \dd F(\hat{\sigma}_P^2) = \int_{\mathbb{R}_+^2}\frac{f_H^2(\hat{\sigma}_P^2)}{f_{\tilde{H}}(\hat{\sigma}_P^2)}\dd \hat{\sigma}_P^2 \leq 1.
$$
By rearranging the inequality,
$$
0 \geq\int_{\mathbb{R}_+^2}\frac{f_H^2(\hat{\sigma}_P^2)}{f_{\tilde{H}}(\hat{\sigma}_P^2)}\dd \hat{\sigma}_P^2 -1 = \int_{\mathbb{R}_+^2} \left(\frac{f_H(\hat{\sigma}_P^2)}{f_{\tilde{H}}(\hat{\sigma}_P^2)}-1\right)^2 f_{\tilde{H}}(\hat{\sigma}_P^2) \dd \hat{\sigma}_P^2 \geq 0,
$$
with the equality holds iff $f_H = f_{\tilde{H}}$ almost surely. Since both $f_H$ and $f_{\tilde{H}}$ are continuous, we must have $f_H = f_{\tilde{H}}$. Combined with the identifiability results, we have $\tilde{H} = H$. 

Therefore, we can conclude there is a set $B$ with probability $1$ such that for each $\omega \in B$, every subsequence of the sequence $\{\hat{H}_n(\cdot;\omega)\}$ has a convergent subsequence, and all these subsequences have the same weak limit $H$. Then with probability 1, 
$$
 \hat{H} \cd H \; \text{ as }\; n \to \infty.
$$

\end{proof}

\subsection{Proof of Proposition \ref{prop:2D_oracle_uniform}}

\label{proof:prop_2D_oracle_uniform}

\begin{proof}

We omit the subscript $i$ throughout the proof. The first property follows from the probability integral transform applied conditional on $(\hat{\sigma}^2_{A}, \hat{\sigma}^2_{B})$. Fix $s^2_A > 0, s^2_B > 0$ and condition throughout on $(\hat{\sigma}^2_{A}, \hat{\sigma}^2_{B}) = (s^2_{A}, s^2_{B})$. Under the null we have
$\Tbf = \TbfNull$, the null-centered statistic.
Let $X := |\TbfNull|$ and for $x > 0$, define
$$F_{H,s^2_A, s^2_B}(x)\ :=\ \PP[H]{X < x \; \cond \; \hat{\sigma}^2_{A} = s^2_A, \hat{\sigma}^2_{B} = s^2_B},$$
and it is the conditional c.d.f. of $X$ under the hierarchical model
$(\sigma^2_{A}, \sigma^2_{B}) \sim H$. Because the conditional law of $\TbfNull$ given $(\sigma^2_{A}, \sigma^2_{B}, \hat{\sigma}^2_{A},  \hat{\sigma}^2_{B})$ has a continuous density (hence so does $X$),
$F_{H,s^2_A, s^2_B}$ is continuous and strictly increasing.

By definition of the DVEPB tail area,
$$\PDVFunc(t, s^2_A, s^2_B; \sigma^2_A, \sigma^2_B) := \PP[\sigma^2_A, \sigma^2_B]{ \abs{\TbfNull} \geq \abs{t} \; \cond \; \hat{\sigma}^2_{A} = s^2_A, \hat{\sigma}^2_{B} = s^2_B} = 1 - F_{H,s^2_A, s^2_B}(|t|).$$ 
Evaluating this at the random $t = \Tbf = \TbfNull$ under the null gives
$$
U := \PDVFunc(\Tbf, s^2_A, s^2_B; H) = 1 - F_{H,s^2_A, s^2_B}(X).
$$
By the probability integral transform, if $X$ has continuous c.d.f. $F_{H,s^2_A, s^2_B}$ then $F_{H,s^2_A, s^2_B}(X) \sim \text{Unif}(0,1)$. Hence, $U = 1 - F_{H,s^2_A, s^2_B}(X)$ is
also $\text{Unif}(0,1)$. Therefore, for all $\alpha\in(0,1)$,
$$\PP[H]{U < \alpha \; \cond \; \hat{\sigma}^2_{A} = s^2_A, \hat{\sigma}^2_{B} = s^2_B} = \alpha.$$
Since the pair $(s^2_A, s^2_B)$ is arbitrary, the above statement holds almost surely with respect to
$(\hat{\sigma}^2_{A}, \hat{\sigma}^2_{B})$.

The second property follows from the conditional uniformity and iterated expectation, i.e.
$$\PP[H]{U \leq \alpha} = \EE[H]{\PP[H]{U \leq \alpha \; \cond \; \hat{\sigma}^2_A, \hat{\sigma}^2_B}} = \EE[H]{\alpha} = \alpha.$$
\end{proof}

\subsection{Proof of Theorem \ref{thm:2D_uniform_convergence_p_value}}

\label{proof:thm_2D_uniform_convergence_p_value}

\begin{proof}

We start by defining $\Tdm$ as:
\begin{equation*}
    \Tdm := \frac{\hat{\mu}_A - \hat{\mu}_B}{\sqrt{\frac{\sigma^2_A}{K_A} + \frac{\sigma^2_B}{K_B}}},
\end{equation*}
and further define $\hat{\zeta} = \abs{\hat{\mu}_A-\hat{\mu}_B}$. We further denote $\PDVFunc(t, s^2_A, s^2_B; \sigma^2_A, \sigma^2_B)$ as $\PDVFunc(t, s^2_P; \sigma^2_P)$, which can be rewriten as:
\begin{equation}
    \label{eq: p(sigma) explicit}
    \PDVFunc(t, s^2_P; \sigma^2_P) \equiv I(z, \sigma^2_P) := 2\int_{\frac{z}{\sqrt{\sigma^2_A / K_A + \sigma^2_B/K_A}}}^{\infty} \frac{1}{\sqrt{2 \pi}} \exp\left(-\frac{1}{2}u^2\right) \dd u,
\end{equation}
where $z = \abs{t}\sqrt{s^2_A / K_A + s^2_B / K_B}$. With the explicit form of $\PDVFunc(t, s^2_P; \sigma^2_P)$, for all distribution $\tilde{H}$ supported on $\mathbb{R}^2_{+}$, we can then define $\phi_{\tilde{H}}$:
\begin{align*}
    \phi_{\tilde{H}}:\mathbb{R}^3_{+} \to R, \phi_{\tilde{H}}(z, s^2_P) = \EE[\tilde{H}]{I(z, \sigma^2_P) \; \cond \; \hat{\sigma}^2_P = s^2_P} = \frac{\int_{\mathbb{R}^2_{+}} I(z, \sigma^2_P) p(s^2_P \; \cond \; \sigma^2_P)  \dd \tilde{H}(\sigma^2_P)}{\int_{\mathbb{R}^2_{+}} p(s^2_P \; \cond \; \sigma^2_P)  \dd \tilde{H}(\sigma^2_P)}
\end{align*}
where $p(s^2_P \; \cond \; \sigma^2_P) = p(\hat{\sigma}^2_P = s^2_P \; \cond \; \sigma^2_P)$.

Thus, we have $\PDV_i = \PDVFunc(\Tbf_i, \hat{\sigma}^2_{iP} ; \hat{H}) = \phi_{\hat{H}}(\hat{\zeta}_i, \hat{\sigma}^2_{iP}), \PDVFunc(\Tbf_i, \hat{\sigma}^2_{iP} ; H) = \phi_{H}(\hat{\zeta}_i, \hat{\sigma}^2_{iP})$. Notice that $I(z, \sigma^2_P) \in (0, 1]$ by its definition as the two-sided tail area of the standard normal distribution. Therefore, it is easy to see that $\phi_{\tilde{H}} \in (0, 1]$, it follows that $\PDV_i \in (0, 1], \PDVFunc(\Tbf_i, \hat{\sigma}^2_{Pi} ; H) \in (0, 1]$.

Next, we define the compact set
$$S = \left\{(z, s^2_A, s^2_B) \; \cond \; 0\leq m_z\leq z \leq M_z <\infty,  0 < m_{A} \leq s^2_A \leq M_{A} < \infty, 0 < m_{B} \leq s^2_B \leq M_{B} < \infty \right\},$$
where $m_z, M_z, m_{A}, M_{A}, m_B, M_B$ are arbitrary positive numbers.

For the setting where $\sigma^2_{iP} \simiid H$, we fix $i$ and define the event $A_i := \{(\hat{\zeta}_i, \hat{\sigma}^2_{iP}) \in S\}$. With an analogous argument as the proof of Theorem \ref{thm:1D_uniform_convergence_p_value}, we have
\begin{align*}
    \EE[H]{\abs{\PDV_i - \PDVFunc(\Tbf_i, \hat{\sigma}^2_{iP}; H)}} \leq \EE[H]{\sup_S\abs{\phi_{\hat{H}}(z, s^2_P) - \phi_H(z, s^2_P)}} + \PP[H]{A_i^c}.
\end{align*}
Since the choice of $i$ is arbitrary, we further show that:
\begin{align*}
    \max_{1 \leq i \leq n} \EE[H]{\abs{\PDV_i - \PDVFunc(\Tbf_i, \hat{\sigma}^2_{iP}; H)}} \leq \EE[H]{\sup_S\abs{\phi_{\hat{H}}(z, s^2_P) - \phi_H(z, s^2_P)}} + \PP[H]{A^c}.
\end{align*}
Here we omit the subscript $i$ for $\PP[H]{A^c}$.

Next, for any distribution $\Tilde{H}$ supported on $S$, and any $z \geq 0, s^2_A > 0, s^2_B > 0$, let us write:
\begin{align*}
    N(z, s^2_P, \Tilde{H}) &:= \int_{\mathbb{R}^2_{+}} I(z, \sigma^2_P) p(s^2_P \; \cond \; \sigma^2_P)  \dd \Tilde{H}(\sigma^2_P), \\
    D(s^2_P, \Tilde{H}) &:= \int_{\mathbb{R}^2_{+}} p(s^2_P \; \cond \; \sigma^2_P)  \dd \Tilde{H}(\sigma^2_P), \\
    \phi_{\Tilde{H}}(z, s^2_P) &:= N(z, s^2_P, \Tilde{H})/D(s^2_P, \Tilde{H}).
\end{align*}
By Lemma \ref{lemm:transformation}, and replacing $G$ with $H$, $\hat{G}$ with $\hat{H}$, $(z,l)$ with $(z,s_P^2)$, we have the following inequality:
\begin{align*}
    |\phi_{\hat{H}}(z, s^2_P) &- \phi_H(z, s^2_P)| \leq \frac{2\abs{N(z, s^2_P, \hat{H}) - N(z, s^2_P, H)}}{D(s^2_P, H)} + \frac{2\abs{D(s^2_P, \hat{H}) - D(s^2_P, H)}}{D(s^2_P, H)}.
\end{align*}
Combining all results above, we have:
\begin{align*}
    &\EE[H]{\abs{\PDV_i - \PDVFunc(\Tbf_i, \hat{\sigma}^2_{iP}; H)}} \\
    & \leq 2\EE[H]{\sup_S \frac{\abs{N(z, s^2_P, \hat{H}) - N(z, s^2_P, H)}}{D(s^2_P, H)}} + 2\EE[H]{\sup_{S_{P}} \frac{\abs{D(s^2_P, \hat{H}) - D(s^2_P, H)}}{D(s^2_P, H)}} + \PP[H]{A^c},
\end{align*}
where $S_P = \{s_P^2:\text{ there exists } z\in \mathbb{R} \text{ s.t. }(z,s_P^2) \in S\}$. Since $H$ is supported on $[L_A, U_A] \times [L_B, U_B]$, we can further study the lower bound of $D(s^2_P, H)$ on the compact set of $S_P$.

As it has been proved in Supplement \ref{sec:properties DV NPMLE}, given any known $s^2 > 0$, $\sigma^2 \to p(s^2 \; \cond \; \sigma^2, \nu)$ increases for $\sigma^2 < s^2$ and decreases for $\sigma^2 > s^2$. Thus, we have:
\begin{align*}
    p(s^2_P \; \cond \; \sigma^2_P) \geq \text{min}\left\{p(s^2_A \; \cond \; L_A, \nu_A),\; p(s^2_A \; \cond \; U_A, \nu_A)\right\} \times \text{min}\left\{p(s^2_B \; \cond  L_B, \nu_B),\; p(s^2_B \; \cond \;  U_B, \nu_B)\right\},
\end{align*}
and furthermore:
\begin{align*}
D(s^2_P, H) &= \int_{\mathbb{R}^2_{+}} p(s^2_P \; \cond \; \sigma^2_P)  \dd H(\sigma^2_P) \\
&\geq \text{min}\left\{p(s^2_A \; \cond \; L_A, \nu_A),\; p(s^2_A \; \cond \; U_A, \nu_A)\right\} \times \text{min}\left\{p(s^2_B \; \cond \; L_B, \nu_B),\; p(s^2_B \; \cond \; U_B, \nu_B)\right\}.
\end{align*}
To bound this, we study $d p(s^2 \; \cond \; \sigma^2, \nu)/ds^2$:
\begin{equation}
\label{eq:p(s2|sigma) derivative wrt s2}
    \frac{dp(s^2 \; \cond \; \sigma^2, \nu)}{d s^2}=p(s^2 \; \cond \; \sigma^2, \nu) \left(\frac{\frac{\nu}{2} - 1}{s^2} - \frac{\nu}{2\sigma^2}\right).
\end{equation}
For cases where $\nu=2$, given any known $\sigma^2 > 0$, it is easy to varify that $d p(s^2 \; \cond \; \sigma^2, \nu)/ds^2 <0$ for all $s^2>0$, and therefore on the compact set $S$:
$$p(s^2_A \; \cond \; \sigma^2_A, \nu_A) \geq p(M_A \; \cond \; \sigma^2_A, \nu_A) = \text{min}\left\{p(m_A \; \cond \; \sigma^2_A, \nu_A),\; p(M_A \; \cond \; \sigma^2_A, \nu_A)\right\} := \mu_A(\sigma^2_A,\nu_A),$$
$$p(s^2_B \; \cond \; \sigma^2_B, \nu_B) \geq p(M_B \; \cond \; \sigma^2_B, \nu_B) = \text{min}\left\{p(m_B \; \cond \; \sigma^2_B, \nu_B),\; p(M_B \; \cond \; \sigma^2_B, \nu_B)\right\} := \mu_B(\sigma^2_B,\nu_B).$$
For $\nu \geq 3$,  $p(s^2 \; \cond \; \sigma^2, \nu) \uparrow$ for $s^2 < \frac{\nu - 2}{\nu} \sigma^2$ and $p(s^2 \; \cond \; \sigma^2, \nu) \downarrow$ for $s^2 \geq \frac{\nu - 2}{\nu} \sigma^2$. Again, consider the compact set $S$, combined with the results that hold for $\nu_A,\nu_B=2$, we have
$$
p(s^2_A \; \cond \; \sigma^2_A, \nu_A) \geq \mu_A(\sigma^2_A,\nu_A),\quad p(s^2_B \; \cond \; \sigma^2_B, \nu_B) \geq \mu_B(\sigma^2_B,\nu_B),\;  \; \text{for all } \nu_A,\nu_B \geq 2,
$$
and hence it follows that on the compact set $S_P$
$$D(s_P^2, H) \geq  \text{min}\{\mu_A(L_A,\nu_A), \mu_B(U_A,\nu_A)\} \times \text{min}\{\mu_A(L_B,\nu_B), \mu_B(U_B,\nu_B)\} = C^{'}_{S, L, U} > 0.$$
Further combining the results above, we have:
\begin{align*}
    &\max_{1 \leq i \leq n} \EE[H]{\abs{\PDV_i - \PDVFunc(\Tbf_i, \hat{\sigma}^2_{iP}; H)}} \\
    &\quad\leq  \frac{2}{C^{'}_{S, L, U}}\left(\EE[H]{\sup_{S} \abs{N(z, s^2_P, \hat{H}) - N(z, s^2_P, H)}}+\EE[H]{\sup_{S_P} \abs{D(s^2_P, \hat{H}) - D(s^2_P, H)}}\right) + \PP[H]{A^c}\\
    &\quad:= \frac{2}{C^{'}_{S, L, U}} \left(\text{I} + \text{II}\right) + \text{III}.
\end{align*}
In the three upcoming Lemmas, we study terms I, II, and III.

\begin{lemm}
    \label{lemma: bound I 2D}It holds that:
    $$\lim_{n \to \infty} \EE[H]{\sup_{(z,s_P^2) \in S} \abs{N(z, s^2_P, \hat{H}) - N(z, s^2_P, H)}} = 0.$$
\end{lemm}

\begin{proof}

First, we want to show that for any $\tilde{H}$ supported on $\mathbb{R}^2_{+}$, $N(z, s^2_P, \tilde{H})$ is bounded on the compact set $S = \{(z, s^2_A, s^2_B) \; \cond \; 0<m_z\leq z \leq M_z <\infty,  0 < m_{A} \leq s^2_A \leq M_{A} < \infty, 0 < m_{B} \leq s^2_B \leq M_{B} < \infty \}$ where $m_z, M_z,m_{A}, M_{A}, m_B, M_B$ are arbitrary positive numbers.

Since $I(z, s^2_P, \sigma^2_P) \in (0, 1]$, it is east to verify that 
$N(z, s^2_P, \tilde{H}) \in (0, D(s^2_P, \tilde{H})],$
which means we need to upper bound $D(s^2_P, \tilde{H})$. Similar to what we did in Supplement \ref{sec:properties DV NPMLE} and use \eqref{eq:f(sigma) bounded}, we can show that:
\begin{equation}
    \label{eq: bound D 2D}
    D(s^2_P, \tilde{H})\leq \int_{\mathbb{R}^2_{+}} C_{\nu_A}'C_{\nu_B}' (s^2_A s^2_B)^{-1}  \dd \tilde{H}(\sigma^2_P)= C_{\nu_A}'C_{\nu_B}' (s^2_A s^2_B)^{-1}\leq C_{\nu_A}'C_{\nu_B}' (m_A m_B)^{-1}.
\end{equation}
Thus, we have shown that $N(z, s^2_P, \tilde{H}) \in (0, C_{\nu_A}'C_{\nu_B}' (m_A m_B)^{-1}]$.

Next, we want to show that for any $\tilde{H}$, $N(\cdot,\cdot,\tilde{H})$ has uniformly bounded first derivatives in $(z,s_P^2)$.

We will start by studying $\abs{\partial N(z, s^2_P, \tilde{H})/\partial z}$:
\begin{align*}
    \abs{\frac{\partial N(z, s^2_P, \tilde{H})}{\partial z}} = 2\abs{\int_{\mathbb{R}^2_{+}} f_{Z}\Big(\frac{z}{\sqrt{\frac{\sigma^2_A}{K_A} + \frac{\sigma^2_B}{K_B}}}\Big)\frac{1}{\sqrt{\frac{\sigma^2_A}{K_A} + \frac{\sigma^2_B}{K_B}}}p(s^2_P \; \cond \; \sigma^2_P)\dd \tilde{H}(\sigma^2_P)},
\end{align*}
where $f_{Z}$ is the p.d.f. of a standard normal distribution. It is easy to verify that $f_Z\left(z \Big /\sqrt{\frac{\sigma^2_A}{K_A} + \frac{\sigma^2_B}{K_B}}\right) \leq f_Z(0) = \frac{1}{\sqrt{2 \pi}}$.
By applying AM-GM inequality, we have the following:
\begin{align*}
    \abs{\frac{\partial N(z, s^2_P, \tilde{H})}{\partial z}} &\leq 2\abs{\int_{\mathbb{R}^2_{+}} \frac{1}{\sqrt{2 \pi}}  \frac{1}{\sqrt{2}} \left(\frac{\sigma^2_A}{\nu_A + 1}\right)^{-\frac{1}{4}} \left(\frac{\sigma^2_B}{\nu_B + 1}\right)^{-\frac{1}{4}}p(s^2_P \; \cond \; \sigma^2_P)\dd \tilde{H}(\sigma^2_P)}\\
    &= \frac{1}{\sqrt{\pi}} \abs{\int_{\mathbb{R}^2_{+}} \psi(s^2_A, \sigma^2_A, \nu_A) \psi(s^2_B, \sigma^2_B, \nu_B)  \dd \tilde{H}(\sigma^2_P)},
\end{align*}
where $\psi(s^2, \sigma^2, \nu) = \left(\frac{\sigma^2}{\nu+1}\right)^{-1/4} p(s^2 \; \cond \; \sigma^2, \nu)$. Consider the map: $\sigma^2 \rightarrow \psi(s^2, \sigma^2, \nu)$, it is easy to verify that:
$$\psi'(\sigma^2) = \frac{(\nu+1)^{\frac{1}{4}}\nu^{\frac{\nu}{2}}}{2^{\frac{\nu}{2}} \Gamma(\frac{\nu}{2})} (s^2)^{\frac{\nu}{2}-1} (\sigma^2)^{-\frac{\nu}{2}-\frac{9}{4}}\exp\left(-\frac{\nu s^2}{\sigma^2}\right) \left(\frac{1}{2}\nu s^2 - \left(\frac{1}{4} + \frac{\nu}{2}\right)\sigma^2\right),$$
which implies that $\psi(\sigma^2) \uparrow$ for $\sigma^2 > \frac{2\nu}{1 + 2\nu} s^2$, and $\psi(\sigma^2) \downarrow$ for $\sigma^2 < \frac{2 \nu}{1 + 2\nu} s^2$. Thus, for any given $s^2 > 0$, $$\psi(\sigma^2) \leq \psi(\frac{2\nu}{1 + 2\nu} s^2) = C_1(\nu) (s^2)^{-\frac{5}{4}},$$
where $C_1(\nu)$ is a constant given $\nu$. Therefore, we can further upper bound $\abs{\partial N(z, s^2_P, \tilde{H})/\partial z}$ by:
\begin{align*}
    \abs{\frac{\partial N(z, s^2_P, \tilde{H})}{\partial z}} &\leq  \frac{1}{\sqrt{\pi}} C_1(\nu_A) (s^2_A)^{-\frac{5}{4}} C_1(\nu_B) (s^2_B)^{-\frac{5}{4}}  \leq \frac{1}{\sqrt{\pi}}C_1(\nu_A) (m_A)^{-\frac{5}{4}} C_1(\nu_B) (m_B)^{-\frac{5}{4}} < \infty.
\end{align*}
Next, we focus on studying $\abs{\partial N(z, s^2_P, G)/\partial s^2_A}$. By \eqref{eq:p(s2|sigma) derivative wrt s2}, it is easy to verify that
\begin{align*}
    \abs{\frac{\partial N(z, s^2_P, \tilde{H})}{\partial s^2_A}} &\leq \abs{\int_{\mathbb{R}^2_{+}} p(s^2_A | \sigma^2_A, \nu_A) \left(\frac{\frac{\nu_A}{2} - 1}{s^2_A} - \frac{\nu_A}{2\sigma^2_A}\right) p(s^2_B \; \cond \; \sigma^2_B, \nu_B)  \dd \tilde{H}(\sigma^2_P)} \\
    &\leq  C_{\nu_B}' (s^2_B)^{-1} \abs{\int_{\mathbb{R}^2_{+}} p(s^2_A | \sigma^2_A, \nu_A) \left(\frac{\frac{\nu_A}{2} - 1}{s^2_A} - \frac{\nu_A}{2\sigma^2_A}\right)  \dd \tilde{H}(\sigma^2_P)},
\end{align*}
where the first inequality uses the fact that $\abs{I(z, \sigma^2_P)} \leq 1$ for $\sigma^2_P \in S$, and the second inequality uses the result of \eqref{eq:f(sigma) bounded}. Next, define the map:
$$h: (0, \infty) \rightarrow R: h(\sigma^2) = p(s^2 \; \cond \; \sigma^2, \nu) \left(\frac{\frac{\nu}{2} - 1}{s^2} - \frac{\nu}{2\sigma^2}\right),$$
and consider $h'(\sigma^2)$:
$$
    h'(\sigma^2) = p(s^2 \; \cond \; \sigma^2, \nu) \frac{\nu}{4s^2 (\sigma^2)^3}\left[-(\nu-2)(\sigma^2)^2 + 2\nu s^2\sigma^2 - \nu (s^2)^2\right].
$$
When $\nu=2$, $h(\sigma^2) \uparrow$ for $\sigma^2 < s^2/2$, and $h(\sigma^2) \downarrow$ for $\sigma^2 > s^2/2$, then for any given $s^2 > 0$:
\begin{equation}
\label{ineq:bound_h_2}
    h(\sigma^2) \leq q(s^2/2) = 4 C_{\nu} \exp(-2) (s^2)^{-2},
\end{equation}
which leads to the following inequality when $\nu_A = 2$,
\begin{align*}
    \abs{\frac{\partial N(z, s^2_P, \tilde{H})}{\partial s^2_A}} 
    &\leq C_{\nu_B}' (s^2_B)^{-1} \int_{\mathbb{R}^2_{+}} \abs{q(\sigma^2_A)}  \dd \tilde{H}(\sigma^2_P) \\
    &\leq C_{\nu_B}' (s^2_B)^{-1} 4 C_{\nu_A} \exp(-2) (s_A^2)^{-2} \\
    &\leq C_{\nu_B}' (m_B)^{-1} 4 C_{\nu_A} \exp(-2) (m_A)^{-2} < \infty.
\end{align*}
When $\nu\geq3$, $h(\sigma^2) \uparrow$ for $\sigma^2 \in \left(\frac{\nu - \sqrt{2\nu}}{\nu - 2}s^2, \frac{\nu + \sqrt{2\nu}}{\nu - 2}s^2\right)$, and $h(\sigma^2) \downarrow$ for $\sigma^2 \in \left(0, \frac{\nu - \sqrt{2\nu}}{\nu - 2}s^2\right] \cup \left[\frac{\nu + \sqrt{2\nu}}{\nu - 2}s^2, \infty\right)$, $h(\sigma^2) > 0$ for $\sigma^2 > \frac{\nu}{\nu-2} s^2$, $h(\sigma^2) < 0$ for $\sigma^2 < \frac{\nu}{\nu-2} s^2$. which implies that for any given $s^2 > 0$:
\begin{align*}
    \abs{h(\sigma^2)} &\leq \text{max}\left\{\abs{\lim_{\sigma^2 \rightarrow 0}h(\sigma^2)}, \abs{h\left(\frac{\nu - \sqrt{2\nu}}{\nu - 2}s^2\right)}, \abs{h\left(\frac{\nu + \sqrt{2\nu}}{\nu - 2}s^2\right)}, \abs{\lim_{\sigma^2 \rightarrow \infty}h(\sigma^2)}\right\} \\
    &= \text{max}\left\{-\lim_{\sigma^2 \rightarrow 0}h(\sigma^2), -h\left(\frac{\nu - \sqrt{2\nu}}{\nu - 2}s^2\right), h\left(\frac{\nu + \sqrt{2\nu}}{\nu - 2}s^2\right), \lim_{\sigma^2 \rightarrow \infty}h(\sigma^2)\right\}.
\end{align*}
It is also easy to verify that $\lim_{\sigma^2 \rightarrow 0}h(\sigma^2) = 0$ and $\lim_{\sigma^2 \rightarrow \infty}h(\sigma^2) = 0$. Moreover, we have:
$$
h\left(\frac{\nu - \sqrt{2\nu}}{\nu - 2}s^2\right)= -C_{2}(\nu)(s^2)^{-2},\quad h\left(\frac{\nu + \sqrt{2\nu}}{\nu - 2}s^2\right)= C_{3}(\nu)(s^2)^{-2},
$$
where $C_{2}(\nu), C_{3}(\nu)$ are positive constants given $\nu$. Thus, for any $s^2 > 0$, we have:
\begin{equation}
\label{ineq:bound_h_3}
 \abs{h(\sigma^2)} \leq \text{max}\{C_{2}(\nu), C_{3}(\nu)\} (s^2)^{-2} = C_{4}(\nu)(s^2)^{-2},   
\end{equation}
where $C_{4}(\nu)$ is constant given $\nu$. The following inequality follows for $\nu_A \geq 2$,
\begin{align*}
    \abs{\frac{\partial N(z, s^2_P, \tilde{H})}{\partial s^2_A}} 
    &\leq C_{\nu_B}' (s^2_B)^{-1} \int_{\mathbb{R}^2_{+}} \abs{h(\sigma^2_A)}  \dd \tilde{H}(\sigma^2_P) \\
    &\leq  C_{\nu_B}' (s^2_B)^{-1} C_{4}(\nu_A)(s^2_A)^{-2} \\
    &\leq C_{\nu_B}' (m_B)^{-1} C_{4}(\nu_A)(m_A)^{-2} < \infty.
\end{align*}
By symmetry, we can also show that:
$$\abs{\frac{\partial N(z, s^2_P, \tilde{H})}{\partial s^2_B}} < \infty.$$
Hence, we have shown the continuity of $N(z, s^2_P, \tilde{H})$ with respect to $z, s^2_P$. Applying the mean value theorem, it is easy to verify the equicontinuity of $\{N(z, s^2_P, \tilde{H}):\tilde{H} \text{ supported on }\mathbb{R}_{+}^2\}$ on the compact set $S$. Thus, by Arzelà-Ascoli \citep[Theorem 7.25]{Rudin1976Analysis}, we can prove that on the compact set $S$, any subsequence of $\{N(z, s^2_P, \tilde{H})\}$ has a uniformly convergent subsequence on $S$ whose limit is continuous. Since $\hat{H}_n \to H$, and the fact that $f_{z,s_P^2}(\sigma_P^2):=I(z,\sigma_p^2)p(s_P^2 \;\cond\;\sigma_P^2)$ is bounded and continuous for $\sigma_P^2 \in \mathbb{R}_{+}^2$ for any fixed $(z,\sigma_P^2) \in S$, by the definition of weak convergence, we have
$$
N(z,s_P^2,\hat{H}_n) = \int f_{z,s_P^2}(\sigma_P^2)\dd \hat{H}_n(\sigma_p^2) \to \int f_{z,s_P^2}(\sigma_P^2)\dd H(\sigma_p^2) = N(z,s_P^2,H).
$$
Since the pointwise convergence ensures that the only possible limit for the uniformly convergent subsequence of $\{N(z,s_P^2,\hat{H}_n\}$ is $N(z,s_P^2,H)$, it follows that the full sequence converges uniformly on $S$:

$$\lim_{n \to \infty} \sup_S \abs{N(z, s^2_P, \hat{H}_n) - N(z, s^2_P, H)} = 0\quad \text{almost surely}.$$ 
Thus, by Fatou's lemma \citep[Theorem 16.3]{billingsley1995probability}
\begin{align*}
    &\lim_{n \to \infty} \EE[H]{\sup_{S}\abs{N(z, s^2_P, \hat{H}) - N(z, s^2_P, H)}} \leq \EE[H]{\lim_{n \to \infty} \sup_S \abs{N(z, s^2_P, \hat{H}) - N(z, s^2_P, H)}} = 0.
\end{align*}
\end{proof}

\begin{lemm}
    \label{lemma:bound II 2D} It holds that:
    $$\lim_{n \to \infty} \EE[H]{\sup_{s_P^2 \in S_P} \abs{D(s^2_P, \hat{H}) - D(s^2_P, H)}} = 0.$$
\end{lemm}

\begin{proof}

First, we want to show that for any $\tilde{H}$ supported on $\mathbb{R}_
{+}^{2}$, $D(s^2_P, \tilde{H}) = \int_{\mathbb{R}^2_{+}} p(s^2_P \; \cond \; \sigma^2_P)  \dd \tilde{H}(\sigma^2_P)$ is bounded on the compact set $S_P$. In Lemma \ref{lemma: bound I 2D}, we have already proved that 
$
D(s^2_P, \tilde{H}) \in (0,C_{\nu_A}'C_{\nu_B}' (m_A m_B)^{-1}]   
$.

Now, we want to study the upper bound of $\abs{\partial D(s^2_P, \tilde{H})/\partial s^2_A}$. 
We can reuse the result in \eqref{eq:p(s2|sigma) derivative wrt s2} and show that for $\nu_A =2$, by \eqref{ineq:bound_h_2}, we have
\begin{align*}
    \abs{\frac{\partial D(s^2_P, \tilde{H})}{\partial s^2_A}} &= \abs{\int_{\mathbb{R}^2_{+}} p(s^2_A | \sigma^2_A, \nu_A) \left(- \frac{1}{\sigma^2_A}\right) p(s^2_B \; \cond \; \sigma^2_B, \nu_B)  \dd \tilde{H}(\sigma^2_P)} \\
    &\leq C_{\nu_B}' (m_B)^{-1} 4 \exp(-2) C_{\nu_A}(m_A)^{-2} < \infty,
\end{align*}
For $\nu_A \geq 3$, by \eqref{ineq:bound_h_3}, we have
\begin{align*}
    \abs{\frac{\partial D(s^2_P, \tilde{H})}{\partial s^2_A}} &= \abs{\int_{\mathbb{R}^2_{+}} p(s^2_A | \sigma^2_A, \nu_A) \left(\frac{\frac{\nu_A}{2} - 1}{s^2_A} - \frac{\nu_A}{2\sigma^2_A}\right) p(s^2_B \; \cond \; \sigma^2_B, \nu_B)  \dd \tilde{H}(\sigma^2_P)} \\
    &\leq C_{\nu_B}' (m_B)^{-1} C_{4}(\nu_A)(m_A)^{-2} < \infty.
\end{align*}
By symmetry, we can also show that:
$$\abs{\frac{\partial D(s^2_P, \tilde{H})}{\partial s^2_B}} < \infty.$$
Hence, we have shown the continuity of $D(s^2_P, \tilde{H})$ with respect to $s^2_P$. Applying the mean value theorem, it is easy to verify the equicontinuity of  $\{D(s^2_P, \tilde{H}):\tilde{H} \text{ supported on } \mathbb{R}_{+}^2\}$ on the compact set $S_P$. Thus, by Arzelà-Ascoli \citep[Theorem 7.25]{Rudin1976Analysis}, we can prove that on the compact set $S_P$, any subsequence of $\{D(s^2_P, \tilde{H})\}$ has a uniformly convergent subsequence on $S_P$ whose limit is continuous. Since $\hat{H}_n \to H$, and the fact that $p(s_P^2 \;\cond\;\sigma_P^2)$ is bounded and continuous with $\lambda \in \mathbb{R}_{+}$ for any fixed $s_P^2 \in S_p$, by the definition of weak convergence, we have
$$
D(s_P^2,\hat{H}_n) = \int p(s_P^2 \;\cond\;\sigma_P^2)\dd \hat{H}_n(\sigma_P^2) \to \int p(s_P^2 \;\cond\;\sigma_P^2)\dd H(\sigma_P^2) = D(s_P^2,H).
$$
Since the pointwise convergence ensures that the only possible limit for the uniformly convergent subsequence of $\{D(s_P^2,\hat{H}_n)\}$ is $D(s_P^2,H)$, it follows that the full sequence converges uniformly:
$$\lim_{n \to \infty} \sup_{S_P} \abs{D(s^2_P, \hat{H}_n) - D(s^2_P, H)} = 0\quad\text{almost surely}.$$
Thus, by Fatou's lemma \citep[Theorem 16.3]{billingsley1995probability}
$$\lim_{n \to \infty} \EE[H]{\sup_{S_P}\abs{D(s^2_P, \hat{H}) - D(s^2_P, H)}} \leq \EE[H]{\lim_{n \to \infty} \sup_{S_P} \abs{D(s^2_P, \hat{H}) - D(s^2_P, H)}} = 0.$$
\end{proof}

\begin{lemm}
    \label{lemma: P(Ac) = 0 2D}
    As $m_z,m_A,m_B \to 0$ and $M_z,M_A,M_B \to \infty$, $\PP[H]{A^c} \to 0$.
\end{lemm}

\begin{proof}

Recall that $A = \{(\hat{\zeta}, \hat{\sigma}_{P}^2) \in S\}$, where $S = \{(z, s_A^2,s_B^2) \; \cond \; 0<m_z\leq z \leq M_z < \infty,  0 < m_{A} \leq s_A^2 \leq M_{A} < \infty,0 < m_{B} \leq s_B^2 \leq M_{B} < \infty \}$. Notice that $\PP[H]{A^c} \leq \PP[H]{\hat{\zeta} < m_z}+\PP[H]{\hat{\zeta} > M_z}+\PP[H]{\hat{\sigma}_A^2 < m_A} + \PP[H]{\hat{\sigma}_A^2 > M_A}+\PP[H]{\hat{\sigma}_B^2 < m_B} + \PP[H]{\hat{\sigma}_B^2 > M_B}$, therefore we focus on studying these terms separately.

First, we study $\PP[H]{\hat{\zeta} < m_z}$ and $\PP[H]{\hat{\zeta} > M_z}$. Based on the law of the sufficient statistics \eqref{eq:ss_distribution}, we have:
\begin{align*}
    \hat{\mu}_A - \hat{\mu}_B \sim C(\sigma^2_A,\sigma^2_B)\mathrm{N}(0, 1),\quad C(\sigma^2_A,\sigma^2_B) =\sqrt{\frac{\sigma^2_A}{K_A} + \frac{\sigma^2_B}{K_B}}.
\end{align*}
Thus, we can express $\PP[H]{\hat{\zeta} < m_z}$ explicitly as:
\begin{align*}
    \PP[H]{\hat{\zeta} < m_z} = \PP[H]{\abs{ Z} < \frac{m_z}{C(\sigma^2_A,\sigma^2_B)}} \leq \PP[]{\abs{ Z} < \frac{m_z}{C(L_A,L_B)}} 
   = 2  \int_0^{\frac{m_z}{C(L_A,L_B)}} \frac{1}{\sqrt{2 \pi}} \exp\left(-\frac{t^2}{2}\right) \dd t,
\end{align*}
where $Z\sim \mathrm{N}(0,1)$, and the inequality holds due to the fact that H is supported on $ = [L_A, U_A]\times[L_B, U_B]$, therefore $C(U_A,U_B)\geq C(\sigma^2_A,\sigma^2_B) \geq C(L_A,L_B)$. Since $\exp(-\frac{t^2}{2}) \leq \exp(0) = 1$ for $t \geq 0$, we have:
\begin{align*}
    \PP[H]{\hat{\zeta} < m_z}  \leq 2\int_0^{\frac{m_z}{C(L_A,L_B)}} \frac{1}{\sqrt{2 \pi}}\dd t = \sqrt{\frac{2}{\pi}}\frac{m_z}{C(L_A,L_B)} ,
\end{align*}
Hence, we show that as $m_z \to 0$, $\PP[H]{\hat{\zeta} < m_z} \to 0$.

For $\PP[H]{\hat{\zeta} >M_z}$, we have:
\begin{align*}
    \PP[H]{\hat{\zeta} >M_z} = \PP[H]{\abs{Z} > \frac{M_z}{C(\sigma^2_A,\sigma^2_B) }} &\leq \PP[]{\abs{Z} > \frac{M_z}{C(U_A,U_B) }}= 2\left(1-\Phi\left(\frac{M_z}{C(U_A,U_B) }\right)\right),
\end{align*}
where $\Phi$ is the c.d.f. of standard normal distribution. Since as $x \to \infty$, $\Phi(x) \to 1$, then we have that as $M_z \to \infty$, $\PP[H]{\hat{\zeta} >M_z} \to 0$. 

Now, we move on to study $\PP[H]{\hat{\sigma}_A^2 < m_A} $ and $\PP[G]{\hat{\sigma}_A^2 > M_A}$. Recall that $\hat{\sigma}_A^2 \cond \sigma_A^2 \sim \sigma_A^2\chi_{\nu_A}^2/\nu_A$, we have
\begin{align*}
    \PP[H]{\hat{\sigma}_A^2 < m_A} = \PP[H]{\chi_{\nu_A}^2 < \frac{\nu_Am_A}{\sigma_A^2}} \leq\PP[]{\chi_{\nu_A}^2 < \frac{\nu_Am_A}{L_A}}= \int_{0}^{\frac{\nu_Am_A}{L_A}}f_{\chi^2_{\nu_A}}(t)dt,
\end{align*}
where $f_{\chi^2_{\nu_A}}$ is the density function of $\chi_{\nu_A}^2$ distribution with $\nu_A$ degrees of freedom. The inequality comes from the fact that $\sigma_A^2 \in [L_A,U_A]$ and $f_{\chi^2_{\nu_A}}(t) >0$ for all $t>0$. Note that when $\nu_A \geq 2$, $f_{\chi^2_{\nu_A}}$ is bounded on the compact intervals $[0,\nu_Am_A/L_A]$, we define
$$
M_2 := \sup_{\substack{t \in [0,\nu_Am_A/L_A]}}f_{\chi^2_{\nu_A}}(t).
$$
Then we arrive at:
$$
\PP[H]{\hat{\sigma}_A^2 < m_A}  \leq \int_{0}^{\frac{\nu_Am_A}{L_A}}M_2\dd t  = \frac{\nu_A M_2}{L_A}m_A.
$$
Hence, as $m_A \to 0$, $\PP[H]{\hat{\sigma}_A^2 < m_A} \to 0$.

For $\PP[H]{\hat{\sigma}_A^2 > M_A}$, we have
\begin{align*}
    \PP[H]{\hat{\sigma}_A^2 > M_A} &= \PP[H]{\chi_{\nu_A}^2 > \frac{\nu_AM_A}{\sigma_A^2}} \leq \PP[]{\chi_{\nu_A}^2 > \frac{\nu_AM_A}{U_A}}= \int_{0}^{\frac{\nu_AM_A}{U_A}}f_{\chi^2_{\nu_A}}(t)dt,
\end{align*}
where the inequality comes from the fact that $\sigma_A^2 \in [L_A,U_A]$ and $f_{\chi^2_{\nu_A}}(t) >0$ for all $t>0$. Consider the moment generating function of $X \sim \chi_{\nu_A}^2$, 
$$
\EE[]{e^{tX}} = (1-2t)^{-v_A/2}, \quad \text{for } 0 < t < \frac{1}{2},
$$
and by applying the Chernoff bound,
$$
\PP[]{X>a} = \PP[]{e^{tX}>e^{ta}} \leq e^{-ta}\EE[]{e^{tX}},
$$
hence, 
$$
\PP[]{X>a} \leq \inf_{\substack 0<t<1/2}\left\{e^{-ta}(1-2t)^{-\nu_A/2}\right\}.
$$
Let $s(t) := -ta-(\nu_A/2)log(1-2t)$, and $s'(t) = -a+\nu_A(1-2t)^{-1}$, we have that for $a > \nu_A$, $s(t) \downarrow$ when $t \in (0,\frac{1}{2}(1-\frac{\nu_A}{a}))$, $s(t) \uparrow$ when $t \in (\frac{1}{2}(1-\frac{\nu_A}{a}),\frac{1}{2})$. Therefore,
$$
\PP[]{X>a} \leq \exp{\left\{-\frac{\nu_A}{2}\left({\frac{a}{\nu_A}}-\log{\left(\frac{a}{\nu_A}\right)-1}\right)\right\}}:=g(a).
$$
Take $a = \frac{\nu_AM_A}{U_A}$, and notice that as $M_A \to \infty$, $g(\frac{\nu_AM_A}{U_A}) \to 0$, and hence
$$
\PP[H]{\hat{\sigma}_A^2 > M_A} \leq \PP[]{X>\frac{\nu_AM_A}{U_A}} \leq g\left(\frac{\nu_AM_A}{U_A}\right) \to 0 \text{ as } M_A \to \infty.
$$
By symmetry, we can also show that:
$$
\PP[H]{\hat{\sigma}_B^2 < m_B} \to 0 \text{ as  } m_B \to 0; \quad  \PP[H]{\hat{\sigma}_B^2 > M_B} \to 0 \text{ as  } M_B \to \infty.
$$
Combining all the results above, we prove that as $m_z, m_A,m_B \to 0$, and $M_z,M_A, M_B \to \infty$,
\begin{align*}
    \PP[H]{A^c} &\leq \PP[H]{\hat{\zeta} < m_z}+\PP[H]{\hat{\zeta} > M_z}+\PP[H]{\hat{\sigma}_A^2 < m_A} \\
    &\quad\quad+ \PP[H]{\hat{\sigma}_A^2 > M_A}+\PP[H]{\hat{\sigma}_B^2 < m_B} + \PP[H]{\hat{\sigma}_B^2 > M_B} \to 0.
\end{align*}

\end{proof}
\noindent With above three lemmas, we can now studying $\EE[H]{\abs{\PDV_i - \PDVFunc(\Tbf_i, \hat{\sigma}^2_{iP} ; H)}}$
\begin{align*}
    &\lim_{n \to \infty}\max_{1 \leq i \leq n} \EE[H]{\abs{\PDV_i - \PDVFunc(\Tbf_i, \hat{\sigma}^2_{iP} ; H)}} \\
    &\quad\leq \lim_{n \to \infty}\frac{2}{C^{'}_{S, L, U}}\left(  \EE[H]{\sup_S \abs{N(z, s^2_P, \hat{H}) - N(z, s^2_P, H)}} + \EE[H]{\sup_{S_P} \abs{D(s^2_P, \hat{H}) - D(s^2_P, H)}}\right) + \PP[H]{A^c}\\
    &\quad = \PP[H]{A^c}.
\end{align*}
Since the choices of $m_z,M_z,m_A,M_A,m_B,M_B$ are arbitrary, as $m_z,m_A,m_B \to 0$ and $M_z,M_A,M_B \to \infty$, $\PP[H]{A^c} \to 0$, which leads to:
$$ \lim_{n \to \infty}\max_{1 \leq i \leq n} \EE[H]{\abs{\PDV_i - \PDVFunc(\Tbf_i, \hat{\sigma}^2_{iP} ; H)}} = 0.$$
\end{proof}

\subsection{Proof of Theorem \ref{thm:2D asymptotic uniformity}}

\label{proof:thm_2D asymptotic uniformity}
\begin{proof}
First, we focus on fixed $i \in \mathcal{H}_0$. Fix any $\delta \in (0, 1)$, by Lemma \ref{lem:indicator}, it can be shown that for any $\alpha \leq 1 - \delta$:
$$\mathbf{1}(\PDV_i \leq \alpha) \leq \frac{1}{\delta} \abs{\PDV_i - \PDVFunc(\Tbf_i, \hat{\sigma}^2_{iP} ; H)} + \mathbf{1}(\PDVFunc(\Tbf_i, \hat{\sigma}^2_{iP} ; H) \leq \alpha + \delta).$$
By taking the conditional expectation on both sides of the inequality and based on the conditional uniformity of the oracle p-values $\PDVFunc(\Tbf_i, \hat{\sigma}^2_{iP} ; H)$, we have the following inequality after rearranging:
$$\PP[H]{\PDV_i \leq \alpha \; \cond \; \hat{\sigma}^2_{1P}, \dots, \hat{\sigma}^2_{nP}} - \alpha \leq  \delta + \frac{1}{\delta}\EE[H]{\abs{\PDV_i - \PDVFunc(\Tbf_i, \hat{\sigma}^2_{iP} ; H)} \; \cond \; \hat{\sigma}^2_{1P}, \dots, \hat{\sigma}^2_{nP}}.$$
For any $\alpha \in (1 - \delta, 1)$, we also have 
$\PP[H]{\PDV_i \leq \alpha \; \cond \; \hat{\sigma}^2_{1P}, \dots, \hat{\sigma}^2_{nP}} - \alpha \leq 1 - \alpha \leq \delta.$
Define $(a)_{+} = \max\{a, 0\}$. Above two inequalities lead to:
\begin{align*}
&\sup_{\alpha \in [0, 1]} \left(\PP[H]{\PDV_i \leq \alpha \; \cond \; \hat{\sigma}^2_{1P}, \dots, \hat{\sigma}^2_{nP}} - \alpha\right)_{+} \\
&\quad\quad \leq \delta + \frac{1}{\delta} \EE[H]{\abs{\PDV_i - \PDVFunc(\Tbf_i, \hat{\sigma}^2_{iP} ; H)} \; \cond \; \hat{\sigma}^2_{1P}, \dots, \hat{\sigma}^2_{nP}}.
\end{align*}
Analogously, by defining $(a)_{-} = \max\{-a, 0\}$, we can prove that 
\begin{align*}
&\sup_{\alpha \in [0, 1]} \left(\PP[H]{\PDV_i \leq \alpha \; \cond \; \hat{\sigma}^2_{1P}, \dots, \hat{\sigma}^2_{nP}} - \alpha\right)_{-} \\
&\quad\quad \leq \delta + \frac{1}{\delta} \EE[H]{\abs{\PDV_i - \PDVFunc(\Tbf_i, \hat{\sigma}^2_{iP} ; H)} \; \cond \; \hat{\sigma}^2_{1P}, \dots, \hat{\sigma}^2_{nP}}.
\end{align*}
Notice that $\abs{a} \leq (a)_{+} + (a)_{-}$, then
\begin{align*}
&\sup_{\alpha \in [0, 1]} \abs{\PP[H]{\PDV_i \leq \alpha \; \cond \; \hat{\sigma}^2_{1P}, \dots, \hat{\sigma}^2_{nP}} - \alpha} \\
&\quad\quad \leq 2\left(\delta + \frac{1}{\delta} \EE[H]{\abs{\PDV_i - \PDVFunc(\Tbf_i, \hat{\sigma}^2_{iP} ; H)} \; \cond \; \hat{\sigma}^2_{1P}, \dots, \hat{\sigma}^2_{nP}}\right).
\end{align*}
With iterated expectation, and the fact that $i$ is picked arbitrarily, the inequality above should hold for all $i \in \mathcal{H}_0$, which leads to:
\begin{align*}
    &\max_{i \in \mathcal{H}_0} \left\{\EE[H]{\sup_{\alpha \in [0, 1]} \abs{\PP[H]{\PDV_i \leq \alpha \; \cond \; \hat{\sigma}^2_{1P}, \dots, \hat{\sigma}^2_{nP}} - \alpha}}\right\} \\
    &\quad\quad \leq 2\left(\delta + \frac{1}{\delta} \max_{i \in \mathcal{H}_0} \left\{\EE[H]{\abs{\PDV_i - \PDVFunc(\Tbf_i, \hat{\sigma}^2_{iP} ; H)}}\right\}\right).
\end{align*}
Based on Theorem \ref{thm:1D_uniform_convergence_p_value}, by taking $n \to \infty$, we have
\[\lim_{n \to \infty} \max_{i \in \mathcal{H}_0} \left\{\EE[H]{\sup_{\alpha \in [0, 1]} \abs{\PP[H]{\PDV_i \leq \alpha \; \cond \; \hat{\sigma}^2_{1P}, \dots, \hat{\sigma}^2_{nP}} - \alpha}} \right\} \leq 2\delta.\]
Now, we can take $\delta \to 0$ to conclude that:
\[\lim_{n \to \infty} \max_{i \in \mathcal{H}_0} \left\{\EE[H]{\sup_{\alpha \in [0, 1]} \abs{\PP[H]{\PDV_i \leq \alpha \; \cond \; \hat{\sigma}^2_{1P}, \dots, \hat{\sigma}^2_{nP}} - \alpha}}\right\} = 0.\]
Asymptotic uniformity follows from the above result:
\begin{align*}
    &\max_{i \in \mathcal{H}_0} \left\{\sup_{\alpha \in [0,1]} \abs{\PP[H]{\PDV_i \leq \alpha} - \alpha}\right\} = \max_{i \in \mathcal{H}_0} \left\{\sup_{\alpha \in [0,1]} \abs{\EE[H]{\PP[H]{\PDV_i \leq \alpha \; \cond \; \hat{\sigma}^2_{1P}, \dots, \hat{\sigma}^2_{nP}} - \alpha}}\right\} \\
    &\quad\quad \leq \max_{i \in \mathcal{H}_0} \left\{\sup_{\alpha \in [0,1]} \EE[H]{\abs{\PP[H]{\PDV_i \leq \alpha \; \cond \; \hat{\sigma}^2_{1P}, \dots, \hat{\sigma}^2_{nP}} - \alpha}}\right\} \\
    & \quad\quad \leq \max_{i \in \mathcal{H}_0} \left\{\EE[H]{\sup_{\alpha \in [0,1]}\abs{\PP[H]{\PDV_i \leq \alpha \; \cond \; \hat{\sigma}^2_{1P}, \dots, \hat{\sigma}^2_{nP}} - \alpha}} \right\} \to 0 \text{ as } n \to \infty.
\end{align*}

\end{proof}

\end{document}